\documentclass[reprint, onecolumn, notitlepage, superscriptaddress,footinbib]{revtex4-2}
\usepackage{amsmath, amssymb, amsthm}
\usepackage[utf8]{inputenc}
\usepackage{bm}
\usepackage{xparse}
\usepackage{physics}
\usepackage{mathtools}
\usepackage{graphicx}
\usepackage{upgreek}
\usepackage{stmaryrd}
\usepackage[caption=false]{subfig}
\usepackage{qcircuit}
\usepackage{mathpazo}
\usepackage{libertine}
\usepackage[scale=.9]{inconsolata}
\usepackage{dsfont}
\usepackage{mathrsfs}
\usepackage{xcolor}
\definecolor{maroon}{HTML}{800020}
\definecolor{darkcyan}{HTML}{008060}
\usepackage[colorlinks=true,
allcolors=blue]{hyperref} % hyperlinking
\usepackage{cleveref} % avoid \ref, use \cref instead

\crefformat{figure}{FIG.~#2#1#3}
\crefformat{equation}{Eq.~(#2#1#3)}
\crefrangeformat{equation}{Eqs.~(#3#1#4) to~(#5#2#6)}
\crefformat{algorithm}{Algorithm~#2#1#3}
\crefformat{section}{\S#2#1#3}
\crefformat{definition}{Definition~#2#1#3}
\crefformat{theorem}{Theorem~#2#1#3}
\crefformat{proposition}{Proposition~#2#1#3}
\crefformat{corollary}{Corollary~#2#1#3}
\crefformat{appendix}{Appendix~#2#1#3}
\crefformat{table}{Table~#2#1#3}

\usepackage{algorithm} 
\usepackage[noend]{algpseudocode}

\algnewcommand{\Yield}{\textbf{yield}~}
\makeatletter

\renewcommand*{\p@subsection}{}

\renewcommand*{\p@subsubsection}{}
\makeatother

\theoremstyle{plain}
\newtheorem{theorem}{Theorem}[section]

\newtheorem{corollary}[theorem]{Corollary}
\newtheorem{proposition}[theorem]{Proposition}
\theoremstyle{definition}
\newtheorem{definition}[theorem]{Definition}

\newcommand{\integers}{\ensuremath{\mathds{Z}}}
\newcommand{\reals}{\ensuremath{\mathds{R}}}
\newcommand{\cmplex}{\ensuremath{\mathds{C}}}
\newcommand{\binary}{\ensuremath{\mathds{B}}}
\newcommand{\sqintR}{\ensuremath{\mathscr{L}^2 (\reals)}}
\newcommand{\sqintS}{\ensuremath{\mathscr{L}^2\left( \mathds{S} \right)}}
\newcommand{\id}{\ensuremath{\mathds{1}}} % Identity Operator
\newcommand{\dbIndex}{\ensuremath{\mathcal{K}}} % Daubechies wavelet index
\newcommand{\T}{\ensuremath{{\mathsf{T}}}}
\newcommand{\upi}{\ensuremath{\mathrm{i}}}
\newcommand{\e}{\ensuremath{\mathrm{e}}}
\newcommand{\oneDG}{\textup{1DG}}   % upright 1DG
\newcommand{\G}{\textup{G}}         % upright G
\newcommand{\s}{\text{s}}           % upright s
\newcommand{\w}{\text{w}}           % upright w
\newcommand{\sS}{\text{ss}}         % upright ss
\newcommand{\sw}{\text{sw}}         % upright sw
\newcommand{\ww}{\text{ww}}         % upright ww

\DeclarePairedDelimiter{\ceil}{\lceil}{\rceil} % typesets ceiling 
\DeclarePairedDelimiter{\floor}{\lfloor}{\rfloor} % typesets floor
\newcommand{\cBraket}[1]{\left\{#1\right\}} % typesets curly bracket
\newcommand{\pars}[1]{\left(#1\right)}% typesets parentheses
\newcommand{\upbm}[1]{\bm{\mathrm{#1}}} % upright bold math

\DeclareMathOperator{\cas}{cas}

\newcommand{\out}{\texttt{out}}
\newcommand{\ang}{\texttt{ang}}
\newcommand{\mean}{\texttt{mean}}
\newcommand{\scratch}{\texttt{scratch}}
\newcommand{\stDev}{\texttt{std}}
\newcommand{\spacing}{\texttt{spc}}
\newcommand{\vac}{\texttt{vac}}
\newcommand{\temp}{\texttt{tmp}}
\newcommand{\shear}{\texttt{shear}}
\newcommand{\reff}{\texttt{ref}}
\newcommand{\anc}{\texttt{anc}}
\newcommand{\ineq}{\texttt{ineq}}
\newcommand{\MUL}{\operatorname{\textsc{mul}}}

\newcommand{\softO}[1]{\ensuremath{\widetilde{\mathcal{O}}\!\left(#1\right)}}

\DeclareFontFamily{U} {MnSymbolB}{}
\DeclareFontShape{U}{MnSymbolB}{m}{n}{
  <-6> MnSymbolB5
  <6-7> MnSymbolB6
  <7-8> MnSymbolB7
  <8-9> MnSymbolB8
  <9-10> MnSymbolB9
  <10-12> MnSymbolB10
  <12-> MnSymbolB12}{}
\DeclareSymbolFont{MnSyB} {U} {MnSymbolB}{m}{n}
\DeclareMathSymbol{\nleftrightline}{\mathrel}{MnSyB}{208}

\newcommand{\qwire}{\text{\Large\ensuremath \nleftrightline}}

\newcounter{library} % set a new counter for libraries
 \newcounter{algCountSaved} % the real counter for algorithms
\makeatletter
\newenvironment{library}[1][htb]{%
\setcounter{algCountSaved}{\value{algorithm}} % saves alg counter
  \setcounter{algorithm}{\value{library}} % locally changes alg counter
  \renewcommand{\ALG@name}{Library}% Updates algorithm name
   \counterwithin*{algorithm}{library}%
\addtocounter{library}{1}
  \begin{algorithm}[#1]%
  }{\end{algorithm}
  \setcounter{algorithm}{\value{algCountSaved}} % brings back the alg counter
  }
\makeatother
\makeatletter
\def\@ssect@ltx#1#2#3#4#5#6[#7]#8{%
  \def\H@svsec{\phantomsection}%
  \@tempskipa #5\relax
  \@ifdim{\@tempskipa>\z@}{%
    \begingroup
      \interlinepenalty \@M
      #6{%
       \@ifundefined{@hangfroms@#1}{\@hang@froms}{\csname @hangfroms@#1\endcsname}%
       {\hskip#3\relax\H@svsec}{#8}%
      }%
      \@@par
    \endgroup
    \@ifundefined{#1smark}{\@gobble}{\csname #1smark\endcsname}{#7}%
  }{%
    \def\@svsechd{%
      #6{%
       \@ifundefined{@runin@tos@#1}{\@runin@tos}{\csname @runin@tos@#1\endcsname}%
       {\hskip#3\relax\H@svsec}{#8}%
      }%
      \@ifundefined{#1smark}{\@gobble}{\csname #1smark\endcsname}{#7}%
      \addcontentsline{toc}{#1}{\protect\numberline{}#8}%
    }%
  }%
  \@xsect{#5}%
}%
\makeatother

\newcommand{\IQST}{\affiliation{%
    Institute for Quantum Science and Technology,
    University of Calgary,
    Alberta, Canada}}
    
\newcommand{\chemUofT}{\affiliation{%
    Chemical Physics Theory Group, Department of Chemistry, University of Toronto, Ontario, Canada}
    }
    
\newcommand{\csUofT}{\affiliation{%
    Department of Computer Science, University of Toronto, Ontario, Canada}}
    
\newcommand{\PhysMQ}{\affiliation{%
    Department of Physics and Astronomy,
    Macquarie University,
    New South Wales, Australia}}

\newcommand{\UTS}{\affiliation{%
    Centre for Quantum Software and Information,
    University of Technology Sydney,
    New South Wales, Australia}}

\newcommand{\EQUS}{\affiliation{%
Centre of Excellence in Engineered Quantum Systems, Macquarie University, New South Wales, Australia}}

\begin{document}

\title{Nearly optimal quantum algorithm for generating the ground state\\
	of a free quantum field theory}
	
\author{Mohsen Bagherimehrab}   \email{mohsen.bagherimehrab@utoronto.ca}\IQST \chemUofT \csUofT
\author{Yuval R.\ Sanders}       \PhysMQ \UTS
\author{Dominic W.\ Berry}       \PhysMQ
\author{Gavin K.\ Brennen}       \PhysMQ \EQUS
\author{Barry C.\ Sanders}       \IQST
\date{\today}

%%%%%%%%%%%%%%%%%%%%%%%%%%%%%%%%%%%%%%%%%%%%%%%%%%%%%%%%%
%%%%%% Abstract %%%%%%%%%%%%%%%%%%%%%%%%%%%%%%%%%%%%%%%%%
%%%%%%%%%%%%%%%%%%%%%%%%%%%%%%%%%%%%%%%%%%%%%%%%%%%%%%%%%

\begin{abstract}
We devise a quasilinear quantum algorithm for generating an approximation for the ground state of a quantum field theory~(QFT).
Our quantum algorithm delivers a super-quadratic speedup over the state-of-the-art quantum algorithm for ground-state generation, overcomes the ground-state-generation bottleneck of the prior approach and is optimal up to a polylogarithmic factor.
Specifically, we establish two quantum algorithms---Fourier-based and wavelet-based---to generate the ground state of a free massive scalar bosonic QFT with gate complexity quasilinear in the number of discretized-QFT modes.
The Fourier-based algorithm is limited to translationally invariant QFTs. Numerical simulations show that the wavelet-based algorithm successfully yields the ground state for a QFT with broken translational invariance.
Furthermore, the cost of preparing particle excitations in the wavelet approach is independent of the energy scale.
Our algorithms require a routine for generating one-dimensional Gaussian~(1DG) states.
We replace the standard method for 1DG-state generation,
which requires the quantum computer to perform lots of costly arithmetic,
with a novel method based on inequality testing that significantly reduces the need for arithmetic.
Our method for 1DG-state generation is generic and could be extended to preparing states whose amplitudes can be computed on the fly by a quantum computer.
\end{abstract}

 % PACS, the Physics and Astronomy Classification Scheme.
%\keywords{Suggested keywords}
\maketitle

%%%%%%%%%%%%%%%%%%%%%%%%%%%%%%%%%%%%%%%%%%%%%%%%%%%%%%%%%
% This block provides dotted subsections and subsubsections
% Take a look at here: https://tex.stackexchange.com/questions/37385
\makeatletter
\renewcommand*\l@subsection{\@dottedtocline{2}{1.5em}{2em}}
\renewcommand*\l@subsubsection{\@dottedtocline{3}{3.5em}{3em}}
\makeatother
%%%%%%%%%%%%%%%%%%%%%%%%%%%%%%%%%%%%%%%%%%%%%%%%%%%%%%%%%

{
  \hypersetup{linkcolor=black}
  \tableofcontents
}

%%%%%%%%%%%%%%%%%%%%%%%%%%%%%%%%%%%%%%%%%%%%%%%%%%%%%%%%%
% Take a look at here: https://tex.stackexchange.com/questions/37385
\makeatletter
\let\toc@pre\relax
\let\toc@post\relax
% activate next line if you want to eliminate the dots in the list of figures
%\renewcommand\@dotsep{10000}
\makeatother 
%%%%%%%%%%%%%%%%%%%%%%%%%%%%%%%%%%%%%%%%%%%%%%%%%%%%%%%%%

%%%%%%%%%%%%%%%%%%%%%%%%%%%%%%%%%%%%%%%%%%%%%%%%%%%%%%%%%
%%%%%% Introduction %%%%%%%%%%%%%%%%%%%%%%%%%%%%%%%%%%%%%
%%%%%%%%%%%%%%%%%%%%%%%%%%%%%%%%%%%%%%%%%%%%%%%%%%%%%%%%%
\section{Introduction}
\label{sec:introduction}

Quantum algorithms for simulating a quantum field theory~(QFT) comprise three main steps:
generating an initial state,
simulating time evolution and
measuring observables,
with the initial-state generation being the most expensive step~\cite{JLP12,JLP14}.
The conventional approach to generating the initial state for simulating a QFT, particularly for simulating particle scattering, is first to generate the ground state of the free (i.e. non-interacting) field theory.
Then prepare free wavepackets,
which refers to spatially localized non-interacting particles
and, finally, turn on the field interaction adiabatically~\cite{JLP12,JLP14,JLP14-2}.
An alternative approach is to avoid adiabatic evolution in initial-state preparation, which is used in~\cite{MJ18} to remove the state-preparation as the bottleneck of simulating fermionic QFTs.

everal works have also developed variational~\cite{RLC+20,DMH18} and stochastic~\cite{HLL20,GL21} methods for simulating QFTs that avoid the need to prepare the full quantum state.
These methods, however, are not purely quantum but rather are hybrid quantum-classical,
which are more suitable for simulation on near-term quantum computers.
Despite the advances in simulating QFTs and development of recent algorithmic techniques for preparing quantum states~\cite{CLB+21,LBG+20,Gus20},
state preparation remains the computationally expensive step for simulating a certain class of bosonic QFTs called the massive scalar bosonic QFTs.
In particular, preparing the free-field ground state is the most expensive part in one or two spatial dimensions and is the second-most-expensive part for simulating these theories in three spatial dimensions~\cite{JLP12,JLP14}.

In this paper, we establish a quasilinear quantum algorithm for generating an approximation for the ground state of a massive scalar-bosonic free QFT. 
Our algorithm is optimal, up to a polylogarithmic factor, provides a super-quadratic speedup over the best prior quantum algorithm for ground-state generation and overcomes the ground-state-generation bottleneck.

Specifically, we develop two quantum algorithms, one Fourier-based and the other wavelet-based, to generate the free-field ground state with gate complexity quasilinear in the number of discretized-QFT modes.
For the case of broken translational invariance, e.g., due to mass defects~\cite{BJV11}, the Fourier-based algorithm is inappropriate.
We show, by numerical simulation, that the wavelet-based algorithm successfully yields an approximation for the free-field ground state with a quasilinear gate complexity.
Our algorithms require a routine for preparing one-dimensional Gaussian~(1DG) states, 
which are required for preparing multi-dimensional Gaussian states. 
The standard method for generating a 1DG state is based on Zalka-Grover-Rudolph state preparation~\cite{Zal96,GR02},
which requires the quantum computer to perform costly arithmetic operations.
We replace this method with a novel method based on inequality testing~\cite{SLSB19} that significantly reduces the need for arithmetic.

%%%%%%%%%%%%%%%%%%%%%%%%%%%%%%%%%%%%%%%%%%%%%%%%%%%%%%%%%
\subsection{Nontechnical background}

To simulate a QFT, we first need to discretize it~\cite{JLP14}.
The discretization is needed for regulating infinite-dimensional Hilbert spaces involved in the continuum field theory.
Once discretized, the QFT becomes a many-body quantum system whose time evolution can be efficiently simulated on a quantum computer~\cite{Lloyd96,JLP14,BY06}.
The two other steps of a full quantum simulation, namely initial-state generation and measurement, strongly depend on which QFT is being simulated and must be analyzed on a case-by-case basis~\cite[p.~1017]{JLP14}.
Two approaches are used to discretize a continuum massive scalar-bosonic QFT: lattice-based and wavelet-based approaches.
The conventional lattice-based approach is used in the seminal quantum algorithm~\cite{JLP12, JLP14}, and the wavelet-based approach is used in the subsequent quantum algorithm~\cite{BRSS15}.
By discretizing the QFT, the free-field ground state becomes a $N$-dimensional Gaussian~($N$DG) state, where~$N$ is the number of modes in the discretized QFT.

A method proposed in\cite{JLP12} to generate a multi-dimensional Gaussian state is using the Kitaev-Webb method~\cite{KW09}.
This method has three main steps to generate an $N$DG state on a quantum computer.
The first step is to compute the LDL matrix decomposition of the Gaussian state's inverse-covariance matrix (ICM) by a classical computation;
note that any Gaussian state is fully described by a covariance matrix or its inverse.
The second step is to prepare~$N$ different 1DG states where the standard deviation of each 1DG state is obtained from a diagonal element of the diagonal matrix in the LDL decomposition.
The last step is to perform a basis transformation based on the LDL matrix decomposition that maps the previously generated state, a Gaussian state with a diagonal ICM, to a Gaussian state with the desired ICM.

The Kitaev-Webb method~\cite{KW09} for preparing a 1DG state is an application of the standard state-preparation method by Zalka~\cite{Zal96}, Grover and Rudolph~\cite{GR02}.
In this method for preparing a 1DG state on a quantum register, the continuous state is approximated by a discrete 1DG state over a one-dimensional lattice with unit spacing.
Then a recursive description of the approximated 1DG state is used to generate the approximate state that requires the quantum computer to perform a controlled rotation for each qubit of the quantum register.
For each controlled rotation, the quantum computer needs to coherently compute the rotation angle by a large amount of coherent arithmetic on the quantum computer.

The space and time complexities of the Kitaev-Webb method, both for 1DG- and $N$DG-state generation, were not analyzed in the original paper~\cite{KW09}.
However, authors of~\cite{JLP12} state that the method's time complexity for $N$DG-state generation is dominated by the classical complexity of the LDL matrix decomposition, which is~$\softO{N^{2.373}}$~\cite{AW21} if we use Coppersmith-Winograd-style matrix multiplication;
here~$\widetilde{\mathcal{O}}$ suppresses logarithmic factors.
There is also a quantum complexity of~$\softO{N^2}$ to perform the basis transformation needed to yield the Gaussian state with the desired ICM.
This quantum complexity is effectively the cost of performing an in-place matrix-vector multiplication.

%%%%%%%%%%%%%%%%%%%%%%%%%%%%%%%%%%%%%%%%%%%%%%%%%%%%%%%%%
\subsection{Overview of methods and results}

Our overall approach to reducing the time-complexity for ground-state generation is to exploit known sparsity properties of the matrices involved.
This approach allows us to reduce both classical and quantum time-complexities by replacing dense matrix operations with sparse ones.
In our wavelet-based approach, for example, we reduce the classical cost of the LDL matrix decomposition by approximating the ground-state ICM with a matrix containing $\order{N \log N}$ nonzero entries.
Similarly, we reduce the quantum cost by replacing the general approach for in-place matrix-vector multiplication with a sparse approach.
In our Fourier-based approach, we avoid even these sparse-approach optimizations by exploiting the translational invariance presumed by Jordan-Lee-Preskill~\cite{JLP12,JLP14} to eliminate the need for classical LDL decomposition.
We replace the quantum in-place matrix-vector multiplication with a coherent fast Fourier transform.

Specifically, in our Fourier-based algorithm, we exploit the translational invariance of the free QFT to reduce the cost of ground-state generation from~$\softO{N^{2.373}}$ down to~$\softO{N}$.
We discretize the continuum QFT in a fixed-scale basis and utilize the circulant structure of the ground state's ICM to compute its eigenvalues by a discrete Fourier transform~(DFT).
We then generate a~$N$DG with a diagonal ICM whose diagonals are the eigenvalues by preparing~$N$ different 1DG states.
Finally, we transform this state into the ground state by a basis transformation.

A choice for basis transformation is to execute a DFT on a quantum computer by reversible arithmetic operations along the lines of~\cite{ASY20}.
The problem with this basis transformation is that the resulting representation for the Gaussian state requires us to have complex-valued coordinates.
We avoid the complex numbers required by the DFT and instead use a discrete Hartley transform~(DHT) that only involves real numbers~\cite[Theorem~1]{BF93}.
The real, symmetric and circulant properties of the ground-state ICM guarantees that by performing a DHT, we obtain the desired Gaussian state.
Analogous to the quantum fast Fourier transform in~\cite{ASY20}, we construct a quantum fast Hartley transform~(QFHT) algorithm with a quasilinear gate complexity.

In the wavelet-based algorithm, we discretize the continuum free QFT in a multi-scale wavelet basis.
The ICM of the ground state in this basis has many elements that have an exponentially close-to-zero value.
We truncate this matrix by replacing the near-zero elements with exactly zero.
This truncation introduces a systematic error that distorts the ground state of the discretized theory.
We show that not only the truncated ICM remains a positive-definite matrix, which is required for the ground state to be a Gaussian state, but also the infidelity between the Gaussian state with the truncated ICM and the free-field ground state is within the pre-specified error-tolerance for preparing the ground state.

By the truncation, the ICM becomes a sparse matrix with a particular structure known as the `fingerlike' sparsity structure~\cite{BCR91};
see~\cref{fig:truncICM}.
A matrix with this structure has a number of nonzero elements that is quasilinear in the dimension of the matrix.
In order to exploit the fingerlike sparse structure, we replace the LDL decomposition in the Kitaev-Webb method~\cite{KW09} with the UDU decomposition.
By exploiting the sparsity structure, we perform the UDU matrix decomposition of the truncated ICM in a quasilinear time.

\begin{figure}[htb]
\centering
    \includegraphics[width=.98\textwidth]{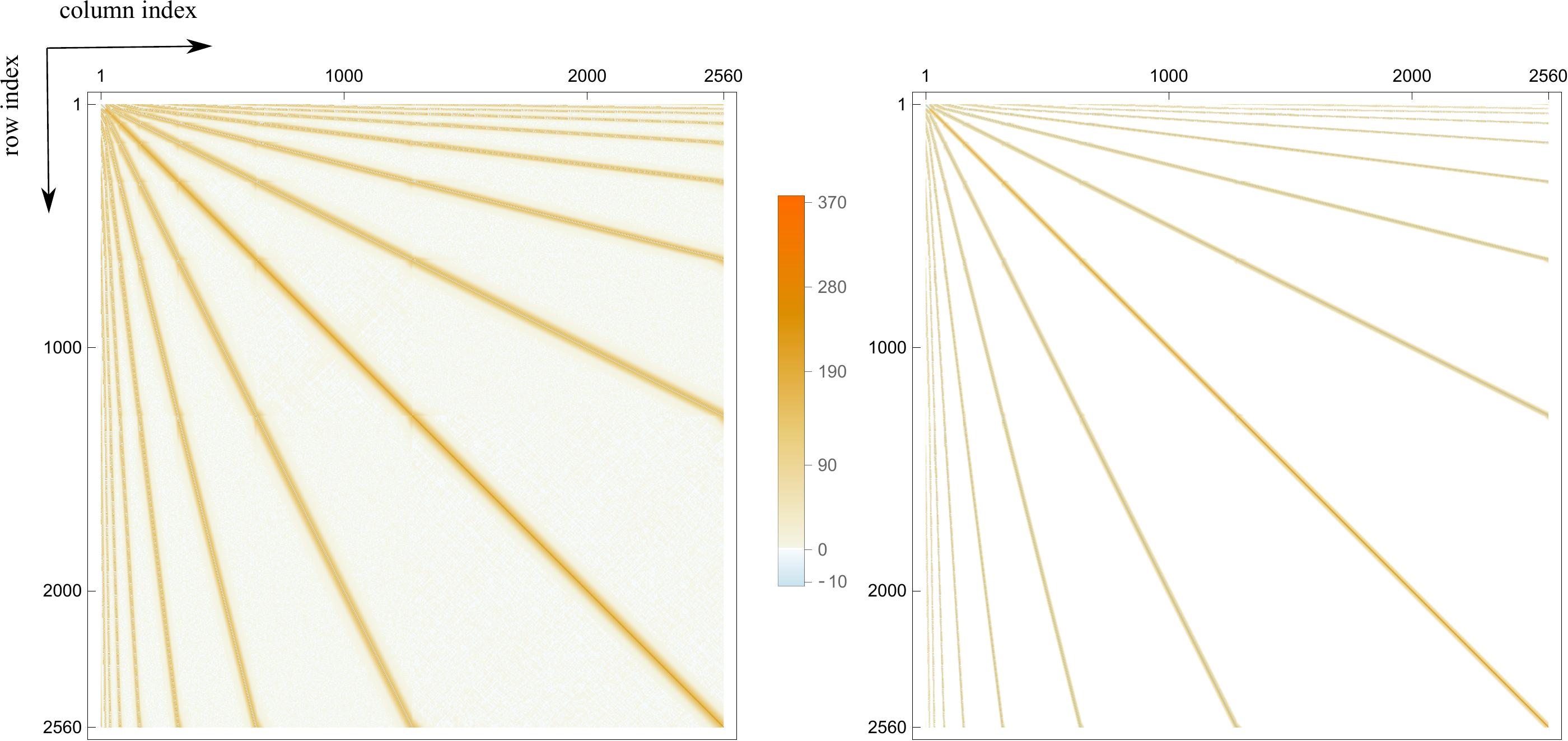}
\caption[Truncated ICM]{
Visual representation of the exact (Left) and approximate ICM (Right) for the ground state of the discretized free~QFT in a multi-scale wavelet basis.
The column (row) index of the visualized matrix is ordered from left (top) to right (bottom), just as the typical format of a matrix.
The values of matrix elements are represented by different colors with the shown color scale.
The number of modes for the discretized QFT is $N=2560$, free mass is~$m_0=1$, and the Daubechies wavelet with index $\dbIndex =3$ is used.
The exact ICM is a quasi-sparse matrix, meaning most of its elements are nearly zero.
Elements of the exact matrix with a magnitude less than~$10^{-8}$ are replaced with zero to obtain the approximate ICM.
The approximate ICM has a specific sparsity structure known as the `fingerlike' sparsity structure.
\label{fig:truncICM}
}
\end{figure}

The ground state in the wavelet-based approach is generated as follows. First we prepare a $N$DG state whose ICM is the diagonal matrix in the UDU decomposition.
We then transform this state into the ground state by performing a quantum shear transform~(QST) on a quantum computer.
The shear-transform matrix has the same sparse structure as the truncated ICM in our UDU decomposition.
We exploit this sparse structure and perform the needed QST with a quasilinear gate complexity.

Our algorithms require a routine for generating an approximation to a 1DG state, and we devise two methods to perform this task.
Our first method is based on the Kitaev-Webb method~\cite{KW09}.
In our method, we approximate the 1DG state over a finite lattice with a non-unit lattice spacing.
We choose the lattice spacing and the number of lattice points in terms of the 1DG-state standard deviation and an error tolerance for the generated state.
The Kitaev-Webb method is restricted to 1DG states with extremely large standard deviations.
Our method does not have this restriction and can generate a 1DG state with any standard deviation.
We show that space and time complexities for preparing an approximation to each 1DG state required for generating the ground state are both logarithmic in the number of modes in the discretized QFT.

A key contributor to gate complexity for the Kitaev-Webb method is the need to perform a large amount of coherent arithmetic.
Their method requires performing a sequence of controlled rotations for every qubit of the quantum register used for preparing a Gaussian state.
For each controlled rotation, the quantum computer needs to coherently compute the third Jacobi theta function twice, perform a division, and compute one square-root and one arccosine function.
In contrast, we provide a significantly improved state-preparation technique based on inequality testing~\cite{SLSB19}, where the most complicated computation needed is a single exponential.
We go beyond the inequality-testing method in~\cite{SLSB19} by showing how to prepare a state with mostly small amplitudes except for a peak in a known location.
We exploit the known location in our method and perform a single step of amplitude amplification~\cite{BHM+02} rather than multiple steps of amplitude amplification required for preparing a general state by inequality testing~\cite{SLSB19}. 

We perform a numerical study to justify why the wavelet-based approach could be preferred over the Fourier-based approach for QFTs with broken translational invariance.
We consider a simple case where the translational invariance is broken due to a mass defect~\cite{BJV11}.
The Fourier-based approach is not applicable in this case because the ground-state ICM cannot be diagonalized by a discrete Fourier transform.
However, our numerical experiment demonstrates that the wavelet-based approach accommodates such QFT with broken translational invariance.
Specifically, the fingerlike sparse structure of the ground-state ICM is not affected by the mass defect,
thereby the wavelet-based algorithm successfully yields an approximate ground state with a quasilinear gate complexity.

We go beyond ground-state generation and construct procedures for preparing free-field particle states in the Fourier- and wavelet-based approaches.
We show that, unlike the Fourier approach, the wavelet approach enables preparing particle states at different energy scales without an additional cost required for the Fourier approach.
Specifically, we show that preparing a free-field single-particle state at a given scale is more expensive than preparing the same state in the wavelet approach.

%%%%%%%%%%%%%%%%%%%%%%%%%%%%%%%%%%%%%%%%%%%%%%%%%%%%%%%%%
\subsection{Organization}

Our paper is organized as follows.
We begin, in~\cref{sec:background}, by elaborating the key background pertinent to subsequent sections.
In this section, we review wavelet bases, various methods for discretizing a continuum QFT, a standard approach for generating a Gaussian state and a computation model used for analyzing an algorithm's time complexity.
Next we describe our approach for generating the free-field ground state in~\cref{sec:approach} where we discuss our model for describing a scalar QFT, a metric that we use to analyze our algorithms' time complexity and our methods for generating the free-field ground state.

We then present our results in~\cref{sec:results}.
In this section, we construct our ground-state-generation algorithms, analyze their classical and quantum time complexities, and show that our algorithms are optimal up to polylogarithmic factors.
Furthermore, we determine the space required to represent the free-field ground state and compare the Fourier- and wavelet-based algorithms for a QFT with broken translational invariance and for generating states beyond the free-field ground state.
We finally discuss our results in~\cref{sec:discussion} and conclude in~\cref{sec:conclusions}.

%%%%%%%%%%%%%%%%%%%%%%%%%%%%%%%%%%%%%%%%%%%%%%%%%%%%%%%%%
%%%%%% Background %%%%%%%%%%%%%%%%%%%%%%%%%%%%%%%%%%%%%%%
%%%%%%%%%%%%%%%%%%%%%%%%%%%%%%%%%%%%%%%%%%%%%%%%%%%%%%%%%

\section{Background}
\label{sec:background}

This section covers the key background pertinent to subsequent sections.
We begin by describing wavelet bases in~\cref{subsec:wavelet_bases}. 
Next, in~\cref{subsec:discretization}, we discuss different approaches for discretizing a continuum QFT.
Then we review the Kitaev-Webb method for generating a Gaussian state in~\cref{subsec:Kitaev_Webb_method}.
Finally, in~\cref{subsec:QRAM}, we review the quantum random-access machine~(QRAM) model for computation (not to be confused with quantum random-access memory~\cite{GLM08}).

%%%%%%%%%%%%%%%%%%%%%%%%%%%%%%%%%%%%%%%%%%%%%%%%%%%%%%%%%
\subsection{Wavelet bases}
\label{subsec:wavelet_bases}

Here we briefly review wavelet bases and define terms frequently used in this paper.
For a detailed description, we refer to~\cite[Sec.~2.1]{mbm} and~\cite{Mal09,Dau92}.
We begin by explaining salient features of wavelet bases.
Then we describe a fixed-scale wavelet basis followed by a multi-scale wavelet basis.

% Description
The wavelet bases constitute an orthonormal basis for the Hilbert space~$\sqintR$ of square-integrable functions on~\reals.
A wavelet basis is defined in terms of two functions:
\emph{scaling}~$s(x)$ and \emph{wavelet}~$w(x)$ functions.
Here we focus on Daubechies~\dbIndex\ wavelets and refer to them as~db\dbIndex\ wavelets.
The index~$\dbIndex\in\integers^+$ specifies the number of vanishing moments of~$w(x)$.
The scaling and wavelet functions become smoother and less localized
by increasing~$\dbIndex$;
see~\cref{fig:dbwavelets}.

\begin{figure}[htb]
\centering
    \includegraphics[width=.98\textwidth]{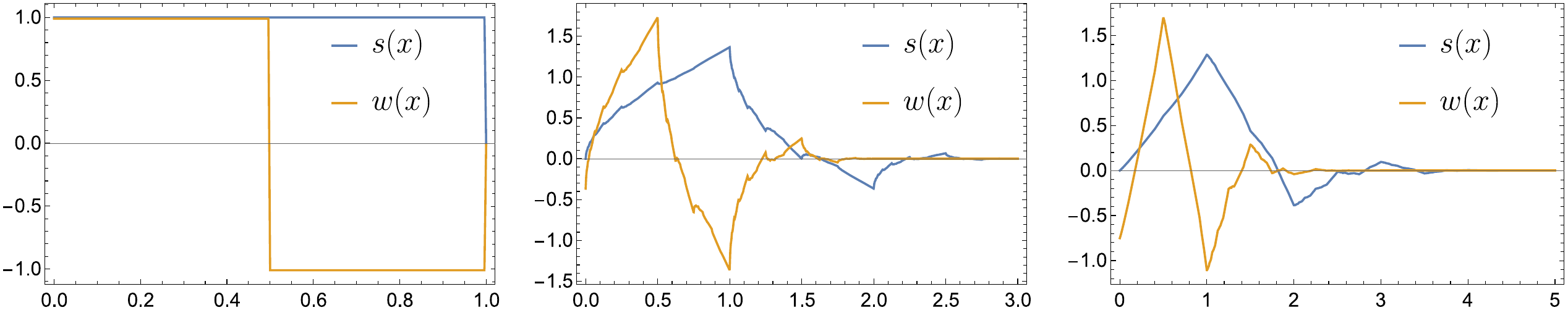}
\caption[Daubechies scaling and wavelet functions]{
From left to right: scaling~$s(x)$ and wavelet~$w(x)$ functions for the db\dbIndex\ wavelet with index~$\dbIndex=1, 2,3$;
db1 wavelet is identical to Haar wavelet.
The support of these functions is~$[0,2\dbIndex-1]$.
The~$s(x)$ and~$w(x)$ for db3 wavelet have continuous first derivatives.}
\label{fig:dbwavelets}
\end{figure}

\noindent
The db\dbIndex\ scaling function is the solution of the linear equation
\begin{equation}
\label{eq:scaling_function}
s(x)=\sqrt{2}\sum_{\ell=0}^{2\dbIndex-1} h_\ell\, s(2x-\ell),
\end{equation}
where, for a unique solution,~$s(x)$ is normalized to have a unit area.
The~$2\dbIndex$ real coefficients~$\{h_0, \ldots, h_{2\dbIndex-1}\}$ are called the \emph{low-pass filter} coefficients.
These coefficients uniquely determine a wavelet basis, and several methods are available to compute their numerical values;
see~\cref{appx:lowpass_filter}.
Given~$s(x)$, the wavelet function is
\begin{equation}
\label{eq:wavelet_function}
w(x):= \sqrt{2}\sum_{\ell=0}^{2\dbIndex-1} g_\ell\, s(2x-\ell),
\quad
g_\ell:=(-)^{\ell} h_{2\dbIndex-1-\ell},
\end{equation}
and~$\{g_0,\ldots,g_{2\dbIndex-1}\}$ are called the \emph{high-pass filter} coefficients.
The scaling and wavelet functions at scale~$k\in\integers$ are defined~as
\begin{equation}
\label{eq:scale_k_scaling_wavelet_functions}
s^{(k)}_\ell(x):= \sqrt{2^k}\, s\left( 2^k x-\ell \right),
\quad
w^{(k)}_\ell(x):= \sqrt{2^k}\, w\left( 2^k x- \ell \right).
\end{equation}
These functions are orthonormal and have support on~$\left[\ell/2^k, (\ell+2\dbIndex-1)/2^k\right]$.
For convenience, we henceforth use the term `scale' to refer to the parameter $k$, so higher (lower) scale means larger (smaller) $k$.

% Fixed-scale basis
The scaling function at a fixed scale~$k\in \integers$ and its integer translation span a subspace~$\mathcal{S}_{k}$ of~$\sqintR$.
We refer to this subspace as the \emph{scale subspace} of~$\sqintR$. Note that $\sqintR \cong \lim_{k \to \infty} \mathcal{S}_k$.
A fixed-scale wavelet basis is an orthonormal basis whose basis vectors are comprised of a scaling function and its integer translations at a fixed scale.

% Multi-scale wavelet basis
The wavelet function at a fixed scale~$k$ and its integer translations span a subspace~$\mathcal{W}_k$ of~$\sqintR$, which we call the \emph{wavelet subspace}.
The wavelet subspace~$\mathcal{W}_k$ is the orthogonal complement of~$\mathcal{S}_k$ in~$\mathcal{S}_{k+1}$, i.e.,
$\mathcal{S}_{k+1}\cong \mathcal{S}_k \oplus \mathcal{W}_k$.
This property of the scale and wavelet subspaces leads to a multi-scale decomposition of~\sqintR\ through the relation~$\sqintR \cong
\mathcal{S}_{s_0} \bigoplus_{r=s_0}^{\infty} \mathcal{W}_r$~\cite{BRSS15},
where~$s_0$ is coarsest scale.
The scaling function and its integer translations at a fixed scale along with the wavelet function and its integer translations at all finer scales construct an orthonormal basis for~\sqintR, which we call the multi-scale wavelet~basis.

%%%%%%%%%%%%%%%%%%%%%%%%%%%%%%%%%%%%%%%%%%%%%%%%%%%%%%%%%
\subsection{Discretization of a continuum QFT}
\label{subsec:discretization}

In this subsection, we discuss different approaches for discretizing a continuum QFT.
First we explain the common approach of discretizing space as a regular lattice.
Then we describe the alternative wavelet approach for discretizing a QFT.
For simplicity, we consider quantum fields in one spatial dimension.
These discretizations can be extended to higher dimensions using known techniques~\cite{BP13, BRSS15}.

% Lattice discretization
We begin by describing a massive scalar bosonic free QFT
using a Hamiltonian formalism,
which is typical for quantum simulation of a QFT~\cite{JLP14,FS20}.
The time-independent Hamiltonian for a massive scalar bosonic free QFT with bare mass~$m_0$, confined to one spatial dimension, 
is
\begin{equation}
\label{eq:free_Hamiltonian}
    \hat{H}_{\mathrm{free}} := \frac1{2} \int_{\reals}\dd{x}
    \left[
        \hat{\Pi}^2 (x) +
        \left(\nabla \hat{\Phi} (x)\right)^2 +
        m_0^2 \hat{\Phi}^2 (x)
    \right],
\end{equation}
with~$x$ a position
and with field operator~$\hat{\Phi}(x)$ and conjugate momentum~$\hat{\Pi}(x)$ satisfying
\begin{equation}
\label{eq:CCRs}
\left[\hat{\Phi}(x),\hat{\Pi}(y)\right]=\upi \updelta(x-y)\mathds1, \quad
\left[\hat{\Phi}(x),\hat{\Phi}(y)\right]=
\left[\hat{\Pi}(x),\hat{\Pi}(y)\right]=0.
\end{equation}
By QFT conventions, we fix constants $c\equiv1\equiv\hbar$.
Discretizing QFT is achieved by replacing the spatial continuum with a hyper-cubic lattice of `locations' with finite lattice spacing~$a\in\reals^+$.
For a 1D scalar field over~\reals, the discrete domain is~$\Omega:=a\integers$,
and discrete field operators over the lattice satisfy
$\left[\hat{\Phi}(x),\hat{\Pi}(y)\right]=\upi
\updelta_{x,y}\mathds1/a$, and all other commutators are zero.

% Discretized QFT
For lattice QFT~\cite{Rot05},
the Laplacian in~\cref{eq:free_Hamiltonian} and integration measure are
\begin{equation}
\nabla^2 \hat{\Phi} (x) \mapsto
\nabla^2_a \hat{\Phi} (x):=\frac1{a^2}(\hat{\Phi} (x+a)+\hat{\Phi}(x-a)-2\hat{\Phi} (x)),
\quad
\int_{\reals}\dd{x} \mapsto \sum_{x \in \Omega} a,
\end{equation}
respectively.
Defining~$\hat{\Phi}_\ell:=\hat{\Phi}(\ell a)$ and~$\hat{\Pi}_\ell:=\hat{\Pi}(\ell a)$ for each integer~$\ell$, the discretized version of the free QFT~\eqref{eq:free_Hamiltonian} is~\cite{KS19}
\begin{equation}
\label{eq:lattice_Hamiltonian}
    \hat{H}_{a} := \frac{a}{2} \sum_{x \in \Omega}
     \left[
        \hat{\Pi}^2 (x) -
       \hat{\Phi}(x) \nabla^2_a \hat{\Phi} (x) +
        m_0^2 \hat{\Phi}^2 (x)
    \right]
    = 
    \frac{a}{2} \sum_{\ell \in \integers} \hat{\Pi}^2_\ell +
    \frac{a}{2} \sum_{\ell, \ell^\prime \in \integers} \hat{\Phi}_\ell K_{\ell \ell^\prime} \hat{\Phi}_{\ell^\prime},
\end{equation}
where
\begin{equation}
    K_{\ell \ell^\prime}:= m_0^2 \updelta_{\ell \ell^\prime}-
    \left( \updelta_{\ell, \ell^\prime+1} +
    \updelta_{\ell, \ell^\prime-1} -
    2 \updelta_{\ell \ell^\prime}\right)/a^2,
\end{equation}
is the coupling between the two localized $\ell$ and $\ell'$ modes.

% Field discretization
The field~$\hat{\Phi}(x)$ and conjugate~$\hat{\Pi}(x)$ operators can be expressed in a wavelet basis by projections onto the scaling and wavelet functions as~\cite{BRSS15,BP13}
\begin{align}
\label{eq:discrete_fields}
&\hat{\Phi}_{\s;\, \ell}^{(k)} := 
\int_\reals \dd{x} s_\ell^{(k)}(x) \hat{\Phi}(x),\quad
\hat{\Pi}_{\s;\, \ell}^{(k)}:=
\int_\reals \dd{x} s_\ell^{(k)}(x) \hat{\Pi}(x),\\
&\hat{\Phi}_{\w;\, \ell}^{(r)}:=
\int_\reals \dd{x} w_\ell^{(r)}(x) \hat{\Phi}(x),\quad
\hat{\Pi}_{\w;\, \ell}^{(r)}:=
\int_\reals \dd{x} w_\ell^{(r)}(x) \hat{\Pi}(x) \quad \forall\, r\geq k.
\end{align}
We refer to field operators with subscript `s' and `w' as scale-field and wavelet-field operators, respectively.
These operators have a compact support, determined by the support of their associated wavelet or scaling functions, and satisfy a set of commutation relations analogous to those of~\cref{eq:CCRs}, but with the Dirac~$\updelta$ replaced with the Kronecker~$\updelta$~\cite[p.~3]{BRSS15}.

% Fixed-scale discretization
To discretize the free QFT in a fixed-scale basis, we project the continuum Hamiltonian~\eqref{eq:free_Hamiltonian} onto a scale subspace~$\mathcal{S}_k$ of~$\sqintR$ for some~$k\in \integers^+$.
By~\cref{eq:discrete_fields}, we first project~$\hat{\Phi}(x)$ and~$\hat{\Pi}(x)$ onto this subspace.
Substituting the projected field and momentum operators into~\cref{eq:free_Hamiltonian}, we obtain the expression
\begin{equation}
\label{eq:fixediscale_Hamiltonian}
    \hat{H}^{(k)}_\s
    := \frac1{2} \sum_\ell \hat{\Pi}_{\s;\ell}^{(k)} \hat{\Pi}_{\s;\ell}^{(k)}
    + \frac1{2} \sum_{\ell\ell^\prime}
    \hat{\Phi}_{\s;\ell}^{(k)} K^{(k)}_{\sS; \ell \ell^\prime} \hat{\Phi}_{\s;\ell^\prime}^{(k)},
\end{equation}
for projected Hamiltonian onto a scale subspace~$\mathcal{S}_k$.
Here
\begin{equation}
\label{eq:coupling_matrix_elements}
    K^{(k)}_{\sS; \ell \ell^\prime}
    := m_0^2 \updelta_{\ell \ell^\prime} - 4^k \Delta^{(2)}_{\ell^\prime - \ell}\, ,
\end{equation}
are the coupling between different modes and ~$\Delta^{(2)}_{\ell^\prime-\ell}$ are the coefficients
\begin{equation}
\label{eq:derivative_overlaps}
    \Delta^{(n)}_{\ell-\ell^\prime} :=
    \int \dd{x} s_{\ell-\ell^\prime}(x) \dv[n]{}{x} s(x)
    \quad \forall\, n\geq 1,
\end{equation}
with $n=2$.
We refer to these coefficients as the $n$-order \emph{derivative overlaps} and use~$\Delta_\ell:=\Delta^{(2)}_\ell$ for simplicity.

% Multi-scale discretization
The multi-scale wavelet discretization of the free QFT is obtained by projecting~$\hat{\Phi}(x)$ and~$\hat{\Pi}(x)$ onto a multi-scale basis and substituting the projected operators into~\cref{eq:free_Hamiltonian}.
In this case, the projected Hamiltonian is
\begin{equation}
\label{eq:multiscale_Hamiltonian}
    \hat{H}_\w:=
    \hat{H}^{(s_0)}_\s+
    \frac1{2} \sum_{\ell,r\geq s_0} \hat{\Pi}_{\w;\ell}^{(r)} \hat{\Pi}_{\w;\ell}^{(r)}
    + \frac1{2}
    \sum_{\ell\ell^\prime, rr^\prime}
    \hat{\Phi}_{\w;\ell}^{(r)} K^{(r,r^\prime)}_{\ww; \ell \ell^\prime} \hat{\Phi}_{\w;\ell^\prime}^{(r^\prime)}
    + \frac1{2}
    \sum_{\ell\ell^\prime, r \geq s_0}
    \hat{\Phi}_{\s;\ell}^{(s_0)} K^{(s_0,r)}_{\sw; \ell \ell^\prime} \hat{\Phi}_{\w;\ell^\prime}^{(r)},
\end{equation}
where~$K^{(r,r^\prime)}_{\ww; \ell \ell^\prime}$ are coupling between the wavelet fields at scales~$r$ and~$r^\prime$, and~$K^{(s_0,r)}_{\sw; \ell \ell^\prime}$ are coupling between the scale fields at scale~$s_0$ and wavelet fields at scale~$r$.
These couplings are systematically computed from the derivative overlaps~\eqref{eq:derivative_overlaps}~\cite{BP13,SB16}.

% Momentum discretization
Wavelet discretization for the momentum operator
$\hat{P}:=- \int_{\reals}\dd{x}
\hat{\Pi}(x)\nabla \hat{\Phi} (x)$
of the free QFT~\cite[p.~6]{BP13} is similar to wavelet discretization for Hamiltonian.
Therefore, we only consider fixed-scale discretization of~$\hat{P}$.
In this discretization, $\hat{P}$ is projected onto a scale subspace~$\mathcal{S}_k$.
The projected momentum operator is
\begin{equation}
\label{eq:projected_momentum}
    \hat{P}^{(k)}_\s := - \sum_{\ell, \ell^\prime}
    \hat{\Pi}_{\s;\ell}^{(k)} P^{(k)}_{\ell \ell^\prime} \hat{\Phi}_{\s;\ell^\prime}^{(k)}, \quad
    P^{(k)}_{\ell \ell^\prime}:= 2^k 
    \Delta^{(1)}_{\ell^\prime- \ell},
\end{equation}
where~$\Delta^{(1)}_{\ell^\prime-\ell}$ are the derivative overlaps in~\cref{eq:derivative_overlaps} with~$n=1$.
We use this expression in~\cref{subsubsec:fixedscale_groundstate} to obtain the number of modes of the discretized free QFT in a fixed-scale basis.

%%%%%%%%%%%%%%%%%%%%%%%%%%%%%%%%%%%%%%%%%%%%%%%%%%%%%%%%%
\subsection{Kitaev-Webb method for Gaussian-state generation}
\label{subsec:Kitaev_Webb_method}

Here we review the Kitaev-Webb method for generating an approximation of a multi-dimensional continuous Gaussian pure state on a quantum register~\cite{KW09}.
First we define a continuous Gaussian pure state and set our notations for particular Gaussian states.
Then we explain the main idea of the Kitaev-Webb method and proceed with describing details of the method.

% Continuous Gaussian state
We begin by defining a continuous Gaussian pure state in~\cref{def:continuous_Gaussian}.
Throughout this paper, we use this definition when we refer to a continuous Gaussian pure state.

\begin{definition}[Continuous Gaussian pure state]
\label{def:continuous_Gaussian}
Let~$N$ be a positive integer,~$\upbm{A}$ be a real-valued $N$-by-$N$ symmetric positive-definite matrix and
\begin{equation}
    p_N(\upbm{A};\bm{x})
    :=\sqrt{\tfrac{\det\upbm{A}}{(2\uppi)^N}}\
    \e^{-\bm{x}^\T\upbm{A}\bm{x}/2},
\end{equation}
be the probability density function of a continuous $N$-dimensional Gaussian~(NDG) distribution, with the ICM~$\upbm{A}$,
for a random variable~$\bm{x}:=(x_0,x_1,\ldots,x_{N-1}) \in \reals^N$.
We define the pure state
\begin{equation}
\label{eq:continuous_NDG}
\ket{\G_N(\upbm{A})}
:=\int_{\reals^N}\dd[N]{\bm{x}}
\sqrt{p_N(\upbm{A};\bm{x})}\ket{\bm{x}},
\end{equation}
where
\begin{equation}
\label{eq:ketx}
\ket{\bm{x}}:=\ket{x_0} \otimes \cdots \otimes \ket{x_{N-1}},
\end{equation}
is a vector of distributions,
as the continuous NDG pure state with the ICM~$\upbm{A}$.
\end{definition}

We use a particular notation for 1DG pure states.
If the Gaussian distribution in~\cref{def:continuous_Gaussian} is one-dimensional with the standard deviation~$\sigma \in \reals^+$ and the mean value $\mu \in \reals^+$,
then we refer to the state
\begin{equation}
\label{eq:continuous1DG}
 \ket{\G(\sigma, \mu)}:= \frac1{\mathcal{N}}
    \int_{\reals}\dd[]{x} \e^{-\frac{(x-\mu)^2}{4\sigma^2}} \ket{x},
    \quad
    \mathcal{N}^2:=\int_\reals \dd{x}\e^{-\frac{(x-\mu)^2}{2\sigma^2}} =
    \sigma\sqrt{2\uppi},
\end{equation}
for~$\ket{x}$ a continuously parameterized position state,
as the continuous 1DG pure state with standard deviation~$\sigma$ and mean value~$\mu$.
For simplicity, 
we denote~$\ket{\G(\sigma, 0)}$ by~$\ket{\G(\sigma)}$.

% Main idea
We now present a high-level description of Kitaev's and Webb's method for generating a multi-dimensional Gaussian state.
Their method's main idea is first to prepare a set of independent 1DG states and then perform a basis transformation to produce the desired Gaussian state.
The parameters needed for preparing the 1DG states and the basis transformation are outputs of a classical algorithm that computes the LDL decomposition of the Gaussian state's ICM~$\upbm{A}$.
Specifically, the classical algorithm returns a diagonal matrix~$\upbm{D}$ and a lower unit-triangular matrix~$\upbm{L}$ such that $\upbm{A}=\upbm{LDL}^\T$. 
Diagonal of~$\upbm{D}$ are parameters needed for preparing the 1DG states, and off-diagonals of~$\upbm{L}$ are parameters needed for the basis transformation.

Their method for generating a $N$DG state with ICM~$\upbm{A}$ can be described by three main steps:
(1)~classically compute~$\upbm{L}$ and~$\upbm{D}$ in the LDL decomposition of the ICM~$\upbm{A}$;
(2)~generate an approximation for~$\ket{\G_N(\upbm{D})}$; and
(3)~implement the basis transformation~$\ket{\bm{x}} \mapsto \ket{\upbm{S}\bm{x}}$, where~$\upbm{S}$ is inverse-transpose of~$\upbm{L}$ and~$\ket{\bm{x}}$~\eqref{eq:ketx} is the basis state.
The basis transformation in the last step is implemented by storing off-diagonal elements of~$\upbm{S}$ on ancillary quantum registers and performing reversible operations on a quantum computer.
The state$\ket{\G_N(\upbm{D})}$ in the second step is generated by preparing~$N$ independent 1DG states with standard deviations $\sigma_\ell:=1/\sqrt{D_{\ell\ell}}$ for $\ell\in\{0,\ldots,N-1\}$.

% 1DG state
We now describe the Kitaev-Webb method for generating 1DG states.
To generate the continuous 1DG state~$\ket{\G(\sigma,\mu)}$~\eqref{eq:continuous1DG} 
on a quantum register, it is first approximated by the discrete 1DG state
\begin{equation}
\label{eq:KW_discrete_1DG}
     \ket{\tilde{\G}_\text{KW}(\sigma,\mu)}:=
     \sum_{i\in \integers} \tilde{\G}_\text{KW}(\sigma,\mu;\, i) \ket{i},
     \quad
     \tilde{\G}_\text{KW}(\sigma,\mu;\, i):=
     \tfrac{1}{\sqrt{f_\text{KW}(\sigma, \mu)}} \e^{-\tfrac{(i-\mu)^2}{4\sigma^2}},
     \quad
     f_\text{KW}(\sigma, \mu):=
     \sum_{i\in \integers} \e^{-\tfrac{(i-\mu)^2}{2\sigma^2}},
\end{equation}
over the 1D infinite lattice with unit lattice spacing. 
This discrete 1DG state is again approximated by the state
\begin{align}
\label{eq:KW_approx_1DG}
    \ket{\xi(\sigma, \mu, m)}:=
    \sum_{i=0}^{2^m-1} \xi(\sigma, \mu, m;\, i)\ket{i},
    \quad
    \xi^2(\sigma, \mu, m;\, i):=
    \sum_{j\in \integers}
    \tilde{\G}^2_\text{KW}(\sigma, \mu;\, i+j 2^m).
\end{align}
This quantum state is used as an approximation for~$\ket{\G(\sigma,\mu)}$~\eqref{eq:continuous1DG} in the Kitaev-Webb method to be generated on a quantum register.
The key point of this method 
is to employ the recursive decomposition
\begin{align}
\label{eq:KW_recursive_formula}
    \ket{\xi(\sigma, \mu, m)} =
    \ket{\xi\left(\tfrac{\sigma}{2}, \tfrac{\mu}{2}, m-1\right)} \otimes \cos \theta \ket{0} +
    \ket{\xi\left(\tfrac{\sigma}{2}, \tfrac{\mu-1}{2}, m-1\right)} \otimes \sin \theta \ket{1}, 
    \quad \theta:= \arccos{\sqrt{\frac{f_\text{KW}(\sigma/2, \mu/2)}{f_\text{KW}(\sigma, \mu)}}},
\end{align}
for generating the approximate 1DG state$\ket{\xi(\sigma, \mu, m)}$~\eqref{eq:KW_approx_1DG}.

By the recursive formula~\eqref{eq:KW_recursive_formula} and classical inputs~$\sigma_0:=\sigma, \mu_0:=\mu$ and $m$, following recursive procedure is used to generate the approximate~\oneDG\ state:
(1)~compute~$\theta/2\uppi$~\eqref{eq:KW_recursive_formula} and store it on an ancillary quantum register;
(2)~perform the single-qubit rotation~$R(\theta):=\exp(-\upi\theta Y)$ on the rightmost qubit, where~$Y$ is the Pauli-Y gate;
(3)~uncompute~$\theta/2\uppi$;
(4)~compute~$\sigma_1:=\sigma_0/2$ and~$\mu_1:=(\mu_0-q_0)/2$, where~$q_0=0$ if the state of the rotated qubit is~$\ket{0}$ and~$q_0=0$ if it is~$\ket{1}$; and
(5)~generate the state~$\ket{\xi(\sigma_1, \mu_1, m-1)}$~\eqref{eq:KW_recursive_formula} on the remaining~$m-1$ qubits.

The Kitaev-Webb method is restricted to 1DG states that possess an extremely large standard deviation.
In~\cref{subsubsec:1DG_method} and~\cref{subsubsec:1DG_generation}, we describe our methods for generating an approximation of a continuous 1DG state for any standard deviation.

%%%%%%%%%%%%%%%%%%%%%%%%%%%%%%%%%%%%%%%%%%%%%%%%%%%%%%%%%
\subsection{Quantum random-access machine model for computation}
\label{subsec:QRAM}

An algorithm's time complexity depends on the model for computation. 
Here we review the quantum random-access machine~(QRAM) model for computation introduced by Knill~\cite{Knill96}.
We use a variant of this model for analyzing our algorithms' time complexity.
First we describe the architecture of the QRAM model. 
Then we discuss `primitive' operations for this model.
Finally, we describe how an algorithm's time complexity is assessed in the QRAM model and discuss the computational complexity for computing some elementary functions.

% Architecture
The QRAM model is an extension of the classical RAM model that allows classical and quantum computations.
We describe a simplified architecture of the QRAM model.
For more details, we refer to~\cite{LU08, Mis12,WY20}.
In this model, as schematically illustrated in~\cref{fig:QRAM},
a computer has a classical and a quantum processor that are respectively connected to classical and quantum registers.
These processors work in a master-slave fashion, where the classical processor is the master that controls the quantum processor.
A hybrid quantum-classical code is first compiled into the classical processor.
The compiled code contains both classical and quantum instructions.
The classical processor performs the classical instructions on classical registers and sends the quantum instructions to be performed by the quantum processor on quantum registers.
The measurement results are sent back to the classical processor by the quantum processor.
This process could be repeated multiple times depending on the~code.

\begin{figure}[htb]
\centering
    \includegraphics[width=.8\textwidth]{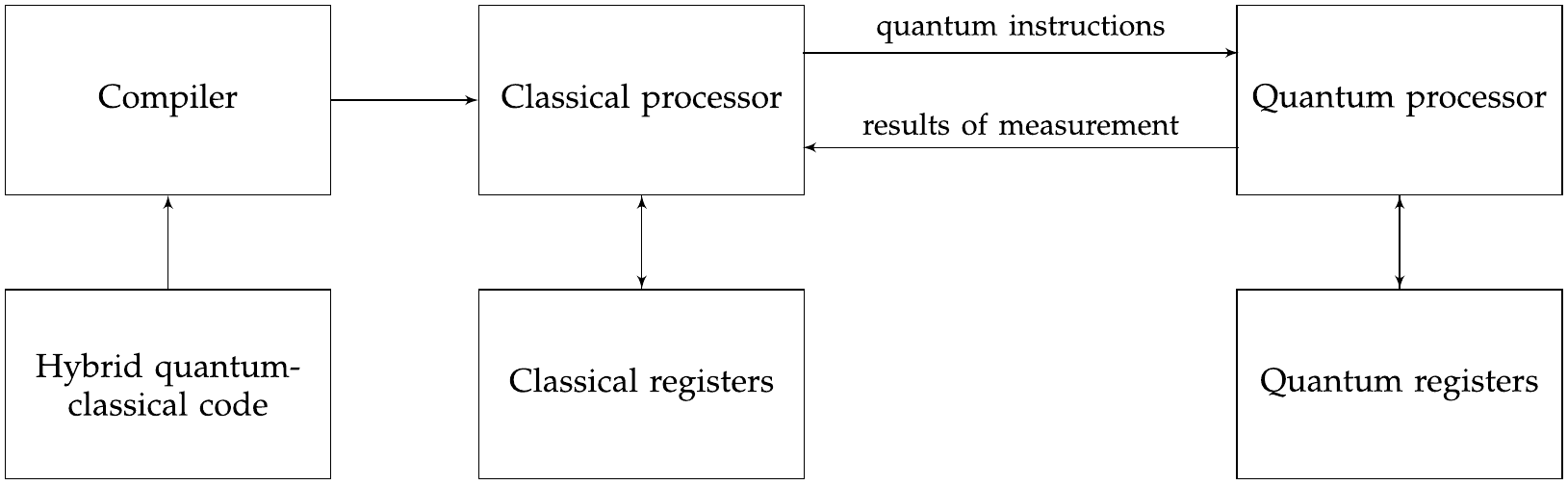}
\caption[Schematic description of the QRAM model]{Schematic description of the QRAM model.
A hybrid code is compiled into the classical processor.
The compiled code provides instructions to be performed by classical and quantum processors on their associated registers.
The quantum processor sends the measurement results back to the classical processor.}
 \label{fig:QRAM}
\end{figure}

% Primitive operations
The classical and quantum processors in the QRAM model can only perform a restricted set of operations on their associated registers.
These operations are called primitive operations, and a unit cost is assigned to each primitive in this model.
As QRAM is an extension of the classical RAM, the QRAM's classical primitives are considered to be the same as the primitives in the classical RAM model, which are the 
following~\cite{GT14}:
(1)~basic arithmetic operations, i.e., addition, subtraction, multiplication and division;
(2)~data-movement operations such as writing data from memory to classical registers and reading data from classical registers to memory;
(3)~Boolean logic operations such as AND and OR; and
(4)~flow-control operations such as calling a function or returning from a function.
The classical primitives, except Boolean logic operations, are high-level operations.
In practice, implementing high-level operations does not have the same cost in terms of low-level (i.e., bit-wise) operations.

% Quantum primitives 
Unlike classical primitives, which are high-level operations, quantum primitives in the QRAM model are low-level operations~\cite{Knill96}.
Specifically, quantum primitives in the QRAM model are quantum gates from a universal set of gates.
In~\cref{subsubsec:complexity_measure}, we describe our alternative for quantum primitives that are high-level operations, similar to classical primitives.

% Complexity analysis in QRAM
The time complexity for a classical algorithm is determined by counting the number of classical primitives that need to be executed in the algorithm.
Similarly, the common approach to analyzing a quantum algorithm's time complexity in the QRAM model is counting the number of quantum primitives.
As quantum primitives in QRAM are quantum gates, the algorithm's gate complexity is typically used as a standard metric to cost out a quantum algorithm.

% Complexity for computing elementary functions
Finally, we state the time complexity for computing four elementary functions with respect to the classical primitives in the QRAM model: logarithm, square-root, inverse-square-root and trigonometric functions.
These functions are used in various subroutines of our ground-state-generation algorithms described in~\cref{subsec:highlevel_description}.
The time complexity for computing each of these functions was analyzed in terms of the time complexity for performing a multiplication in~\cite{Bre10}.
Multiplication is a primitive operation in the QRAM model and has a unit cost.
Therefore, we only list time complexities for the elementary functions with respect to classical primitives in the QRAM model.
The time complexity for computing square-root or inverse-square-root of a number in this model is~$\order{1}$~\cite[pp.~3--4]{Bre10}, and the time complexity for computing a logarithm or trigonometric functions to precision~$p$ is~$\order{\log p}$~\cite[pp.~11,~14]{Bre10}.
We use these complexities in~\cref{subsect:complexity_analysis}, where we analyze our algorithms' time complexity.

%%%%%%%%%%%%%%%%%%%%%%%%%%%%%%%%%%%%%%%%%%%%%%%%%%%%%%%%%
%%%%%% Approach %%%%%%%%%%%%%%%%%%%%%%%%%%%%%%%%%%%%%%%%%
%%%%%%%%%%%%%%%%%%%%%%%%%%%%%%%%%%%%%%%%%%%%%%%%%%%%%%%%%
\section{Approach}
\label{sec:approach}

In this section, we present our approach for generating an approximate quantum-register representation for the ground state of a free massive scalar bosonic QFT, such that the quantum algorithm succeeds deterministically in a quasilinear time with respect to the number of modes for the discretized QFT.
We begin in~\cref{subsec:model} by discussing our model for describing a free massive scalar bosonic QFT.
We introduce a client-server framework for simulating a QFT and discuss a metric to measure our algorithm's time complexity.
Next we describe in~\cref{subsec:mathematics} the mathematical approach inherent in generating the free-field ground state.
We explain how the mass, momentum cutoff, wavelet index and error tolerance are used to discretize the continuum free-field Hamiltonian~\eqref{eq:free_Hamiltonian} in fixed- and multi-scale wavelet bases from which the covariance-matrix description of the discretized-QFT ground state is obtained.
Finally, in~\cref{subsec:methods}, we describe our Fourier and wavelet-based method for ground-state generation and also our methods for generating one-dimensional Gaussian states.

%%%%%%%%%%%%%%%%%%%%%%%%%%%%%%%%%%%%%%%%%%%%%%%%%%%%%%%%%
%%%%%% Model %%%%%%%%%%%%%%%%%%%%%%%%%%%%%%%%%%%%%%%%%%%%
%%%%%%%%%%%%%%%%%%%%%%%%%%%%%%%%%%%%%%%%%%%%%%%%%%%%%%%%%

\subsection{Model}
\label{subsec:model}

In this subsection, we discuss our model for describing a massive scalar bosonic QFT.
We also introduce a client-server framework for simulating a QFT on a quantum computer and discuss the metric we use to measure our algorithms' time complexity.
We begin by explaining our model for describing the QFT in~\cref{subsubsec:discqft}.
We then explain the framework in~\cref{subsubsec:framework} followed by the metric for measuring an algorithm's runtime in~\cref{subsubsec:complexity_measure}.

%%%%%%%%%%%%%%%%%%%%%%%%%%%%%%%%%%%%%%%%%%%%%%%%%%%%%%%%%
\subsubsection{Discretized quantum field theory}
\label{subsubsec:discqft}

Here we explain our model for describing a massive scalar bosonic QFT.
We describe our approach for discretizing a one-dimensional scalar bosonic field in fixed- and multi-scale wavelet bases. 
Furthermore, we compare these discretizations with the Jordan-Lee-Preskill discretization, which is based on the conventional lattice approach~\cite{JLP12}.

In contrast to the usual approach for discretizing a QFT over a non-compact domain, such as the infinite real line for a one-dimensional theory, we discretize for the field on a finite interval of the real line with periodic boundary conditions.
Mathematically, the finite interval with periodic boundary conditions can be treated as a circle domain.
We denote the bare mass of the field theory by~$m_0$, as in~\cref{eq:free_Hamiltonian}.
We consider an ultraviolet momentum cutoff~$\Lambda$ for the field theory, meaning we ignore all field configurations with momentum higher than~$\Lambda$.
For convenience, we work in the normalized scale where the compact domain becomes the unit interval.
In this case, the theory is on the unit interval with periodic boundaries, and all involved parameters such as mass and momentum cutoff become dimensionless parameters.

In fixed-scale discretization, we partition the unit interval into~$2^k$ subintervals of length~$1/2^k$.
We choose the positive integer~$k$ such that~$2^k \geq 2(2\dbIndex-1)$, where~$\dbIndex$ is the wavelet index~(\cref{subsec:wavelet_bases}).
This choice is made because the smallest size admissible with periodic boundaries is~$2(2\dbIndex-1)$ to ensure the scale function and its translations are orthogonal~\cite[p.~4]{SB16}.
We assign a discrete scale field~\eqref{eq:discrete_fields} to each of the~$2\dbIndex-1$ cyclically consecutive subintervals according to the following averaging rule. 
Each discrete field is an average of the continuous field over an interval of length~$(2\dbIndex-1)/2^k$ weighted by the db\dbIndex\ scaling function at scale~$k$~\cite[p.~7]{BP13}.
The discrete fields and their conjugate momenta satisfy the commutation relations analogous to those of~\cref{eq:CCRs} but with the Dirac~$\updelta$ replaced by the Kronecker~$\updelta$.
The momentum cutoff~$\Lambda$ in this discretization is proportional to the inverse of the subinterval’s length.

In multi-scale wavelet discretization, we partition the unit interval into~$2^s$ subintervals of length~$1/2^s$ at each scale,
where~$s\in\{s_0,s_0+1,\ldots,k-1\}$.
Here~$k$ is an integer such that~$2^k\geq 2(2\dbIndex-1)$ and~$s_0$ is the smallest integer that satisfies this inequality.
At the smallest scale~$s_0$, we assign a scale field and a wavelet field~\eqref{eq:discrete_fields} to each of the~$2\dbIndex-1$ cyclically consecutive subintervals according to the following averaging rule.
Each scale~(wavelet) field is an average of the continuous field over an interval of length~$(2\dbIndex -1)/2^{s_0}$
weighted by the db\dbIndex\ scaling~(wavelet) function at scale~$s_0$.
At each other scale~$s>s_0$,
we assign a wavelet field to each of the~$2\dbIndex-1$ cyclically consecutive subintervals.
 Each wavelet field is an average of the continuous field over an interval of length~$(2\dbIndex -1)/2^s$ weighted by the db\dbIndex\ wavelet function at scale~$s$.
The discrete fields and their conjugate momenta in a multi-scale wavelet basis also satisfy the commutation relations analogous to those of~\cref{eq:CCRs} but with the Dirac~$\updelta$ replaced by a Kronecker~$\updelta$~\cite[p.~3]{BRSS15}.
The momentum cutoff~$\Lambda$ in this discretization is proportional to the inverse of the subinterval’s length at scale with the largest scale index~$s$.

In the Jordan-Lee-Preskill approach~\cite{JLP12}, the conventional lattice discretization is used to discretize the scalar bosonic quantum field.
The unit interval is approximated by a one-dimensional finite lattice with~$2^k$ points and lattice spacing~$1/2^k$ for some positive integer~$k$.
A discrete field is then assigned to each lattice point, where the discrete fields are samples of the continuous field at lattice points.
In contrast, the discrete fields in the wavelet approach are an average of the continuous field over subintervals of the unit interval.
These discrete fields have overlapping domains, whereas the domains of the discrete fields in the lattice approach, i.e., the lattice points, do not overlap.

%%%%%%%%%%%%%%%%%%%%%%%%%%%%%%%%%%%%%%%%%%%%%%%%%%%%%%%%%
\subsubsection{Client-server framework for simulating a QFT}
\label{subsubsec:framework}

We now describe a framework for simulating a QFT.
To make clear ground-state generation vs other aspects of a QFT simulation,
we employ a framework comprising three components:
a client, a main server and a ground-state-generation server; see~\cref{fig:clientServerModel} for a schematic representation of the framework.
We define each component of the framework and elucidate its relative task in simulating a QFT.
Finally, we describe the information flow between the client and the two servers for simulating a massive scalar bosonic QFT.

% Client
A client is an agent who supplies the needed parameters for solving a computational problem to a server and accepts the solution. 
The client in our framework only communicates with the main server.
She provides the required inputs for simulating a QFT to this server and accepts the solution, which is the simulation's outputs.

% Main server
A server is a computer that provides a function or service to one or many clients; a server could also be a client to another server.
The main server's task in our framework is to simulate the QFT specified by the client and deliver the simulation's output to the client.
The main server chooses a particular basis for simulation to accomplish this task and delegates the ground-state-generation part of the simulation to the ground-state-generation server.
The main server is, therefore, a client to the ground-state-generation server in our framework.
The main server supplies the input parameters needed for generating the ground state of the QFT and accepts a quantum state,
which is an approximation for the free-field ground state.

% Ground-state-generation server
The ground-state-generation server in our framework is an auxiliary server whose task is to generate an approximation for the QFT ground state, which is represented in a particular basis by the main server, on a quantum register and deliver the generated state to the main server.
We consider this auxiliary server in our framework to elucidate the input parameters needed for generating the ground state and separate the ground-state generation part of a QFT simulation from other parts of the full quantum simulation as ground-state generation is a bottleneck for the entire simulation. 
The main server then performs the simulation using the generated state by the ground-state-generation server.
The main server may wish to perform the simulation on a different basis.
In this case, the main sever first executes a basis transformation on the generated state by the ground-state-generation server and then performs the simulation on a new basis.

\begin{figure}[htb]
\centering
\includegraphics[width=.7\textwidth]{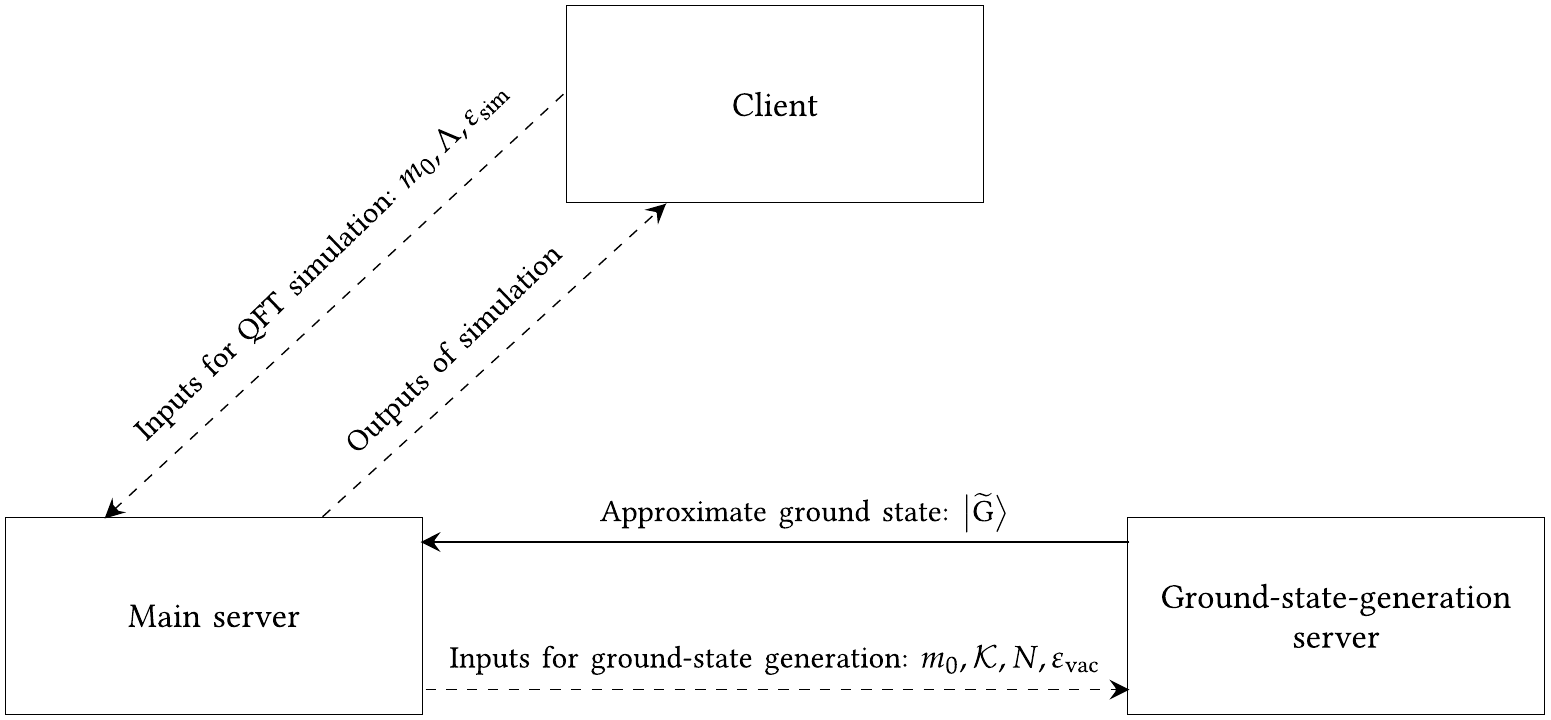}
\caption[A client-server framework for simulating a QFT]{A client-server framework for simulating a QFT.
Dashed lines represent classical communication, and solid lines represent quantum communication.
The classical inputs to the main and ground-state-generation servers are specified in \cref{table:maininputs} and \cref{table:inputsgroundstategen}, respectively, for simulating a massive scalar bosonic QFT.
Output from the ground-state-generation server is an approximation for the free-field ground state.
\label{fig:clientServerModel}
}
\end{figure}

% Information flow
We now discuss the information flow between the client and the two servers in the QFT-simulation framework.
The client in this framework supplies the input parameters for simulating a QFT to the main server.
The required input parameters are those that specify the Hamiltonian (or Lagrangian) describing the QFT, an error tolerance for output of simulation and a parameter specifying the energy at which the simulation is performed.
The free mass $m_0$ is the only parameter that specifies the Hamiltonian of a free massive scalar bosonic QFT;
see~\cref{eq:free_Hamiltonian}.
We use an ultraviolet cutoff on the momentum of the free-QFT particles as the parameter specifying the simulation energy.
\cref{table:maininputs} specifies the required inputs that need to be supplied by the client to the main server for simulating a massive scalar bosonic free QFT.

\begin{table}[htb]
\centering
\begin{tabular}{|c|c|l|}
\hline
Parameter & Type & Description \\
\hline
$m_0$ & $\reals^+$ & Free-QFT mass\\
\hline
$\Lambda$ & $\reals^+$ &
Ultraviolet cutoff on momentum of the free-QFT particles\\
\hline
$\varepsilon_\text{sim}$ & $(0,1)$ & 
Error tolerance for output of simulation\\
\hline
\end{tabular}
\caption{Inputs from client to the main server.}
\label{table:maininputs}
\end{table}

The main server chooses a wavelet basis by selecting a wavelet index $\dbIndex\in \integers_{\geq 3}$ to perform the simulation.
This server approximates the continuum field theory by a finite-mode discretized QFT in the db\dbIndex\ wavelet basis and calculates the sufficient number of modes for the discretized theory using the client’s inputs.
The main server then supplies the inputs specified in \cref{table:inputsgroundstategen} to the ground-state-generation server.

\begin{table}[htb]
\centering
\begin{tabular}{|c|c|l|}
\hline
Parameter & Type & Description\\
\hline
$m_0$ & $\reals^+$ & Free-QFT mass\\
\hline
$\dbIndex$ & $\integers_{\geq 3}$ & Wavelet index\\
\hline
$N$ & $\integers_{\geq 2(2\dbIndex-1)}$ &
Number of modes in the discretized QFT\\
\hline
$\varepsilon_\text{vac}$ & $(0,1)$ &
Error tolerance for generating the discretized-QFT ground state\\
\hline
\end{tabular}
\caption{Inputs from main server to the ground-state-generation server.}
\label{table:inputsgroundstategen}
\end{table}

The ground-state-generation server uses the inputs supplied by the main server to generate an approximation for the ground state of the discretized free QFT.
This server then delivers the generated state to the main server for QFT simulation.
Finally, the simulation outputs are provided to the client by the main server. 
The flow of information between the client and two servers is shown in~\cref{fig:clientServerModel}.

%%%%%%%%%%%%%%%%%%%%%%%%%%%%%%%%%%%%%%%%%%%%%%%%%%%%%%%%%
\subsubsection{Measure for time complexity}
\label{subsubsec:complexity_measure}

Here we describe how we assess an algorithm's time complexity.
We begin by stating the metric that we use for the time complexity of an algorithm.
Then we specify primitive operations for classical and quantum processors of the QRAM model described in~\cref{subsec:QRAM}.
The chosen primitives are high-level operations,
and we explain how to relate these primitives to low-level primitives.
Finally, we discuss how our metric differs from common metrics for analyzing an algorithm's time complexity.

% Metric
We use the number of primitive operations in an algorithm as a metric to quantify the algorithm's time complexity.
That is to say, the number of classical primitives in an algorithm determines its classical complexity, and the number of quantum primitives determines the algorithm's quantum complexity.
By this metric, the time complexity depends on the classical and quantum primitives;
changing the set of primitives yields a different time complexity.
Hence the set of primitive operations needs to be specified.

% primitive operations
We now specify the classical and quantum primitives.
We choose the classical primitives to be the same as the classical primitive operations in the QRAM model;
see~\cref{subsec:QRAM}.
Except for the flow-control operations~(\cref{subsec:QRAM}), we take the quantum primitives to be a quantum version of the classical primitives.
We exclude the flow-control operations as quantum primitives because the quantum processor in the QRAM model is controlled by the classical processor.
Specifically, we choose the following operations as quantum primitives in our time-complexity analysis for a quantum algorithm:
(1)~basic arithmetic operations on quantum registers;
(2)~data-movement operations on quantum registers, such as writing~(preparing) classical data from memory into quantum registers and reading~(measuring) data from quantum registers to memory; and
(3)~quantum logic gates such as the Hadamard, controlled-Not~(CNOT) and Toffoli gates.

% Low-level primitives
With our chosen quantum primitives, similar to the classical RAM model, the QRAM model encapsulates computers' core functionality, not their exact functionalities.
For instance, addition and multiplication operations are each considered a single primitive operation for each processor in this model.
In practice, however, a processor needs to execute more low-level operations---bit-wise operations for the classical processor and qubit-wise operations for quantum processor---to perform multiplication versus addition.
By analyzing the cost of performing high-level primitives in terms of low-level operations, one can obtain an algorithm's time complexity with respect to low-level operations.

% Compare with common metrics
We comment that our approach for analyzing a quantum algorithm's time complexity is not common in literature.
As described in~\cref{subsec:QRAM}, the common approach to cost out a quantum algorithm is to count the number of low-level operations, i.e., quantum gates, in the algorithm~\cite{Knill96,Mon16}.
In~\cref{sec:discussion}, we discuss our algorithms' time complexity with respect to low-level~operations.

%%%%%%%%%%%%%%%%%%%%%%%%%%%%%%%%%%%%%%%%%%%%%%%%%%%%%%%%%
%%%%%% Mathematics %%%%%%%%%%%%%%%%%%%%%%%%%%%%%%%%%%%%%%
%%%%%%%%%%%%%%%%%%%%%%%%%%%%%%%%%%%%%%%%%%%%%%%%%%%%%%%%%

\subsection{Mathematics}
\label{subsec:mathematics}

This subsection describes the mathematical approach for generating the ground state of a massive scalar bosonic QFT.
We begin in~\cref{subsubsec:fixedscale_groundstate} by describing a procedure for approximating the ground state of the continuum QFT~\eqref{eq:free_Hamiltonian} in a fixed-scale basis.
Then we proceed with approximating the ground state in a multi-scale wavelet basis in~\cref{subsubsec:multiscale_groundstate}. 
Lastly, we explain our procedure for discretizing a continuous Gaussian pure state in~\cref{subsubsec:discrete_GaussianState}.

%%%%%%%%%%%%%%%%%%%%%%%%%%%%%%%%%%%%%%%%%%%%%%%%%%%%%%%%%
\subsubsection{Free-field ground state in a fixed-scale basis}
\label{subsubsec:fixedscale_groundstate}

Here we represent the ground state of the continuum free QFT~\eqref{eq:free_Hamiltonian} in a fixed-scale wavelet basis.
To this end, we discretize the continuum theory by projecting its Hamiltonian onto a fixed-scale subspace of~$\sqintS$.
The ground state of the discretized QFT represents the free-field ground state in a fixed-scale basis.
We explain how to select a sufficient number of modes~$N$ for the discretized QFT using the client inputs in~\cref{table:inputsgroundstategen}.
Specifically, we establish sufficiency for~$N$ in terms of the momentum cutoff~$\Lambda$, supplied by the client, such that the magnitude of mean momentum (expectation value of the momentum operator) for a single-particle state in the discretized QFT is no greater than~$\Lambda$.

% fixed-scale ground state
To represent the ground state of the free theory~\eqref{eq:free_Hamiltonian} over the unit interval with periodic boundaries in a fixed-scale wavelet basis, we project the continuum Hamiltonian onto a scale subspace~$\mathcal{S}_k$ of~$\sqintS$ for some integer~$k$ such that~$2^k\geq 2(2\dbIndex-1)$;
see~\cref{subsubsec:discqft}.
The projected Hamiltonian has the form of the discrete Hamiltonian in~\cref{eq:fixediscale_Hamiltonian} but with the coupling matrix
\begin{equation}
\label{eq:fixedscaleK}
    K^{(k)}_{\sS;\, \ell \ell^\prime}
    := m_0^2 \updelta_{\ell \ell^\prime} - N^2
    \left(
      \Delta^{(2)}_{\ell^\prime - \ell}
    + \Delta^{(2)}_{\ell^\prime - (\ell+N)}
    \right),\quad
    N:= 2^k,
\end{equation}
which are matrix elements of~$\upbm{K}_\sS^{(k)}$.
Here~$N$ is the number of modes for the discretized QFT and~$\Delta^{(2)}_{\ell^\prime - \ell}$~\eqref{eq:derivative_overlaps} are the second-order derivative overlaps;
the second term inside the parentheses comes from the periodic boundary condition.
The projected Hamiltonian is quadratic in the field operators and their conjugate momenta akin to the discrete Hamiltonian in~\cref{eq:fixediscale_Hamiltonian} and, therefore, its ground state is a Gaussian state.
Specifically, the ground state of the projected Hamiltonian is
\begin{equation}
\label{eq:fixedscale_groundstate}
    \ket{\G_\text{scale}^{(k)}}:=
    \left( \dfrac{\det \upbm{A}_\sS^{(k)}}{(2\uppi)^N} \right)^{1/4}
    \int_{\reals^N} \dd[N]{\bm{\phi}}
    \e^{ -\tfrac{1}{4} \bm{\phi}^\T \upbm{A}_\sS^{ (k)} \bm{\phi}} \ket{\bm{\phi}},
    \quad
    \upbm{A}_\sS^{(k)}:= \sqrt{\upbm{K}_\sS^{(k)}},
\end{equation}
where $\upbm{A}_\sS^{(k)}$ is the ground state's ICM as per~\cref{def:continuous_Gaussian}.

% Number of modes in discrete QFT
We now establish sufficiency for the number of modes~$N$ using the supplied momentum cutoff~$\Lambda$ by the client to calculate the sufficient~$N$ for the discretized QFT.
We first project the momentum operator~(\cref{subsec:discretization}) of the continuum QFT to the same scale subspace~$\mathcal{S}_k$ that the continuum Hamiltonian is projected.
Then we write an expression for a single-particle state whose mean momentum has the maximum magnitude~$\bar{P}_{\max}$.
Next we bound~$\bar{P}_{\max}$ from above by~$\Lambda$.
The sufficient~$N$ saturates this bound.
\cref{prop:sufficient_modes} provides the
established sufficiency for~$N$ with respect to~$\Lambda$
and the largest first-order derivative overlap
$\Delta^{(1)}_{\max} := \max_{\ell} \abs{\Delta^{(1)}_\ell}$,
for~$\Delta^{(1)}_\ell$~\eqref{eq:derivative_overlaps} the first-order derivative overlaps.

\begin{proposition}
\label{prop:sufficient_modes}
For~$\Lambda$ the momentum cutoff and~$\Delta^{(1)}_{\max}$ the largest first-order derivative overlap,
\begin{equation}
\label{eq:sufficientModes}
    N = \floor*{\frac{2\Lambda}{\Delta^{(1)}_{\max}}},
\end{equation}
modes suffices to guarantee that the mean momentum of a single-particle state in the discretized QFT is bounded from above~by~$\Lambda$.
\end{proposition}

\begin{proof}
Let $a^{(k)\text\textdagger}_\ell$ be the creation operator constructed from the $\ell^\text{th}$ scale-field operator~\eqref{eq:discrete_fields} and its conjugate momentum~\footnote{
Similar to the procedure of constructing creation and annihilation operators from the position and momentum operators of a quantum harmonic oscillator, a set of discrete creation and annihilation operators can be constructed from the discrete field operators and their conjugate momenta in~\cref{eq:discrete_fields}; see~\cite[p.~7]{BP13}.
}.
Acting this operator on the ground state in~\cref{eq:fixedscale_groundstate} creates a single-particle state
with zero-mean momentum whose wavefunction is localized in a compact space of size equal to the support of~$s^{(k)}_\ell$~\cite{BRSS15}.
Single-particle states with finite mean-momentum can be created from a superposition of two zero-mean-momentum single-particle states as
\begin{equation}
    \ket{\Psi^{(k)}_{\ell\ell'}}:=
    \left(\alpha a^{(k)\text\textdagger}_\ell
    + \beta a^{(k)\text\textdagger}_{\ell'} \right)\ket{\G_\text{scale}^{(k)}},    
\end{equation}
with $\alpha, \beta \in \cmplex$ such that
$\abs{\alpha}^2+\abs{\beta}^2 = 1$.
The expectation value of the projected momentum operator~\eqref{eq:projected_momentum}, i.e. the mean momentum, for this state is
\begin{equation}
    \bar{P}:=\ev{\hat{P}^{(k)}}{\Psi^{(k)}_{\ell\ell'}}
    = P^{(k)}_{\ell \ell'} \Im\left(\alpha\beta^\ast\right).
\end{equation}
The magnitude of this expression is maximized
for~$\alpha =\pm\upi \beta =1/\sqrt{2}$.
By this equation and the projected momentum operator~\eqref{eq:projected_momentum}, the maximum magnitude of the mean momentum for a single-particle state is
\begin{equation}
\label{eq:max_momentum}
    \bar{P}_{\max}
    = \frac{2^k}{2} \max_{\ell} \abs{\Delta^{(1)}_\ell}=\frac{2^k}{2}\Delta^{(1)}_{\max}.
\end{equation}
By bounding this expression from above by the momentum cutoff~$\Lambda$, we obtain
\begin{align}
    k = \floor*{\log_2\left(\frac{2\Lambda}{\Delta^{(1)}_{\max}}\right)},
\end{align}
which, by $N=2^k$~\eqref{eq:fixedscaleK}, yields the sufficient number of modes in~\cref{eq:sufficientModes}.
\end{proof}

\noindent
The established sufficiency~\eqref{eq:sufficientModes} for~$N$ is used by the main server to calculate the number of modes for the discretized~QFT.

%%%%%%%%%%%%%%%%%%%%%%%%%%%%%%%%%%%%%%%%%%%%%%%%%%%%%%%%%
\subsubsection{Free-field ground state in a multi-scale wavelet basis}
\label{subsubsec:multiscale_groundstate}

We now represent the ground state of the continuum theory~\eqref{eq:free_Hamiltonian} in a multi-scale wavelet basis.
In this case, we project the continuum Hamiltonian onto a subspace of~$\sqintS$ that is a multi-scale decomposition of the scale subspace~$\mathcal{S}_k$~(\cref{subsec:wavelet_bases}).
The projected Hamiltonian~\eqref{eq:multiscale_Hamiltonian} onto a multi-scale subspace is quadratic, similar to the fixed-scale Hamiltonian~\eqref{eq:fixediscale_Hamiltonian}, but involves~both scale- and wavelet-field operators~\eqref{eq:discrete_fields}.

The coupling matrix~$\upbm{K}^{(k)}$ in the multi-scale Hamiltonian~\eqref{eq:multiscale_Hamiltonian} is obtained by a multi-level wavelet transform from the fixed-scale coupling matrix~$\upbm{K}_\sS^{(k)}$~\eqref{eq:fixedscaleK}.
Let~$s_0<k$ be the scale index for the lowest scale, then~$\upbm{K}^{(k)}$ has the block-matrix structure
\begin{equation}
\label{eq:multiscaleK}
\upbm{K}^{(k)}:=
\begin{bmatrix*}[l]
\upbm{K}^{(s_0)}_\sS & 
\upbm{K}^{(s_0,\, s_0)}_\sw &
\cdots &
\upbm{K}^{(s_0,\, k-1)}_\sw \\
\upbm{K}^{(s_0,\, s_0)\T}_\sw &
\upbm{K}^{(s_0,\, s_0)}_\ww   &
\cdots
&\upbm{K}^{(s_0,\, k-1)}_\ww\\
\quad \vdots & \quad \vdots& \ddots &\quad \vdots\\
\upbm{K}^{(s_0,\, k-1)\T}_\sw &
\upbm{K}^{(s_0,\, k-1)\T}_\ww &
\cdots &
\upbm{K}^{(k-1,\, k-1)}_\ww
\end{bmatrix*},
\end{equation}
imposed by the wavelet transform~\cite{BRSS15}.
We select the level of wavelet transform so that the number of modes~$2^{s_0}$ in the lowest scale~$s_0$,
i.e., the number of rows or columns of the top-left block~$\upbm{K}^{(s_0)}_\sS$, is at least~$2(2\dbIndex-1)$ as per~\cref{subsubsec:discqft}.
The ground state of the projected Hamiltonian onto a multi-scale wavelet basis is
\begin{equation}
\label{eq:multiscale_groundstate}
    \ket{\G_\text{wavelet}^{(k)}}:=
    \left( \dfrac{\det \upbm{A}^{(k)}}{(2\uppi)^N} \right)^{1/4}
    \int_{\reals^N} \dd[N]{\bm{\phi}}
    \e^{ -\tfrac{1}{4} \bm{\phi}^\T \upbm{A}^{(k)} \bm{\phi}} \ket{\bm{\phi}},
\end{equation}
which is a continuous Gaussian pure state akin to the state in~\cref{eq:fixedscale_groundstate} but with the ICM
\begin{equation}
\label{eq:multiscale_ICM}
\upbm{A}^{(k)}:=
\begin{bmatrix*}[l]
\upbm{A}^{(s_0)}_\sS & 
\upbm{A}^{(s_0,\, s_0)}_\sw &
\cdots &
\upbm{A}^{(s_0,\, k-1)}_\sw \\
\upbm{A}^{(s_0,\, s_0)\T}_\sw &
\upbm{A}^{(s_0,\, s_0)}_\ww   &
\cdots
&\upbm{A}^{(s_0,\, k-1)}_\ww\\
\quad \vdots & \quad \vdots& \ddots &\quad \vdots\\
\upbm{A}^{(s_0,\, k-1)\T}_\sw &
\upbm{A}^{(s_0,\, k-1)\T}_\ww&
\cdots &
\upbm{A}^{(k-1,\, k-1)}_\ww
\end{bmatrix*}
= \sqrt{\upbm{K}^{(k)}},
\end{equation}
which has the same block-matrix structure as the coupling matrix in~\cref{eq:multiscaleK}.
For convenience, we henceforth refer to the block with subscript `ss' as the ss block.
Similarly, we refer to the blocks with subscript `sw' as the sw blocks and those with subscript `ww' as the ww blocks.

%%%%%%%%%%%%%%%%%%%%%%%%%%%%%%%%%%%%%%%%%%%%%%%%%%%%%%%%%
\subsubsection{Discretization of continuous Gaussian pure states}
\label{subsubsec:discrete_GaussianState}

Discretization is essential in obtaining a qubit representation for a continuous quantum state.
Here we explain how we discretize a multi-dimensional continuous Gaussian pure state.
We use the described method to discretize the ground state of the free field theory represented in both fixed- and multi-scale wavelet bases to generate the ground state on a quantum register.
First we define a discrete 1DG pure state over a lattice in~\cref{def:discrete1DG}.
Then we explain how to discretize a multi-dimensional Gaussian pure state.

\begin{definition}(discrete~1DG pure state over a lattice)
\label{def:discrete1DG}
For~$\sigma, \delta\in \reals^+$ and $ m\in \integers^+$, let
$\mathds{L}:=\{j\delta\mid j\in[-2^{m-1},2^{m-1})\cap\integers\}$
be a one-dimensional lattice with~$2^m$ points and lattice spacing~$\delta$, and let~$\tilde{\sigma}:=\sigma/\delta$.
We define the pure state 
\begin{equation}
\label{eq:lattice1DG}
    \ket{\G_\textup{lattice}(\tilde{\sigma}, \delta, m)}
    := \frac1{\tilde{\mathcal{N}}}
    \sum_{j=-2^{m-1}}^{2^{m-1}-1} \delta\, \e^{-\frac{j^2}{4\tilde{\sigma}^2}} \ket{j\delta},\quad
    \tilde{\mathcal{N}}^2:= \delta^2\sum_{j=-2^{m-1}}^{2^{m-1}-1}\e^{-\frac{j^2}{2\tilde{\sigma}^2}},
\end{equation}
for~$\ket{j\delta}$ equally spaced lattice states in one dimension, as the discrete 1DG pure state with standard deviation~$\tilde{\sigma}$ over lattice~$\mathds{L}$.
\end{definition}

We use the discrete 1DG state in~\cref{eq:lattice1DG} as a discrete approximation for the continuous 1DG state~\eqref{eq:continuous1DG} with the standard deviation~$\sigma$.
The lattice parameters, i.e., the lattice spacing and the number of lattice points, are chosen based on two given inputs:
the standard deviation and an error tolerance on the infidelity between the discrete and continuous 1DG states.
In~\cref{subsubsec:1DG_space}, we describe how these two inputs are used to calculate the lattice parameters.

% Discretizing multi-dimensional Gaussians
To discretize a continuous multi-dimensional Gaussian pure state~\eqref{eq:continuous_NDG}, first we decompose the state into a tensor product of several continuous 1DG pure states by a basis transformation.
Then we discretize each continuous 1DG pure by a discrete 1DG pure state over a lattice as per~\cref{def:discrete1DG}.
Note that a continuous $N$-dimensional Gaussian pure state~$\ket{\G_N(\upbm{A})}$~\eqref{eq:continuous_NDG} with the ICM $\upbm{A}$ is a linear combination of basis states~$\ket{\bm{x}}:=\ket{x_0} \otimes \cdots \otimes \ket{x_{N-1}}$,
where~$\bm{x}$ is a vector of real numbers.
Let~$\upbm{O}$ be a matrix such that~$\upbm{O}^\T\upbm{AO}$ is a diagonal matrix~$\upbm{D}$.
Then the basis transformation~$\ket{\bm{x}} \mapsto \ket*{\upbm{O}^{-1}\bm{x}}$ yields the continuous Gaussian state with the diagonal ICM~$\upbm{D}$,
which can be decomposed into a tensor product of~$N$ continuous 1DG states;
see~\cref{subsubsec:NDG_space}.

%%%%%%%%%%%%%%%%%%%%%%%%%%%%%%%%%%%%%%%%%%%%%%%%%%%%%%%%%
%%%%%% Methods %%%%%%%%%%%%%%%%%%%%%%%%%%%%%%%%%%%%%%%%%%
%%%%%%%%%%%%%%%%%%%%%%%%%%%%%%%%%%%%%%%%%%%%%%%%%%%%%%%%%
\subsection{Methods}
\label{subsec:methods}

In this subsection, we present our Fourier- and wavelet-based methods for generating an approximation for the free-field ground state on a quantum register.
Both methods are based on a method for generating 1DG states.
We begin in~\cref{subsubsec:1DG_method} by explaining our method for generating a discrete approximation for a continuous 1DG state.
Then we describe the Fourier-based method in~\cref{subsubsec:Fourier_method} and the wavelet-based method in~\cref{subsubsec:wavelet_method}.

%%%%%%%%%%%%%%%%%%%%%%%%%%%%%%%%%%%%%%%%%%%%%%%%%%%%%%%%%
\subsubsection{One-dimensional Gaussian-state generation}
\label{subsubsec:1DG_method}

Here we present two methods for generating a discrete approximation for a continuous 1DG state on a quantum register.
First we specify the inputs along with the task in generating a 1DG state. 
Next we describe our discrete approximation for a 1DG state.
Then we explain our first method for a 1DG-state generation.
Our first method is similar to the Kitaev-Webb method~\cite{KW09}.
However, in contrast to the Kitaev-Webb method, which is restricted to 1DG states with an extremely large standard deviation, our method generates 1DG states with any standard deviation.
We finally describe our second method for generating a 1DG state, which is based on a method for performing inequality testing~\cite{SLSB19} on a quantum computer.

% Specification of task
We begin by specifying the task in a 1DG-state generation.
To generate a continuous 1DG state~$\ket{\G(\sigma)}$~\eqref{eq:continuous1DG}, we are given two inputs:
(1)~the standard deviation~$\sigma \in \reals^+$ of the 1DG state and (2)~an error tolerance~$\varepsilon_\text{1DG} \in (0,1)$.
The task is to generate an approximate 1DG state~$\ket{\tilde{\G}(\sigma)}$ such that the infidelity~\cite{NC11}
\begin{equation}
\label{eq:infidelity}
\text{infid}\left(\ket{\G(\sigma)},\ket{\tilde{\G}(\sigma)}\right)
:=1-\braket{\G(\sigma)}{\tilde{\G}(\sigma)}
\in[0,1),
\end{equation}
between the approximate and exact states is no greater than~$\varepsilon_\text{1DG}$.
We only consider continuous 1DG states with means of zero ($\mu=0$) as we only need to prepare these states in order to generate an approximation for the ground state of the free QFT~\eqref{eq:free_Hamiltonian}.

% Approximate 1DG state
We approximate~$\ket{\G(\sigma)}$ by a discrete 1DG pure state~$\ket{\G_\text{lattice}(\tilde{\sigma},\delta,m)}$~\eqref{eq:lattice1DG} over a lattice with~$2^m$ points and lattice spacing~$\delta$ as per~\cref{def:discrete1DG}. 
We select~$m$ and~$\delta$ based on~$\sigma$ and~$\varepsilon_\text{1DG}$ such that
the infidelity between the continuous and discrete 1DG states is at most~$\varepsilon_\text{1DG}$.
Our approximate 1DG state is different from that of the Kitaev-Webb method. 
A continuous 1DG in their method is first approximated by a discrete 1DG state over an infinite lattice with unit spacing as in~\cref{eq:KW_discrete_1DG}.
The discrete 1DG state is then again approximated by the state in~\cref{eq:KW_approx_1DG} to be generated on a quantum register.
Our approximate 1DG state, however, is a discrete 1DG state over a finite lattice with a non-unit lattice spacing.

% Generating the discrete 1DG state
We now explain our strategy for generating~$\ket{\G_\text{lattice}(\tilde{\sigma},\delta,m)}$~\eqref{eq:lattice1DG} on a quantum register.
Our strategy comprises two steps.
In the first step, we generate the state~$\ket{\G_\text{lattice}(\tilde{\sigma},1,m)}$.
That is to say, we first prepare a discrete 1DG state with the same standard deviation but over a lattice with unit spacing.
This state is a linear combination of basis states~$\ket{j}$, where~$j$ is an~integer;
see~\cref{eq:lattice1DG} with~$\delta=1$.
In the second step, we transform~$\ket{\G_\text{lattice}(\tilde{\sigma},1,m)}$ to~$\ket{\G_\text{lattice}(\tilde{\sigma},\delta,m)}$ by performing the unitary map for which~$\ket{j} \mapsto \ket{j\delta}$ for all~$j$.
Our method for generating~$\ket{\G_\text{lattice}(\tilde{\sigma},1,m)}$ is similar to the Kitaev-Webb method for 1DG-state generation.
We write a recursive decomposition for~$\ket{\G_\text{lattice}(\tilde{\sigma},1,m)}$
analogous to~\cref{eq:KW_recursive_formula} and employ the recursive decomposition to design an iterative algorithm for generating this state.
See~\cref{subsubsec:1DG_generation} for a detailed description of the algorithm.

% inequality-testing method
We now describe our inequality-testing-based method for generating the state~$\ket{\G_\text{lattice}(\tilde{\sigma},1,m)}$~\eqref{eq:lattice1DG}.
To elucidate the method, we write this state
as~$\sum_j g(j)\ket{j}$ with unnormalized amplitude distribution $g(j):=\exp\left(-j^2/4\sigma^2\right)$.
To generate this state by inequality testing, first we prepare a quantum state with amplitude according to the value of~$j$ rounded down to the nearest power of~$2$.
Specifically, we first prepare the state $\sum_j g_\text{round}(j) \ket{j}$ with unnormalized amplitude distribution $g_\text{round}(j):=\exp\left(2^{2\floor{\log_2j}}/4\sigma^2\right)$ for all $j \neq 0$ and $g_\text{round}(0):=1$.
The amplitude distributions~$g(j)$ and~$g_\text{round}(j)$ are shown in~\cref{fig:1DG_by_ineq_testing} by blue and orange points, respectively;
these distributions are only shown for non-negative~$j$ for simplicity.

\begin{figure}[htb]
    \centering
    \includegraphics[width=.445\textwidth]{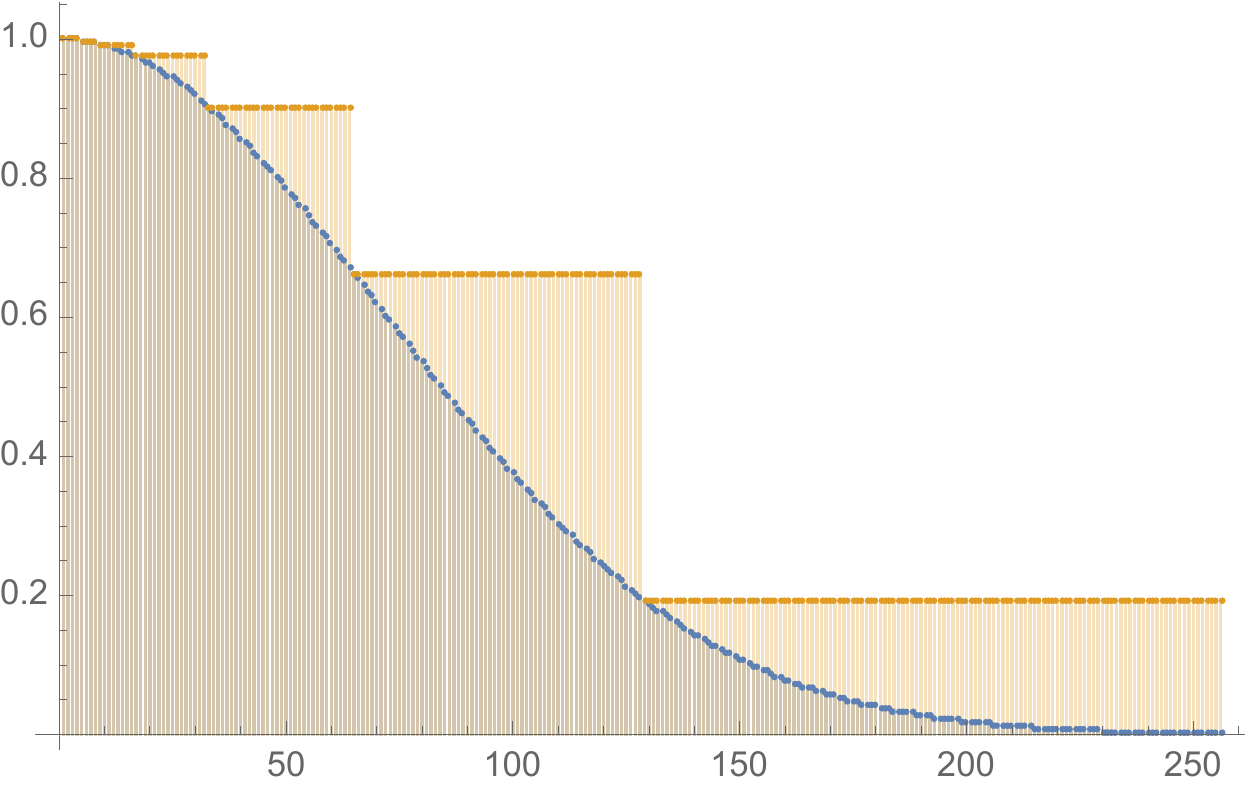}
    \caption[Inequality test for 1DG-state generation]{
    Illustration of state-generation steps before performing an inequality test.
    First we prepare a state with amplitudes according to the orange points.
    Then we test an inequality against the blue points.
    The success probability is at least about~$70\%$.}
    \label{fig:1DG_by_ineq_testing}
\end{figure}

Once the state~$\sum_j g_\text{round}(j) \ket{j}$ is generated, we then coherently compute an approximation for the ratio of the amplitudes $r_j:=g_\text{round}(j)/g(j)$ into a scratch register and prepare a reference quantum register in uniform superposition.
Next we perform an inequality test between the value encoded in the scratch register and the value encoded in the reference register,
and write the result into a flag qubit.
We then erase the reference register and measure the flag qubit.
If the post-measurement state of the flag qubit is $\ket{0}$, then the state generated on the first register is the desired 1DG state.
See~\cref{subsubsec:inequality_testing} for a detailed description of the inequality-testing-based algorithm for generating a 1DG state.

%%%%%%%%%%%%%%%%%%%%%%%%%%%%%%%%%%%%%%%%%%%%%%%%%%%%%%%%%
\subsubsection{Fourier-based method for ground-state generation}
\label{subsubsec:Fourier_method}

We now present our Fourier-based method for generating a discrete approximation for the free-field ground state on a quantum register.
First we explain the rationale for specifying the task in this method and describe the task.
Then we explain our strategy for generating the approximate ground in the Fourier-based method.

% Rationale for task specification
In the Fourier-based method, we discretize the continuum QFT in a fixed-scale wavelet basis and use the ground state of the discretized QFT as an approximation for that of the continuum theory.
The discretized-QFT ground state in this method is a continuous Gaussian state whose ICM~$\upbm{A}_\sS$~\eqref{eq:fixedscale_groundstate} is a circulant matrix.
The ICM, which fully describes the discretized-QFT ground state, is specified by three parameters:
wavelet index~$\dbIndex \in \integers^+$, free-QFT mass~$ m_0 \in \reals^+$ and the number of modes~$N \in \integers^+$ of the discretized QFT.

% Specification of task
For the  Fourier-based method to generate the ground state, we are given an error tolerance~$\varepsilon_\text{vac} \in (0,1)$ along with the parameters that specify the ground-state ICM. 
The task is to generate an approximation for the ground state~$\ket*{\G_\text{scale}}$~\eqref{eq:fixedscale_groundstate} of the discretized QFT in a fixed-scale basis on a quantum register. 
The infidelity between the approximate and exact states is required to be no greater than the error tolerance~$\varepsilon_\text{vac}$.

% Strategy
Our strategy for generating an approximate ground state for the discretized QFT is as follows.
First we construct a classical algorithm for computing the eigenvalues $\bm{\lambda}:=(\lambda_0,\ldots,\lambda_{N-1})$ of the ground-state ICM; eigenvalues are diagonals of
 the diagonal matrix~$\bm{\Lambda}$ in the spectral decomposition of the ICM.
In this algorithm, we exploit the circulant structure of the ICM and compute~$\bm{\lambda}$ by a discrete Fourier transform~\cite[p.~100]{HJ12}.
Next we use~$\bm{\lambda}$ as a classical input to design a quantum circuit for generating an approximation for the state~$\ket{\G_N(\bm{\Lambda})}$~\eqref{eq:continuous_NDG}, i.e., the continuous Gaussian state whose ICM is the diagonal~matrix~$\bm{\Lambda}$.
Finally, we perform a basis transformation by a quantum fast Fourier transform~(QFFT) to map the state~$\ket{\G_N(\bm{\Lambda})}$ to the ground state~$\ket{\G_N(\upbm{A}_\sS)}$.
The state~$\ket{\G_N(\bm{\Lambda})}$ is a linear combination of basis states~$\ket{\bm{x}}:=\ket{x_1} \otimes \cdots \otimes \ket{x_N}$,
where~$\bm{x}=(x_1,\ldots, x_N)$ is a vector of real numbers;
see~\cref{eq:continuous_NDG}.
The QFFT implements the map~$\ket{\bm{x}} \mapsto \ket{\upbm{F}\bm{x}}$ on a quantum computer, where~$\upbm{F}$ is the transformation matrix for the discrete Fourier transform.

% Preparing the diagonal Gaussian state
To design a quantum circuit to preparing an approximation for a continuous $N$DG state~$\ket{\G_N(\upbm{D})}$ with a diagonal inverse-covariance matrix~$\upbm{D}:=\text{diag}(d_0,\ldots,d_{N-1})$,
first we decompose the state
as~$\ket{\G(\sigma_0)} \otimes \cdots \otimes \ket{\G(\sigma_{N-1})}$,
where~$\ket{\G(\sigma_\ell)}$~\eqref{eq:continuous1DG}~is a continuous 1DG state with the standard deviation~$\sigma_\ell:=1/\sqrt{d_\ell}$.
By the method described in~\cref{subsubsec:1DG_method}, we then design a quantum circuit for generating a discrete approximation for each 1DG state.
The combined output of all these quantum circuits is an approximation for the continuous Gaussian state with the diagonal ICM~$\upbm{D}$.

%%%%%%%%%%%%%%%%%%%%%%%%%%%%%%%%%%%%%%%%%%%%%%%%%%%%%%%%%
\subsubsection{Wavelet-based method for ground-state generation}
\label{subsubsec:wavelet_method}

Here we present our wavelet-based method for generating a discrete approximation for the free-field ground state represented in a multi-scale wavelet basis.
Similar to the Fourier-based method, we first provide the rationale for specifying the task in the wavelet-based method and describe the task.
Then we explain the strategy for generating the approximate ground~state.

% Rationale for task specification and specification of task
In the wavelet-based method, we discretize the continuum QFT in a multi-scale wavelet basis and use the ground state of the discretized QFT as an approximation for that of the continuum theory.
The ground state~\eqref{eq:multiscale_groundstate} of the discretized QFT in this method is also a continuous Gaussian state.
The same parameters specify the ground-state ICM here as in the Fourier-based method.
For the wavelet-based method, we are given the same classical inputs as the Fourier-based method.
The task, however, is to generate an approximation for the ground
state~$\ket*{\G_\text{wavelet}}$~\eqref{eq:multiscale_groundstate} of the discretized QFT in a multi-scale wavelet basis such that the infidelity between the approximate and exact states is no greater than the error tolerance~$\varepsilon_\text{vac}$.

% Strategy
The strategy for generating an approximate ground state in the wavelet-based method is as follows.
The ground-state ICM in a multi-scale wavelet basis has many near-zero elements.
We approximate this matrix by replacing its near-zero elements with exactly zero.
Specifically, we replace all matrix elements whose magnitude are less than the threshold value~$\varepsilon_\text{th}=m_0\varepsilon_\text{vac} N^{-3/2}$ with exactly zero.
This approximation enables a fingerlike sparse structure~\cite{Bey92} for the ground-state ICM with a quasilinear number of nonzero elements; see~\cref{fig:truncICM}.

We exploit the fingerlike structure and perform the UDU matrix decomposition of the approximate ICM in a quasilinear time.
In the UDU decomposition, we decompose the fingerlike sparse matrix~$\tilde{\upbm{A}}$ as the product of an upper unit-triangular matrix~$\upbm{U}$, a diagonal matrix~$\upbm{D}$ and transpose of~$\upbm{U}$.
We compute diagonals of~$\upbm{D}$ and shear elements, i.e., nonzero off-diagonal elements, of~$\upbm{U}$ by a classical algorithm, and use them as classical inputs to construct a quantum circuit for generating an approximation for the free-field ground state.

To generate an approximate ground state, first we construct a quantum circuit to preparing an approximation for the $N$DG state~$\ket{\G_N(\upbm{D})}$ whose ICM is the diagonal matrix~$\upbm{D}$ in the UDU decomposition.
The diagonals of~$\upbm{D}$ are used as classical inputs, and the quantum circuit is constructed by the method described in~\cref{subsubsec:Fourier_method}.
Then we transform~$\ket{\G_N(\upbm{D})}$ into the ground state by performing a quantum shear transform~(QST) on a quantum computer;
the shear elements of~$\upbm{U}$ are classical inputs for the QST.
Akin to the QFFT in the Fourier-based method, the QST executes a basis transformation.
Specifically, the QST implements the map~$\ket{\bm{x}} \mapsto \ket{\upbm{S}\bm{x}}$ on a quantum computer, where the shear-transform matrix~$\upbm{S}$ is inverse-transpose of~$\upbm{U}$.

%%%%%%%%%%%%%%%%%%%%%%%%%%%%%%%%%%%%%%%%%%%%%%%%%%%%%%%%%
%%%%%% Results %%%%%%%%%%%%%%%%%%%%%%%%%%%%%%%%%%%%%%%%%%
%%%%%%%%%%%%%%%%%%%%%%%%%%%%%%%%%%%%%%%%%%%%%%%%%%%%%%%%%

\section{Results}
\label{sec:results}

In this section, we present our main results.
We first construct a high-level description of our ground-state-generation algorithms in~\cref{subsec:highlevel_description}.
Next we discuss the number of required qubits for representing an approximation for the discretized-QFT ground state in the Fourier- and wavelet-based methods in~\cref{subsec:space_requirement}.
Our algorithms have classical preprocessing and quantum routine. 
We present the classical preprocessing of our algorithms in~\cref{subsec:classical_preprocessing} and the quantum routine in~\cref{subsec:quantum_algorithms}.
Then we analyze the runtime of our algorithms in~\cref{subsect:complexity_analysis}.
In~\cref{subsec:lowebound}, we establish a lower bound on the gate complexity for ground-state generation.
Finally, in~\cref{subsec:Fouriervswavelet}, we compare the Fourier vs wavelet approach for ground-state generation.

%%%%%%%%%%%%%%%%%%%%%%%%%%%%%%%%%%%%%%%%%%%%%%%%%%%%%%%%%
\subsection{High-level description of our algorithms for ground-state generation}
\label{subsec:highlevel_description}

We begin with a high-level description of our two algorithms for ground-state generation.
The first of these algorithms, described in \cref{subsubsec:highLevel_FBA},
is based on the use of a discrete Fourier transform,
and we refer to it as the Fourier-based algorithm.
The second algorithm, described in~\cref{subsubsec:highLevel_WBA},
is based on the use of a wavelet transform, and we call it the wavelet-based algorithm.

Both algorithms have a similar structure that we now explain.
The algorithms proceed in two stages:
(1)~prepare several independent one-dimensional Gaussian states, and
(2)~perform a collection of arithmetic operations on those Gaussian states.
Both stages require a certain amount of classical information,
much of which is not provided directly by the main server~(\cref{subsubsec:framework}) but
requires a non-negligible amount of computation to produce.
We therefore must analyze not only the quantum complexity
but also the classical complexity of our algorithms in order to
ensure that the resulting procedures are indeed quasilinear in
the number of modes of the discretized QFT.
We refer to the classical part of our algorithm as the
`classical preprocessing' step, as it must be carried out prior
to the execution of our quantum algorithms.
In our descriptions of the Fourier-based and wavelet-based
algorithms, we therefore distinguish between the classical 
preprocessing procedure and the quantum algorithm~itself.

%%%%%%%%%%%%%%%%%%%%%%%%%%%%%%%%%%%%%%%%%%%%%%%%%%%%%%%%%
\subsubsection{High-level description of Fourier-based algorithm}
\label{subsubsec:highLevel_FBA}

Here we construct a high-level description of the Fourier-based algorithm for ground-state generation. 
We begin by explaining this algorithm's classical preprocessing and then describe the quantum routine. 
Finally, we present the algorithm by pseudocode to elucidate the inputs, output and procedure of the Fourier-based algorithm.

% Classical preprocessing
The Fourier-based algorithm generates an approximation for the free-field ground state~\eqref{eq:fixedscale_groundstate} represented in a fixed-scale~basis.
To generate this state, first we generate~$N$ discrete 1DG states over a lattice with spacing~$\delta$.
The task in the classical preprocessing of the Fourier-based algorithm is to compute the standard deviations~$\tilde{\bm{\upsigma}}:=(\tilde{\sigma}_0,\ldots,\tilde{\sigma}_{N-1})$ for the discrete 1DG states as per~\cref{def:discrete1DG}, and the lattice spacing~$\delta$;
these are the needed parameters for the quantum routine in the Fourier-based algorithm.
To compute~$\tilde{\bm{\upsigma}}$ and~$\delta$, first we compute the derivative overlaps~\eqref{eq:derivative_overlaps} for the second-order derivative, i.e., Laplace, operator.
We then use these derivative overlaps and the bare mass~$m_0$ to compute the eigenvalues~$\bm{\lambda}:=(\lambda_0,\ldots,\lambda_{N-1})$ of the ground-state ICM.
Having~$\bm{\lambda}$, we compute the lattice spacing as~$\delta=1/\sqrt{\lambda_{\max}}$ and the 1DG standard deviations as~$\tilde{\bm{\upsigma}}=1/(\delta\sqrt{\bm{\uplambda}})$.
Figure~\ref{fig:FBA}~(top) shows a schematic description of the classical preprocessing in the Fourier-based algorithm.

% Quantum routine
In the quantum routine, we use outputs of the classical preprocessing to generate an approximation for the ground state. 
For each component of~$\tilde{\bm{\upsigma}}$ and the lattice spacing~$\delta$, we generate a discrete 1DG state~\eqref{eq:lattice1DG} corresponding to these inputs on a quantum register. 
Then we execute a quantum fast Fourier transform~(QFFT).
The QFFT performs a discrete Fourier transform by a collection of arithmetic operations on the 1DG states.
The resulting state is an approximation for the ground state~\eqref{eq:fixedscale_groundstate} represented in a fixed-scale basis.
Figure~\ref{fig:FBA}~(bottom) shows a schematic description of the quantum routine in the Fourier-based algorithm.
For clarity, we present the inputs, outputs and procedure of this algorithm as pseudocode in~\cref{alg:FBA}.

\begin{figure}[htb]
\centering
\includegraphics[width=.9\textwidth]{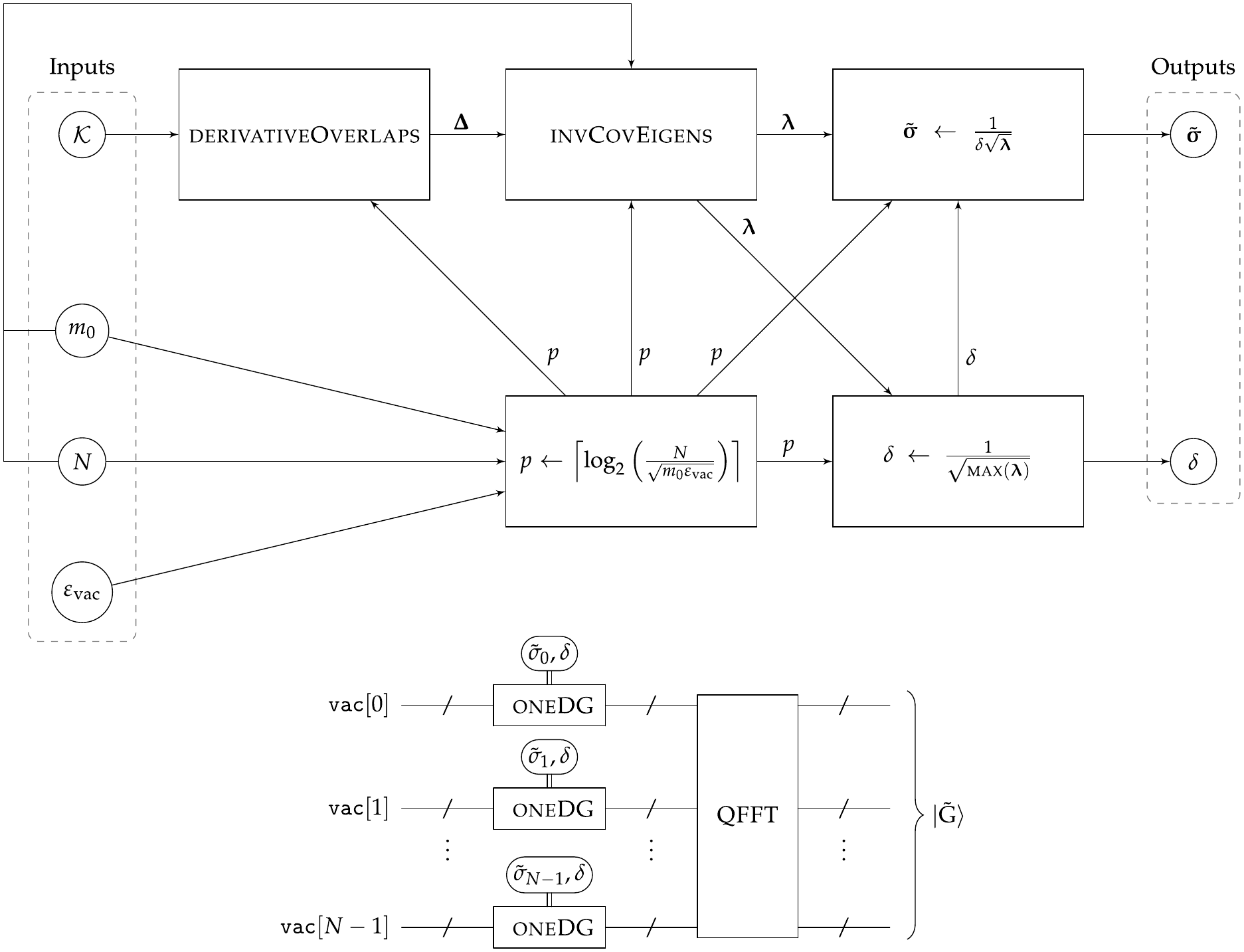}
\caption[Description of the Fourier-based algorithm for ground-state generation]{
    Description of Fourier-based algorithm for ground-state generation.
    Top:~classical preprocessing.
    Inputs: wavelet index~$\dbIndex$, mass~$m_0$,
    number of modes~$N$ in the discretized QFT and error tolerance~$\varepsilon_\text{vac}$ for the output state.
    Each box represents a process; incoming arrows identify the inputs, and outgoing arrows identify the outputs of the process.
    Outputs: standard deviation~$\tilde{\bm{\upsigma}}$ of the approximate one-dimensional Gaussian~(1DG) states and lattice spacing~$\delta$.
    Outputs of intermediate processes:
    working precision~$p$,
    second-order derivative overlaps~$\bm\Delta$ and eigenvalues~$\bm\uplambda$ of the ground-state ICM.
    Bottom:~quantum routine.
    Double lines indicate classical inputs to the quantum routine.
    \vac\ is a quantum register with~$N$ cells, and each cell comprises~$p$ qubits;
    \qwire\ represents multiple qubits.
    Each \textsc{oneDG} accepts~$\delta$ and one component of~$\tilde{\bm{\upsigma}}$ as classical inputs and generates an approximate 1DG state corresponding to these inputs on once cell of \vac.
    The quantum fast Fourier transform~(\textsc{QFFT}) acts collectively on the set of approximate 1DG states and transforms them into the approximate ground state~$\ket*{\tilde{\G}}$.
\label{fig:FBA}
}
\end{figure}

\begin{algorithm}[H]
  \caption{Fourier-based algorithm for ground-state generation}
  \label{alg:FBA}
  \begin{algorithmic}[1]
    \Require{
    \Statex $\dbIndex \in \integers_{\geq 3}$
    \Comment{wavelet index}
    \Statex $m_0\in\reals^+$
    \Comment{free-QFT mass}
    \Statex $N \in \integers_{\geq 2(2\dbIndex-1)}$
    \Comment{number of modes}
    \Statex $\varepsilon_\text{vac} \in (0,1)$
    \Comment{error tolerance for output state}
            }
  \Ensure
    \Statex $\ket{\tilde{\G}} \in \mathscr{H}_2^{N\times \ceil*{\log_2{\left(N/\sqrt{m_0\varepsilon_\text{vac}}\right)}} }$ 
    \Comment{$\left(N\times \ceil*{\log_2{\left(N/\sqrt{m_0\varepsilon_\text{vac}}\right)}}\right)$-qubit approximate ground state}
  
\Function{fourierBasedGSG}{$\dbIndex, m_0, N,\varepsilon_\text{vac}$}
       \Statex \textbf{Classical preprocessing}
       \State $\integers^+ \ni p \gets
       \ceil*{\log_2{
       \left(N/\sqrt{m_0\varepsilon_\text{vac}}\right)}
       }$
       \Comment{computes working precision~$p$}
       \State $\reals^{2\dbIndex-1} \ni \bm{\Delta} \gets \textsc{derivativeOverlaps}(\dbIndex,p)$
       \Comment{computes derivative overlaps~\eqref{eq:derivative_overlaps} for Laplace operator in index-\dbIndex\ wavelet basis}
       \State $ \reals^N \ni \bm{\uplambda} \gets \textsc{invCovEigens}(\dbIndex, m_0, N, \bm{\Delta}, p)$
       \Comment{computes eigenvalues of the ICM by~\cref{alg:invCovEigens}}      
       \State $\reals^+ \ni \delta \gets 1/\sqrt{\textsc{max}(\bm{\uplambda})}$
       \Comment{computes lattice spacing~$\delta$}
       \State $\reals^N \ni \tilde{\bm{\upsigma}} \gets 1/(\delta\sqrt{\bm{\uplambda}})$
       \Comment{computes standard deviation~$\tilde{\bm{\upsigma}}$ of approximate 1DG states}
    \Statex \textbf{Quantum routine}
       \For{$\ell \gets 0$ to $N-1$}
        \State $\mathscr{H}_2^p \ni \vac[\ell] \gets  \textsc{oneDG}~(\tilde{\sigma}_\ell, \delta)\ket{0^p}$
        \Comment{generates $p$-qubit approximate 1DG state
                with standard deviation $\tilde{\sigma}_\ell$ by Algorithm~\ref{alg:1DG}}
       \EndFor
    \State $\mathscr{H}_2^{N\times p}\ni \vac
        \gets \textsc{QFFT}\left(N, \bigotimes\limits_{\ell=0 }^{N-1}\vac[\ell]\right)$
    \Comment{transforms the 1DG~states into the approximate ground state by~\cref{alg:QFHT}}
    \State \Yield \vac
    \Comment{
    we use \Yield instead of \Return for quantum algorithms as the output is a state generated on a quantum register}
\EndFunction
     \end{algorithmic}
\end{algorithm}

%%%%%%%%%%%%%%%%%%%%%%%%%%%%%%%%%%%%%%%%%%%%%%%%%%%%%%%%%
\subsubsection{High-level description of wavelet-based algorithm}
\label{subsubsec:highLevel_WBA}

We now present a high-level description of the wavelet-based algorithm for ground-state generation.
This algorithm, similar to the Fourier-based algorithm, has classical preprocessing and quantum routine.
We begin by describing the classical preprocessing and proceed with explaining the quantum routine. 
We finally present the wavelet-based algorithm as pseudocode to specify the algorithm's inputs, output and procedure.

% Classical preprocessing
In classical preprocessing of the wavelet-based algorithm, we compute the required inputs for the quantum routine.
These inputs are the lattice spacing~$\delta$, the vector~$\tilde{\bm{\upsigma}}$ of standard deviations for the discrete 1DG states and shear elements of the upper unit-triangular matrix~$\upbm{U}$ in the UDU decomposition of the approximate ICM~(\cref{subsubsec:wavelet_method}) for the free-field ground state~\eqref{eq:multiscale_groundstate} represented in a multi-scale wavelet basis.
The first two inputs are needed to generate the discrete 1DG states, and the last input is needed to perform the basis transformation.

The needed classical inputs for quantum routine of the wavelet-based algorithm are computed as follows.
First we compute the second-order derivative overlaps~\eqref{eq:derivative_overlaps}.
These derivative overlaps are then used to compute the unique matrix elements of the ground-state ICM.
The ICM~\eqref{eq:multiscale_ICM} is a block matrix, and each block is a circulant matrix.
The unique matrix elements are, therefore, the circuit row of each block in the block matrix.
Next we use the circulant rows, denoted by~$\upbm{a}$, to compute the vector~$\upbm{d}$ of diagonals in the diagonal matrix~$\upbm{D}$ and the shear elements~$\upbm{S}$ of the upper unit-triangular matrix~$\upbm{U}$ in the UDU decomposition of the approximate ICM.
Having~$\upbm{d}$, we compute the lattice spacing as~$\delta=1/\sqrt{d_{\max}}$ and the vector of standard deviations as~$\tilde{\bm{\upsigma}}=1/(\delta\sqrt{\upbm{d}})$,
where~$d_{\max}$ is the largest element of~$\upbm{d}$.
Figure~\ref{fig:WBA}~(top) shows a schematic description of the classical preprocessing in the wavelet-based algorithm.

% Quantum routine
The quantum routine of the wavelet-based algorithm proceeds as follows.
Similar to the quantum routine of the Fourier-based algorithm, first we generate a discrete 1DG state~\eqref{eq:lattice1DG} for each component of~$\tilde{\bm{\upsigma}}$ and the lattice spacing~$\delta$.
Using the shear elements~$\upbm{S}$ as classical input, we then perform a basis transformation on the discrete 1DG states by executing a quantum shear transform~(QST).
The QST performs a collection of arithmetic operations to map the discrete 1DG states to an approximation for the ground state~\eqref{eq:multiscale_groundstate} of the discretized QFT in a multi-scale wavelet basis.
A schematic description of the quantum routine in the wavelet-based algorithm is shown in~\cref{fig:WBA}~(bottom).
The pseudocode in Algorithm~\ref{alg:WBA} specifies the inputs, output and procedure of the wavelet-based algorithm.

\begin{algorithm}[H]
  \caption{Wavelet-based algorithm for ground-state generation}
  \label{alg:WBA}
  \begin{algorithmic}[1]
    \Require{
    \Statex $\dbIndex \in \integers_{\geq 3}$
    \Comment{wavelet index}
    \Statex $m_0\in\reals^+$
    \Comment{free-QFT mass}
    \Statex $N \in \integers_{\geq 2(2\dbIndex-1)}$
    \Comment{number of modes}
    \Statex $\varepsilon_\text{vac}\in (0,1)$
    \Comment{error tolerance for output state}
            }
  \Ensure
    \Statex $\ket{\tilde{\G}} \in \mathscr{H}_2^{N\times \ceil*{\log_2{\left(N/\sqrt{m_0\varepsilon_\text{vac}}\right)}} }$ 
    \Comment{$\left(N\times \ceil*{\log_2{\left(N/\sqrt{m_0\varepsilon_\text{vac}}\right)}}\right)$-qubit approximate ground state}
  
\Function{waveletBasedGSG}{$\dbIndex, m_0, N,\varepsilon_\text{vac}$}
    \Statex \textbf{Classical preprocessing}
    \State $\integers^+ \ni p \gets \ceil*{\log_2{\left(N/\sqrt{m_0\varepsilon_\text{vac}}\right)}}$
       \Comment{computes working precision~$p$}
       \State $\reals^{2\dbIndex-1} \ni \bm{\Delta} \gets \textsc{derivativeOverlaps}(\dbIndex,p)$
       \Comment{computes the derivative overlaps~\eqref{eq:derivative_overlaps} for Laplace operator by Algorithm~\ref{alg:derivativeOverlaps}}
       \State $\upbm{a}:=\cBraket{\left.
  \upbm{a}^{(s_0)}_\sS\in \reals^{2^{s_0}},
  \upbm{a}^{(s_0, c)}_\sw\in \reals^{2^c},
  \upbm{a}^{(r, c)}_\ww\in \reals^{2^{c-r}}
  \right\vert
  s_0\leq r \leq c < k}
  \gets \textsc{invCovCircRows} \left(\dbIndex, m_0, N, \bm{\Delta}, \varepsilon_\text{vac}, p\right)$
  \Comment{computes circulant row in main and upper-diagonal blocks of multi-scale ICM~\eqref{eq:multiscale_ICM} by~\cref{alg:invCovCircRows}; here $s_0:=\ceil*{\log_2 (4\dbIndex-2)}$ and $k:= \log_2 N$}
    \State $\left\{\reals^N \ni \upbm{d}, \reals^{\softO{N}}\ni\upbm{S}\right\}
    \gets \textsc{invCovUDU}
    \left(m_0, N, \varepsilon_\text{vac},
    \upbm{a}
    \right)$
   \Comment{computes diagonals $\upbm{d}$ of $\upbm{D}$ and shear elements $\upbm{S}$ of~$\upbm{U}$ in UDU decomposition of approximate ICM by~\cref{alg:invCovUDU}.
   Here~$\softO{N}$ denotes a quasilinear number; see~\cref{subsubsec:UDUDecomp} and~\cref{alg:invCovUDU} for details
   }
   \State $\reals^+ \ni \delta \gets 1/\sqrt{\textsc{max}(\upbm{d})}$
       \Comment{computes lattice spacing~$\delta$}
       \State $\reals^N \ni \tilde{\bm{\upsigma}} \gets 1/(\delta\sqrt{\upbm{d}})$
       \Comment{computes standard deviation~$\tilde{\bm{\upsigma}}$ of approximate 1DG states}
    \Statex \textbf{Quantum routine}
       \For{$\ell \gets 0$ to $N-1$}
        \State $\mathscr{H}_2^p \ni \vac[\ell] \gets  \textsc{oneDG}~(\tilde{\sigma}_\ell, \delta)\ket{0^p}$
        \Comment{generates $p$-qubit approximate 1DG~state
                with standard deviation $\tilde{\sigma}_\ell$ by~\cref{alg:1DG}}
       \EndFor
    \State $\mathscr{H}_2^{N\times p}\ni \vac
        \gets \textsc{QST}\left(N, \upbm{S}, \bigotimes\limits_{\ell=0 }^{N-1}\vac[\ell]\right)$
    \Comment{transforms the 1DG~states into the approximate ground state by~\cref{alg:QST}}
    \State \Yield \vac
\EndFunction
     \end{algorithmic}
\end{algorithm}

\begin{figure}[htb]
\centering
\includegraphics[width=.98\textwidth]{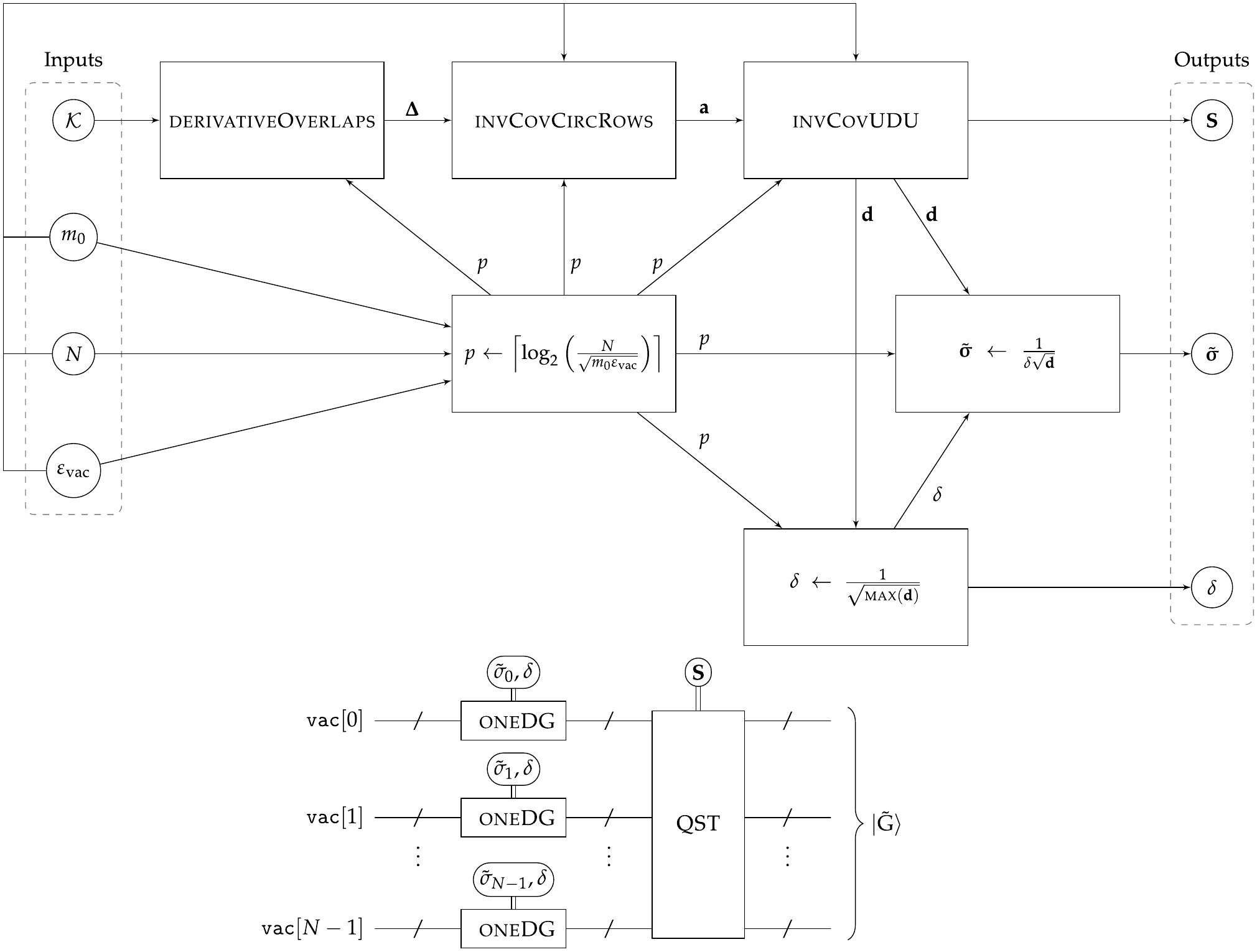}
\caption[Description of the wavelet-based algorithm for ground-state generation]{
    Description of wavelet-based algorithm for ground-state generation.
    Top:~classical preprocessing.
    Inputs are the same as the inputs to the Fourier-based algorithm
    Outputs: shear elements~$\upbm{S}$ of the upper unit-triangular matrix in the UDU decomposition of the approximate ICM,
    standard deviation~$\tilde{\bm{\upsigma}}$ of the approximate one-dimensional Gaussian~(1DG) states and lattice spacing~$\delta$. 
    Outputs of intermediate processes:
    working precision~$p$,
    second-order derivative overlaps~$\bm\Delta$,
    circulant rows~$\upbm{a}$ of the upper-triangular blocks in the ICM and diagonals~$\upbm{d}$ of the diagonal matrix in the UDU decomposition.
    Bottom:~quantum routine.
    Double lines indicate classical inputs to the quantum~routine.
    \vac\ is a quantum register with~$N$ cells, and each cell comprises~$p$ qubits;
    \qwire\ represents multiple qubits.
    Each \textsc{oneDG} accepts~$\delta$ and one component of~$\tilde{\bm{\upsigma}}$ as classical inputs and generates an approximate 1DG state corresponding to these inputs on one cell of \vac.
    The quantum shear transform~(\textsc{QST}) acts collectively on the set of approximate 1DG states and transforms them into the approximate ground state~$\ket{\tilde{\G}}$.
\label{fig:WBA}
}
\end{figure}

%%%%%%%%%%%%%%%%%%%%%%%%%%%%%%%%%%%%%%%%%%%%%%%%%%%%%%%%%
\subsection{Space requirement to represent the ground state}
\label{subsec:space_requirement}

In this subsection, we determine how the number of qubits required to represent an approximation for the free-field ground state in both Fourier- and wavelet-based methods scales with respect to the basis-independent input parameters in~\cref{table:inputsgroundstategen} specified by the main server for ground-state generation.
The wavelet index~\dbIndex\ in~\cref{table:inputsgroundstategen} is basis-dependent and is typically chosen to be a small constant number~\cite{BRSS15,SB16,BP13}.
Hence we exclude this parameter in our analysis.

We begin in~\cref{subsubsec:1DG_space} by discussing the number of required qubits to represent a discrete approximation for a continuous 1DG state~\eqref{eq:continuous1DG}.
Having determined the space required to represent a 1DG state, we then analyze the space required to represent an approximation for a continuous multi-dimensional Gaussian state in~\cref{subsubsec:NDG_space}. 
Our result on the space requirement for representing a multi-dimensional Gaussian state allows us to show how the space required to represent the free-field ground state scales in terms of the inputs specified by the main server.

%%%%%%%%%%%%%%%%%%%%%%%%%%%%%%%%%%%%%%%%%%%%%%%%%%%%%%%%%
\subsubsection{Space requirement to represent a one-dimensional Gaussian state}
\label{subsubsec:1DG_space}

Here we determine the minimal number of qubits required to represent an approximation for a continuous 1DG state~\eqref{eq:continuous1DG} in terms of its standard deviation and an error tolerance on the infidelity~\eqref{eq:infidelity} between the approximate and continuous states.
First we explain how we approximate a continuous 1DG state.
Then we establish a bound on the infidelity between our approximate 1DG state and the continuous 1DG state.
Having this bound, we then show that the number of qubits needed~to represent the approximate 1DG state is logarithmic in the ratio of the standard deviation to the square root of the error tolerance.

% Approximate 1DG
We begin by discussing how we approximate a continuous 1DG pure state.
We are given the standard deviation~$\sigma\in \reals^+$ of a continuous 1DG pure state and an error tolerance~$\varepsilon_\oneDG \in (0,1)$.
We approximate the continuous 1DG pure state by a discrete 1DG pure state over a one-dimensional lattice with~$2^m$ points and lattice spacing~$\delta$ as per~\cref{def:discrete1DG}.
We determine~$\delta$ and~$m$ in terms of~$\sigma$ and~$\varepsilon_\oneDG$ such that the infidelity~\eqref{eq:infidelity} between the discrete and continuous 1DG states is at most~$\varepsilon_\oneDG$.
We show, in~\cref{prop:1DG_fidelity}, how to determine $\delta$ and $m$ in terms of $\sigma$ and $\varepsilon_\oneDG$.

\begin{proposition}
\label{prop:1DG_fidelity}
Let $\ket{\G(\sigma)}$~\eqref{eq:continuous1DG} be a continuous 1DG state with the standard deviation~$\sigma \in \reals^+$ and
let~$\ket{\G_\textup{lattice}(\tilde{\sigma}, \delta, m)}$~\eqref{eq:lattice1DG} be a discrete 1DG state with standard deviation~$\tilde{\sigma}:=\sigma/\delta$ over a one-dimensional lattice with~$2^m$ points and lattice spacing~$\delta \in \reals^+$.
For~$\varepsilon_\oneDG \in (0,1)$, if
\begin{equation}
\label{eq:lattice_paramtrs}
    \delta\leq \min(1/2,\sigma) \quad \text{and} \quad
    2^m \delta\geq 2\sigma/\sqrt{\varepsilon_\oneDG},
\end{equation}
then the infidelity~\eqref{eq:infidelity} between the discrete and continuous 1DG states is bounded above by~$\varepsilon_\oneDG$.
\end{proposition}

\begin{proof}
Let $J:=2^m$.
Then the fidelity between the continuous~\eqref{eq:continuous1DG} and discrete~\eqref{eq:lattice1DG} 1DG states is
\begin{equation}
\label{eq:1DGfidelityBound}
    F:=\braket{\G_\text{lattice}(\tilde{\sigma}, \delta, m)}{\G(\sigma)}
    = \frac1{\mathcal{N}\tilde{\mathcal{N}}} \sum_{j=-J/2}^{J/2-1} \delta\, \e^{-\frac{j^2}{2\tilde{\sigma}^2}}
    \geq \int_{-J\delta/2}^{J\delta/2} \dd{x} \rho(x)
    =\text{Pr}\left(\abs{x}\leq J\delta/2\right),
\end{equation}
where the inequality follows from~\cref{prop:sum_integral_relation}.
Here $\rho(x):=\mathcal{N}^{-2} \e^{-x^2/(2\sigma^2)}$ is the probability density function of a Gaussian distribution and~$\text{Pr}\left(\abs{x}\leq J\delta/2\right)$ is the probability that~$x \in \left[-J\delta/2, J\delta \right/2]$.
By~\cref{eq:1DGfidelityBound},~$J\delta/2 \geq
\sigma/\sqrt{\varepsilon_\oneDG}$~\eqref{eq:lattice_paramtrs}
and~$\text{Pr}\left(\abs{x} \leq J\delta/2\right)=1-\text{Pr}\left(\abs{x}> J\delta/2\right)$, we have
\begin{equation}
\label{eq:1DGfidelityBound2}
    F\geq 1-\text{Pr}\left(
    \abs{x} \geq \sigma/\sqrt{\varepsilon_\oneDG}
    \right).
\end{equation}
Using Chebyshev's inequality~\cite[p.~609]{NC11},
\begin{equation}
\label{eq:Chebyshev_inequality}
    \text{Pr}\left(\abs{x} \geq \sigma/\sqrt{\varepsilon_\oneDG}\right)
    \leq \varepsilon_\oneDG,
\end{equation}
which together with~\cref{eq:1DGfidelityBound2} yield~$F\geq 1-\varepsilon_\oneDG$.
The infidelity~$1-F$ is therefore bounded from above by~$\varepsilon_\oneDG$.
\end{proof}

% Space requirement for 1DG state
\cref{prop:1DG_fidelity} allows us to obtain the minimal number of qubits to represent the discrete 1DG state~\eqref{eq:lattice1DG}.
We provide the result for the minimal number of qubits in the following proposition and proceed with a proof.

\begin{proposition}
\label{prop:1DG_qubits}
For~$\sigma \in \reals^+$ the standard deviation of a continuous 1DG state~\eqref{eq:continuous1DG} and $\varepsilon_\oneDG \in (0,1)$ an error tolerance, at~least
\begin{equation}
\label{eq:discrete1DG_qubits}
    n_\oneDG= \ceil{\log_2\left(\sigma/\sqrt{\varepsilon_\oneDG}\right)} + \max\left(1,\, \ceil{\log_2(1/\sigma)}\right),
\end{equation}
qubits are required to represent a discrete approximation for the continuous state with infidelity~\eqref{eq:infidelity} no greater than~$\varepsilon_\oneDG$.
\end{proposition}

\begin{proof}
The discrete 1DG state~\eqref{eq:lattice1DG} is a superposition of lattice states~$\ket{j\delta}$, where $j\in [-2^{m-1},2^{m-1})\cap \integers$ and $j\delta$ is a real number.
Hence the smallest nonzero value for the real numbers is~$\delta$,
and the largest value is~$2^{m-1}\delta$. 
By virtue of~\cref{prop:1DG_fidelity}, taking $\delta = \min(1/2, \sigma)$ and $2^{m-1}\delta= \sigma/\sqrt{\varepsilon_\oneDG}$ ensures that the infidelity between the approximate and continuous 1DG states is at most~$\varepsilon_\oneDG$.
As $\delta<1$, the number of qubits needed to represent the fractional part of the real numbers
is~$\max\left(1,\, \ceil{\log_2(1/\sigma)}\right)$.
The largest value for the real numbers
is~$\sigma/\sqrt{\varepsilon_\oneDG}$,
so $\ceil{\log_2\left(\sigma/\sqrt{\varepsilon_\oneDG}\right)}$ qubits are needed to represent the integer part of the real numbers.
The total number of qubits needed to represent the approximate 1DG state is obtained by adding the number of qubits required to represent the real numbers' integer and fractional parts.
\end{proof}
\noindent
Having determined the space requirement for representing a 1DG state, next we establish a bound on the number of qubits needed to represent a multi-dimensional Gaussian state.

%%%%%%%%%%%%%%%%%%%%%%%%%%%%%%%%%%%%%%%%%%%%%%%%%%%%%%%%%
\subsubsection{Space requirement to represent a multi-dimensional Gaussian state}
\label{subsubsec:NDG_space}

We now determine the minimal number of qubits needed to represent an approximation for a continuous $N$-dimensional Gaussian state in terms of the condition number of its ICM,~$N$ and an error tolerance on the infidelity between the approximate and continuous states.
We present the result in~\cref{theorem:space_requirement} and proceed with a proof.
Then we invoke this theorem to obtain the space requirement for representing the free-field ground state in both Fourier- and wavelet-based methods in terms of the inputs specified by the main server.

\begin{theorem}
\label{theorem:space_requirement}
Let~$\upbm{A} \in \reals^{N\times N}$ be the ICM for a continuous $N$-dimensional Gaussian state~$\ket{\G_N(\upbm{A})}$~\eqref{eq:continuous_NDG} and let~$\varepsilon_\G\in (0,1)$ be an error tolerance.
Also let $\kappa \in \reals^+$ be the condition number of~$\upbm{A}$, quantified as the ratio of the largest to smallest eigenvalues of~$\upbm{A}$.
Then
\begin{equation}
\label{eq:NDGqubits}
   n \in \Omega\left(N\log_2\left(\dfrac{N \kappa}{\varepsilon_\G}\right)\right),
\end{equation}
qubits are required to represent a discrete approximation for~$\ket{\G_N(\upbm{A})}$ such that the infidelity
between the discrete and continuous states is bounded from above by~$\varepsilon_\G$.
\end{theorem}

\begin{proof}
Let~$\upbm{D}$ be the diagonal matrix in either the spectral or the UDU decomposition of the ICM~$\upbm{A}$.
Then the space required to represent~$\ket{\G_N(\upbm{A})}$ is the same as the space required to represent~$\ket{\G_N(\upbm{D})}$ because the former state is obtained from the latter by a basis transformation.
Therefore, we determine the required space to represent~$\ket{\G_N(\upbm{D})}$.
We decompose this state into a tensor product of~$N$ continuous 1DG states as
\begin{equation}
\label{eq:NDG_decomposition}
    \ket{\G_N(\upbm{D})}= \bigotimes_{\ell=0}^{N-1} \ket{\G(\sigma_\ell)}, \quad
    \sigma_\ell:=1/\sqrt{D_{\ell\ell}},
\end{equation}
where~$\sigma_\ell$ is the standard deviation of the~$\ell^\text{th}$ continuous 1DG state~$\ket{\G(\sigma_\ell)}$~\eqref{eq:continuous1DG}.
We approximate each 1DG state by a discrete 1DG state over a lattice with $2^m$ points and spacing $\delta$ as per~\cref{def:discrete1DG}.
The approximate Gaussian state is then
\begin{equation}
\label{eq:NDG_approximation}
    \ket{\tilde{\G}_N(\upbm{D})} := \bigotimes_{\ell=0}^{N-1} \ket{\G_\text{lattice}(\tilde{\sigma}_\ell, \delta, m)},
\end{equation}
where $\tilde{\sigma}_\ell:=\sigma_\ell/\delta$ is the standard deviation of the $\ell^\text{th}$ discrete 1DG state~\eqref{eq:lattice1DG}.
By~\cref{prop:1DG_fidelity}, if
\begin{equation}
    \varepsilon_\oneDG=\varepsilon_\G/N,\quad
    \delta \leq \min\left(1/2, \sigma_{\min}\right) \quad \text{and} \quad
    2^m\delta\geq 2\sigma_{\max}/\sqrt{\varepsilon_\oneDG},
\end{equation}
then for each $\ell$
\begin{equation}
\label{eq:1DG_fidelity}
    \braket{\G_\text{lattice}(\tilde{\sigma}_\ell, \delta, m)}{\G(\sigma_\ell)} \geq 1- \varepsilon_\G/N.
\end{equation}
\crefrange{eq:NDG_decomposition}{eq:1DG_fidelity} yield
\begin{equation}
\label{eq:fidelity_bound}
    \braket{\tilde{\G}_N(\upbm{D})}{\G_N(\upbm{D})} = \prod_{\ell=0}^{N-1} \braket{\G_\text{lattice}(\tilde{\sigma}_\ell, \delta, m)}{\G(\sigma_\ell)}
    \geq \left(1- \varepsilon_\G/N\right)^N
    \geq 1- \varepsilon_\G.
\end{equation}
Each discrete 1DG state~\eqref{eq:lattice1DG} is a superposition of lattice states~$\ket{j\delta}$, where $j\delta$ is a real number. 
To ensure that~\cref{eq:1DG_fidelity} holds~for each 1DG state, we need at least
$\ceil*{\log_2\left(\sigma_{\max}\sqrt{N/\varepsilon_\G}\right)}$ qubits to represent the integer part and at least~$\ceil{\log_2{(1/\sigma_{\min})}}$ qubits to represent the fractional part of the real numbers.
Thus, the minimal number of qubits to represent each 1DG state scales as
\begin{equation}
\label{eq:1DG_qubits}
    n_\oneDG\in \Omega\left(\log_2\left((\sigma_{\max}/\sigma_{\min})\sqrt{N/\varepsilon_\G}\right)\right).
\end{equation}
Let~$d_{\max}$ and~$d_{\min}$ be respectively the largest and smallest diagonal elements of $\bm{\mathrm{D}}$,
then by~\cref{eq:NDG_decomposition}
\begin{equation}
    \frac{\sigma_{\max}}{\sigma_{\min}}=\sqrt{\frac{d_{\max}}{d_{\min}}},
\end{equation}
and by~\cref{prop:dmax_dmin_bound}, $d_{\max}/d_{\min} \in O(\kappa)$.
The combination of these equations with~\cref{eq:1DG_qubits} yield~$n_\oneDG\in\Omega(\log_2(N\kappa/\varepsilon_\G))$.
Therefore, the total number of qubits to represent an approximation for a $N$-dimensional Gaussian state scales as~\cref{eq:NDGqubits}.
\end{proof}

% Space requirement for representing ground state
We now determine the space required to represent an approximation for the ground state in both Fourier- and wavelet-based methods.
To this end, by~\cref{theorem:space_requirement}, we only need to bound the condition number~$\kappa$ of the ground-state ICM for each method to obtain the space requirement in terms of the parameters specified by the main server.
By~\cref{prop:condition_bound}, the condition number of the ground state's ICM for both methods scales~as $\kappa\in \Theta(N/m_0)$.
Therefore, by \cref{theorem:space_requirement}, the number of qubits needed to represent the ground state scales as
\begin{equation}
  \label{eq:qubits_for_ground_state}
   n \in \Omega\left(N\log_2\left(N/\sqrt{m_0\varepsilon_\text{vac}}\right)\right),
\end{equation}
with respect to the parameters specified by the main server in~\cref{table:inputsgroundstategen}.
Notice that~$n$ is quasilinear in the number of modes~$N$.
The logarithmic factor here is the number of qubits required to represent each 1DG state for generating the free-field ground state.
For simplicity, we use
\begin{equation}
    \label{eq:p}
    p = \ceil{\log_2\left(N/\sqrt{m_0\varepsilon_\text{vac}}\right)},
\end{equation}
in description of our ground-state-generation algorithms for number of qubits to represent each 1DG state.

%%%%%%%%%%%%%%%%%%%%%%%%%%%%%%%%%%%%%%%%%%%%%%%%%%%%%%%%%
%%%%%% Classical preprocessing %%%%%%%%%%%%%%%%%%%%%%%%%%
%%%%%%%%%%%%%%%%%%%%%%%%%%%%%%%%%%%%%%%%%%%%%%%%%%%%%%%%%

\subsection{Classical preprocessing}
\label{subsec:classical_preprocessing}

In this subsection, we construct key subroutines of the classical preprocessing in the Fourier- and wavelet-based algorithms for ground-state generation.
We begin, in~\cref{subsubsec:ICM_eigenvalues}, by constructing a classical algorithm for computing the eigenvalues of the ground-state ICM in a fixed-scale wavelet basis;
this algorithm is used as a subroutine in classical preprocessing of both ground-state-generation algorithms.
Then, in~\cref{subsubsec:ICM_circulant-rows}, we present a classical algorithm for computing the circulant row in unique blocks of the ground-state ICM in a multi-scale wavelet basis.
Finally, in~\cref{subsubsec:UDUDecomp}, we construct a classical algorithm for computing the UDU decomposition of the ICM in a multi-scale wavelet basis.
The last two algorithms are subroutines of the classical preprocessing in the wavelet-based algorithm.

%%%%%%%%%%%%%%%%%%%%%%%%%%%%%%%%%%%%%%%%%%%%%%%%%%%%%%%%%
\subsubsection{Eigenvalues of the ground state's inverse-covariance matrix~(ICM)}
\label{subsubsec:ICM_eigenvalues}

Here we devise a classical algorithm for computing the eigenvalues of the ground-state ICM~$\upbm{A}_\sS$~\eqref{eq:fixedscale_groundstate} in a fixed-scale wavelet basis.
First we state the properties of this matrix used to compute the eigenvalues and explain the parameters that specify unique elements of the matrix.
We then provide the rationale for computing the eigenvalues of~$\upbm{A}_\sS$ by a discrete Hartley transform~(DHT) and present our algorithm as pseudocode.
Finally, we elaborate on the relationship between the eigenvalues of the ground-state ICM in fixed- and multi-scale wavelet bases.

% Properties of fixed-scale ICM
We begin by stating properties of the fixed-scale ICM~\eqref{eq:fixedscale_groundstate} used for computing the eigenvalues.
This matrix is the principal square root of the fixed-scale coupling matrix~$\upbm{K}_\sS$~\eqref{eq:fixedscaleK}, which itself is a circulant, real and symmetric matrix.
Being the coupling matrix's principal square root, the ICM is also a circulant, real and symmetric matrix.
A circulant matrix is fully specified by its first row, which we call the `circulant row' of the circulant matrix.
Four parameters specify the unique elements in the circulant row of~$\upbm{K}_\sS$:
the wavelet index~$\dbIndex$,
the second-order derivative overlaps~$\Delta_\ell$~\eqref{eq:derivative_overlaps},
the number of modes~$N$ and the free mass~$m_0$;
see~\cref{eq:fixedscaleK}.
The same parameters specify the unique matrix elements of the fixed-scale ICM.

% Rationale
We now provide the rationale for computing the eigenvalues of~$\upbm{A}_\sS$~\eqref{eq:fixedscale_groundstate} by a DHT.
Any circulant matrix is diagonalizable by a discrete Fourier transform~\cite[p.~100]{HJ12}.
By the real and symmetric properties of the ICM, we diagonalize this matrix by a DHT.
In particular, the eigenvalues of any real, symmetric and circulant matrix are obtained by computing the DHT of the matrix's circulant row~\cite[p.~100]{HJ12}.
As the ICM is the principal square root of the coupling matrix, first we compute the coupling-matrix eigenvalues by the DHT of its circulant row. 
Then we take the square root of the coupling-matrix eigenvalues to obtain the eigenvalues of the ICM~$\upbm{A}_\sS$. 
Algorithm~\ref{alg:invCovEigens} provides the procedure for computing the eigenvalues of the ICM using the parameters that specify unique elements of the coupling matrix.

\begin{algorithm}[H]
  \caption{Classical algorithm for computing eigenvalues of ground-state ICM}
  \label{alg:invCovEigens}
  \begin{algorithmic}[1]
  \Require{
  \Statex $\dbIndex \in \integers_{\geq 3}$
  \Comment{wavelet index}
  \Statex $m_0\in \reals^+$
  \Comment{mass of free QFT}
  \Statex $N \in \integers_{\geq 2(2\dbIndex-1)}$
  \Comment{number of modes}
  \Statex $\bm{\Delta} \in \reals^{2\dbIndex-1}$
  \Comment{derivative overlaps~\eqref{eq:derivative_overlaps} for second-order derivative operator in a wavelet basis with index~\dbIndex}
    }
  \Ensure
  \Statex $\bm{\uplambda} \in \reals^N$
  \Comment{eigenvalues of ground-state ICM~$\upbm{A}_\sS$~\eqref{eq:fixedscale_groundstate}}
    \Function{invCovEigens}{$\dbIndex, m_0, N, \bm{\Delta}$}
       \For{$j \gets 0$ to $N-1$}
            \State \label{line:eigen1}
            $\text{tmp}_j
            \gets m_0^2 - N^2 \Delta_0
            -2\sum\limits_{\ell=1}^{2\dbIndex-1} N^2 \Delta_\ell \cos\left(\frac{2\uppi \ell}{N}j\right)$ 
            \Comment{computes $j^\text{th}$ eigenvalue of the coupling matrix $\upbm{K}_\sS$~\eqref{eq:fixedscaleK}}
            \State \label{line:eigen2}
            $\lambda_j \gets \sqrt{\text{tmp}_j}$
            \Comment{computes $j^\text{th}$ eigenvalue of $\upbm{A}_\sS$}
       \EndFor
       \State \Return $\bm{\uplambda}$
    \EndFunction
     \end{algorithmic}
\end{algorithm}

% Eigenvalues of multi-scale ICM
The coupling matrix~\eqref{eq:multiscaleK} in a multi-scale wavelet basis is obtained from the coupling matrix~\eqref{eq:fixedscaleK} in a fixed-scale wavelet basis by a wavelet transform, which is a unitary transformation.
For each basis, the ground-state ICM is the principal square root of the coupling matrix.
Consequently, the ICM in a multi-scale wavelet basis is obtained by the same unitary wavelet transform from the ICM in a fixed-scale basis, and they have identical eigenvalues.
Therefore, we use Algorithm~\ref{alg:invCovEigens} as a subroutine in classical preprocessing of both Fourier- and wavelet-based algorithms for computing the eigenvalues of the ground-state ICM.

%%%%%%%%%%%%%%%%%%%%%%%%%%%%%%%%%%%%%%%%%%%%%%%%%%%%%%%%%
\subsubsection{Elements of ground state's ICM}
\label{subsubsec:ICM_circulant-rows}

We now construct a classical algorithm to compute the unique matrix elements of the ground-state ICM in a multi-scale wavelet basis.
First we state key properties of the coupling matrix~\upbm{K}~\eqref{eq:multiscaleK} represented in this basis.
Then we explain how we approximate the multi-scale ICM and discuss the classical memory requirement to store unique elements of the approximate ICM.
Next we explain our algorithm's procedure for computing the unique matrix elements, and finally, we present our algorithm as pseudocode.

% Properties of the coupling matrix
We begin with a few observations about the multi-scale coupling matrix~$\upbm{K}$~\eqref{eq:multiscaleK}.
We then extend these observations to the approximate ICM.
The coupling matrix~$\upbm{K}$~\eqref{eq:multiscaleK} has the following key properties:
\begin{enumerate}
    \item \emph{Block structure.}
    The matrix~$\upbm{K}$ has a block-matrix structure imposed by a wavelet transform, with three types of blocks:~ss, sw and ww;
    see~\cref{eq:multiscaleK}.
    We therefore present entries of~$\upbm{K}$ based on their block's location in the block matrix.
    
    \item \emph{Symmetry.}
    The coupling matrix is symmetric.
    Therefore, we only consider the unique blocks, which are the main and upper-diagonal blocks of~$\upbm{K}$.
    
    \item \emph{Banded circulant blocks.}
    The main and upper-diagonal blocks of~$\upbm{K}$ have a circulant structure;
    specifically, each block is a banded~$2^k$-circulant matrix for some non-negative integer~$k$.
    Each diagonal block is a banded~$1$-circulant matrix with bandwidth $w_\text{d}:=2(2\dbIndex-2)+1$.
    The ww block at entry~$(r,c)$ of the block matrix~$\upbm{K}$~\eqref{eq:multiscaleK}, for any~$c>r$, is a banded~$2^{c-r}$-circulant matrix with bandwidth~$w_\text{d}+(2^{c-r}-1)(2\dbIndex-1)$.
    The sw block at entry $(s_0,c)$ is a banded~$2^{c-s_0}$-circulant matrix with the same bandwidth as the ww block at entry~$(s_0,c)$.
\end{enumerate}
These observations are also true for the multi-scale ICM~$\upbm{A}$~\eqref{eq:multiscale_ICM},
except that the blocks are no longer banded.
However, imposing a cutoff condition on elements of~$\upbm{A}$ by some threshold value~$\varepsilon_\text{th}$ reimposes the banded structure for blocks of this matrix.
To describe the cutoff condition and how we approximate the ICM, first we define an $\varepsilon_\text{th}$-approximate ICM as follows.

\begin{definition}[$\varepsilon_\text{th}$-approximate ICM]
\label{def:Avarepsilonth}
Given any ICM~$\upbm{A} \in \reals^{N\times N}$
and any $\varepsilon_\text{th}>0$,
an $\varepsilon_\text{th}$-approximate ICM is a symmetric matrix~$\upbm{A}_{\varepsilon_\text{th}}$ such that
\begin{equation}
\label{eq:cutoffcondition}
    \left[A_{\varepsilon_\text{th}}\right]_{ij} :=
    \begin{cases}
        0      & \text{if}\ \abs*{A_{ij}} < \varepsilon_\text{th}, \\
        A_{ij} & \text{otherwise.}
    \end{cases}
\end{equation}
\end{definition}
\noindent
As per~\cref{def:Avarepsilonth}, we approximate the ground-state ICM by replacing its near-zero matrix elements, i.e., the elements with magnitude less than some close-to-zero threshold value~$\varepsilon_\text{th}$, with exactly zero.
This replacement rule enables a sparse structure
for the ICM that we exploit to perform its UDU decomposition in quasilinear time.

We show, in~\cref{prop:truncated_Gaussian}, that not only the approximate ICM obtained by imposing the cutoff condition is a positive-definite matrix but also the infidelity between the Gaussian state with the approximate ICM and the free-field ground state is bounded from above by the error tolerance~$\varepsilon_\text{vac}$ in~\cref{table:inputsgroundstategen}.

\begin{proposition}
\label{prop:truncated_Gaussian}
Given any $\varepsilon_\textup{vac}\in(0,1)$
and any ICM~$\upbm{A} \in \reals^{N\times N}$,
then for any $\varepsilon_\textup{th}$
satisfying
\begin{equation}
0<\varepsilon_\textup{th}\leq \varepsilon_\textup{vac} N^{-3/2}
\min \operatorname{spec}\upbm{A},
\end{equation}
every $\varepsilon_\textup{th}$-approximate ICM
$\upbm{A}_{\varepsilon_\textup{th}}$
is a positive-definite matrix such that
\begin{equation}
\label{eq:infidGNAGNAth}
\operatorname{infid}\left(\ket{\G_N(\upbm{A})},
\ket{\G_N(\upbm{A}_{\varepsilon_\textup{th}})}\right)
\leq\varepsilon_\textup{vac},
\end{equation}
where~$\ket{\G_N(\upbm{A})}$ is a $N$-dimensional continuous Gaussian state with the ICM~$\upbm{A}$ as per~\cref{def:continuous_Gaussian}.
\end{proposition}
\noindent
This proposition is proven in~\cref{proofprop:truncated_Gaussian}.

The approximate ICM obtained by imposing the cutoff condition is sparse because most elements of the exact ICM~$\upbm{A}$~\eqref{eq:multiscale_ICM} have an exponentially close-to-zero value.
In particular, we show in~\cref{prop:decayingDiagBlocks} that diagonal blocks of~$\upbm{A}$ decay exponentially away from the diagonal elements.
A corollary of the exponential decay, shown in~\cref{corol:bandedCirculantBlocks}, is that the diagonal blocks are banded circulant matrices with a bandwidth that is logarithmic in the number of modes~$N$;
here we refer to the number of cyclically nonzero elements in any row of a circulant matrix as `bandwidth' of the matrix.

\begin{proposition}[exponentially decaying diagonal blocks]
\label{prop:decayingDiagBlocks}
Let $m_0 \in \reals^{+}$ be the free mass and $\upbm{A}\in\reals^{N\times N}$~\eqref{eq:multiscale_ICM} be the ground-state ICM in a multi-scale wavelet basis with the wavelet index $\dbIndex \in \integers_{\geq 3}$.
Also let \{$\upbm{A}^{(r,r)}_\textup\ww:\ceil*{\log_2(4\dbIndex-2)}\leq r <\log_2 N\}$ be the diagonal ww blocks of~$\upbm{A}$ as in~\cref{eq:multiscale_ICM}.
Then for any $r$ and $j\geq 2\dbIndex-1$
\begin{equation}
\label{eq:decayingDiagBlocks}
    \abs{A^{(r,r)}_{\textup{ww};\, 0, j}}
    \leq 16\dbIndex m_0\kappa^{(r+1)} 2^{-\abs{j}/\xi^{(r+1)}},
    \quad
    \xi^{(r)}:=(2\dbIndex-1)2^{r+1}/m_0,
\end{equation}
where $\kappa^{(r)}>1$ is the spectral condition number of~$\upbm{K}_\sS^{(r)}$~\eqref{eq:fixedscaleK}.
\end{proposition}
\noindent
This proposition is proven in~\cref{proofprop:decayingDiagBlocks}.

\begin{corollary}[banded circulant blocks]
\label{corol:bandedCirculantBlocks}
The diagonal blocks of the approximate ICM in a multi-scale wavelet basis are banded matrices with the upper bandwidth
\begin{equation}
\label{eq:bandwidth}
    w=\ceil*{\frac{2(2\dbIndex-1)}{m_0}
    \log N \log_2\left(\frac{4\dbIndex N}{m_0\varepsilon_\text{vac}}\right)},
\end{equation}
where all parameters are specified in~\cref{table:inputsgroundstategen}.
\end{corollary}
\noindent
By this corollary, proven in~\cref{appx:proofcorol:bandedCirculantBlocks}, each diagonal ww block of the approximate ICM~$\tilde{\upbm{A}}$ is a banded $1$-circulant matrix~with bandwidth~$2w+1$.
The wavelet transform implies that the off-diagonal ww blocks are also banded matrices.
Specifically, the off-diagonal ww block at entry~$(r, c)$ of the block matrix~$\tilde{\upbm{A}}$ is a banded~$2^{c-r}$-circulant matrix with bandwidth
\begin{equation}
\label{eq:block_bandwidth}
    \textsc{width}(r,c):= 2w+1+(2^{c-r}-1)(2\dbIndex-1) \quad \forall\, c>r.
\end{equation}
However, the ss and sw blocks of $\tilde{\upbm{A}}$ are not necessarily banded matrices.
We therefore treat these blocks as dense matrices.

% Data structure for representing block matrix in a multi-scale wavelet basis
We now describe a data structure for representing the bock matrix in a multi-scale wavelet basis.
The date structure that we describe here takes advantage of the matrix's block and circulant structures for efficient storage.
In particular, we use the data structure for storing the multi-scale ICM~$\upbm{A}$~\eqref{eq:multiscale_ICM}. 
We store the block matrix in a multi-scale wavelet basis by an associative array, i.e., by a collection of (key, value) pairs.
Each key is a tuple $(z,r,c)$, where $z\in\{\text{ss,sw,ww}\}$ specifies if the block belongs to the ss, sw or ww part of the block matrix.
The positive integers~$r$ and~$c$ respectively specify the row and column indices of the block matrix;
these integers also specify the scales to which the block belongs.
The value of the associative array specifies the block at entry~$(z, r, c)$ of the block matrix.
Each block, being a circulant matrix, is specified by a vector, which is the first row of the block, and two additional parameters:
the block's size and the amount by which the vector is shifted when moving from one row to the next.
For our application, each key specifies the block's size and the amount by which the vector is shifted.
Therefore, the values are vectors that specify blocks of the block matrix.

% Memory requirements for storing a symmetric block matrix with circulant blocks
The memory requirement for storing a symmetric block matrix with circulant blocks using the described data stricture is quasilinear in the matrix's dimension.
The fact that each block matrix is circulant means
that we only need to store one row of each circulant block;
the other rows are shifted versions of the stored row.
Assuming unit cost for storing each element, the cost of storing a block matrix with circulant blocks is equal to the sum of the row size for each block, which is altogether quasilinear in the matrix's dimension.
The symmetry of the matrix reduces this cost by a factor of two.
The memory cost can be further reduced by noting that most of the entries
in the vector specifying the circulant blocks are equal to zero.
We do not use this technique because we are only concerned with producing a quasilinear algorithm for ground-state generation.
The memory-cost reduction only delivers a constant-factor improvement and makes the algorithm much more complicated.
However, it is an obvious way to improve the classical preprocessing of the wavelet-based algorithm further.

% Procedure
We now describe our algorithm's procedure for computing the circulant row in unique blocks of the multi-scale ICM~$\upbm{A}$~\eqref{eq:multiscale_ICM}.
To elucidate the algorithm's procedure, first we state recursive relations for blocks of the multi-scale ICM.
The ICM in a multi-scale wavelet basis is obtained by a wavelet transform from the ICM in a fixed-scale wavelet basis.
Specifically, for $d := k-s_0$,
\begin{equation}
    \upbm{A}^{(k)}=\upbm{W}^{(k)}_d\upbm{A}^{(k)}_\sS \upbm{W}^{(k)\T}_d,
\end{equation}
where $\upbm{W}^{(k)}_d$~\eqref{eq:WTM} is the $d$-level wavelet-transform matrix at scale~$k$.
By this equation, for scales~$r$ and~$c$ with $s_0 \leq r \leq c <k $, the wavelet transform imposes the recursive relations
\begin{align}
 &\label{eq:Ass}
 \upbm{A}^{(c)}_\sS = \upbm{H} \upbm{A}^{(c+1)}_\sS \upbm{H}^\T, \\
 &\label{eq:Asw}
 \upbm{A}^{(r,c)}_\sw = \upbm{H}^{c-r+1} \upbm{A}^{(c+1)}_\sS \upbm{G}^\T, \\
 &\label{eq:Aww}
 \upbm{A}^{(r, c)}_\ww = \upbm{G} \upbm{H}^{c-r} \upbm{A}^{(c+1)}_\sS \upbm{G}^\T,
\end{align}
for blocks of the multi-scale ICM~$\upbm{A}$ and the fixed-scale ICM~$\upbm{A}_\sS^{(c+1)}$, where
$\upbm{H}$ and $\upbm{G}$ are the upper and lower half of~$\upbm{W}^{(k)}_d$~\eqref{eq:WTM}, respectively.
We use these recursive formulae to compute the circulant in unique blocks of the multi-scale ICM~$\upbm{A}$.

We start by computing the circulant row of the bottom-right block in~$\upbm{A}$ and proceed to compute the circulant row of the top-left block column by column.
For each column~$c$ of the block matrix, first we compute the circulant row of~$\upbm{A}^{(c)}_\sS$ by~\cref{eq:Ass}.
Next we compute circulant row of the diagonal block in column~$c$ using~\cref{eq:Aww} with~$r=c$.
For each ww block above the diagonal block,
i.e., for ww blocks with row index~$r$ form~$c-1$ to~$s_0$, we then iteratively update~$\upbm{A}^{(c+1)}_\sS$ as~$\upbm{A}^{(c+1)}_\sS \gets \upbm{H}\upbm{A}^{(c+1)}_\sS$ and compute circulant row of~$\upbm{A}^{(r, c)}_\ww$ by~$\upbm{A}^{(r, c)}_\ww = \upbm{G} \upbm{A}^{(c+1)}_\sS \upbm{G}^\T$.
Finally, we compute the circulant row of the sw block in column~$c$ by~$\upbm{A}^{(s_0,c)}_\sw = \upbm{A}^{(c+1)}_\sS \upbm{G}^\T$ because~$\upbm{A}^{(c+1)}_\sS$ in~\cref{eq:Asw} is updated~$c-s_0+1$ times while computing the circulant row of the ww blocks in column~$c$.
The explicit procedure of our algorithm for computing the circulant rows is presented in Algorithm~\ref{alg:invCovCircRows}.

\begin{algorithm}[H]
  \caption{Classical algorithm for computing circulant row in blocks of the ICM~\eqref{eq:multiscale_ICM} in multi-scale wavelet basis}
  \label{alg:invCovCircRows}
  \begin{algorithmic}[1]
  \Require{
  \Statex $\dbIndex \in \integers_{\geq 3}$
    \Comment{wavelet index}
    \Statex $m_0\in\reals^+$
    \Comment{free-QFT mass}
    \Statex $N \in \integers_{\geq 2(2\dbIndex-1)}$
    \Comment{number of modes}
    \Statex $\bm{\Delta} \in \reals^{2\dbIndex-1}$
    \Comment{derivative overlaps~\eqref{eq:derivative_overlaps} for Laplace operator}
    \Statex $p \in \integers^+$
    \Comment{working precision}
  }
  \Ensure{
  \Statex $\upbm{a}:=\cBraket{\left.
  \upbm{a}^{(s_0)}_\sS\in \reals^{2^{s_0}},
  \upbm{a}^{(s_0, c)}_\sw\in \reals^{2^c},
  \upbm{a}^{(r, c)}_\ww\in \reals^{2^{c-r}}
  \right\vert
  \ceil*{\log_2(4\dbIndex-2)} =:s_0\leq r \leq c < k:=\log_2 N}$
  \Comment{circulant row in main and upper-diagonal blocks of the multi-scale ICM~$\upbm{A}$~\eqref{eq:multiscale_ICM}}
  }
\Function{invCovCircRows}{$\dbIndex, m_0, N, \bm{\Delta}, p$}
\State \label{line:invCovEigens}
$\reals^N \ni \bm{\uplambda} \gets
\textsc{invCovEigens}(\dbIndex, m_0, N, \bm{\Delta}, p)$
\Comment{computes eigenvalues of~$\upbm{A}$ by~\cref{alg:invCovEigens}}   
\For{$j \gets 0$ to $N-1$} \label{line:scalekCircrow-start}
    \Comment{computes circulant row of $\upbm{A}_\sS^{(k)}$}
    \State $a_{\sS;\,j }^{(k)} \gets
    \tfrac{1}{N} \sum\limits_{i=0}^{N-1}
    \lambda_i \cos\pars{\tfrac{2\uppi j}{N}i}$
\EndFor \label{line:scalekCircrow-end}
\State \label{line:lowpass}
$\reals^{2\dbIndex} \ni \upbm{h} \gets
\textsc{lowPassFilter}(\dbIndex,p)$
\Comment{computes low-pass filter for index-\dbIndex\ wavelet by~\cref{alg:lowPass}}
\For{$i \gets 0$ to $2\dbIndex-1$} \label{line:highpass-start}
    \Comment{computes high-pass filter}
    \State $g_i \gets (-)^i h_{2\dbIndex-1-i}$
\EndFor \label{line:highpass-end}
\For{$c \gets k-1$ to $s_0$}
\label{line:ssCircrow-start}
\Comment{iterates over column index of the block matrix~$\upbm{A}$}
    \For{$j \gets 0$ to $2^c-1$}
        \State $a_{\sS;\, j}^{(c)} \gets
        \sum \limits_{m,n=0}^{2\dbIndex-1}
         h_m h_n a_{\sS;\, (m+2j-n)\bmod 2^{c+1}}^{(c+1)}$
         \Comment{computes circulant row of $\upbm{A}^{(c)}_\sS$ from the circulant row of $\upbm{A}^{(c+1)}_\sS$ by~\cref{eq:Ass}}
    \EndFor \label{line:ssCircrow-end}
    \For{$r \gets c$ to $s_0$}
    \Comment{iterates over row index of the block matrix~$\upbm{A}$}
        \For{$j \gets 0$ to $2^c-1$}
             \State $a_{\ww;\, j}^{(r,c)}\gets
             \sum \limits_{m,n=0}^{2\dbIndex-1}
             g_m g_n a^{(c+1)}_{\sS;\, (m+2j-2^{c-r} n)\bmod 2^{c+1}}$
            \Comment{computes circulant row of $\upbm{A}^{(r,c)}_\ww$}
        \EndFor
        \State $\upbm{x} \gets \upbm{a}^{(c+1)}_\sS$
        \Comment{to be updated, the circulant row of $\upbm{A}^{(c+1)}_\sS$ is stored on scratchpad as a temporary vector $\upbm{x}$
        }
        \For{$j \gets 0$ to $2^{c+1}-1$}
            \State $a^{(c+1)}_{\sS;\, j} \gets \sum\limits_{m=0}^{2\dbIndex-1} h_m x_{(j-2^{c-r}m)\bmod 2^{c+1}}$
        \EndFor
    \EndFor
    \For{$j \gets 0$ to $2^c-1$}
    \Comment{computes circulant row of the sw block in column $c$}
    \State
    \label{line:swCircrow-end}
    $a^{(s_0,c)}_{\sw;\, j\bmod 2^c} \gets \sum\limits_{m=0}^{2\dbIndex-1} g_m a^{(c+1)}_{\sS;\, (m+2j)\bmod 2^{c+1} }$
    \EndFor
\EndFor
\State \Return $\upbm{a}$
\EndFunction
\end{algorithmic}
\end{algorithm}

%%%%%%%%%%%%%%%%%%%%%%%%%%%%%%%%%%%%%%%%%%%%%%%%%%%%%%%%%
\subsubsection{UDU decomposition of ground state's ICM}
\label{subsubsec:UDUDecomp}

Here we present our classical algorithm for computing the UDU decomposition of the approximate ICM in a multi-scale wavelet basis.
First we describe a block variant of the UDU decomposition for a real-symmetric matrix.
Next we explain how we approximate the UDU decomposition of the approximate ICM.
We then specify the inputs and outputs for our UDU-decomposition algorithm for a sparse matrix and explain an efficient method for storing the outputs.
Finally, we describe our algorithm's procedure and present the algorithm as pseudocode.

% Block variant for UDU decomposition
We begin by describing a block variant for the UDU decomposition of a dense real-symmetric matrix~$\upbm{A}$ that has a block structure as the matrix in~\cref{eq:multiscale_ICM};
see~\cref{appx:UDU} for the standard UDU decomposition of~$\upbm{A}$.
By~\cref{eq:UDU}, we write the UDU decomposition of the block matrix~$\upbm{A}$~\eqref{eq:multiscale_ICM} as
\begin{equation}
\label{eq:blockUDU}
\upbm{A} = \upbm{UDU}^\T := \upbm{UV}^\T =
\begin{bmatrix*}[l]
\upbm{U}^{(s_0)}_\sS & \upbm{U}^{(s_0,\,s_0)}_\sw & \cdots & \upbm{U}^{(s_0,\,k-1)}_\sw \\
                     & \upbm{U}^{(s_0,\,s_0)}_\ww & \cdots & \upbm{U}^{(s_0,\,k-1)}_\ww \\
                     &                            & \ddots & \quad\vdots                \\
                     &                            &        & \upbm{U}^{(k-1,\,k-1)}_\ww
\end{bmatrix*}
\begin{bmatrix*}[l]
\upbm{V}^{(s_0)}_\sS & \upbm{V}^{(s_0,\,s_0)}_\sw & \cdots & \upbm{V}^{(s_0,\,k-1)}_\sw \\
                     & \upbm{V}^{(s_0,\,s_0)}_\ww & \cdots & \upbm{V}^{(s_0,\,k-1)}_\ww \\
                     &                            & \ddots & \quad\vdots                \\
                     &                            &        & \upbm{V}^{(k-1,\,k-1)}_\ww
\end{bmatrix*}^\T,
\end{equation}
where
\begin{equation}
\label{eq:blcokDiagonalD}
   \upbm{D}:= \upbm{D}^{(s_0)}_\sS
\bigoplus_{s=s_0}^{k-1} \upbm{D}^{(s)}_\ww, 
\end{equation}
is a block-diagonal matrix and
\begin{equation}
\label{eq:matrixV}
\upbm{V}^{(s_0)}_\sS:= \upbm{U}^{(s_0)}_\sS \upbm{D}^{(s_0)}_\sS,\;\;\;
\upbm{V}^{(s_0,c)}_\sw:= \upbm{U}^{(s_0,c)}_\sw \upbm{D}^{(c)}_\ww,\;\;\;
\upbm{V}^{(r,c)}_\ww:= \upbm{U}^{(r,c)}_\ww \upbm{D}^{(c)}_\ww \;\;\;
(s_0 \leq r\leq c<k),
\end{equation}
are blocks of the block matrix~$\upbm{V}$, which has the same block structure as~$\upbm{U}$.
Diagonal elements of~$\upbm{D}$ and shear elements of~$\upbm{U}$ are computed in the UDU matrix decomposition.
By~\cref{eq:blockUDU}, the $i^\text{th}$ diagonal element of~$\upbm{D}^{(c)}_\ww$ and~$\upbm{D}^{(s_0)}_\sS$ are
\begin{align}
& \label{eq:wwDiags}
    d^{(c)}_{\ww;\,i} = a^{(c,c)}_{\ww;\,i,i}
    -\sum_{s=c}^{k-1}
    \upbm{u}^{(c,s)}_{\ww;\,
    i,(i+1)\updelta_{sc}:2^s-1}
    \cdot \upbm{v}^{(c,s)}_{\ww;\, i,(i+1)\updelta_{sc}:2^s-1},\\
& \label{eq:ssDiags}
    d^{(s_0)}_{\sS;\,i} = a^{(s_0)}_{\sS;\,i,i}
    -\upbm{u}^{(s_0)}_{\sS;\, i,i+1:2^{s_0}-1}
    \cdot \upbm{v}^{(s_0)}_{\sS;\, i,i+1:2^{s_0}-1}
    -\sum_{s=s_0+1}^{k-1}
    \upbm{u}^{(s_0,s)}_{\sw;\,i,0:2^s-1}
    \cdot \upbm{v}^{(s_0,s)}_{\sw;\,i,0:2^s-1},
\end{align}
respectively;
see also~\cref{eq:diagsOfD}.

We now establish formulae to compute shear elements in various blocks of~$\upbm{U}$.
For $0< i <2^{s_0}$, the shear elements in the~$i^\text{th}$ column of~$\upbm{U}$ are the same as the shear elements in the~$i^\text{th}$ column of~$\upbm{U}^{(s_0)}_\sS$.
By \cref{eq:shearsOfU}, these elements are
\begin{equation}
\label{eq:ssShears}
    \upbm{u}^{(s_0)}_{\sS;\, 0:i-1, i} = \frac1{d^{(s_0)}_{\sS;\, i}}
    \left[\upbm{a}^{(s_0)}_{\sS;\, 0:i-1, i}
    -\upbm{u}^{(s_0)}_{\sS;\, 0:i-1, i+1:2^{s_0}-1}
    \cdot\upbm{v}^{(s_0)}_{\sS;\, i, i+1:2^{s_0}-1}
    -\sum_{s=s_0+1}^{k-1}
    \upbm{u}^{(s_0,s)}_{\sw;\, 0:i-1, 0:2^s-1}
    \cdot\upbm{v}^{(s_0,s)}_{\sw;\, i, 0:2^s-1}
    \right]
    \quad \forall i \neq 0.
\end{equation}
For $c\geq s_0$ and $0 \leq i < 2^c$,
the shear elements in the~$(2^c+i)^\text{th}$ column of $\upbm{U}$ are the elements in the $i^\text{th}$ column of~$\upbm{U}^{(s_0,c)}_\sw$ and~$\upbm{U}^{(r,c)}_\ww$
for $s_0\leq r<c$,
and the elements above the diagonal entries of~$\upbm{U}^{(c,c)}_\ww$.
By~\cref{eq:shearsOfU}, these elements are
\begin{align}
&\label{eq:swShears}
    \upbm{u}^{(s_0,c)}_{\sw;\, 0:2^{s_0}-1, i}
    =\frac1{d^{(c)}_{\ww;\,i}}
    \left[\upbm{a}^{(s_0, c)}_{\sw;\, 0:2^{s_0}-1, i}
    -\sum_{s=c}^{k-1}
    \upbm{u}^{(s_0,s)}_{\sw;\, 0:2^{s_0}-1, (i+1)\updelta_{sc}:2^s-1}
    \cdot \upbm{v}^{(c,s)}_{\ww;\, i,(i+1)\updelta_{sc}:2^s-1}
    \right],\\
&\label{eq:wwShears}
    \upbm{u}^{(r,c)}_{\ww;\, 0:2^r-1, i}
    =\frac1{d^{(c)}_{\ww;\,i}}
    \left[\upbm{a}^{(r,c)}_{\ww;\, 0:2^r-1, i}
    -\sum_{s=c}^{k-1}
    \upbm{u}^{(r,s)}_{\ww;\, 0:2^r-1, (i+1)\updelta_{sc}:2^s-1}
    \cdot \upbm{v}^{(c,s)}_{\ww;\, i,(i+1)\updelta_{sc}:2^s-1}
    \right],\\
&\label{eq:diagonalwwShears}
    \upbm{u}^{(c,c)}_{\ww;\, 0:i-1, i} = \frac1{d^{(c)}_{\ww;\,i}}
    \left[\upbm{a}^{(c, c)}_{\ww;\, 0:i-1, i}
    -\sum_{s=c}^{k-1}
    \upbm{u}^{(c,s)}_{\ww;\, 0:i-1, (i+1)\updelta_{sc}:2^s-1}
    \cdot \upbm{v}^{(c,s)}_{\ww;\, i,(i+1)\updelta_{sc}:2^s-1}
    \right]
    \quad \forall i\neq 0,
\end{align}
respectively.
We use these formulae in our UDU-decomposition algorithm to compute the nonzero shear elements of the upper unit-triangular matrix~$\upbm{U}$ in the UDU decomposition of the approximate ICM in a multi-scale wavelet basis.

% Incomplete UDU decomposition by position
In our algorithm for the UDU decomposition of~$\tilde{\upbm{A}}$, we compute an approximation of the upper unit-triangular matrix $\upbm{U}$.
We take the above-diagonal nonzero elements of~$\upbm{U}$ to be in the same position as the above-diagonal nonzero elements of $\tilde{\upbm{A}}$;
we set any entry of~$\upbm{U}$ to zero if the corresponding entry in~$\tilde{\upbm{A}}$ is also zero.
We refer to this decomposition as the incomplete UDU decomposition by position.

% Helper functions for location of nonzero entries
We specify the location of nonzero entries of the approximate ICM by helper functions.
The helper functions return various parameters that specify nonzero entries in each ww block of the approximate ICM.
As shown in~\cref{fig:offDiagBlock}, each ww block has nonzero elements in three parts:
top-right part~(TRP),
bottom-left part~(BLP) and
main part~(MP).
For each block, we compute the following parameters by the five helper functions in~Library~\ref{lib:positionHelperFunctions}:
(1)~the number of nonzero (NNZ) entries in the last column at the TRP of the block;
(2)~bandwidth of the block, which is the number of cyclically consecutive nonzero entries in each row of the block;
(3)~vertical bandwidth of the block, which is the number of cyclically consecutive nonzero entries in each column of the block;
(4)~the column index of the first and last nonzero entries in the block's main part for a given row index; and
(5)~The row index of the first and last nonzero entries in the block's main part for a given column index.
These parameters fully specify the location of nonzero entries in each ww block.

\begin{figure}[htb]
\centering
    \includegraphics[width=\textwidth]{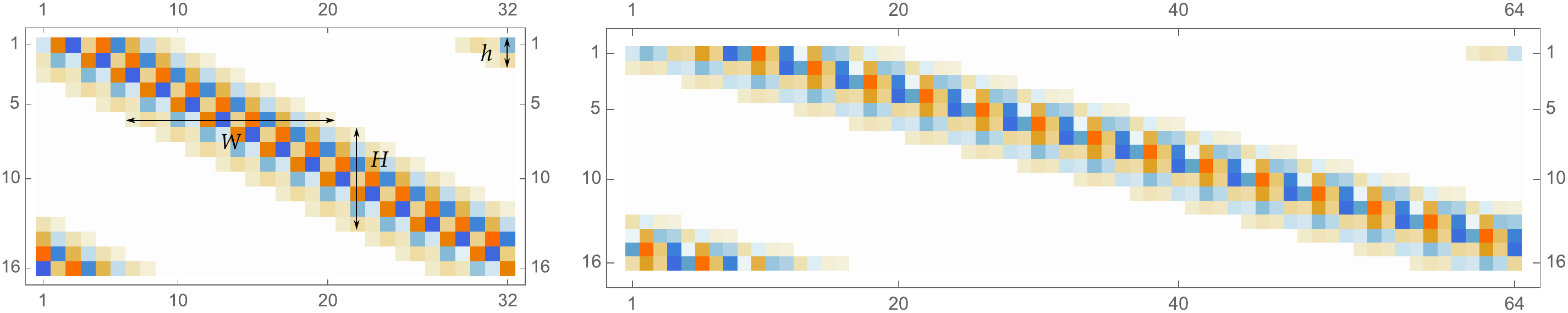}
\caption[Visualization of a ww block]{
Visualization of two ww blocks of the approximate ICM.
The left block is a $2$-circulant matrix and the right block is a $4$-circulant matrix.
Each ww block has nonzero elements in the top-right, bottom-left and main part.
The parameter~$h$ is the NNZ elements in the last column at the top-right part of the block;
$W$ is the bandwidth and $H$ is the vertical bandwidth of the block.
These parameters specify the location of nonzero entries in the block.}
\label{fig:offDiagBlock}
\end{figure}

\begin{library}[H]
  \caption{Helper functions for location of nonzero elements in blocks of the approximate ICM in multi-scale wavelet basis}
  \label{lib:positionHelperFunctions}
  \begin{algorithmic}[1]
\State \textbf{global} \dbIndex, $w$
\Comment{wavelet index \dbIndex\ and upper bandwidth $w$ of the diagonal \ww\ blocks of $\tilde{\upbm{A}}$ are global variables}
\Function{lastColNNZ}{$r,c$}
\Comment{computes NNZ elements in last column in top-right part of $(r,c)$ block}
\State \Return $\ceil*{\frac{w}{2^{c-r}}}$
\EndFunction
\Function{width}{$r,c$}
\Comment{computes bandwidth of $(r,c)$ block}
\State \Return $2w+1+(2^{c-r}-1)(2\dbIndex-1)$
\EndFunction
\Function{vertWidth}{$r,c$}
\Comment{computes vertical bandwidth of $(r,c)$ block}
\State \Return $\ceil*{
\frac{\textsc{width}(r,c)}{2^{c-r}}
}$
\EndFunction
\Function{mainCol}{$i,r,c$}
\Comment{computes column index of first and last nonzero entry in $i^\text{th}$ row in main part of $(r,c)$ block}
    \State $h \gets \textsc{lastColNNZ}(r,c)$
    \State $W \gets \textsc{width}(r,c)$
    \State $j_\text{F} \gets (i-h)2^{c-r}\times \Theta(i-h)$
    \Comment{$\Theta(n)$ is the unit-step function whose value is zero for $n<0$ and is one for $n\geq0$}
    \State $j_\text{L} \gets
    \min\left(j_\text{F}+W-1,2^c-1\right)\times\Theta(i-h)
    +\left(W-w+i2^{c-r}-1\right)\times\Theta(h-i-1)$
    \State \Return $(j_\text{F},j_\text{L})$
\EndFunction
\Function{mainRow}{$j,r,c$}
\Comment{computes row index of first and last nonzero entry in $j^\text{th}$ column in main part of $(r,c)$ block}
    \State $h \gets \textsc{lastColNNZ}(r,c)$
    \State $W \gets \textsc{width}(r,c)$
    \State $H \gets \textsc{vertWidth}(r,c)$
    \State $i_\text{F} \gets \ceil*{\frac{j-(W-w-1)}{2^{c-r}}}
    \times \Theta(j-(W-w))$
    \State $i_\text{L} \gets \min (i_\text{F} + H-1, 2^r-1) \times \Theta(j-(W-w))
    +\left(h + \floor*{\frac{j}{2^{c-r}}}\right)
    \times\Theta((W-w)-j-1)$
    \State \Return $(i_\text{F},i_\text{L})$
\EndFunction
\end{algorithmic}
\end{library}

% Inputs/output
We now specify inputs and outputs of our UDU-decomposition algorithm.
The algorithm's inputs are parameters that specify unique elements of the approximate ICM~$\tilde{\upbm{A}}$:
the wavelet index~$\dbIndex$,
the number of modes~$N$ in the discretized QFT,
the upper bandwidth~$w$ for diagonal blocks of~$\tilde{\upbm{A}}$
and the circulant row in unique blocks of~$\tilde{\upbm{A}}$. 
Outputs are shear elements of the approximate upper unit-triangular matrix~$\tilde{\upbm{U}}$ and diagonals of the diagonal matrix~$\upbm{D}$ in the incomplete UDU decomposition~of~$\tilde{\upbm{A}}$.

% Storing shear elements
The approximate upper unit-triangular~$\tilde{\upbm{U}}$ matrix in the incomplete UDU decomposition has the same sparse and block-matrix structure as the matrix~$\tilde{\upbm{A}}$.
We exploit these structures and store the shear elements of~$\tilde{\upbm{U}}$ by a sparse representation.
Similar to sparse representation of~$\tilde{\upbm{A}}$,
we represent shear elements in blocks of~$\tilde{\upbm{U}}$ by an associative array, i.e., a collection of (key, value) pairs.
Each key is a tuple $(z, r, c)$, where $z\in\{\text{ss,sw,ww}\}$ specifies if the block belongs to the ss, sw or ww part of the block matrix.
The positive integers $r$ and $c$ specify the row and column indices of the block matrix, respectively.
The associative array's value specifies the shear elements of the block at entry~$(z, r, c)$ of the block matrix.
Unlike blocks of $\tilde{\upbm{A}}$, blocks of~$\tilde{\upbm{U}}$ are not circulant matrices, so we cannot specify each block by a single row of the block.
Therefore, each value in the associative array is a matrix that specifies the shear elements of a block in the block matrix.

The ss block of~$\tilde{\upbm{U}}$ is a unit-triangular matrix.
As illustrated in~\cref{fig:shearstoring}~(b), we store this block's shear elements as a vector of size~$2^{s_0}(2^{s_0}-1)/2$ by concatenating the shear elements from top-left to bottom-right corner sequentially row by row;
note that the ss block is a matrix of size~$2^{s_0}\times 2^{s_0}$.
Similar to the sw blocks of~$\tilde{\upbm{A}}$, the sw blocks of~$\tilde{\upbm{U}}$ are not sparse matrices.
Therefore, we store the shear elements in these blocks by a matrix of the same size.
The diagonal ww blocks are each an upper unit-triangular matrix.
The shear element of~$\tilde{\upbm{U}}_\ww^{(r,r)}$
are in the same position of the above-diagonal nonzero elements of~$\tilde{\upbm{A}}_\ww^{(r,r)}$;
see~\cref{fig:shearstoring}~(c).
As illustrated in~\cref{fig:shearstoring}~(d), we store the shear elements of~$\tilde{\upbm{U}}_\ww^{(r,r)}$ by a matrix of size~$2^r \times w$, where~$w$ is the upper bandwidth of the diagonal ww blocks of~$\tilde{\upbm{A}}$.
The off-diagonal ww blocks of~$\tilde{\upbm{U}}$ are sparse with the same sparse structure as the off-diagonal ww blocks of~$\tilde{\upbm{A}}$.
Specifically, the block~$\tilde{\upbm{U}}_\ww^{(r,c)}$ with~$c>r$ is a banded matrix with bandwidth $\textsc{width}(r,c)$~\eqref{eq:block_bandwidth}.
We store shear element of this block by a matrix of size~$2^r \times \textsc{width}(r,c)$, which is illustrated in~\cref{fig:shearstoring}~(f).
For convenience, we use the helper functions in~Library~\ref{lib:storingHelperFunctions} for storing the shear elements in the ss and ww blocks of~$\tilde{\upbm{U}}$.

\begin{figure}[htb]
\centering
\includegraphics[width=.7\linewidth]{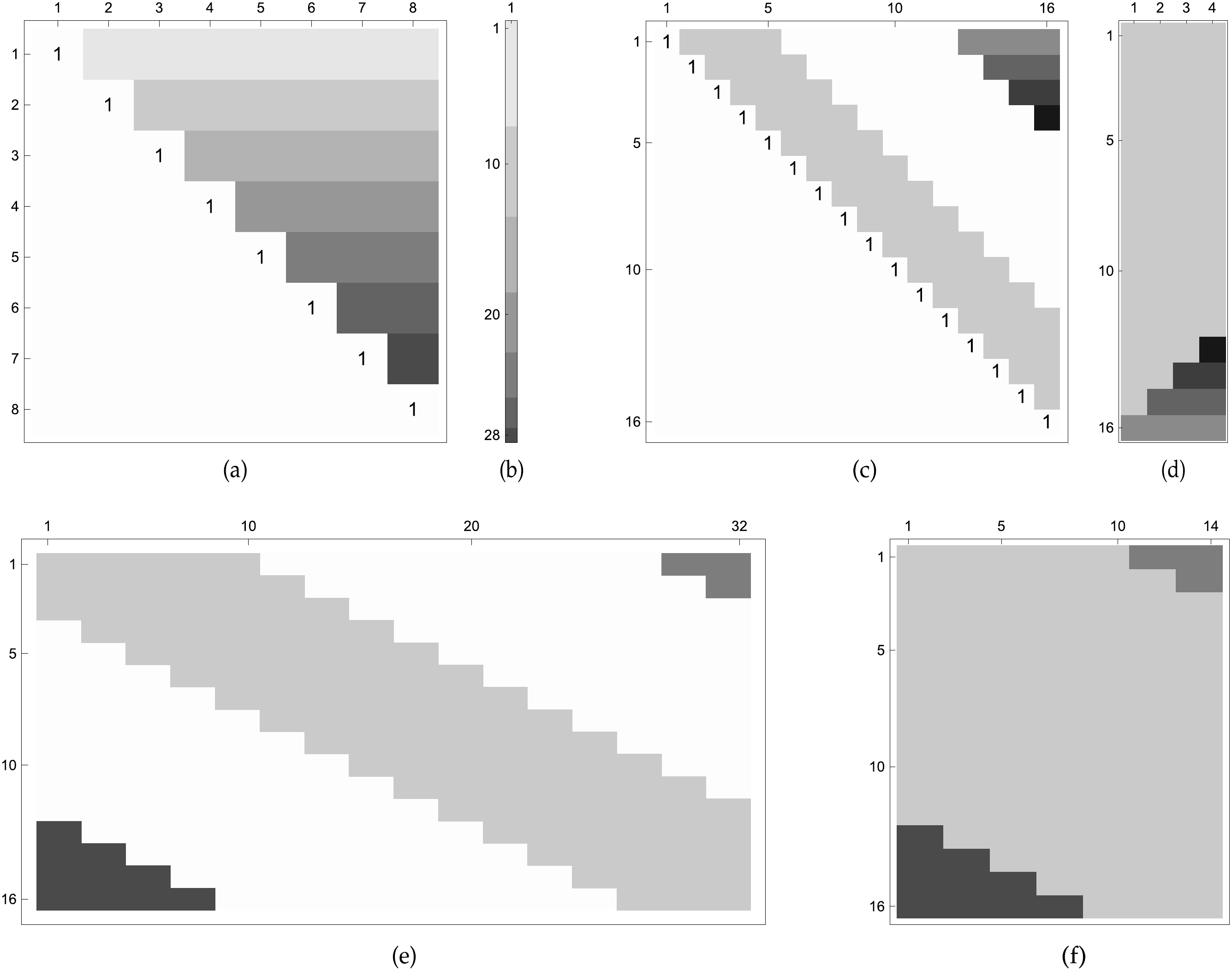}
 \caption[Schematic illustration of our method for storing the shear elements]{Schematic illustration of our method for storing shear elements in the ss and ww blocks of the upper unit-triangular matrix in the incomplete UDU decomposition of the approximate ICM in a multi-scale wavelet basis.
 To elucidate how we store the shear elements, we use gray color with different scales to show the location of nonzero elements in different parts of a block;
 gray color's scale here does not represent the relative magnitude of the elements
 (a)~Visualization of the ss block.
 This block is an upper unit-triangular matrix.
 (b)~Visualization of the vector storing the ss block's shear elements.
 (c)~Visualization of a diagonal ww block.
 This block is also an upper unit-triangular matrix.
 (d)~Visualization of the matrix storing shear elements of the diagonal ww block.
 (e)~Visualization of an off-diagonal ww block.
 (f)~Visualization of the matrix storing the shear elements of the off-diagonal ww block.
 \label{fig:shearstoring}
 }
\end{figure}

% Procedure
We now describe the procedure of our UDU-decomposition algorithm.
We employ the block variant for the UDU decomposition in our algorithm, but we only compute the nonzero shear elements in each column of~$\upbm{U}$.
The ss and sw blocks of~$\upbm{U}$ are dense matrices, so we compute all elements of these blocks.
The ww blocks, however, are sparse matrices similar to the ww blocks of~$\tilde{\upbm{A}}$.
We therefore use the helper functions in Library~\ref{lib:positionHelperFunctions} to specify the location of nonzero elements in the ww blocks.
We start from the top-right block of~$\upbm{U}$ and proceed to compute the nonzero shear elements in each block of~$\upbm{U}$ column by column.
For each column~$c$ of the block matrix~$\upbm{U}$, we compute diagonals~$d^{(c)}_{\ww;\,i}$ of~$\upbm{D}^{(c)}_\ww$, all elements of~$\upbm{U}_\sw^{(s_0,c)}$ and nonzero shear elements of $\upbm{U}_\ww^{(r,c)}$ for $r\in\{s_0,\ldots, c\}$.

% Computing diagonal elements
To compute the diagonal element $d^{(c)}_{\ww;\,i}$, first we compute nonzero elements in the~$i^\text{th}$ row of~$\upbm{V}^{(c,s)}_\ww$~\eqref{eq:matrixV} for $s\in\{c,\ldots,k-1\}$.
Then we multiply each nonzero element in the~$i^\text{th}$ row of~$\upbm{V}^{(c,s)}_\ww$ to its corresponding nonzero element in the~$i^\text{th}$ row of $\upbm{U}^{(c,s)}$ and add the results.
Next we negate the result and add $a^{(c,c)}_{\ww;\,i,i}$ to obtain~$d^{(c)}_{\ww;\,i}$.
The~$i^\text{th}$ diagonal element in the ss block of~$\upbm{D}$ is computed by a similar procedure. 
In this case, we compute all elements in the~$i^\text{th}$ row of the ss and sw blocks of~$\upbm{V}$~\eqref{eq:matrixV}
because these blocks are dense matrices.

% Computing shear elements
To compute the nonzero shear element at entry~$(i,j)$ of a ww block at entry~$(r,c)$ of the block matrix $\upbm{U}$,
we multiply nonzero elements in the~$j^\text{th}$ row of~$\upbm{V}^{(c,s)}$ by their corresponding entries in the~$i^\text{th}$ row of~$\upbm{U}^{(c,s)}$ for $s\in\{c,\ldots,k-1\}$, as per~\cref{eq:wwShears};
note that off-diagonal elements of diagonal blocks and all elements of off-diagonal blocks are shear elements.
We then add them all and negate the result.
Next we add~$a_{\ww;\,i,j}^{(r,c)}$ to the obtained value and divide the result by~$d^{(c)}_{\ww;\, i}$.
The final result is the shear element at the~$(i,j)$ entry of~$(r,c)$ block of~$\upbm{U}$.
The shear elements in the ss and sw blocks are computed by a similar procedure, but being dense matrices, all elements in these blocks must be computed.

\begin{algorithm}[H]
\caption{Classical algorithm for UDU decomposition of the approximate ICM in multi-scale wavelet basis}
\label{alg:invCovUDU}
\begin{algorithmic}[1]
\Require{
    \Statex $\dbIndex \in \integers_{\geq 3}$
    \Comment{wavelet index}
    \Statex $N \in \integers_{\geq 2(2\dbIndex-1)}$
    \Comment{number of modes; for convenience, we assume $N$ is a power of 2}
    \Statex $w \in \integers^{+}$
    \Comment{upper bandwidth of diagonal \ww\ blocks in the approximate ICM $\tilde{\upbm{A}}$}
    \Statex $\upbm{a}:=\cBraket{\left.
    \upbm{a}^{(s_0)}_\sS \in \reals^{2^{s_0}},
    \upbm{a}^{(s_0, c)}_\sw \in \reals^{2^c},
    \upbm{a}^{(r, c)}_\ww \in \reals^{2^{c-r}}
    \right\vert
    \ceil*{\log_2(4\dbIndex-2)} =:s_0\leq r \leq c < k:=\log_2 N}$
    \Comment{circulant rows of ICM~$\upbm{A}$}
}
\Ensure{
    \Statex $\upbm{d} \in \reals^N$
    \Comment{diagonals of the diagonal matrix~$\upbm{D}$ in UDU decomposition of $\tilde{\upbm{A}}$}
    \Statex
    $\upbm{S}:=\cBraket{\left.
    \upbm{s}^{(s_0)}_\sS \in \reals^{\frac1{2}2^{s_0}(2^{s_0}-1)},
    \upbm{S}^{(s_0,c)}_\sw \in \reals^{2^{s_0}\times 2^c},
    \upbm{S}^{(r,c)}_\ww \in \reals^{2^r\times\left(\textsc{width}(r,c)
    -(w+1)\updelta_{rc}\right)}
    \right\vert s_0 \leq r \leq c < k
    }$
    \Comment{shear elements in main and upper-diagonal blocks of~$\upbm{U}$ in UDU decomposition of $\tilde{\upbm{A}}$;
    here $\textsc{width}(r,c)$~\eqref{eq:block_bandwidth} is bandwidth of $\tilde{\upbm{A}}^{(r,c)}_\ww$}
}    
\Function{invCovUDU}{$\dbIndex,N,w,\upbm{a}$}
\For{$c \gets k-1$ to $s_0$}
\Comment{iterates over columns of $\upbm{U}$}
\For{$i \gets 2^c-1$ to $0$}
\Comment{iterates over columns/rows of $\upbm{U}_\ww^{(c,c)}$ from right/bottom to left/top}
\State $d^{(c)}_{\ww;\,i} \gets a^{(c,c)}_{\ww;\,i,i}$
\For{$s \gets c$ to $k-1$}
    \Comment{lines~(\ref{line:nonzeroVkr-1}--\ref{line:nonzeroVkr-2}) compute
    $d^{(c)}_i$ and nonzero elements in $i^\text{th}$ row of $\upbm{V}_\ww^{(c,s)}$}
    \State $(d^{(c)}_{\ww;\,i}, \upbm{v}^{(c,s)}_{\ww,i})\stackrel{-}{\gets}
    \textsc{diagAndV}\left(i,c,s,\upbm{u}^{(c,s)}_{\ww;\,i}, \upbm{d}_\ww^{(s)}\right)$
\EndFor
% sw block
    \State $(i_1,i_2) \gets (0, 2^{s_0}-1)$
    \State $\upbm{u}^{(s_0,c)}_{\sw;\,i_1:i_2,i}
            \gets \upbm{a}^{(s_0,c)}_{\sw;\,i_1:i_2,i}/d^{(c)}_{\ww;\,i}$
\For{$s \gets c$ to $k-1$}
    \State $\upbm{u}^{(s_0,c)}_{\sw;\,i_1:i_2,i} \stackrel{+}{\gets} \textsc{mainShears}\left(i_1,i_2,i,r,c,s,d^{(c)}_{\ww;\,i}, \upbm{u}^{(s_0,s)}_{\sw;\,i_1:i_2}, \upbm{v}^{(c,s)}_{\ww;\,i}\right)$
\EndFor
\For{$r \gets s_0$ to $c$} \label{line:Ulk-1}
\Comment{lines~(\ref{line:Ulk-1}--\ref{line:Ulk-2}) compute nonzeros in TRP, MP and BLP of $\upbm{U}_\ww^{(r, c)}$ ($r < c$)}
%%% main part
\If{$r \neq c$}
\State $(i_\text{F},i_\text{L}) \gets \textsc{mainRow}(i,r,c)$ \label{line:mainUlk-1}
\State $\upbm{u}^{(r,c)}_{\ww;\,i_\text{F}:i_\text{L},i} \gets    
        \upbm{a}^{(r, c)}_{\ww;\,i_\text{F}:i_\text{L},i}
        /d^{(c)}_{\ww;\,i}$
\For{$s \gets c$ to $k-1$}
    \State $\upbm{u}^{(r, c)}_{\ww;\, i_\text{F}:i_\text{L},i} \stackrel{+}{\gets} \textsc{mainShears}\left(i_\text{F},i_\text{L},i,r,c,s, d^{(c)}_{\ww;\,i}, \upbm{u}^{(r,s)}_{\ww;\,i_\text{F}:i_\text{L}}, \upbm{v}^{(c,s)}_{\ww;\,i}\right)$
\EndFor
\EndIf
%%% top-right part
    \If{$i\geq 2^c - w$}
    \Comment{lines~(\ref{line:topRightUlk-1}--\ref{line:topRightUlk-2}) compute nonzeros in TRP of $\upbm{U}_\ww^{(r, c)}$}
        \State $(i_\text{F},i_\text{L})\gets
        \left(0, \textsc{lastColNNZ}(r,c)-1
        -\floor*{\frac{2^c-1-i}{2^{c-r}}}\right)$
        \label{line:topRightUlk-1}
        \Comment{row index of last nonzero in $i^\text{th}$ column in TRP of $\upbm{U}_\ww^{(r, c)}$}
        \State $\upbm{u}^{(r,c)}_{\ww;\,i_\text{F}:i_\text{L},i}
        \gets\upbm{a}^{(r, c)}_{\ww;\,i_\text{F}:i_\text{F},i}
        /d^{(c)}_{\ww;\,i}$
        \For{$s \gets c$ to $k-1$}
           \State $\upbm{u}^{(r,c)}_{\ww;\,i_\text{F}:i_\text{L},i}
           \stackrel{+}{\gets} \textsc{topRightShears}
           \left(i_\text{F},i_\text{L},i,r,c,s, d^{(c)}_{\ww;\,i}, \upbm{u}^{(r,s)}_{\ww;\,i_\text{F}:i_\text{L}}, \upbm{v}^{(c,s)}_{\ww;\,i}\right)$
        \EndFor
    \EndIf
%%% bottom part
\If{$i< \textsc{width}(r,c) -w-2^{c-r}$ and $r \neq c$}
    \label{line:bottomLeftUlk-1}
    \Comment{lines~(\ref{line:bottomLeftUlk-1}--\ref{line:bottomLeftUlk-2}) compute nonzeros in BLP of $\upbm{U}_\ww^{(r, c)}$}
    \State $(i_1,i_2) \gets \left(i_\text{F}+\floor*{\frac{i}{2^{c-r}}},\, 2^r-1\right)$
    \State $\upbm{u}^{(r,c)}_{\ww;\, i_1:i_2,\, i} \gets
    \upbm{a}^{(r,c)}_{\ww;\, i_1:i_2,i}/d^{(c)}_i$
\For{$s \gets c$ to $k-1$}
    \State $\upbm{u}^{(r,c)}_{\ww;\, i_1:i_2,\, i} \stackrel{+}{\gets} \textsc{bottomLeftShears}(r,c,i,i_1,i_2, d^{(c)}_i, \upbm{U}^{(r,s)}_{\ww;\, i_1:i_2}, \upbm{v}^{(c,s)}_i,\dbIndex,w)$
\EndFor
\EndIf
\EndFor \label{line:Ulk-2}
\If{$i > 0$}
\label{line:mainUkk-1}
    \Comment{lines (\ref{line:mainUkk-1}--\ref{line:mainUkk-2})
    ) compute nonzeros in MP of $\upbm{U}^{(c,c)}_\ww$}
    \State $(i_\text{F},i_\text{L}) \gets (\max(0,i-w),i-1)$
    \State $\upbm{u}^{(c,c)}_{\ww;\,i_\text{F}:i_\text{L},i}\gets
            \upbm{a}^{(c,c)}_{\ww;\,i_\text{F}:i_\text{L},i}
            /d^{(c)}_{\ww;\,i}$
    \For{$s \gets c$ to $k-1$}
        \State
        \label{line:mainUkk-2}
        $\upbm{u}^{(c,c)}_{\ww;\,i_\text{F}:i_\text{L},i} \stackrel{+}{\gets}\textsc{mainShears}
        \left(i_\text{F},i_\text{L},i,r,c,s, d^{(c)}_{\ww;\,i}, \upbm{u}^{(c,s)}_{\ww;\,i_\text{F}:i_\text{L}}, \upbm{v}^{(c,s)}_{\ww;\,i}\right)$
    \EndFor
\EndIf
\EndFor
\EndFor
\For{$i \gets 2^{s_0}-1$ to $0$}
\label{line:nonzeroVsw-1}
\Comment{lines~(\ref{line:nonzeroVsw-1}--\ref{line:nonzeroVsw-2}) compute $d_i$ and nonzero elements in $i^\text{th}$ row of $\upbm{V}_\sS^{(s_0)}$ and $\upbm{V}_\sw^{(s_0,c)}$}
    \State $(j_\text{F},j_\text{L}) \gets (i+1,2^{s_0}-1)$
    \State $\upbm{v}^{(s_0)}_{\sS;\,i,j_\text{F}:j_\text{L}}
            \gets \upbm{u}^{(s_0)}_{\sS;\,i,j_\text{F}:j_\text{L}}
            \odot \upbm{d}_{j_\text{F}:j_\text{L}}$
    \State $d_i \gets a^{(s_0)}_{\sS;\,i,i}
            -\upbm{u}^{(s_0)}_{\sS;\,i,j_\text{F}:j_\text{L}}
            \odot \upbm{v}^{(s_0)}_{\sS;\,j_\text{F}:j_\text{L}}$
    \For{$s \gets s_0$ to $k-1$}
        \State $v^{(s_0,s)}_{\sw;\,i,0:2^s-1}
            \gets \upbm{u}^{(s_0,c)}_{\sw;\,i,0:2^s-1}
            \odot \upbm{d}_{0:2^s-1}$
        \State $d_i \stackrel{-}{\gets}
                \upbm{u}^{(s_0,s)}_{\sw;\,i,0:2^s-1}
                \cdot \upbm{v}^{(s_0,s)}_{\sw;\,i,0:2^s-1}$
    \EndFor \label{line:nonzeroVsw-2}
\algstore{invCovLDLPart2}
\end{algorithmic}
\end{algorithm}
\begin{algorithm}[H]
\ContinuedFloat
\caption{Classical algorithm for UDU decomposition of the approximate ICM in multi-scale wavelet basis (continued)}
\begin{algorithmic}
\algrestore{invCovLDLPart2}
    \If{$i>0$} \label{line:upperDiagonalUss-1}
        \Comment{lines~(\ref{line:upperDiagonalUss-1}--\ref{line:upperDiagonalUss-2}) compute upper-diagonal elements of $\upbm{U}^{(s_0)}_\sS$}
        \State $\upbm{u}^{(s_0)}_{\sS;\,0:i-1,i} \gets
        \left[ a^{(s_0)}_{\sS;\,0:i-1,i}
        -\upbm{u}^{(s_0)}_{\sS;\,0:i-1,i+1:2^{s_0}-1}
        \cdot \upbm{v}^{(s_0)}_{\sS;\,i,i+1:2^{s_0}-1} 
        \right]/d_i$
        \For{$s \gets s_0$ to $k-1$}
            \State $\upbm{u}^{(s_0)}_{\sS;\,0:i-1,i}
            \stackrel{-}{\gets}
            \left[\upbm{u}^{(s_0,s)}_{\sw;\,0:i-1,0:2^s-1}
            \cdot \upbm{v}^{(s_0,s)}_{\sw;\,i,0:2^s-1} 
            \right]/d_i$
        \EndFor
    \EndIf \label{line:upperDiagonalUss-2}
\EndFor

\State $\upbm{s}^{(s_0)} \gets\textsc{upperUnitriangStore}
        \left(2^{s_0}, \upbm{U}_\sS^{(s_0)}\right)$
\Comment{stores shear elements of $\upbm{U}_\sS^{(s_0)}$ as a vector; see line~\ref{func:upperUnitriangStore} of~Library~\ref{lib:storingHelperFunctions}}
\For{$c \gets s_0$ to $k-1$}
    \State $\upbm{d}_{2^c:2^{c+1}-1}
            \gets \upbm{d}^{(c)}_{\ww;\,0:2^c-1}$
    \State $\upbm{S}^{(s_0,c)}_\sw \gets \upbm{U}^{(s_0,c)}_\sw$
    \State $\upbm{S}^{(c,c)}_\ww \gets \textsc{diagStore}
            \left(c,w,\upbm{U}^{(c,c)}_\ww\right)$
    \Comment{stores shear elements of $\upbm{U}_\ww^{(c,c)}$ as a matrix; see line~\ref{func:diagStore} of~Library~\ref{lib:storingHelperFunctions}}
\EndFor
\For{$r \gets s_0$ to $k-1$}
    \For{$c \gets r+1 $ to $k-1$}
        \State $\upbm{S}^{(r,c)}_\ww \gets \textsc{offDiagStore}
                \left(r,c,\dbIndex,w, \upbm{U}^{(r,c)}_\ww\right)$
        \Comment{stores shear elements of $\upbm{U}_\ww^{(r,c)}$ as a matrix; see line~\ref{func:offDiagStore} of~Library~\ref{lib:storingHelperFunctions}}
    \EndFor
\EndFor
\State \Return
$\left(\upbm{d}, \upbm{S}
\right)$
\EndFunction
\end{algorithmic}
\end{algorithm}

\begin{library}[H]
\caption{Helper functions for storing shear elements of the upper unit-triangular matrix in UDU decomposition}
\label{lib:storingHelperFunctions}
\begin{algorithmic}[1]
\Function{upperUnitriangStore}{$N,\upbm{U}$}
\label{func:upperUnitriangStore}
\Comment{stores shear elements of an $N\times N$ upper unit-triangular matrix~$\upbm{U}$}
\For{$i \gets 0$ to $N-2$}
    \For{$j \gets i+1 $ to $N-1$}
        \State $s_{iN-i(i+1)/2 +j-(i+1)} \gets u_{ij}$
    \EndFor
\EndFor
\State \Return $\upbm{s}$
\EndFunction
\Function{diagStore}{$c,\upbm{U}_\ww^{(c,c)}$}
\label{func:diagStore}
\Comment{stores shear elements of a diagonal \ww\ block of~$\upbm{U}$}
\For{$i \gets 0 $ to $2^c-1$}
    \Comment{iterates over rows of $\upbm{U}^{(c,c)}_\ww$}
    \State $(j_\text{FM},j_\text{LM}) \gets
            (i+1, \min(i+w,2^c-1))$
    \Comment{computes column index of first and last shear elements of $i^\text{th}$ row in MP of $\upbm{U}^{(c,c)}_\ww$}
    \State $\upbm{s}^{(c,c)}_{\ww;\,i,0:j_\text{LM}-j_\text{FM}} \gets \upbm{u}^{(c,c)}_{\ww;\,i,j_\text{FM}:j_\text{LM}}$
    \If{$i\leq w-1$}
    \Comment{$w-1$ is row index of last nonzero in last column of $\upbm{U}^{(c,c)}_\ww$}
        \State $\upbm{s}^{(c,c)}_{\ww;\,2^c-1-i,\,i:w-1}
        \gets \upbm{u}^{(c,c)}_{\ww;\,i,\,2^c-w-i:2^c-1}$
        \Comment{stores shear elements of $i^\text{th}$ row in TRP of $\upbm{U}^{(c,c)}_\ww$ into $(2^c-1-i)^\text{th}$ row of~$\upbm{S}^{(c,c)}_\ww$}
    \EndIf
\EndFor
\State \Return $\upbm{S}_\ww^{(c,c)}$
\EndFunction
\Function{offDiagStore}{$r,c,\upbm{U}_\ww^{(r,c)}$}
\label{func:offDiagStore}
\Comment{stores shear elements of the off-diagonal \ww\ block at entry $(r,c)$ of~$\upbm{U}$}
\State $h \gets \textsc{lastColNNZ}(r,c)$
\State $H \gets \textsc{vertWidth}(r,c)$
\State $W \gets \textsc{width}(r,c)$
\State $(i_\text{LT},i_\text{FB}) \gets
        \left(h-1, 2^r-(H-h-1)\right)$
\Comment{$i_\text{LT}$ is row index of last column's last nonzero in TRP and $i_\text{FB}$ is row index of first column's first nonzero in BLP of $\upbm{U}^{(r,c)}_\ww$}
\For{$i \gets 0$ to $2^r-1$}
\Comment{iterates over rows of $\upbm{U}^{(r,c)}_\ww$}
    \State $(j_\text{FM},j_\text{LM}) \gets
            \textsc{mainCol}(i,r,c)$
    \Comment{computes column index of first and last nonzero in $i^\text{th}$ row in MP of $\upbm{U}^{(r,c)}_\ww$}
    \If{$i \leq i_\text{LT}$}
        \State $\upbm{s}^{(r,c)}_{\ww;\,i,0:j_\text{LM}}
                \gets \upbm{u}^{(r,c)}_{\ww;\,i,0:j_\text{LM}}$
                \Comment{note that $j_\text{FM}=0$ for $i< i_\text{LT}$}
        \State $(j_\text{FT}, j_\text{LT}) \gets
                (2^c-w+i2^{c-r},2^c-1)$
        \Comment{computes column index of first and last nonzero in $i^\text{th}$ row in TRP of $\upbm{U}^{(r,c)}_\ww$}
        \State $\upbm{s}^{(r,c)}_{\ww;\,i,j_\text{LM}+1:W-1}\gets
                \upbm{u}^{(r,c)}_{\ww;\,i,j_\text{FT}:j_\text{LT}}$
    \ElsIf{$i_\text{LT} <i< i_\text{FB}$}
        \State $\upbm{s}^{(r,c)}_{\ww;\, i, 0:W-1}\gets
                \upbm{u}^{(r,c)}_{\ww;\,i,j_\text{FM}:j_\text{LM}}$
    \Else
        \State $(j_\text{FB},j_\text{LB})\gets
                (0, (i-i_\text{FB}) 2^{c-r}-1)$
        \Comment{computes column index of first and last nonzero in $i^\text{th}$ row in BLP of $\upbm{U}^{(r,c)}_\ww$}
        \State $\upbm{s}^{(r,c)}_{\ww;\,i,0:j_\text{LB}}
                \gets \upbm{u}^{(r,c)}_{\ww;\,i,0:j_\text{LB}}$
        \State $\upbm{s}^{(r,c)}_{\ww;\,i,j_\text{LB}+1:W-1}\gets
                \upbm{u}^{(r,c)}_{\ww;\,i,j_\text{FM}:j_\text{LM}}$ 
    \EndIf
\EndFor
\State \Return $\upbm{S}_\ww^{(r,c)}$
\EndFunction
\end{algorithmic}
\end{library}

\begin{library}[H]
\caption{Helper functions for computing shear elements}
\label{lib:shearsHelperFunctions}
\begin{algorithmic}[1]
\Function{diagAndV}{$i,c,s,\upbm{u}^{(c,s)}_{\ww;\,i}, \upbm{d}_\ww^{(s)}$}
\State $(j_\text{F},j_\text{L}) \gets \textsc{mainCol}(i,c,s)$ \label{line:nonzeroVkr-1}
\State $\upbm{v}^{(c,s)}_{\ww;\,i,j_\text{F}:j_\text{L}}
        \gets \upbm{u}^{(c,s)}_{\ww;\,i,j_\text{F}:j_\text{L}}
        \odot \upbm{d}^{(s)}_{\ww;\,j_\text{F}:j_\text{L}}$
\State $d_{\ww;\,i}^{(c)} \stackrel{+}{\gets}
        \upbm{u}^{(c,s)}_{\ww;\,i,j_\text{F}:j_\text{L}}
        \cdot \upbm{v}^{(c,s)}_{\ww;\,i,j_\text{F}:j_{\text{L}}}$
\Comment{$\stackrel{+}{\gets}$ denotes the addition assignment, i.e., $a \stackrel{+}{\gets} b := a \gets a+b$
}
\State $(h,H) \gets (\textsc{lastColNNZ}(c,s),\, \textsc{vertWidth}(c,s))$
\State $(i_\text{LT}, i_\text{FB}) \gets (h-1,\, 2^c-(H-h-1))$
\If{$i\leq i_\text{LT}$}
\Comment{$i_\text{L}$ is row index of last nonzero in last column in TRP of $\upbm{U}_\ww^{(c,s)}$}
    \State $(j_\text{F},j_\text{L}) \gets (2^s-w+i2^{s-c},2^s-1)$
    \Comment{$j_\text{F}$ is column index of first nonzero in $i^\text{th}$ row in TRP of $\upbm{U}_\ww^{(r,c)}$}
    \State $\upbm{v}^{(c,s)}_{\ww;\,i,j_\text{F}:j_\text{L}}
        \gets \upbm{u}^{(c,s)}_{\ww;\,i,j_\text{F}:j_\text{L}}
        \odot \upbm{d}^{(s)}_{\ww;\,j_\text{F}:j_\text{L}}$
    \State $d_{\ww;\,i}^{(c)} \stackrel{+}{\gets}
        \upbm{u}^{(c,s)}_{\ww;\,i,j_\text{F}:j_\text{L}}
        \cdot \upbm{v}^{(c,s)}_{\ww;\,i,j_\text{F}:j_{\text{L}}}$
    \Comment{multiplies TRP of $\upbm{U}_\ww^{(c,s)}$ by TRP of $\upbm{V}_\ww^{(c,s)}$}
\EndIf
\If{$i \geq i_\text{FB}$ and $s \neq c$}
\Comment{$i_\text{F}$ is row index of first nonzero in first column in BLP}
    \State $(j_\text{F},j_\text{L}) \gets (0,(i-i_\text{FB})2^{s-c}-1)$
    \Comment{$j_\text{L}$ is column index of last nonzero in $i^\text{th}$ row in BLP of $\upbm{U}_\ww^{(c,s)}$}
    \State $\upbm{v}^{(c,s)}_{\ww;\,i,j_\text{F}:j_\text{L}}
        \gets \upbm{u}^{(c,s)}_{\ww;\,i,j_\text{F}:j_\text{L}}
        \odot \upbm{d}^{(s)}_{\ww;\,j_\text{F}:j_\text{L}}$
    \State $d_{\ww;\,i}^{(c)} \stackrel{+}{\gets}
        \upbm{u}^{(c,s)}_{\ww;\,i,j_\text{F}:j_\text{L}}
        \cdot \upbm{v}^{(c,s)}_{\ww;\,i,j_\text{F}:j_{\text{L}}}$
    \Comment{multiplies BLP of $\upbm{U}_\ww^{(c, s)}$ by BLP of $\upbm{V}_\ww^{(c,s)}$}
\EndIf \label{line:nonzeroVkr-2}
\State \Return $(d^{(c)}_{\ww;\,i}, \upbm{v}^{(c,s)}_{\ww;\,i})$
\EndFunction
\Function{mainShears}{$i_1,i_2,i,r,c,s,d^{(c)}_{\ww;\,i}, \upbm{u}^{(r,s)}_{\ww;\,i_1:i_2}, \upbm{v}^{(c,s)}_{\ww;\,i}$}
    \State $(j_\text{F},j_\text{L}) \gets \textsc{mainCol}(i,c,s)$
    \State $\upbm{u}^{(r, c)}_{\ww;\,i_1:i_2,i}
    \stackrel{+}{\gets} 
    \left[\upbm{u}^{(r,s)}_{\ww;\,i_1:i_2,j_\text{F}:j_\text{L}}
          \cdot \upbm{v}^{(c,s)}_{\ww;\,i,j_\text{F}:j_\text{L}}
    \right]/d^{(c)}_{\ww;\,i}$
    \Comment{contribution to shear elements from multiplying MP of $\upbm{U}_\ww^{(r,s)}$ to MP of~$\upbm{V}_\ww^{(c,s)}$}
    \State $(h,H) \gets (\textsc{lastColNNZ}(c,s),\, \textsc{vertWidth}(c,s))$
    \State $(i_\text{LT}, i_\text{FB}) \gets (h-1,\, 2^c-(H-h-1))$
    \If{$i\leq i_\text{LT}$}
    \Comment{$i_\text{L}$ is row index of last nonzero in last column in TRP of $\upbm{V}_\ww^{(c,s)}$}
        \State $(j_\text{F},j_\text{L}) \gets (2^s-w+i2^{s-c},2^s-1)$
        \Comment{$j_\text{F}$ is column index of first nonzero in $j^\text{th}$ row in TRP of $\upbm{V}_\ww^{(c,s)}$}
        \State $\upbm{u}^{(r,c)}_{\ww;\,i_1:i_2,i}
                \stackrel{+}{\gets}
                \left[\upbm{u}^{(r,s)}_{\ww;\,i_1:i_2,j_\text{F}:j_\text{L}}
                \cdot \upbm{v}^{(c,s)}_{\ww;\,i,j_\text{F}:j_\text{L}}
                \right]/d^{(c)}_{\ww;\,i}$
        \Comment{contribution from multiplying TRP of $\upbm{U}_\ww^{(r,s)}$ to TRP of $\upbm{V}_\ww^{(c,s)}$}
    \EndIf
    \If{$i \geq i_\text{FB}$ and $s \neq c$}
    \Comment{$s \neq c$ as $\upbm{V}_\ww^{(c,c)}$ does not have BLP part}
        \State $(j_\text{F},j_\text{L}) \gets (0,(i-i_\text{FB})2^{s-c}-1)$
        \State $\upbm{u}^{(r, c)}_{\ww;\, i_1:i_2,\, i}
        \stackrel{+}{\gets}
        \left[\upbm{u}^{(r,s)}_{\ww;\,i_1:i_2,j_\text{F}:j_\text{L}}
              \cdot \upbm{v}^{(c,s)}_{\ww;\,i,j_\text{F}:j_\text{L}}
        \right]/d^{(c)}_{\ww;\,i}$
        \Comment{contribution to shear elements from multiplying BL of $\upbm{U}_\ww^{(r,s)}$ to BLP of $\upbm{V}_\ww^{(c,s)}$}
    \EndIf \label{line:mainUlk-2}
\State \Return $\upbm{u}^{(r,c)}_{\ww;\,i_1:i_2,i}$
\EndFunction
\Function{topRightShears}{$i_1,i_2,i,r,c,s, d^{(c)}_{\ww;\,i}, \upbm{u}^{(r,s)}_{\ww;\,i_1:i_2}, \upbm{v}^{(c,s)}_{\ww;\,i}$}
    \State $(j_\text{F},j_\text{L}) \gets \textsc{mainCol}(i,c,s)$
    \State $\upbm{u}^{(r, c)}_{\ww;\,i_1:i_2,i}
    \stackrel{+}{\gets} 
    \left[\upbm{u}^{(r,s)}_{\ww;\,i_1:i_2,j_\text{F}:j_\text{L}}
          \cdot \upbm{v}^{(c,s)}_{\ww;\,i,j_\text{F}:j_\text{L}}
    \right]/d^{(c)}_{\ww;\,i}$
    \Comment{contribution to shear elements from multiplying TRP of $\upbm{U}_\ww^{(r,s)}$ to MP of~$\upbm{V}_\ww^{(c,s)}$}
    \State $(h,H) \gets (\textsc{lastColNNZ}(c,s),\, \textsc{vertWidth}(c,s))$
    \State $i_\text{FB} \gets 2^c-(H-h-1)$
    \Comment{$i_\text{F}$ is row index of first nonzero in first column in BLP of $\upbm{V}_\ww^{(c,s)}$}
    \If{$i\geq i_\text{FB}$ and $s \neq c$}
    \Comment{BLP of $\upbm{V}_\ww^{(c,s)}$ has nonzeros only for rows $i>i_\text{F}$;
    $s \neq c$ as $\upbm{V}_\ww^{(c,c)}$ does not have BLP part}
        \State $(j_\text{F},j_\text{L}) \gets (0,(i-i_\text{FB})2^{s-c}-1)$
        \State $\upbm{u}^{(r, c)}_{\ww;\, i_1:i_2,\, i}
        \stackrel{+}{\gets}
        \left[\upbm{u}^{(r,s)}_{\ww;\,i_1:i_2,j_\text{F}:j_\text{L}}
              \cdot \upbm{v}^{(c,s)}_{\ww;\,i,j_\text{F}:j_\text{L}}
        \right]/d^{(c)}_{\ww;\,i}$
        \Comment{multiplies TLP of $\upbm{U}_\ww^{(s,r)}$ by BLP of $\upbm{V}_\ww^{(k,r)}$}
    \EndIf \label{line:topRightUlk-2}
    \State \Return $\upbm{u}^{(r,c)}_{\ww;\, i_1:i_2,\, i}$
\EndFunction
\Function{bottomLeftShears}{$i_1,i_2,i,r,c,s,d^{(c)}_{\ww;\,i}, \upbm{u}^{(r,s)}_{\ww;\,i_1:i_2}, \upbm{v}^{(c,s)}_{\ww;\,i}$}
    \State $(j_\text{F},j_\text{L}) \gets \textsc{mainCol}(i,c,s)$
    \State $\upbm{u}^{(r, c)}_{\ww;\,i_1:i_2,i}
    \stackrel{+}{\gets} 
    \left[\upbm{u}^{(r,s)}_{\ww;\,i_1:i_2,j_\text{F}:j_\text{L}}
          \cdot \upbm{v}^{(c,s)}_{\ww;\,i,j_\text{F}:j_\text{L}}
    \right]/d^{(c)}_{\ww;\,i}$
\Comment{contribution to shear elements from multiplying BLP of $\upbm{U}_\ww^{(r,s)}$ to MP of $\upbm{V}_\ww^{(c,s)}$}
\If{$i\leq i_\text{TL}\gets \textsc{lastColNNZ}(c,s,w)-1$}
    \State $(j_\text{F},j_\text{L}) \gets (2^s-w+i2^{s-c},2^s-1)$
    \State $\upbm{u}^{(r,c)}_{\ww;\,i_1:i_2,i}
                \stackrel{+}{\gets}
                \left[\upbm{u}^{(r,s)}_{\ww;\,i_1:i_2,j_\text{F}:j_\text{L}}
                \cdot \upbm{v}^{(c,s)}_{\ww;\,i,j_\text{F}:j_\text{L}}
                \right]/d^{(c)}_{\ww;\,i}$
    \Comment{contribution to shear elements from multiplying MP of $\upbm{U}_\ww^{(r,s)}$ to TRP of $\upbm{V}_\ww^{(c,s)}$}
\EndIf \label{line:bottomLeftUlk-2}
\State \Return $\upbm{u}^{(r,c)}_{\ww;\, i_1:i_2,\, i}$
\EndFunction
\end{algorithmic}
\end{library}

%%%%%%%%%%%%%%%%%%%%%%%%%%%%%%%%%%%%%%%%%%%%%%%%%%%%%%%%%
%%%%%% Quantum algorithms %%%%%%%%%%%%%%%%%%%%%%%%%%%%%%%
%%%%%%%%%%%%%%%%%%%%%%%%%%%%%%%%%%%%%%%%%%%%%%%%%%%%%%%%%

\subsection{Quantum algorithms}
\label{subsec:quantum_algorithms}

In this subsection, we construct the quantum routine of our Fourier- and wavelet-based algorithms for ground-state generation.
The quantum routine of our algorithms has two subroutines.
The first subroutine generates an approximation for a continuous 1DG state, and the second subroutine executes a basis transformation.

We have two quantum algorithms for generating a 1DG state.
The first algorithm is presented in~\cref{subsubsec:1DG_generation}, and the second algorithm, which is based on inequality testing,
is described in~\cref{subsubsec:inequality_testing}.
We present our algorithm for quantum fast Fourier transform~(QFFT) in~\cref{subsubsec:QFFT}, and the algorithm for quantum shear transform~(QST) in~\cref{subsubsec:QST}.
The QFFT and QST algorithms serve as the basis-transformation subroutine in the Fourier- and wavelet-based algorithms, respectively.

%%%%%%%%%%%%%%%%%%%%%%%%%%%%%%%%%%%%%%%%%%%%%%%%%%%%%%%%%
\subsubsection{One-dimensional Gaussian-state generation}
\label{subsubsec:1DG_generation}

Here we present our first quantum algorithm for generating a discrete approximation for a continuous 1DG state~\eqref{eq:continuous1DG} on~a quantum register.
We begin by explaining the inputs and output of the algorithm.
We then describe the involved quantum registers in our algorithm and explain the algorithm's procedure.
Finally, we present our algorithm as pseudocode.

% Inputs/output
The output of our 1DG-state-generation algorithm is a discrete 1DG state with the standard deviation~$\tilde{\sigma}$ over a lattice with~$2^m$ points and lattice spacing~$\delta$; see~\cref{subsubsec:1DG_method} for a description of our method for approximating a continuous 1DG state.
The discrete 1DG state~\eqref{eq:lattice1DG} is a linear combination of basis states $\ket{j\delta}$, where~$\delta$ is a real number.
Therefore, the discrete 1DG state can be regarded as a superposition of real numbers.

For convenience, we use the fixed-point number representation~\cite[p.~255]{HH16} to treat the real numbers in our algorithm.
Specifically, we consider each real number as a $p$-bit number in this representation.
The positive integer~$p$ is the working precision in our main algorithms for ground-state generation. 
As the largest value for the real numbers is $2^m$, we use lattice parameter~$m$ to be the radix-point position, i.e., the number of bits to the left of the radix point in the fixed-point number representation.
Therefore, the following parameters are taken as classical inputs to our 1DG-state-generation algorithm:
the working precision~$p$,
the radix-point position~$m$,
and the $p$-bit numbers~$\tilde{\sigma}$ and~$\delta$.

% Involved quantum registers
To elucidate the procedure of our quantum algorithm, we now describe various quantum registers involved in the algorithm.
Our algorithm involves the following quantum registers;
all registers start in the all-zero state and have~$p$ qubits unless otherwise specified.
(1)~\out: a register that is to be prepared in the discrete 1DG state~\eqref{eq:lattice1DG},
(2)~\stDev: a register that stores the value of~$\tilde{\sigma}$,
(3)~\spacing: a register that stores the value of~$\delta$,
(4)~\mean: a register that stores a mean value~$\mu$,
(5)~\ang: a register that stores a single-qubit rotation angle, and
(6)~\scratch: a $\order{p}$-qubit register used to
assist in various operations throughout the algorithm.
We index the qubits of each register from~$0$ to~$p-1$, where the~$0^\text{th}$ qubit is the rightmost qubit.
The first $p-m$ qubits in each $p$-qubit register represent the fractional part of a number, and the last~$m$ qubits represent its integer part in the fixed-point number representation.

% Main strategy
The main strategy to generate the discrete 1DG state~\eqref{eq:lattice1DG} is as follows.
First we prepare the state
\begin{align}
\label{eq:1DG_over_integers}
\ket{\Psi(\tilde{\sigma}, \mu, m)}:=
\frac1{\sqrt{f(\tilde{\sigma},\mu,m)}}
\sum_{j=-2^{m-1}}^{2^{m-1}-1}
\e^{-\tfrac{(j+\mu)^2}{4\tilde{\sigma}^2}}\ket{j},
\quad
f(\tilde{\sigma},\mu,m):=
\sum_{j=-2^{m-1}}^{2^{m-1}-1}
\e^{-\tfrac{(j+\mu)^2}{2\tilde{\sigma}^2}},
\end{align}
with the initial value~$\mu=0$ on the leftmost~$m$ qubits of the $p$-qubit \out\ register, i.e. the qubits that represent the integer part of a number.
This state is a discrete 1DG state with the mean value~$\mu$ over a lattice with unit spacing.
The parameter~$\mu$ is used here because our algorithm for generating the state in~\cref{eq:1DG_over_integers} is iterative, and the value of~$\mu$ changes in each iteration.
Having prepared the~$\ket{\Psi(\tilde{\sigma},0,m)}$, we then transform it into the discrete 1DG state~$\ket{\G_\text{lattice}(\tilde{\sigma},\delta,m)}$ by multiplying~$j$ to~$\delta$.
To this end, we write the classical input~$\delta$ into the \spacing\ register and implement the transformation~$\ket{j}_\out \ket{\delta}_\spacing \mapsto \ket{j\delta}_\out \ket{\delta}_\spacing$.
The state prepared on the \out\ register is then~$\ket{\G(\tilde{\sigma},\delta,m)}$.

%Procedure for generating $\ket{\Psi(\tilde{\sigma}, \mu, m)}$
We now proceed with a detailed description of the procedure for generating~$\ket{\Psi(\tilde{\sigma}, \mu, m)}$.
This state is a linear combination of basis states~$\ket{j}$ where~$j$ is a signed integer.
Using two's complement to represent signed integers~\cite[p.~16]{HH16}, we recursively decompose the state as
\begin{equation}
\label{eq:our_recursive_formula}
\ket{\Psi(\tilde{\sigma}_\ell, \mu_\ell, m_\ell)}
=
\ket{\Psi(\tilde{\sigma}_\ell/2, \mu_\ell/2,m_\ell-1}
\otimes
\cos\theta_\ell \ket{0}
+
\ket{\Psi(\tilde{\sigma}_\ell/2, (\mu_\ell +1)/2, m_\ell-1)}
\otimes
\sin\theta_\ell \ket{1},
\end{equation}
where we define the four terms $\tilde{\sigma}_0:=\tilde{\sigma}$,
$\mu_0:=\mu$,
$m_\ell:=m-\ell$ and
\begin{equation}
\label{eq:rot_angle}
\theta_\ell:=
\arccos\sqrt{
\dfrac{f\left(\tilde{\sigma}_\ell/2, \mu_\ell/2, m_\ell-1\right)}
{f\left(\tilde{\sigma}_\ell, \mu_\ell, m_\ell\right)}
}
\quad \forall \ell\in\{0,\ldots,m-1\}.
\end{equation}
We use the recursive formula~\eqref{eq:our_recursive_formula} to devise an iterative algorithm for generating ~$\ket{\Psi(\tilde{\sigma}, \mu, m)}$.
We start by writing~$\tilde{\sigma}$ into the \stDev\ register and $\mu$ into the \mean\ register.
For $\ell$ from~$0$ to~$m-1$, we iteratively perform the following operations;
see the quantum circuit in~\cref{fig:1DGCircuit}.
\begin{enumerate}
    \item Compute a $p$-bit approximation for $\eta_\ell:=\theta_\ell/2\uppi$~\eqref{eq:rot_angle} and store the result in the \ang\ register.
    We perform this operation~by
    \begin{equation}
    \label{eq:angle_operation}
         \operatorname{\textsc{angle}}: \ket{\tilde{\sigma}}_\stDev\ket{\mu}_\mean \ket{0}_\ang \mapsto \ket{\tilde{\sigma}}_\stDev\ket{\mu}_\mean \ket{\eta}_\ang,
    \end{equation}
    which we describe by by~$\operatorname{\textsc{angle}}(\stDev, \mean, \ang)$ in our quantum algorithm.
    \item Perform a single-qubit rotation on~$\out[\ell]$,
    the $\ell^\text{th}$ qubit of \out,
    where the angle of rotation is read from the \ang\ register.
    The rotation is performed by implementing the operation
    \begin{equation}
    \label{eq:rot_operation}
    \operatorname{\textsc{rot}}:
    \ket{\eta}_\ang \ket{0}_{\out[\ell]}
    \mapsto
    \ket{\eta}_\ang \left(
        \cos(2\pi\eta) \ket{0}_{\out[\ell]} +
        \sin(2\pi\eta) \ket{1}_{\out[\ell]}
    \right),
    \end{equation}
    which we describe by~$\operatorname{\textsc{rot}}\left(\ang, \out[\ell]\right)$ in our algorithm.
    \item Erase \ang\ by uncomputing~$\eta_\ell$.
    We uncompute~$\eta_\ell$ by performing the $\textsc{angle}$~\eqref{eq:angle_operation} operation. 
    \item Divide the numbers stored in \stDev\ and \mean\ by two.
    To divide the number stored in \stDev\ by two, we cyclically shift the qubits of this register one qubit to the right by performing
    \begin{equation}
    \label{eq:shift_operation}
    \operatorname{\textsc{shift}}: \ket{b_{p-1}\ldots b_1b_0}_\out \mapsto \ket{b_0b_{p-1}\ldots b_1}_\out,
    \end{equation}
    and then flip the leftmost qubit of \stDev\ if the~$\ell^\text{th}$ bit of the classical input~$\tilde{\sigma}$ is~$1$.
    As we start by the initial value~$\mu=0$, the rightmost qubit of \mean\ remains in the zero state throughout the computation.
    Therefore, to divide the number stored in \mean\ by two, we only perform \textsc{shift}~\eqref{eq:shift_operation} on this register.
    \item
    Add $1/2$ to \mean\ if the state of $\out[\ell]$ is $\ket{1}$.
    As we start by~$\mu=0$, after dividing the value encoded in \mean\ by two in the previous step, the~$(p-m-1)^\text{th}$ qubit of this register is~$\ket{0}$.
    Therefore, we add~$1/2$ to \mean\ by flipping the state of this qubit from~$\ket{0}$ to~$\ket{1}$.
\end{enumerate}
To simplify the readability of our 1DG-state-generation algorithm, we describe these iterative operations by a helper function in~Library~\ref{lib:1DG}.

\begin{library}[H]
  \caption{Helper function for \oneDG-state generation in Algorithm~\ref{alg:1DG}}
  \label{lib:1DG}
  \begin{algorithmic}[1]
\Function{iterOneDG}{$\ell,m,p,\stDev, \mean, \ang$}
    \Comment{performs the iterative steps described in~\cref{subsubsec:1DG_generation}}
   \State $\operatorname{\textsc{angle}}(\stDev, \mean, \ang)$
    \Comment{computes $\theta_\ell/2\uppi$~\eqref{eq:rot_angle} into the \ang\ register as per~\cref{eq:angle_operation}}
    \State $\operatorname{\textsc{rot}}(\ang, \out[\ell])$
    \Comment{rotates the $\ell^\text{th}$ qubit
    of \out\ register as per~\cref{eq:rot_operation}}
    \State $
    \operatorname{\textsc{angle}}(\stDev, \mean, \ang)$
    \Comment{erases \ang\ register} 
    \State $\operatorname{\textsc{shift}}(\stDev)$
    \Comment{along with the next two lines, divides the number stored in \stDev\ by two}
    \If{$\tilde{\sigma}_\ell=1$}
        \State $X(\stDev[p-1])$
        \Comment{flips the leftmost qubit of \stDev}
    \EndIf
    \State $\operatorname{\textsc{shift}}(\mean)$
    \Comment{divides the number stored in \mean\ by two; the leftmost qubit of \mean\ is always $\ket{0}$}
    \State $\operatorname{\textsc{cnot}}(\out[\ell],\mean[p-m-1])$
    \Comment{adds $1/2$ to \mean\ register if $\out[\ell]$ is $\ket{1}$}
    \State \Yield $\out[\ell]$
\EndFunction
\end{algorithmic}
\end{library}

\begin{figure}[htb]
\centering
\includegraphics[width=.76\linewidth]{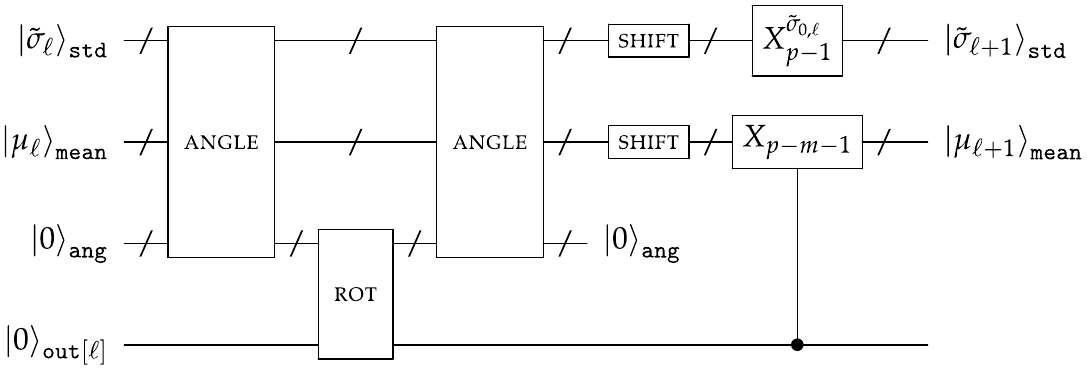} 
\caption[Quantum circuit for 1DG-state generation]{Quantum circuit for implementing the iterative steps for 1DG-state generation.
Positive integers~$p$ and~$m$ are respectively the total number of bits and position of the radix point in the fixed-point number representation;
\qwire\ represents~$p$ qubits.
$\operatorname{\textsc{angle}}$~\eqref{eq:angle_operation} computes a rotation angle~\eqref{eq:rot_angle} into the \ang\ register,
$\operatorname{\textsc{rot}}$~\eqref{eq:rot_operation} rotates the~$\ell^\text{th}$ qubit of \out,~$\out[\ell]$, by the rotation angle
stored in \ang\ and $\operatorname{\textsc{shift}}$~\eqref{eq:shift_operation} cyclically shifts qubits of a register one qubit to the right.
The Pauli-X gate~$X^b_i$ acts on~$i^\text{th}$ qubit of a register if the binary~$b$ is 1.
Intermediate results are uncomputed by running the appropriate operations in reverse order.}
\label{fig:1DGCircuit}
\end{figure}

By performing the described iterative operations, the quantum state~$\ket{\Psi(\tilde{\sigma}, 0, m)}$~\eqref{eq:1DG_over_integers} is prepared on the first~$m$ qubits of the \out\ register.
To transform this state into the integer part of \out, we swap the first~$m$ qubits with the last~$m$ qubits of \out.
Specifically, we swap the~$\ell^\text{th}$ qubit of \out\ with the~$(p-m+\ell)^\text{th}$ qubit
for~$\ell\in\{0, \ldots, m-1$\}.

% Transforming $\ket{\Psi(\tilde{\sigma}, \mu, m)}$ to $\ket{\G(\tilde{\sigma},\delta,m)}$
The last step is to transform~$\ket{\Psi(\tilde{\sigma}, 0, m)}$ to the desired discrete 1DG state~$\ket{\G_\text{lattice}(\tilde{\sigma},\delta,m)}$~\eqref{eq:lattice1DG}.
To this end, we write the classical input~$\delta$ into the \spacing\ register and perform the operation
\begin{equation}
\label{eq:mul_operation}
    \operatorname{\textsc{mul}}: \ket{j}_\out\ket{\delta}_\spacing \ket{0}_\temp \mapsto \ket{j}_\out\ket{\delta}_\spacing \ket{j\delta}_\temp,
\end{equation}
which we describe by~$\operatorname{\textsc{mul}}(\out, \spacing, \temp)$ in our algorithm.
This operation is an out-of-place multiplication and,
therefore, we need to uncompute the \out\ register.
To uncompute \out, we perform the operations that generates~$\ket{\Psi(\tilde{\sigma}, 0, m)}$ in the reverse order.
We then swap qubits of \temp\ with qubits of \out.
By this swapping, \temp\ is erased and~$\ket{\G(\tilde{\sigma}, \delta, m)}$ is transformed into \out.

Finally, we erase the \spacing\ and \stDev\ registers.
Note that the values stored in these registers do not change throughout the computation.
Therefore, we erase \spacing\ and \stDev\ by writing the classical inputs~$\delta$ and~$\tilde{\sigma}$ into them, respectively.
The full description of our 1DG-state-generation algorithm is presented as pseudocode in Algorithm~\ref{alg:1DG}.

\begin{algorithm}[H]
  \caption{Quantum algorithm for generating a one-dimensional Gaussian state}
  \label{alg:1DG}
  \begin{algorithmic}[1]
    \Require{
        \Statex $p\in \integers^{+}$
        \Comment{working precision}
        \Statex $m \in \integers^{+}$
        \Comment{position of the radix point in fixed-point number representation}
        \Statex $\tilde{\sigma}\in \binary^p$
        \Comment{$p$-bit approximate standard deviation of the approximate~1DG pure state}
        \Statex $\delta \in \binary^p$ \Comment{$p$-bit approximate lattice spacing}
            }
    \Ensure
        \Statex $\out \in \mathscr{H}_2^p$
        \Comment{$p$-qubit approximate 1DG~state \eqref{eq:lattice1DG} prepared on \out\ register}
\Function{oneDG}{$p, m, \tilde{\sigma}, \delta$}
    \State $ \mathscr{H}_2^p \ni \stDev\gets \tilde{\sigma}$
   \Comment{encodes $\tilde{\sigma}$ into the \stDev\ register}
    \State $\mathscr{H}_2^p \ni \spacing\gets \delta$
    \Comment{encodes $\delta$ into the \spacing\ register}
\For{$\ell \gets 0$ to $m-1$}\label{line:computeOut1}
    \State \label{line:iterOneDG}
    \textsc{iterOneDG}($\ell,m,p,\stDev,\mean, \ang$)
    \Comment{see the helper function in Library~\ref{lib:1DG}}
\EndFor
\For{$\ell \gets 0$ to $m-1$} \label{line:swap_start}
 \State $\operatorname{\textsc{swap}}(\out[\ell],\out[p-m+\ell])$
 \label{line:computeOut2}
\EndFor \label{line:swap_end}
\State \label{line:mul}
$\operatorname{\textsc{mul}}(\out,\spacing,\temp)$
\For{$\ell \gets 0$ to $m-1$}\label{line:uncomputeOut1}
\Comment{Lines (\ref{line:uncomputeOut1}--\ref{line:uncomputeOut2}) are the reverse of lines (\ref{line:computeOut1}--\ref{line:computeOut2}) and uncomputes \out}
 \State $\operatorname{\textsc{swap}}(\out[\ell],\out[p-m+\ell])$
\EndFor
\For{$\ell \gets m-1$ to $0$}
\State \textsc{invIterOneDG}($\ell, m,p,\stDev, \mean, \ang$) \label{line:uncomputeOut2}
\Comment{preforms inverse of the helper function in line~\ref{line:iterOneDG}}
\EndFor
\For{$\ell \gets 0$ to $p-1$}
\Comment{swaps the qubits of \out\ and \temp\ registers}
 \State $\operatorname{\textsc{swap}}(\out[\ell],\temp[\ell])$
\EndFor
\State $ \stDev\gets \tilde{\sigma}$
   \Comment{erases \stDev\ register}
    \State $\spacing\gets \delta$
    \Comment{erases \spacing\ register}
    \State \Yield \out
\EndFunction
\end{algorithmic}
\end{algorithm}

%%%%%%%%%%%%%%%%%%%%%%%%%%%%%%%%%%%%%%%%%%%%%%%%%%%%%%%%%
\subsubsection{One-dimensional Gaussian-state generation by inequality testing}
\label{subsubsec:inequality_testing}

We now construct our alternative quantum algorithm for generating a 1DG state.
Our algorithm is based on testing an inequality on a quantum computer.
We begin with a high-level description for state generation by inequality testing and our algorithm for 1DG-state generation.
Then we proceed with a detailed description of the algorithm, and finally, we present our algorithm as pseudocode.

The general principle to prepare a state of the form $\sum_j f(j) \ket{j}$ by inequality testing is as follows~\cite{SLSB19}.
First prepare an equal superposition over $\ket{j}$.
Then, for each $j$, compute~$f(j)$ into a new quantum register.
Next perform an inequality test between the value encoded in this register and the value encoded in an ancillary register prepared in an equal superposition over all possible values of the function.
Then erase the ancillary register and measure the qubit storing result of the inequality test.
The prepared state by this method is the desired state with certain probability. 
The success probability can then be boosted by amplitude amplification~\cite{BHM+02}.
The complexity of this approach could be large in the case where the distribution of amplitudes has a sharp peak in an unknown location because the amplitude amplification would then essentially be solving a Grover search, which has a square root speed limit~\cite{BHM+02}.

The amplitude distribution for a Gaussian state has a single peak but in a known location, and we take advantage of the known location in preparing a Gaussian state.
Our approach for using the known location is similar to that used for preparing a state with amplitudes~$1/\|k_\nu\|$ in~\cite{BBM+19}.
Instead of initially preparing a state with an equal superposition over $j$, we prepare a state with approximate amplitudes upper bounding the amplitudes to be prepared, as shown by the orange points in~\cref{fig:1DG_by_ineq_testing}.
Instead of computing $f(j)$, we compute the ratio between $f(j)$ and the initial upper bound on $f(j)$ and perform an inequality test with that ratio.
The inequality test corrects the amplitudes and results in a much larger amplitude for success.
Hence only a single step of amplitude amplification can be used, which significantly reduces the algorithm's complexity.

We now proceed with a detailed description of the inequality-testing-based algorithm for Gaussian-state generation.
The state that we aim to prepare by inequality testing is
\begin{equation}
\label{eq:pn_intgs}
    \ket{\Psi(\sigma,m)} := \frac{1}{\sqrt{\mathcal{N}(\sigma,m)}}
    \sum_{j=-(2^{m-1}-1)}^{2^{m-1}-1}
    \e^{-\frac{j^2}{4\sigma^2}} \ket{j},\quad
    \mathcal{N}(\sigma,m) := 1+2\sum_{j=1}^{2^{m-1}-1}
    \e^{-\frac{j^2}{2\sigma^2}},
\end{equation}
with classical inputs~$m\in \integers^+$ and~$\sigma\in\reals^+$.

Having prepared this state on a quantum register, labelled \out, we then incorporate a lattice spacing~$\delta \in \reals^+$ by 
implementing~$\ket{j}_\out\mapsto\ket{j\delta}_\out$ as per~\cref{eq:mul_operation}.
Thereby an approximation for the desired 1DG state~\eqref{eq:lattice1DG} is prepared on~\out.
Considering the range of the index~$j$ in~\cref{eq:pn_intgs} and~\cref{eq:lattice1DG}, the infidelity between the approximate and the desired 1DG states is exponentially close to zero.

Our strategy to prepare the state in~\cref{eq:pn_intgs} by inequality testing is as follows.
First we prepare the state
\begin{equation}
\label{eq:p_intgs}
    \ket{\Psi_+(\sigma,m)} = \frac{1}{\sqrt{\mathcal{N}(\sigma,m)}}
    \left(\ket{0}+\sqrt{2}\sum_{j=1}^{2^{m-1}-1}
    \e^{-\frac{j^2}{4\sigma^2}} \ket{j}\right),
\end{equation}
on $m$ qubits.
Controlled on~$j\neq 0$, we then perform a Hadamard on the leftmost qubit.
The controlled operation gives a sign bit for $j$ being a signed integer with positive and negative values.
Then we convert from signed integer to two's complement representation~\cite[p.~16]{HH16} to obtain the state in~\cref{eq:pn_intgs}.

To prepare the state in~\cref{eq:p_intgs}, we use an initial amplitude according to the value of $j$ rounded down to the nearest power of two;
the initial amplitude is illustrated by orange points in~\cref{fig:1DG_by_ineq_testing}.
Specifically, the initial state that we prepare is
\begin{equation}
\label{eq:p_intgs_approx}
    \ket{\tilde{\Psi}_+(\sigma,m)} = \frac{1}{\sqrt{\tilde{\mathcal{N}}(\sigma,m)}}
    \left(\ket{0}+\sqrt{2}\sum_{j=1}^{2^{m-1}-1}
    \e^{-\frac{2^{2\floor{\log_2 j}}}{4\sigma^2}} \ket{j}\right),
    \quad
    \tilde{\mathcal{N}}(\sigma,m) :=
    1+2 \sum_{j=1}^{2^{m-1}-1}
    \e^{-\frac{2^{2\floor{\log_2 j}}}{2\sigma^2}},
\end{equation}
where~$\tilde{\mathcal{N}}(\sigma,m)$ is the state's normalization factor.
Having prepared this initial state, we then transform it to the state in~\cref{eq:p_intgs} by inequality testing.
To prepare the initial state, first we prepare the state
\begin{equation}
\label{eq:stepped_Gaussian}
    \frac{1}{\sqrt{\mathcal{M}(\sigma,m)}}
    \left(\ket{0}^{\otimes m} +
    \sum_{j=1}^{m-1} 2^{j/2} \e^{-\frac{2^{2(j-1)}}{4\sigma^2}} \ket{0}^{\otimes(m-j)}\ket{1}^{\otimes j}\right),
    \quad
    \mathcal{M}(\sigma,m):=  1 + \sum_{j=1}^{m-1} 2^j \e^{-\frac{2^{2(j-1)}}{2\sigma^2}},
\end{equation}
on an $m$-qubit register \texttt{G}.
Then we sequentially perform a \textsc{cnot} followed by a controlled-Hadamard~(\textsc{chad}) from~$\texttt{G}[\ell]$ to~$\texttt{G}[\ell-1]$ for~$\ell$ from $1$ to~$m-1$;
note that quits are ordered from right to left, so the rightmost qubit is $0^\text{th}$ qubit.
Upon performing these operations, the initial state~\eqref{eq:p_intgs_approx} is generated on register \texttt{G}.

We now describe how to transform the initial state~\eqref{eq:p_intgs_approx} into the state in~\cref{eq:p_intgs} by inequality testing.
Let us define $r_0:=1$,
\begin{equation}
\label{eq:ratio}
    r_j := 
      \exp(\frac{2^{2\lfloor \log_2 j \rfloor} - j^2}{4\sigma^2})
      \quad
      \forall j\in\{j=1,\ldots, 2^m-1\},
\end{equation}
for the ratio of the unnormalized amplitudes from~\cref{eq:p_intgs} to the unnormalized amplitudes from~\cref{eq:p_intgs_approx}.
For some positive integer~$t$,
we compute a~$t$-bit approximation of~$r_j$ into a $t$-qubit temporary register, labelled~\temp, by an operation defined as
\begin{equation}
\label{eq:ratio_operation}
  \textsc{ratio}: \ket{j}_\out \ket{\sigma}_\stDev \ket{0}_\temp \mapsto \ket{j}_\out \ket{\sigma}_\stDev \ket{r_j}_\temp,
\end{equation}
which we describe by~$\textsc{ratio}(\out,\stDev,\temp)$ in our quantum algorithm.
As the ratio~$r_j\leq 1$ for each~$j$, the encoded $t$-bit string in \temp\ represents the value of~$r_j$~\eqref{eq:ratio} with an implied division by~$2^t$.
Next we prepare a $t$-qubit reference register, labelled \reff, in the uniform superposition state~$2^{-t/2} \sum_{z=0}^{2^t - 1} \ket{z}$ using~$t$ Hadamard gates.
With an implied division of the encoded value by~$2^t$, the register \reff\ can be viewed as being prepared in a uniform superposition of all possible values from~$0$ to~$1$.
Finally, we test an inequality between the value encoded in \temp\
and the value encoded in \texttt{ref} with the result of the inequality test written to a fresh qubit labelled \ineq.
Specifically, we perform a comparison operation defined as
\begin{equation}
    \label{eq:comp_operation}
    \textsc{comp}: \ket{r}_\temp\ket{z}_\reff \ket{0}_\ineq
    \mapsto
    \begin{cases}
        \ket{r}_\temp\ket{z}_\reff \ket{0}_\ineq \quad
        \text{if $r<z$},\\
        \ket{r}_\temp\ket{z}_\reff \ket{1}_\ineq \quad
        \text{if $r\geq z$},
    \end{cases}
\end{equation}
where $r$ and $z$ are $t$-bit numbers;
this operation is described by $\textsc{comp}(\temp,\reff,\ineq)$ in our quantum algorithm.
The state after inequality testing is
\begin{equation}
\ket{\Psi_\text{comp}}:=\frac 1{\sqrt{\tilde{\mathcal{N}}(\sigma,m) 2^{t}}}
\sum_{j=0}^{2^{m-1}-1}
g(j) \ket{j}_\out \ket{r_j}_\temp
\left(
\sum_{z=0}^{r^{(t)}_j-1}
\ket{z}_\reff\ket{0}_\ineq
+
\sum_{z=r^{(t)}_j}^{2^{t-1}-1}
\ket{z}_\reff\ket{1}_\ineq
\right),
\end{equation}
where~$\tilde{\mathcal{N}}(\sigma,m)$ is the normalization factor in~\cref{eq:p_intgs_approx}, $r^{(t)}:=\floor{2^t r_j}$ and
\begin{equation}
    g(0):=1, \quad g(j):=\sqrt{2}
    \exp\left(
    -\frac{2^{2\floor{\log_2j}}}{4\sigma^2}
    \right) \quad \forall\, j\in\{1,\ldots,2^m-1\},
\end{equation}
are the unnormalized amplitudes in~\cref{eq:p_intgs_approx}.
Next we unprepare the uniform superposition on \reff\ with $t$ Hadamard gates.
Then projecting the single qubit \ineq\ onto the success state~$\ket{0}_\ineq$ yields
\begin{equation}
    \ket{\Psi_\text{out}}:=
    \frac 1{\sqrt{\tilde{\mathcal{N}}(\sigma,m)}2^t} \sum_j g(j) r^{(t)}_j \ket{j},
\end{equation}
with success probability
\begin{equation}
    P_\text{success} := \frac{1}{\tilde{\mathcal{N}}(\sigma,m)2^{2t}} \sum_j \left(g(j) \floor*{2^t r_j}\right)^2.
\end{equation}

\cref{fig:successProb} shows the success probability for a wide range of the standard deviation, and realistic size (number of qubits $m$) for the quantum register encoding the 1DG state.
The success probability is greater than $0.67\approx 2/3$ for any realistic example.
Therefore, one step of amplitude amplification is sufficient to achieve a high success probability.
However, as 1DG-state preparation is at the beginning of the main algorithm for ground-state generation, it would be sufficient to prepare the state probabilistically and repeat until success.
As the probability of success is at least about $70\%$, the repeat-until-success procedure is more efficient on average than amplitude amplification.

\begin{figure}[H]
    \centering
    \includegraphics[width=.76\textwidth]{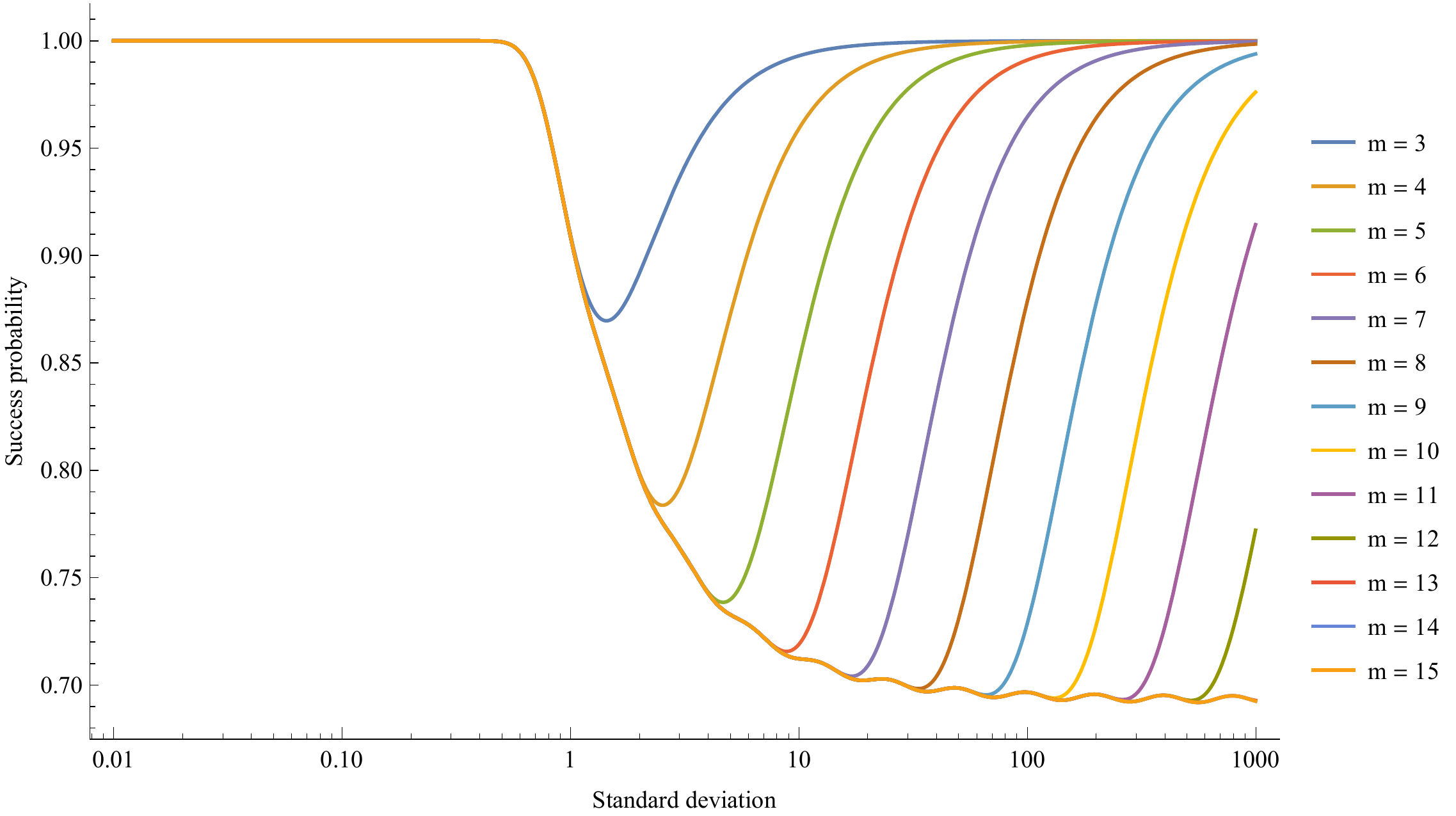}
    \caption[Success probability for inequality testing]{
    The success probability for preparing a 1DG state by inequality testing for a wide range of standard deviation and number of qubits~$m$; see~\cref{eq:pn_intgs}.
    The success probability is at least about~$70\%$.}
    \label{fig:successProb}
\end{figure}

The state in~\cref{eq:stepped_Gaussian} is a unary state which can be prepared by rotations and controlled rotations~\cite[pp.~7--8]{BBK+16}.
We now describe how to prepare this state on an $m$-qubit register \out\ by these operations.
To this end, first we compute the rotation angles~$\theta_\ell$~in
\begin{equation}
\label{eq:ineq_rot_angs}
    \sin(\theta_\ell) =\frac{\sqrt{2}2^{\ell/2}
    \e^{-\frac{2^{2\ell}}{4\sigma^2}}}{\sqrt{\mathcal{M}(\sigma,\ell+2)}}
    \quad
    \forall \ell \in\{0,\ldots,m-2\},
\end{equation}
with $\mathcal{M}(\sigma,\ell)$ given in~\cref{eq:stepped_Gaussian}, by a classical algorithm.
Then we perform the following operations:
\begin{enumerate}
    \item Perform a single-qubit rotation with angle $\theta_{m-2}$ on $\out[m-2]$;
    $\out[i]$ denotes the $i^\text{th}$ qubit of \out. 
    \item For~$\ell$ from~$m-3$ to zero,
    perform a \textsc{cnot} from~$\out[\ell+1]$ to~$\out[\ell]$, and a single-qubit rotation with angle~$\theta_\ell$ on~$\out[\ell]$ if state of~$\out[\ell+1]$ is $\ket{0}$. 
    We perform the controlled rotation by
    \begin{equation}
        \operatorname{\textsc{crot}}:
        \ket{b}_{\out[\ell+1]} \ket{\eta}_\ang \ket{0}_{\out[\ell]}
        \mapsto
        \ket{b}_{\out[\ell+1]}
        \operatorname{\textsc{rot}}^{1-b}
        \ket{\eta}_\ang \ket{0}_{\out[\ell]},
    \end{equation}
    where $\operatorname{\textsc{rot}}^{1}:=\operatorname{\textsc{rot}}$~\eqref{eq:rot_operation} and $\operatorname{\textsc{rot}}^{0}:=\id$.
    This operation is described by~$\operatorname{\textsc{crot}}(\out[\ell+1],\ang,\out[\ell])$ in our algorithm.
    \item If $j\neq 0$, perform a Hadamard on $\out[m-1]$, i.e., the leftmost qubit of \out.
    To do this, we test an inequality between the value $j$ encoded in \out\ register and the value $0$ encoded in \reff\ register, and write the result of inequality test to a fresh qubit labelled \texttt{flag}.
    Then we perform a \textsc{chad} from \texttt{flag} to~$\out[m-1]$.
    That is, we perform a Hadamard on~$\out[m-1]$ if state of \texttt{flag} is~$\ket{1}$. 
    Finally, we erase \texttt{flag}.
    \item Convert the signed integer encoded in \out\ to its two's complement representation.
\end{enumerate}
The classical algorithm for computing the rotation angles is formally described by the helper function in Library~\ref{alg:ineq1DG}.
The purpose of this helper function is only to distinguish the classical vs quantum parts of the state-generation algorithm by inequality testing presented as pseudocode in Algorithm~\ref{alg:ineq1DG}.

\begin{library}[H]
  \caption{Helper functions for 1DG-state generation in Algorithm~\ref{alg:ineq1DG}}
  \label{lib:ineq1DG}
  \begin{algorithmic}[1]
\Function{rotAngles}{$m,\sigma$}
\Comment{computes rotation angles $\theta_\ell$~\eqref{eq:ineq_rot_angs} required to prepare the state in~\cref{eq:stepped_Gaussian}}
\For{$\ell\gets 0$ to $m-2$}
    \State $\mathcal{M} \gets 1+ \displaystyle\sum_{j=1}^{\ell-1} 2^j \exp\left(-2^{2(j-1)}/\left(2\sigma^2\right)\right)$
    \State $\theta_\ell \gets 2^{\ell/2}\sqrt{2/\mathcal{M}} \exp\left(-2^{2\ell}/\left(4\sigma^2\right)\right)$ 
\EndFor
\State \Return $\bm{\uptheta}$
\Comment{$\bm{\uptheta}$ is a vector with components $\theta_\ell$}
\EndFunction
\end{algorithmic}
\end{library}

\begin{algorithm}[H]
  \caption{Quantum algorithm for generating a one-dimensional Gaussian state by inequality testing}
  \label{alg:ineq1DG}
  \begin{algorithmic}[1]
    \Require{
    \Statex $p\in \integers^{+}$
        \Comment{working precision}
        \Statex $m \in \integers^{+}$
        \Comment{position of the radix point in fixed-point number representation}
        \Statex $\sigma\in \binary^p$
        \Comment{$p$-bit approximate standard deviation of the approximate~1DG pure state}
        \Statex $\delta \in \binary^p$ \Comment{$p$-bit approximate lattice spacing}
            }
    \Ensure
    \Statex $\out \in \mathscr{H}_2^p$
        \Comment{$p$-qubit approximate 1DG~state \eqref{eq:lattice1DG} prepared on \out\ register}
\Function{ineqBasedOneDG}{$p, m, \tilde{\sigma}, \delta$}
    \State $\reals^{m-1} \ni \bm{\uptheta} \gets \textsc{rotAngles}(m, \sigma)$
    \Comment{classically computes rotation angles for preparing the state in~\cref{eq:stepped_Gaussian}; see~Library~\ref{lib:ineq1DG}}
    \State $\ang \gets \theta_{m-2}/2\uppi$
    \State \textsc{rot}($\ang, \out[m-2]$)
    \Comment{rotates $(m-2)^\text{th}$ qubit of \out\ register by angle~$\theta_{m-2}$; see~\cref{eq:rot_operation}}
    \State $\ang \gets \theta_{m-2}/2\uppi$
    \Comment{erases \ang\ register}
    \For{$\ell \gets m-3$ to $0$}
        \State \textsc{cnot}($\out[\ell+1], \out[\ell]$)
        \State $\ang \gets \theta_\ell/2\uppi$
        \State \textsc{crot}($\out[\ell+1], \ang, \out[\ell]$)
        \Comment{rotates $(\ell-1)^\text{th}$ qubits of \out\
        if $\ell^\text{th}$ qubit of \out\ is $\ket{0}$}
        \State $\ang \gets \theta_\ell/2\uppi$
        \Comment{erases \ang\ register}
    \EndFor
    \For{$\ell \gets 1$ to $m-1$}
        \Comment{lines~(\ref{line:p_intgs_approx1}--\ref{line:p_intgs_approx2}): prepares the state in~\cref{eq:p_intgs_approx} on \out\ register}
        \State \label{line:p_intgs_approx1}
        \textsc{cnot}($\out[\ell], \out[\ell-1]$)
        \State \label{line:p_intgs_approx2}
        \textsc{chad}($\out[\ell],\out[\ell-1]$)
    \EndFor
    \State $b\gets 1$
    \Comment{sets the initial value for binary~$b$}
    \While{$b=1$}
        \For{$\ell \gets 0$ to $t-1$}\label{line:repeat}
        \Comment{prepares $t$-qubit \reff\ register in uniform superposition state}
        \State \textsc{H}($\reff[\ell]$)
        \EndFor
        \State \textsc{ratio}(\out, \stDev, \temp)
        \Comment{computes the ratio $r_j$~\eqref{eq:ratio} into \temp\ as per~\cref{eq:ratio_operation}}
        \State \textsc{comp}(\temp, \reff,\ineq)
        \Comment{performs the inequality test~\eqref{eq:comp_operation} between the values encoded in \temp\ and \reff}
        \For{$\ell \gets 0$ to $t-1$}
        \Comment{erases \reff\ register}
        \State \textsc{H}($\reff[\ell]$)
        \EndFor
        \State $b\gets \textsc{measure}(\ineq)$
        \Comment{measures the single qubit \ineq\ in computation basis; $b$ is measurement's output}
        \If{$b=1$}
        \Comment{if measurement's output is 1, the state in~\cref{eq:p_intgs} is generated on \out}
        \State \textsc{comp}(\out, \reff, \texttt{flag})
        \Comment{performs the inequality test between the values encoded in \out\ and \texttt{flag}}
        \State $\textsc{chad}(\texttt{flag}, \out[m-1])$
        \Comment{controlled on \texttt{flag} applies a Hadamard on the leftmost qubit of \out}
        \State \textsc{comp}(\out, \reff, \texttt{flag})
        \Comment{erases \texttt{flag}}
        \State $\spacing \gets \delta$
        \Comment{writes lattice spacing~$\delta$ into \spacing\ register}
        \State \textsc{mul}(\out, \spacing, \anc)
        \State $\spacing \gets \delta$
        \Comment{erases \spacing\ register}
        \State \textsc{erase}(\out)
        \Comment{erases \out\ register by performing reverse of operations used in preparation}
         \For{$\ell \gets 0$ to $p-1$}
            \State \textsc{swap}($\out[\ell], \anc[\ell]$)
            \Comment{swaps the qubits of \out\ and \anc\ registers; erases \anc}
        \EndFor
        \State \Yield \out
        \EndIf
    \EndWhile    
\EndFunction
\end{algorithmic}
\end{algorithm}

\newpage
%%%%%%%%%%%%%%%%%%%%%%%%%%%%%%%%%%%%%%%%%%%%%%%%%%%%%%%%%
\subsubsection{Quantum fast Fourier transform}
\label{subsubsec:QFFT}

Here we present our quantum algorithm for performing a discrete Fourier transform~(DFT) on a quantum computer. 
We begin by specifying the task here and distinguishing it from the standard quantum Fourier transform~\cite[Chap.~5]{NC11}.
We then explain how a classical fast Fourier transform algorithm can be modified for execution on a quantum computer.
We next explain that, for our purposes, the usual complex-valued DFT can be replaced with the discrete Hartley transform~(DHT).
Finally, we present a detailed description of our quantum algorithm for performing a DHT on a quantum computer.

For a given positive integer~$N$ and a complex-valued vector~$\bm{x} = \left( x_0, x_1, \ldots, x_{N-1} \right)$,
the DFT of~$\bm{x}$ is a complex-valued
vector~$\tilde{\bm{x}}=
\left(\tilde{x}_0, \tilde{x}_1,\ldots, \tilde{x}_{N-1}\right)$ with components
\begin{equation}
\label{eq:DFT}
    \tilde x_\ell := \frac1{\sqrt{N}}
    \sum_{\ell^\prime = 0}^{N-1} x_{\ell^\prime}
    \e^{-2\uppi\upi\ell\ell^\prime/N} .
\end{equation}
A quantum fast Fourier transform~(QFFT) is then a quantum circuit 
$\operatorname{\textsc{qfft}}$ designed so that
\begin{equation}
\label{eq:QFFT}
    \operatorname{\textsc{qfft}}: \ket{x_0} \ket{x_1} \cdots \ket{x_{N-1}}
    \mapsto \ket{\tilde{x}_0} \ket{\tilde{x}_1} \cdots \ket{\tilde{x}_{N-1}},
\end{equation}
where~$\tilde x_\ell$~\eqref{eq:DFT} are components of the transformed vector~$\tilde{\bm{x}}$.
Note that QFFT is distinct from the standard quantum Fourier transform~\cite[p.~217]{NC11}, which is defined as
\begin{equation}
\label{eq:QFT}
    \operatorname{\textsc{qft}}: \ket{\ell}
    \mapsto \frac1{\sqrt{N}}
    \sum_{\ell^\prime = 0}^{N-1}
    \e^{-2\uppi\upi\ell\ell^\prime/N} \ket{\ell^\prime}
    \qquad \forall \ell\in \{0, 1, \ldots, N - 1\},
\end{equation}
for any $N$-dimensional qudit.
Specifically, in contrast to the standard quantum Fourier transform, the QFFT requires the quantum computer to perform arithmetic operations.

There are well-known techniques to execute the DFT on a classical computer using~$\order{N \log_2 N}$ arithmetic operations,
rather than the $\order{N^2}$ required by a naive implementation.
Such techniques are called `fast Fourier transforms'.
Fast Fourier transforms can be straightforwardly ported to quantum algorithms by executing the prescribed arithmetic operations reversibly~\cite{ASY20}.
These arithmetic operations always take the form of an in-place transformation~$(a, b) \mapsto (a + \omega b, a - \omega b)$, called `butterfly' operation, for~$a, b \in \cmplex$ read from the input array and~$\omega \in \cmplex$ some precomputed constant.
As per the cost model in~\cref{subsubsec:complexity_measure}, we assign a unit cost for each butterfly operation,
as the overall operation is roughly as computationally expensive as multiplication.
Hence we can execute any fast Fourier transform on a quantum computer and report the same complexity of~$\order{N \log_2 N}$.

% Replacement with fast Hartley transform
The problem with the QFFT is that the resulting representation
for the Gaussian state requires us to have complex-valued position states~\eqref{eq:continuous1DG}.
As all operations for ground-state generation require only real numbers, we avoid the complex numbers required by the DFT and instead use a discrete Hartley transform,
which only requires real numbers~\cite[p.~100]{HJ12}.
The DHT is valid because the ground-state ICM in the Fourier-based method~(\cref{subsubsec:Fourier_method}) is a real, symmetric and circulant matrix, and hence is diagonalizable by a DHT~\cite[Theorem~1]{BF93}.
The fast Hartley transform has a similar structure to the fast Fourier transform but with a different butterfly operation~\cite{Bra84}.
The number of butterfly operations is the same.
Hence, by porting these operations to quantum algorithms, we achieve the same quantum complexity of~$\order{N \log_2 N}$ for executing a fast Hartley transform on a quantum computer.

We replace the discrete Fourier transform circuit $\operatorname{\textsc{qfft}}$~\eqref{eq:QFFT}
by a new circuit $\operatorname{\textsc{qfht}}$ defined so that
\begin{equation}
\label{eq:QFHT}
  \operatorname{\textsc{qfht}}: \ket{x_0} \ket{x_1} \cdots \ket{x_{N-1}}
  \mapsto \ket{\overline{x}_0} \ket{\overline{x}_1} \cdots \ket{\overline{x}_{N-1}},
  \quad
  \overline{x}_\ell:=\frac1{\sqrt{N}} 
  \sum_{\ell^\prime = 0}^{N-1}
  x_{\ell^\prime} \cas(2\uppi\ell\ell^\prime/N),
\end{equation}
where the~$\cas$ function is defined as $\cas(\theta):=\cos(\theta)+\sin(\theta)$ and~$\overline{x}_\ell$ are components of the vector~$\overline{\bm{x}}$, the DHT of~$\bm{x}$.
Analogous to the distinction between QFFT and the standard quantum Fourier transform, the QFFT is distinct from the standard quantum Hartley transform~\cite{KR03,TH05}, which is defined as
\begin{equation}
    \textsc{qht}: \ket{\ell}\mapsto \frac{1}{\sqrt{N}}
    \sum_{\ell^\prime=0}^{N-1} \cas(2\uppi\ell\ell^\prime/N) \ket{\ell^\prime}
    \quad\forall\, \ell\in\{0,1,\ldots,N-1\},
\end{equation}
for any $N$-dimensional qudit.

To elucidate our QFHT algorithm, first we describe a recursive decomposition for the DHT.
Let~$\overline{\bm{x}}_\text{E}$ and~$\overline{\bm{x}}_\text{O}$ be vectors comprised of the even- and odd-indexed components of~$\overline{\bm{x}}$~\eqref{eq:QFHT}, respectively,
and let~$\overline{\bm{x}}_\text{L}$ and~$\overline{\bm{x}}_\text{R}$ be the left half and the right half of~$\overline{\bm{x}}$, respectively.
Then the Hartley-transformed vector~$\overline{\bm{x}}$~\eqref{eq:QFHT} is written as~\cite[Chap.~25]{Arn11}
\begin{equation}
\label{eq:fhtDecomposition}
    \overline{\bm{x}}_\text{L} = \overline{\bm{x}}_\text{E}
    +\textsc{chs}\left(\overline{\bm{x}}_\text{O}\right),
    \quad
    \overline{\bm{x}}_\text{R} = \overline{\bm{x}}_\text{E}
    -\textsc{chs}\left(\overline{\bm{x}}_\text{O}\right),
\end{equation}
where \textsc{chs} is the classical Hartley-shift operation with action
\begin{equation}
\label{eq:CHS}
    \operatorname{\textsc{chs}}:
    (x_0,x_1,\ldots,x_{N-1} )\mapsto
    (x^\text{s}_0, x^\text{s}_1,\ldots,x^\text{s}_{N-1}),
    \quad
    x^\text{s}_\ell:= x_\ell \cos(\uppi\ell/N) +x_{N-\ell} \sin(\uppi\ell/N).
\end{equation}
The decomposition in~\cref{eq:fhtDecomposition} enables a recursive algorithm for the DHT.
The recursive algorithm requires a temporary workspace for writing the results of intermediate computations.
We avoid the need for a temporary workspace by writing the algorithm in a non-recursive way and performing in-place operations.
Our in-place algorithm requires a quantum-data reordering, similar to the data-reordering operation known as `bit reversal' in the classical FHT algorithm.
The bit reversal reorders the input data-vector~$\bm x$ such that the data~$x_i$ at an index $i$ is swapped with the data~$x_{\textsc{rev}(i)}$ at index $\textsc{rev}(i)$, where $\textsc{rev}(i)$ is an integer obtained from~$i$ by reversing its binary digits.

Our QFHT algorithm is based on three key quantum operations that we define here.
The first operation is a quantum version of the classical bit-reversal operation and is defined as
\begin{equation}
    \label{eq:qrev}
    \textsc{qrev}:
    \ket{x_0}_{\vac[0]} \ket{x_1}_{\vac[1]} \cdots \ket{x_{N-1}}_{\vac[N-1]} 
    \mapsto
    \ket{x_{\textsc{rev}(0)}}_{\vac[0]}
    \ket{x_{\textsc{rev}(1)}}_{\vac[1]}
    \cdots \ket{x_{\textsc{rev}(N-1)}}_{\vac[N-1]},
\end{equation}
where \textsc{rev} is the classical bit reversal.
This operation swaps the values encoded in the registers~$\vac[i]$ and~$\vac[\textsc{rev}(i)]$ by swapping qubits of these registers;
see the helper function \textsc{qrev} in line~\eqref{line:qrev} of~Library~\ref{lib:QFHT}.
The second key operation in our algorithm is a quantum Hartley shift defined as
\begin{equation}
    \label{eq:qhs}
    \textsc{qhs}: \ket{x_0} \ket{x_1} \cdots \ket{x_{N-1}}
    \mapsto \ket{x^\text{s}_0} \ket{x^\text{s}_1} \cdots \ket{x^\text{s}_{N-1}},
\end{equation}
where~$x^\text{s}_\ell$ is given in~\cref{eq:CHS}. Notice that this operation is a quantum version of the classical Hartley shift~\eqref{eq:CHS}. 
We implement \textsc{qhs} by $N/2$ applications of a primitive quantum Hartley-shift operation defined as
\begin{equation}
\label{eq:hs}
    \textsc{hs}: \ket{c}_{\texttt{cos}} \ket{s}_{\texttt{sin}}
    \ket{x}_{\vac[i]} \ket{y}_{\vac[j]}
    \mapsto \ket{c}_{\texttt{cos}} \ket{s}_{\texttt{sin}}
    \ket{xc+ys}_{\vac[i]} \ket{xs-yc}_{\vac[j]},
\end{equation}
where the values encoded in the quantum registers labelled \texttt{cos} and \texttt{sin} are respectively the cosine and sine in~\cref{eq:CHS};
see the helper function \textsc{qhs} in line~\eqref{line:qhs} of~Library~\ref{lib:QFHT}.
The third key operation used in the QFHT algorithm is a quantum butterfly operation defined as
\begin{equation}
\label{eq:qbf}
    \textsc{qbf}: \ket{x}_{\vac[i]} \ket{y}_{\vac[j]}
    \mapsto \ket{x+y}_{\vac[i]} \ket{x-y}_{\vac[j]},
\end{equation}
which we describe by $\textsc{qbf}(\vac[i], \vac[j])$ in our quantum algorithm.
All of these operations can be implemented by quantum arithmetic, specifically the quantum multiplication \textsc{mul}~\eqref{eq:mul_operation};
see~\cref{subsubsec:QFHT_complexity}

\begin{library}[H]
  \caption{Helper functions for quantum fast Hartley transform in Algorithm~\ref{alg:QFHT}}
  \label{lib:QFHT}
  \begin{algorithmic}[1]
\Function{qrev}{$N,p,\vac$}
\label{line:qrev}
\Comment{\vac\ is a qubit register with $N$ cells, each of which has $p$ qubits;
$N$ is assumed to be an even integer}
\For{$i\gets 1$ to $N/2-1$}
    \For{$j\gets 0$ to $p-1$}
        \State $\textsc{swap}\left(\vac[i,j],\vac[\textsc{rev}(i),j]\right)$
        \Comment{$\textsc{rev}(i)$ is an integer whose binary representations is revers of binary representation of $i$}
    \EndFor
\EndFor
\State \Yield \vac
\EndFunction
\Function{qhs}{$N,p,\vac$}
\label{line:qhs}
\If{$N>1$}
\Comment{\textsc{qhs}~\eqref{eq:qhs} is trivial for $N=1$}
\For{$\ell\gets 1$ to $N/2-1$}
    \State (\texttt{cos}, \texttt{sin}) $\gets \left(\cos(\uppi\ell/N),\sin(\uppi\ell/N)\right)$
    \Comment{writes cosine and sine in~\cref{eq:CHS} to \texttt{cos} and \texttt{sin} registers}
    \State \textsc{hs}$\left(\texttt{cos}, \texttt{sin}, \vac[\ell], \vac[N-\ell]\right)$
    \Comment{performs the primitive quantum Hartley transform~\eqref{eq:hs}}
    \State (\texttt{cos}, \texttt{sin}) $\gets \left(\cos(\uppi\ell/N),\sin(\uppi\ell/N)\right)$
    \Comment{erases \texttt{cos} and \texttt{sin} registers}
\EndFor
\EndIf
\State \Yield \vac
\EndFunction
\end{algorithmic}
\end{library}

The QFHT algorithm proceeds as follows;
see~\cref{fig:QFHT} for a schematic description of the algorithm.
First we reorder the size-$N$ input state-vector to the algorithm by implementing the \textsc{qrev}~\eqref{eq:qrev} operation.
Then the algorithm has~$\log_2 N$ stages.
For each stage $s\in\{1,\ldots, \log_2 N\}$, we group the reordered state-vector into~$N/2^s$ blocks of size-$2^s$ state-vectors and perform the \textsc{qhs}~\eqref{eq:qhs} on the right half of the state-vector in each block.
Then, for $\ell\in\{0,\ldots,2^s/2-1\}$, we perform \textsc{qbf}~\eqref{eq:qbf} on the~$\ell^\text{th}$ component of the left and right halves of the state-vector in each block.
Figure~\ref{fig:QFHT} shows a schematic representation of the QFHT algorithm for~$N=8$.
Components of the input state-vector in this figure is represented from top to bottom.
Hence the top and bottom halves of the state vector in each block represented in~\cref{fig:QFHT} corresponds to the left and right halves of the state-vector in the block when the state-vector is represented from left to right as in~\cref{eq:qhs}.

\begin{figure}[htb]
\centering
\includegraphics[width=.86\linewidth]{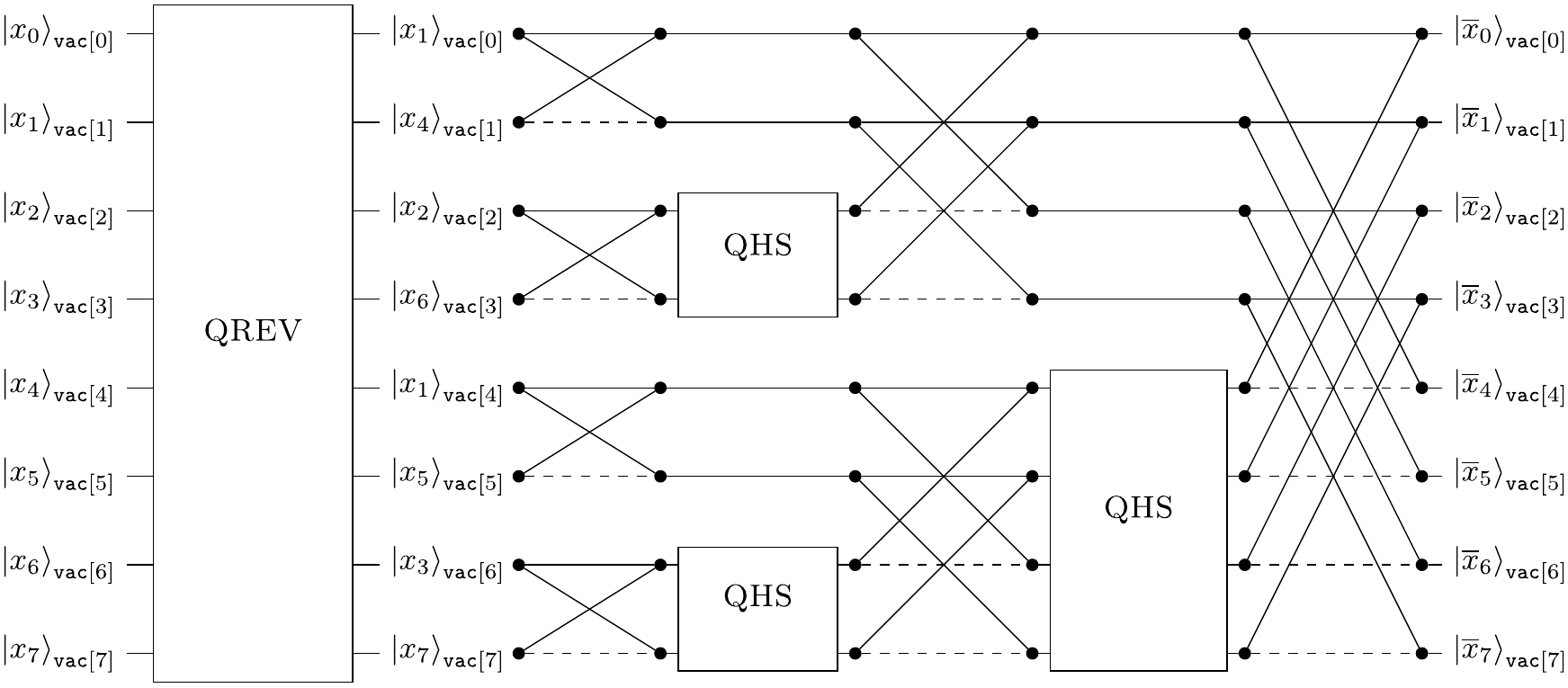} 
\caption[Schematic description of the quantum fast Hartley-transform algorithm]{
Schematic description of our quantum fast Hartley-transform algorithm for~$N=8$;
see~\eqref{eq:QFHT}.
Here \vac\ is a quantum register with~$8$ cells,
and each cell comprises $p$ qubits encoding a $p$-bit number.
\textsc{qrev}~\eqref{eq:qrev} reorders the values encoded in the input state-vector by swapping qubits of~$\vac[i]$ and~$\vac[\textsc{rev}(i)]$, where~$\textsc{rev}(i)$ is an integer obtained from $i$ by reversing its binary digits.
The algorithm has~$\log_2 N$ stages after \textsc{qrev}, each of which comprises~$N/2$ \textsc{qbf}~\eqref{eq:qbf} operations represented by crossed lines;
there are $N/2^s$ blocks of \textsc{qbf} operations at stage $s\in\{1,\ldots,\log_2 N\}$ and each block comprises~$2^s/2$ \textsc{qbf}.
\textsc{qhs} performs the Hartley-shift operation on its input state-vector as per~\cref{eq:qhs}.
\label{fig:QFHT}
}
\end{figure}

\begin{algorithm}[H]
\caption{Quantum fast Hartley transform
\label{alg:QFHT}}
\begin{algorithmic}[1]
    \Require{
    \Statex $N \in \integers_{\geq 2(2\dbIndex-1)}$
    \Comment{number of modes; for convenience, we assume that~$N$ is a power of~$2$}
    \Statex $p\in \integers^{+}$
    \Comment{working precision}
    \Statex $\vac \in \bigotimes_{\ell=0}^{N-1} \mathscr{H}_2^p$
    \Comment{tensor product of $N$ $p$-qubit 1DG states}
    }
    \Ensure
    \Statex $\vac \in \mathscr{H}_2^{N\times p}$
    \Comment{$(N\times p)$-qubit approximate ground state prepared on \vac\ register}
\Function{qfht}{$N, p, \vac$}
\State $\textsc{qrev}(N,p,\vac)$
\Comment{performs the quantum bit-reversal operation in~\cref{eq:qrev}; see line~\eqref{line:qrev} of~Library~\ref{lib:QFHT}}
\For{$s\gets 1$ to $\log_2 N$}
\Comment{iterates over stages of \textsc{qfht}; see~\cref{fig:QFHT}}
    \For{$b\gets 0$ to $N/2^s-1$}
        \Comment{iterates over blocks of butterflies at state~$s$}
        \State $\textsc{qhs}(2^s/2, p, \vac[(b+1/2)2^s:(b+1)2^s-1])$
        \Comment{see the helper function in line~\ref{line:qhs} of Library~\ref{lib:QFHT}}
        \For{$\ell \gets 0$ to $2^s-1$}
            \Comment{iterates over butterflies within block $b$ and performs butterfly operations}
            \State $\textsc{qbf}\left(\vac[b2^s+\ell], \vac[(b+1/2)2^s+\ell]\right)$
            \Comment{performs the quantum butterfly operation in~\cref{eq:qbf}}
        \EndFor
    \EndFor
\EndFor
\State \Yield \vac 
\EndFunction
\end{algorithmic}
\end{algorithm}

%%%%%%%%%%%%%%%%%%%%%%%%%%%%%%%%%%%%%%%%%%%%%%%%%%%%%%%%%
\subsubsection{Quantum shear transform}
\label{subsubsec:QST}

We now present a quantum algorithm for performing a shear transform on a quantum computer.
We begin by specifying the task in a shear transform and describe the quantum shear transform~(QST).
Next we explain how a sequence of shear transforms with exactly one shear element can implement a general shear transform.
We then discuss the shear transform required for the basis transformation in the wavelet-based algorithm.
We specify the inputs and outputs of our QST algorithm and proceed with explaining the procedure.
Finally, we present the QST algorithm as pseudocode.

% QST
For a shear transform, we are given a scalar $N \in \integers^+$,
a vector $\bm{x} = \left( x_0, x_1, \ldots, x_{N-1} \right) \in \reals^N$ and a shear matrix~$\upbm{S} \in \reals^{N\times N}$, where the shear matrix is  either a  lower or an upper unit-triangular matrix.
The shear transform
$\tilde{\bm{x}} = \left( \tilde{x}_0, \tilde{x}_1, \ldots, \tilde{x}_{N-1} \right)$
of~$\bm{x}$ specified by the shear matrix~$\upbm{S}$ is defined so that $\tilde{\bm{x}}= \upbm{S}\bm{x}$.
For our application, we only consider a shear transform with an upper unit-triangular shear matrix.
In this case,
\begin{align}
    \tilde{x}_i = x_i + \sum_{j=i+1}^{N-1} S_{ij} x_j,
\end{align}
is the shear transform~$\tilde{\bm{x}}$ of $\bm{x}$.

The quantum shear transform is similar to the quantum fast Fourier transform.
Specifically, the QST specified by~$\upbm{S}$ is a quantum circuit $\operatorname{\textsc{qst}}$ designed so that
\begin{equation}
\label{eq:qst}
\operatorname{\textsc{qst}}:
\ket{x_0} \ket{x_1} \cdots \ket{x_{N-1}}
\mapsto
\ket{\tilde{x}_0} \ket{\tilde{x}_1} \cdots \ket{\tilde{x}_{N-1}},
\end{equation}
where~$\tilde{x}_i$ are components of the shear-transformed vector.
A shear matrix~$\upbm{S}$ is not unitary, so the map $\ket{\bm{x}} \mapsto \ket{\upbm{S}\bm{x}}$ cannot be directly implemented by a quantum circuit.
However, by storing the shear elements of~$\upbm{S}$ on an ancillary quantum register and performing quantum arithmetic, this map can be implemented by a quantum circuit.

% Decomposing a shear matrix
To elucidate implementation of \textsc{qst}~\eqref{eq:qst} by a quantum circuit, first we describe how to decompose a dense shear matrix as a product of shear matrices with exactly one shear element.
Let~$\upbm{U}$ be an~$N$-by-$N$ upper unit-triangular matrix and let~$\upbm{S}(i,j,s)$ be a $N$-by-$N$ upper unit-triangular shear matrix with one shear element~$s \in \reals$ at entry~$(i,j)$ with~$i<j$.
We decompose~$\upbm{U}$ into a product of shear matrices with exactly one shear element as
\begin{equation}
\label{eq:shearDecomp}
\upbm{U}=\prod _{i=N-2}^{0}
\prod_{j=N-1}^{i+1}\upbm{S}\left(i,j,u_{ij}\right),
\end{equation}
where~$u_{ij}$ are elements of the matrix~$\upbm{U}$.
This decomposition implies that a general shear transform can be implemented by performing a sequence of shear transforms with exactly one shear element.
Therefore, we only consider a shear transform with one nonzero shear element.

% Required shear transform for basis transform
The required shear transform for basis transformation in the wavelet-based algorithm is the inverse-transpose of the upper unit-triangular matrix~$\upbm{U}$ in the UDU decomposition for the ground-state ICM in a multi-scale wavelet basis.
By~\cref{eq:shearDecomp}, we decompose this matrix as
\begin{equation}
\label{eq:shear_transform}
\left(\upbm{U}^\T\right)^{{-1}}=
\prod _{i=N-2}^{0}
\prod_{j=N-1}^{i+1}\upbm{S}^\T \left(i,j,-u_{ij}\right),
\end{equation}
where we used the fact that the inverse of a shear matrix with one shear element is a shear matrix with the shear element negated.
To implement the shear transform specified by the shear matrix~$\upbm{S}^\T(i,j,-u_{ij})$ by a quantum circuit,
we write the shear element~$s=-u_{ij}$ into a~$p$-qubit register labelled \shear\ and implement
\begin{equation}
\label{eq:shear_primitive}
    \ket{s}_\shear \ket{x}_{\vac[i]} \ket{y}_{\vac[j]}
    \mapsto 
    \ket{s}_\shear \ket{x}_{\vac[i]} \ket{y+sx}_{\vac[j]},
\end{equation}
where $\vac[i]$ and $\vac[j]$ are two $p$-qubit quantum registers that store the two $p$-bit numbers~$x$ and~$y$, respectively.
Notice that this map is identical to the quantum multiplication $\textsc{mul}$ in~\cref{eq:mul_operation}.
Hence we implement the map by performing one $\textsc{mul}$~\eqref{eq:mul_operation} operation and describe it by~$\textsc{mul}(\shear, \vac[i], \vac[j])$ in our quantum algorithm.

% Inputs/output
We now discuss the inputs and output of our quantum algorithm for performing the required shear transform for ground-state generation.
Inputs to the algorithm are $N$ $p$-qubit quantum states and the parameters that specify the shear matrix~$\upbm{U}$ in the UDU decomposition of the approximate ICM:
wavelet index~\dbIndex, number of modes~$N$, the upper bandwidth~$w$~\eqref{eq:bandwidth}
and shear elements of $\upbm{U}$.
The algorithm's output is an~$(N\times p)$-qubit quantum state that represents an approximation for the ground state~\eqref{eq:multiscale_groundstate} of the discretized QFT in a multi-scale wavelet basis.

% Procedure
We now describe the procedure of our quantum algorithm for performing the shear transform specified by the shear matrix in~\cref{eq:shear_transform}.
To perform this shear transform, we need to perform a sequence of shear transforms with one shear element;
order of the shear transforms is imposed by~\cref{eq:shear_transform}.
We start by performing the rightmost shear transform in \cref{eq:shear_transform} and proceed to perform the leftmost shear transform.
That is, we start by $i=0$ and proceed to $i=N-1$, and for each $i$ we perform the shear transform defined by $\upbm{S}^\T(i,j, -u_{ij})$ for $j>i$.

The inputs, output and explicit procedure of our quantum algorithm for performing the needed shear transform for ground-state generation is presented by the pseudocode in~\cref{alg:QST}.
In this algorithm, we use the helper functions in Library~\ref{lib:QST} and our method for storing the nonzero shear elements---described in~\cref{subsubsec:UDUDecomp} and illustrated in~\cref{fig:shearstoring}---to find the location of the nonzero shear element in the shear~matrix and only perform the shear transform induced by these elements.

\begin{library}[H]
  \caption{Helper functions for quantum shear transform in~\cref{alg:QST}
  \label{lib:QST}
  }
  \begin{algorithmic}[1]
\Function{ssQST}{$i,s_0,\upbm{s}^{(s_0)}_\sS,p,\vac[0:2^{s_0}-1]$}
\label{line:ssQST}
% \Comment{\vac\ is a qubit register with $N$ cells, each of which has $p$ qubits;
% $N$ is assumed to be an even integer}
\For{$j \gets i+1$ to $2^{s_0}-1$} \label{line:shearUss-1}
\Comment{lines~(\ref{line:shearUss-1}--\ref{line:shearUss-2}) perform the shear transform induced by the ss block of inverse-transpose of $\upbm{U}$}
    \State $\shear \gets -s^{(s_0)}_{\sS;\, i2^{s_0}-i(i+1)/2+j-(i+1)}$
    \Comment{writes negate of the shear element at $(i,j)^\text{th}$ entry of $\upbm{U}^{(s_0)}_\sS$ into the \shear\ register}
    \State $\MUL(\shear, \vac[i], \vac[j])$
    \State \label{line:shearUss-2}
    $\shear \gets -s^{(s_0)}_{\sS;\, i2^{s_0}-i(i+1)/2+j-(i+1)}$
    \Comment{erases the \shear\ register by re-writing negate of the shear element}
\EndFor
\State \Yield $\vac[0:2^{s_0}-1]$
\EndFunction
\Function{swQST}{$i,s_0,c,\upbm{S}^{(s_0,c)}_\sw,p,\vac[2^c:2^{c+1}-1]$}
\label{line:swQST}
\For{$j \gets 0 $ to $2^c-1$}
\Comment{iterates over columns of $\upbm{S}_{\sw}^{(s_0,c)}$}
    \State $\shear \gets -s^{(s_0,c)}_{\sw;\,i,j}$
    \State $\MUL\left(\shear, \vac[i], \vac\left[2^c+j\right]\right)$
    \Comment{note that $j^\text{th}$ column of $\upbm{U}_\sw^{(s_0,c)}$ is $(2^k+j)^\text{th}$ column of $\upbm{U}$}
    \State \label{line:shearUsw-2}
    $\shear \gets  -s^{(s_0,c)}_{\sw;\,i,j}$
\EndFor
\State \Yield $\vac[2^c:2^{c+1}-1]$
\EndFunction
\Function{diagwwQST}{$i,r,\upbm{S}^{(r,r)}_\ww,p,w,\vac[2^r:2^{r+1}-1]$}
\label{line:diagwwQST}
\For{$j \gets 0$ to $w-1$}
\Comment{iterates over columns of $\upbm{S}_\ww^{(r,r)}$}
    \State $\shear \gets -s^{(r,r)}_{\ww;\,i,j}$
    \If{$i<2^r-w$}
        \State $\MUL\left(\shear, \vac\left[2^r+i\right], \vac\left[2^r+i+1+j\right]\right)$
        \Comment{note that $i^\text{th}$ row of $\upbm{U}_\ww^{(r,r)}$ is $(2^r+i)^\text{th}$ row of $\upbm{U}$}
    \EndIf
    \If{$i\geq 2^r-w$ and $j< 2^r-1-i$}
        \State $\MUL\left(\shear, \vac\left[2^r+i\right], \vac\left[2^{r+1}-w+j\right]\right)$
    \EndIf
    \State $\shear \gets -s^{(r,r)}_{\ww;\,i,j}$
\EndFor
\If{$i\leq w-1$}
\Comment{perform shear transform corresponding to shear elements in TRP of $\upbm{S}_\ww^{(r,r)}$}
    \For{$j\gets i$ to $w-1$}
        \State $\shear \gets -s^{(r,r)}_{\ww;\, 2^r-1-i,j}$
        \State $\MUL\left(\shear, \vac\left[2^r+i\right], \vac\left[2^{r+1}-w+j\right]\right)$
        \State $\shear \gets -s^{(r,r)}_{\ww;\, 2^r-1-i,j}$
    \EndFor
\EndIf
\State \Yield $\vac[2^r:2^{r+1}-1]$
\EndFunction
\Function{offDiagwwQST}{$i,r,c,\upbm{S}^{(r,c)}_\ww,p,w,\vac[2^c:2^{c+1}-1]$}
\label{line:offDiagwwQST}
\State $h \gets \textsc{lastColNNZ}(r,c)$
\State $H \gets \textsc{vertWidth}(r,c)$
\State $W \gets \textsc{width}(r,c)$
\State $(i_\text{LT},i_\text{FB}) \gets
        \left(h-1, 2^r-(H-h-1)\right)$
\Comment{$i_\text{LT}$ is row index of last column's last nonzero in TRP and $i_\text{FB}$ is row index of first column's first nonzero in BLP of $\upbm{U}^{(r,c)}_\ww$}
\State $j_\text{FT}\gets \textsc{width}(r,c)-w + i2^{c-r}$
\Comment{computes column index of first nonzero in $i^\text{th}$ row in TRP of $\upbm{S}^{(r,c)}_\ww$}
\For{$j \gets 0 $ to $\textsc{width}(r,c)-1$}
\Comment{iterates over columns of $\upbm{S}_\ww^{(r,c)}$}
    \State $\shear \gets -s^{(r,c)}_{\ww;\, i,j}$
    \If{$i \leq i_\text{LT}$ and $j < j_\text{FT}$}
    \label{line:MPshearStart}
    \Comment{lines (\ref{line:MPshearStart}---\ref{line:MPshearEnd}): perform shear transform corresponding to shear elements in MP of $\upbm{S}_\ww^{(r,c)}$}
        \State $\MUL(\shear,\vac\left[2^r+i\right],\vac\left[2^c+j\right])$
    \EndIf
    \If{$i_\text{LT} < i < i_\text{FB}$}
        \State $\MUL\left(\shear, \vac\left[2^r+i\right], \vac\left[2^c+2^{c-r}(i-i_\text{LT})+j\right]\right)$
    \EndIf
    \If{$i\geq i_\text{FB}$ and $j\geq 2^{c-r}(i-i_\text{FB})$}
        \State $\MUL\left(\shear, \vac\left[2^r+i\right], \vac\left[2^{c+1}-w-\left(2^c-i\right)2^{c-r}+j\right]\right)$
    \EndIf
     \label{line:MPshearEnd}
    \If{$i\leq i_\text{LT}$ and $j \geq j_\text{FT}$}
    \Comment{perform shear transform corresponding to shear elements in TRP of $\upbm{S}_\ww^{(r,c)}$}
        \State $\MUL\left(\shear, \vac\left[2^r+i\right], \vac\left[2^{c+1}-W+j\right]\right)$
    \EndIf
    \If{$i\geq i_\text{FB}$ and $j< 2^{c-r}(i-i_\text{FB})$}
    \Comment{perform shear transform corresponding to shear elements in BLP of $\upbm{S}_\ww^{(r,c)}$}
        \State $\MUL\left(\shear, \vac{\left[2^r+i\right]}, \vac{\left[2^c+j\right]}\right)$
    \EndIf
    \State \label{line:shearUww-2}$
    \shear \gets -s^{(r,c)}_{\ww;\,i,j}$
\EndFor
\State \Yield $\vac[2^c:2^{c+1}-1]$
\EndFunction
\end{algorithmic}
\end{library}

\begin{algorithm}[H]
  \caption{Quantum algorithm for a shear transform}
  \label{alg:QST}
  \begin{algorithmic}[1]
    \Require{
        \Statex $\dbIndex \in \integers_{\geq 3}$
        \Comment{wavelet index}
        \Statex $N \in \integers_{\geq 2(2\dbIndex-1)}$
        \Comment{number of modes; for convenience, we assume~$N$ is a power of~$2$}
        \Statex $p\in \integers^{+}$
        \Comment{working precision}
        \Statex $w \in \integers^{+}$
        \Comment{upper bandwidth of diagonal \ww\ blocks in approximate ICM $\tilde{\upbm{A}}$}
        \Statex 
         $\upbm{S}:=\cBraket{\left.
    \upbm{s}^{(s_0)}_\sS \in \reals^{\frac1{2}2^{s_0}(2^{s_0}-1)},
    \upbm{S}^{(s_0,c)}_\sw \in \reals^{2^{s_0}\times 2^c},
    \upbm{S}^{(r,c)}_\ww \in \reals^{2^r\times\left(\textsc{width}(r,c)
    -(w+1)\updelta_{rc}\right)}
    \right\vert 
    \ceil*{\log_2(4\dbIndex-2)} =:s_0 \leq r \leq c < k:=\log_2 N
    }$
    \Comment{shear elements in main and upper-diagonal blocks of~$\upbm{U}$ in UDU decomposition of $\tilde{\upbm{A}}$;
    here $\textsc{width}(r,c)$~\eqref{eq:block_bandwidth} is bandwidth of $\tilde{\upbm{A}}^{(r,c)}_\ww$}
    \Statex $\vac \in \bigotimes_{\ell=0}^{N-1} \mathscr{H}_2^p$
    \Comment{tensor product of $N$ $p$-qubit 1DG states}
            }
    \Ensure
        \Statex $\vac \in \mathscr{H}_2^{N\times p}$
        \Comment{$(N\times p)$-qubit approximate ground state prepared on \vac\ register}
\Function{QST}{$\dbIndex, N, p, w, \upbm{S}$}
\For{$i \gets 0$ to $2^{s_0}-1$}
    \Comment{perform the shear transform induced by the ss block of inverse-transpose of $\upbm{U}$}
    \State $\textsc{ssQST}(i,s_0,\upbm{s}^{(s_0)}_\sS,p,\vac[0:2^{s_0}-1])$
    \Comment{see line~\eqref{line:ssQST} of~Library~\ref{lib:QST}}
    \For{$c \gets s_0 $ to $k-1$} \label{line:shearUsw-1}
    \Comment{perform the shear transform induced by the sw blocks of inverse-transpose of $\upbm{U}$}
     \State $\textsc{swQST}(i,s_0,c,\upbm{S}^{(s_0,c)}_\sw,p,\vac[2^c:2^{c+1}-1])$
     \Comment{see line~\eqref{line:swQST} of~Library~\ref{lib:QST}}
    \EndFor
\EndFor
\For{$r \gets s_0$ to $k-1$}\label{line:shearUww-1}
\Comment{perform the shear transform induced by the ww blocks of inverse-transpose of $\upbm{U}$}
    \For{$i \gets 0 $ to $2^r-1$}
    \Comment{iterates over rows of $\upbm{U}_\ww^{(r,r)}$}
        \State $\textsc{diagwwQST}(i,r,\upbm{S}^{(r,r)}_\ww,p,w,\vac[2^r:2^{r+1}-1])$
        \Comment{see line~\eqref{line:diagwwQST} of~Library~\ref{lib:QST}}
    \For{$c \gets r+1 $ to $k-1$}
        \State $\textsc{offDiagwwQST}(i,r,c,\upbm{S}^{(r,c)}_\ww,p,w,\vac[2^c:2^{c+1}-1])$
        \Comment{see line~\eqref{line:offDiagwwQST} of~Library~\ref{lib:QST}}
    \EndFor
\EndFor    
\EndFor
\State \Yield \vac
\EndFunction
\end{algorithmic}
\end{algorithm}

%%%%%%%%%%%%%%%%%%%%%%%%%%%%%%%%%%%%%%%%%%%%%%%%%%%%%%%%%
%%%%%% Complexity analysis %%%%%%%%%%%%%%%%%%%%%%%%%%%%%%%
%%%%%%%%%%%%%%%%%%%%%%%%%%%%%%%%%%%%%%%%%%%%%%%%%%%%%%%%%
\subsection{Complexity analysis}
\label{subsect:complexity_analysis}

In this subsection, we analyze time complexity for our two ground-state-generation algorithms with respect to the primitive operations discussed in~\cref{subsubsec:complexity_measure}.
We begin, in~\cref{subsubsec:FBA_classical_complexity}, by analyzing time complexity for classical preprocessing of the Fourier-based algorithm and present time complexity for the wavelet-based algorithm's classical preprocessing  in~\cref{subsubsec:WBA_classical_complexity}.
The time complexity for our quantum algorithm for generating a one-dimensional Gaussian state is discussed in~\cref{subsubsec:1DG_complexity}.
We analyze time complexity for the quantum fast Fourier-transform algorithm and the quantum shear-transform algorithm in~\cref{subsubsec:QFHT_complexity} and~\cref{subsubsec:QST_complexity}, respectively.
Finally, we put all complexities together in~\cref{subsubsec:overall_complexity} and discuss the overall time complexity for the Fourier- and wavelet-based algorithms.

%%%%%%%%%%%%%%%%%%%%%%%%%%%%%%%%%%%%%%%%%%%%%%%%%%%%%%%%%
\subsubsection{Classical preprocessing in Fourier-based algorithm}
\label{subsubsec:FBA_classical_complexity}

Here we analyze time complexity for classical preprocessing in the Fourier-based algorithm.
Following Algorithm~\ref{alg:FBA}, classical preprocessing involves computing various elementary functions, such as logarithm and trigonometric functions, and has two key subroutines:
computing the second-order derivative overlaps~\eqref{eq:derivative_overlaps} and eigenvalues of the ground state's ICM.
First we discuss time complexity for computing the elementary functions and then analyze the key subroutines' time complexity.
Finally, we build on these complexities and discuss the overall complexity of classical preprocessing.

% Complexity for computing elementary functions
The elementary functions used in various subroutines of the classical preprocessing in Algorithm~\ref{alg:FBA} are logarithm, square-root, inverse-square-root, and trigonometric functions.
The time complexity for computing these functions is analyzed with respect to time complexity for performing multiplication~\cite{Bre10}.
Multiplication is a primitive operation in our cost model and has a unit cost.
Therefore, as per~\cref{subsec:QRAM}, the time complexity for computing square-root or inverse-square-root of a number in our cost model is~$\order{1}$, and the time complexity for computing logarithm or trigonometric functions to precision~$p$ are each~$\order{\log p}$.

% Complexity for computing derivative overlaps
We now proceed with analyzing time complexity for computing the second-order derivative overlaps~\eqref{eq:derivative_overlaps}.
We compute these overlaps by Algorithm~\ref{alg:derivativeOverlaps} using Beylkin's method~\cite{Bey92}.
The derivative overlaps in this method are elements of the unique solution-vector for a system of~$2\dbIndex$ linear algebraic equations with~$2\dbIndex-1$ unknowns; see~\cref{appx:derivative_overlaps}.
The standard algorithm for solving the system of linear equations is based on Gaussian elimination that requires~$\order{\dbIndex^3}$ basic arithmetic operations.
Hence the time complexity for computing the derivative overlaps is cubic with respect to the wavelet index~\dbIndex.

% Complexity for computing the eigenvalues
We compute eigenvalues of the ground state's ICM by~Algorithm~\ref{alg:invCovEigens}.
By the lines~(\ref{line:eigen1}--\ref{line:eigen2}) of this algorithm, computing each eigenvalue requires performing~$\order{\dbIndex}$ basic arithmetic operations and computing two elementary functions: square root and cosine.
The square-root function is computed once, and cosine is computed~$\order{\dbIndex}$ times.
By time complexity for computing these functions, computing each eigenvalue requires~$\order{\dbIndex \log p}$ basic arithmetic operations.
Therefore, time complexity for computing the eigenvalues is 
\begin{equation}
\label{eq:T_eigens}
    T_\text{eigens} \in \order{N\dbIndex\log p},
\end{equation}
because~$N$ eigenvalues are computed.

% Overall complexity
We lastly put all complexities together to achieve the overall time complexity for classical preprocessing in the Fourier-based algorithm.
Following Algorithm~\ref{alg:FBA}, classical preprocessing requires computing the working precision~$p$,
the second-order derivative overlaps~$\bm\Delta$,
eigenvalues~$\bm\uplambda$ of the ICM,
lattice spacing~$\delta$ and
the standard deviation~$\tilde{\bm\upsigma}$ for the discrete 1DG states.
Computing~$p$ requires computing one inverse-square-root function and one logarithm function.
Computing~$\delta$ and~$\tilde{\bm\upsigma}$ respectively require computing one and~$N$ inverse-square-root functions.
Therefore, by the time complexities for computing logarithm and inverse-square-root functions,
the time complexity for computing~$p$,~$\delta$ and~$\tilde{\bm\upsigma}$ are $\order{\log p}, \order{1}$ and~$\order{N}$, respectively.
The combination of these complexities with the complexity for computing~$\bm\Delta$ and~$\bm\lambda$ yields $\order{\dbIndex^3+N\dbIndex\log p}$. 
As $p$~\eqref{eq:p} is logarithmic in $N$, the overall time complexity for classical preprocessing is quasilinear in the number of modes~$N$.

%%%%%%%%%%%%%%%%%%%%%%%%%%%%%%%%%%%%%%%%%%%%%%%%%%%%%%%%%
\subsubsection{Classical preprocessing in wavelet-based algorithm}
\label{subsubsec:WBA_classical_complexity}

We now analyze time complexity for classical preprocessing in the wavelet-based algorithm.
In contrast to Algorithm~\ref{alg:FBA}, classical preprocessing in the wavelet-based algorithm, Algorithm~\ref{alg:WBA}, has two unique subroutines:
computing the circulant row in unique blocks of the ground state's ICM and the UDU decomposition for the approximate ICM.
First we elaborate on the time complexity for these unique subroutines and then discuss the overall time complexity for classical preprocessing in the wavelet-based algorithm.

% Complexity for computing the circulant rows
We begin by analyzing time complexity for Algorithm~\ref{alg:invCovCircRows}, which computes the circulant row in unique blocks of the approximate ICM.
This algorithm involves computing eigenvalues of the ground-state ICM,
the low-pass filters~\eqref{eq:scaling_function} and the circulant row of the fixed-scale ICM $\upbm{A}^{(c)}_\sS$~\eqref{eq:fixedscale_groundstate} at scale~$c$ for~$c\in\{s_0, \ldots, k\}$.
The time complexity to compute the low-pass filters for Daubechies~$\dbIndex$ wavelets with precision $p$ is
\begin{equation}
\label{eq:T_lowpass}
    T_\text{lowpass} \in
    \order{\dbIndex \log^5\dbIndex \log{(3\dbIndex+p)}},
\end{equation}
see~\cref{appx:lowpass_filter}.
As per
lines~(\ref{line:scalekCircrow-start}--\ref{line:scalekCircrow-end}) of Algorithm~\ref{alg:invCovCircRows}, computing circulant row of~$\upbm{A}^{(k)}_\sS$ requires performing
one division, $\Theta(N)$ multiplications, $N$ additions and computing~$N$ cosine functions.
Computing trigonometric functions with precision~$p$ require~$\order{\log p}$ basic arithmetic operations;
see~\cref{subsec:QRAM}.
Therefore, time complexity to compute circulant row of~$\upbm{A}^{(k)}_\sS$ is~$\order{N\log p}$.
Computing circulant row of~$\upbm{A}^{(c)}_\sS$ by circulant row of~$\upbm{A}^{(c+1)}_\sS$ requires~$\Theta(\dbIndex 2^c)$ basic arithmetic operations.
Therefore, the overall time complexity to compute the circulant rows of~$\upbm{A}^{(c)}_\sS$ for~$c\in\{s_0, \ldots, k-1\}$ is
\begin{equation}
\label{eq:T_circss}
    T^\sS_\text{circ} \in \Theta(\dbIndex 2^k)
    =\Theta(\dbIndex N),
\end{equation}
where we used~$k=\log_2 N$.
By a similar analysis, following lines~(\ref{line:ssCircrow-start}--\ref{line:swCircrow-end}) of Algorithm~\ref{alg:invCovCircRows}, we obtain
\begin{equation}
\label{eq:T_circww}
    T^\sw_\text{circ} \in \Theta(\dbIndex N\log_2 N),
    \quad T^\ww_\text{circ} \in \Theta(\dbIndex N\log_2 N),
\end{equation}
for time complexity to compute circulant rows in the sw and ww blocks of the ground-state ICM.
The time complexity to perform the rest of computation is
\begin{equation}
\label{eq:T_highpass}
    T_\text{highpass}\in \order{\dbIndex},
\end{equation}
the high-pass filters in lines~(\ref{line:highpass-start}--\ref{line:highpass-end}) of~\cref{alg:invCovCircRows}.
Altogether, Eqs~\eqref{eq:T_eigens} to~\eqref{eq:T_highpass} and the fact that~$p$ is logarithmic in the numbers of modes~$N$, yield
\begin{equation}
\label{eq:circ_complexity}
    T_\text{circ}=
    T_\text{eigens}+ T_\text{lowpass}+
    T^\sS_\text{circ}+T^\sw_\text{circ}+T^\ww_\text{circ}
    +T_\text{highpass}
    \in \order{\dbIndex N \log_2 N},
\end{equation}
for the overall time complexity to compute the circulant rows.

% Complexity for the UDU decomposition
We now proceed with analyzing time complexity for our UDU-decomposition algorithm presented in Algorithm~\ref{alg:invCovUDU}.
This algorithm computes diagonal elements of the diagonal matrix~$\upbm{D}$ and shear elements of the upper unit-triangular matrix~$\upbm{U}$ in the UDU decomposition of the approximate ICM $\tilde{\upbm{A}}$.
We separately analyze time complexity for computing the diagonal elements and time complexity for computing the shear elements.
The sum of these complexities then yields the UDU decomposition's time complexity.

% Complexity for computing the diagonal elements
First we analyze time complexity for computing diagonal elements of~$\upbm{D}$ which,
as per~\cref{eq:blcokDiagonalD},
is a block-diagonal matrix whose blocks are diagonal matrices~$\upbm{D}^{(s_0)}_\sS$ and $\upbm{D}^{(c)}_\ww$ for $s_0\leq c<k$.
Computing~$i^\text{th}$ diagonal element~$d^{(c)}_{\ww;\,i}$ of~$\upbm{D}^{(c)}_\ww$ according to~\cref{eq:wwDiags} requires computing nonzero elements in~$i^\text{th}$ row of~$\upbm{V}^{(c,s)}_\ww$~\eqref{eq:matrixV} for
each~$s\in\{c,\ldots, k-1\}$,
multiplying each nonzero element in~$i^\text{th}$ row of~$\upbm{U}^{(c,s)}_\ww$~\eqref{eq:blockUDU} by its corresponding element in~$i^\text{th}$ row of $\upbm{V}^{(c,s)}_\ww$ and adding them all.
The number of nonzero~(NNZ) elements in each row of~$\upbm{V}^{(c,s)}_\ww$ is at most $\textsc{width}(c,s)$, where $\textsc{width}(c,s)$~\eqref{eq:block_bandwidth} is bandwidth of the block~$\tilde{\upbm{A}}^{(c,s)}_{\ww}$
and, by~\cref{eq:matrixV}, computing each nonzero element of $\upbm{V}^{(c,s)}_\ww$ requires one multiplication.
Therefore, computing each of the~$2^c$ diagonal elements of~$\upbm{D}^{(c)}_\ww$ requires at most~$3\sum_s \textsc{width}(c,s)$ basic arithmetic operations.
Consequently, time complexity to compute all diagonal elements in the \ww\ blocks of~$\upbm{D}$ is
\begin{equation}
\label{T_diagsww}
    T_\text{diags}^\ww =\sum_{c=s_0}^{k-1}
    2^c\times \left(3\sum_{s=c}^{k-1} \textsc{width}(c,s)\right)
    \in \order{\dbIndex N \log_2 N},
\end{equation}
where we use $k=\log_2N$.

By a similar analysis, we obtain time complexity to compute diagonal elements in the ss block of~$\upbm{D}$.
To compute $i^\text{th}$ diagonal element in this block, we compute all elements in the $i^\text{th}$ row of the ss and sw blocks of~$\upbm{V}$ because these blocks are dense matrices.
Being dense blocks, the sum of nonzero elements in any row of the ss and sw blocks is at most~$N$, the number of columns of~$\upbm{V}$.
Therefore, computing each diagonal element requires at most $3N$ basic arithmetic operations, so we have~$T_\text{diags}^\sS = 2^{s_0}~\times~(3N)$ for time complexity to compute diagonal elements in the ss block of~$\upbm{D}$.
This equation together with~\cref{T_diagsww} yield
\begin{equation}
\label{eq:T_diags}
T_\text{diags}= T_\text{diags}^\sS + T_\text{diags}^\ww \in \order{\dbIndex N \log_2 N},
\end{equation}
for the time complexity to compute all diagonal elements of~$\upbm{D}$.

% Complexity for computing the shear elements
We now analyze time complexity for computing shear elements of~$\upbm{U}$ in the UDU decomposition.
As per~\cref{eq:blockUDU},~$\upbm{U}$ is a block matrix with three types of blocks: ss, sw and ww blocks.
By~\cref{eq:swShears}, computing the shear element at entry~$(m,i)$ of the sw block~$\upbm{U}^{(s_0,c)}_\sw$ requires multiplying each nonzero element in~$m^\text{th}$ row of~$\upbm{U}^{(s_0,s)}_\sw$ by its corresponding element in~$m^\text{th}$ row of~$\upbm{V}^{(c,s)}_\ww$~\eqref{eq:matrixV} for~$s\in\{c,\ldots,k-1\}$ and adding them all.
The NNZ elements in each row of~$\upbm{V}^{(c,s)}_\ww$ is at most $\textsc{width}(c,s)$~\eqref{eq:block_bandwidth}, hence computing each shear element of~$\upbm{U}^{(s_0,c)}_\sw$ requires at most~$2\sum_s\textsc{width}(c,s)$ basic arithmetic operations.
Being a dense matrix,~$\upbm{U}^{(s_0,c)}_\sw$ has~$2^{s_0}\times 2^c$ shear elements and, therefore, we have
\begin{equation}
\label{eq:T_shearssw}
    T_\text{shear}^\sw =
    \sum_{c=k-1}^{s_0}
    \left[
    \left(2^{s_0} \times 2^c\right) \times
    \left(2\sum_{s=c}^{k-1} \textsc{width}(c,s)\right)
    \right] \in \order{\dbIndex N\log_2 N},
\end{equation}
for time complexity to compute all shear elements in the \sw\ blocks of~$\upbm{U}$.
Likewise, by~\cref{eq:wwShears}, computing each shear element of the ww block~$\upbm{U}^{(r,c)}_\ww$ requires at most $2\sum_s\textsc{width}(c,s)$ basic arithmetic operations.
This block is a sparse matrix with at most $2^r \times \textsc{width}(r,c)$ shear elements.
We therefore achieve
\begin{equation}
\label{eq:T_shearsww}
    T_\text{shear}^\ww =
    \sum_{c=k-1}^{s_0}
    \left[
    \left(\sum_{r=s_0}^{c} 2^r\times
    \textsc{width}(r,c)\right) \times
    \left(2\sum_{s=c}^{k-1} \textsc{width}(c,s)\right)
    \right] \in \order{\dbIndex N\log_2 N},
\end{equation}
for time complexity to compute all shear elements in the ww blocks of~$\upbm{U}$.

A similar analysis yields the time complexity~$T_\text{shear}^\sS$ to computing all shear elements in the ss bock of~$\upbm{U}$.
By~\cref{eq:ssShears}, computing the shear element at entry~$(m,i)$ of the ss block requires multiplying each nonzero element in~$m^\text{th}$ row of this block by its corresponding element in~$m^\text{th}$ row of the ss block of~$\upbm{V}$,
multiplying each nonzero element in~$m^\text{th}$ row of~$\upbm{U}^{(s_0,s)}_\sw$ by its corresponding element in~$m^\text{th}$ row of~$\upbm{V}^{(s_0,s)}_\sw$ for every~$s\in\{s_0,\ldots,k-1\}$
and adding them all.
As the ss and sw blocks of~$\upbm{U}$ are dense matrices, computing each shear element in the ss block requires at most~$2N$ basic arithmetic operations and because the ss block has~$2^{s_0}\times (2^{s_0}-1)/2 \in \Theta(\dbIndex^2)$ shear elements, we have~$T_\text{shear}^\sS \in \order{\dbIndex^2N}$.
The sum of this time complexity by the times complexities in~\cref{eq:T_shearssw} and~\cref{eq:T_shearsww} yield
\begin{equation}
\label{eq:T_shear}
T_\text{shear} = T_\text{shear}^\sS + T_\text{shear}^\sw + T_\text{shear}^\ww \in \order{\dbIndex^2 N\log_2 N},
\end{equation}
for time complexity to computing all shear elements of~$\upbm{U}$.
Finally, by the time complexities for computing the diagonal~\eqref{eq:T_diags} and shear~\eqref{eq:T_shear} elements we obtain
\begin{equation}
\label{eq:UDU_complexity}
    T_\text{UDU} = T_\text{diags}+ T_\text{shear}
    \in \order{\dbIndex^2 N\log_2 N},
\end{equation}
for the UDU decomposition's time complexity.

% Overall complexity
We finally discuss the overall time complexity for classical preprocessing in the wavelet-based algorithm.
Classical preprocessing in Algorithms~\ref{alg:FBA} and~\ref{alg:WBA} have four identical subroutines:
computing
(1)~working precision $p$;
(2)~second-order derivative overlaps~$\bm{\Delta}$;
(3)~lattice spacing~$\delta$;
and (4)~the vector~$\tilde{\bm{\upsigma}}$ of standard deviations for the approximate 1DG states.
As per the discussion in~\cref{subsubsec:FBA_classical_complexity}, the overall time complexity for these subroutines is~$\order{N}$.
By the combination of this time complexity with the time complexities for the two unique subroutines,~$T_\text{circ}$~\eqref{eq:circ_complexity} and~$T_\text{UDU}$~\eqref{eq:UDU_complexity},
we conclude that the time complexity for classical preprocessing in the wavelet-based algorithm is quasilinear in the number of modes~$N$.

%%%%%%%%%%%%%%%%%%%%%%%%%%%%%%%%%%%%%%%%%%%%%%%%%%%%%%%%%
\subsubsection{One-dimensional Gaussian-state generation}
\label{subsubsec:1DG_complexity}

Here we analyze time complexity for the quantum algorithm presented in~\cref{alg:1DG} that generates an approximation for a 1DG state on a quantum computer.
We begin with analyzing time complexity for executing the quantum circuit in~\cref{fig:1DGCircuit}, which represents the iterative part of~\cref{alg:1DG}.
Then we discuss the algorithm's overall time complexity.

% Complexity of the iterative part
The quantum circuit in~\cref{fig:1DGCircuit} involves performing
two \textsc{angle}~\eqref{eq:angle_operation} operations,
one \textsc{rot}~\eqref{eq:rot_operation},
two \textsc{shift}~\eqref{eq:shift_operation},
one Pauli~$X$ and
one \textsc{cnot} operation.
The circuit's time complexity~$T_\text{iter}$ is therefore
\begin{align}
    T_\text{iter}=2T_\textsc{angle}+T_\textsc{rot}+2T_\textsc{shift}+2,
\end{align}
where $T_\textsc{angle}, T_\textsc{rot}$ and $T_\textsc{shift}$ are time complexities to implement \textsc{angle}, \textsc{rot} and \textsc{shift} operations, respectively.
The \textsc{rot} and \textsc{shift} operations operate on quantum registers with~$p$ qubits, where~$p$ is logarithmic in the number of modes.
As illustrated in~\cref{fig:rotrShift}, implementing \textsc{rot}~\eqref{eq:rot_operation} requires performing at most~$p$ standard rotations and \textsc{shift}~\eqref{eq:shift_operation} is implemented by performing~$p$ \textsc{swap} gates.
Therefore, $T_\textsc{rot} \in \order{p}$ and $T_\textsc{shift}=p$.
Our analysis in~\cref{appx:rotation_angle} shows that $T_\textsc{angle} \in \order{p^2}$.
The combination of these complexities yields
\begin{equation}
\label{eq:iter_complexity}
    T_\text{iter} \in \order{p^2},
\end{equation}
for time complexity of the quantum circuit in~\cref{fig:1DGCircuit}, which implements the iterative part of Algorithm~\ref{alg:1DG}.
\begin{figure}[htb]
    \centering
    \includegraphics[width=.86\linewidth]{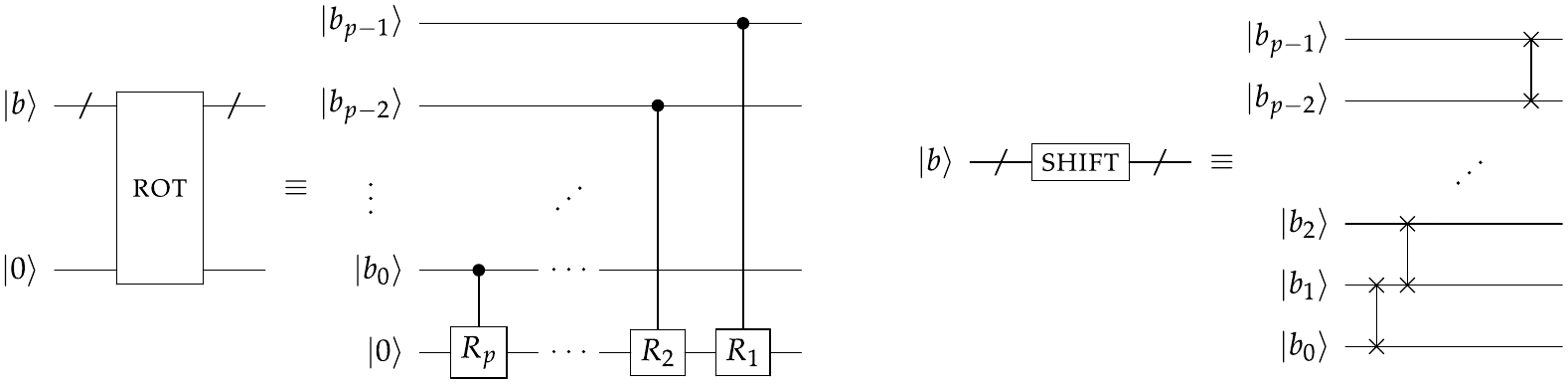}
    \caption[Implementing high-level operations by low-level operations]{Implementation of \textsc{rot}~\eqref{eq:rot_operation} and \textsc{shift}~\eqref{eq:shift_operation} by primitive operations.
    Left: implementing \textsc{rot} using bits of a $p$-bit number $b$ requires at most $p$ standard rotations $R_\ell:=\text{exp}(-2\uppi\upi/2^\ell)$.
    Right: implementing \textsc{shift} requires $p$ \textsc{swap} gates.}
    \label{fig:rotrShift}
\end{figure}

% Overall complexity
The quantum algorithm in Algorithm~\ref{alg:1DG} performs the operations of the quantum circuit in~\cref{fig:1DGCircuit}~$m$ times.
The algorithm also involves performing $m$ \textsc{swap} operations, lines~(\ref{line:swap_start}--\ref{line:swap_end}) of the algorithm, and one \textsc{mul}~\eqref{eq:mul_operation} operation in line~\ref{line:mul}.
All of these operations are also used in the uncomputation part of Algorithm~\ref{alg:1DG}.
Therefore, by~\cref{eq:iter_complexity} and $m<p$, the overall time complexity~$T_\oneDG$ for generating a 1DG state by Algorithm~\ref{alg:1DG} is
\begin{align}
\label{eq:1DGcomplexity}
    T_\oneDG=(2m) \times T_\text{iter} + 2\times (m+1)
    \in \order{p^{3}}.
\end{align}
As~$p$~\eqref{eq:p} is logarithmic in the number of modes~$N$, the overall complexity for generating a 1DG state is logarithmic in~$N$.

%%%%%%%%%%%%%%%%%%%%%%%%%%%%%%%%%%%%%%%%%%%%%%%%%%%%%%%%%
\subsubsection{Quantum fast Hartley transform}
\label{subsubsec:QFHT_complexity}

We now analyze time complexity for our quantum fast Hartley-transform algorithm presented in Algorithm~\ref{alg:QFHT}.
Our algorithm is based on three key quantum operations: \textsc{qrev}~\eqref{eq:qrev}, \textsc{qbf}~\eqref{eq:qbf} and \textsc{hs}~\eqref{eq:hs}; see~\cref{subsubsec:QFFT}.
First we analyze time complexity for performing these key operations with respect to the primitive operations in our cost model~(\cref{subsubsec:complexity_measure}).
Then we build on these complexities and discuss the overall time complexity for the QFHT algorithm.

We begin with analyzing time complexity to perform the quantum data-reordering~\textsc{qrev}~\eqref{eq:qrev}.
This operation reorders the values encoded in the size-$N$ input state-vector $\vac$ by swapping qubits of~$\vac[i]$ and~$\vac[\textsc{rev}(i)]$ for each~$i\in\{1,\ldots,N/2-1\}$, where~$\textsc{rev}(i)$ is an integer obtained form~$i$ by reversing its binary digits.
Hence, we have~$T_\textsc{qrev}\in \order{N}$ for time complexity to perform~\textsc{qrev}~\eqref{eq:qrev}.
Next we analyze time complexity~$T_\textsc{qbf}$ for performing the quantum butterfly~\textsc{qbf}~\eqref{eq:qbf}.
This operation can be implemented as
\begin{equation}
    \ket{x}_{\vac[i]} \ket{y}_{\vac[j]} \mapsto
    \ket{x+y}_{\vac[i]}\ket{y}_{\vac[j]} \mapsto
    \ket{x+y}_{\vac[i]}\ket{2y}_{\vac[j]} \mapsto
    \ket{x+y}_{\vac[i]}\ket{(x+y)-2y}_{\vac[j]},
\end{equation}
on two quantum registers~$\vac[i]$ and~$\vac[j]$ by two additions and one multiplication.
The quantum addition and multiplication here are in-place, and no uncomputation is required.
Therefore,~$T_\textsc{qbf} \in \order{1}$ with respect to the quantum primitive operations.
We use the LDU decomposition
\begin{equation}
\label{eq:qbf_implementation}
    \begin{bmatrix}
    c&s\\
    s&-c
 \end{bmatrix}
 =
 \begin{bmatrix}
    1&0\\
    s/c&1
 \end{bmatrix}
  \begin{bmatrix}
    c&0\\
    0&-1/c
 \end{bmatrix}
  \begin{bmatrix}
    1&s/c\\
    0&1
 \end{bmatrix},
\end{equation}
to implement the primitive Hartley shift~\textsc{hs}~\eqref{eq:hs} by in-place quantum arithmetic operations;
note here~$s^2+c^2=1$.
Assuming that the values of~$c$ and~$s$ are stored into two ancillary quantum registers, \textsc{hs}~\eqref{eq:hs} can be implemented by in-place operations~as
\begin{equation}
\label{eq:hs_implementation}
    \ket{x}_{\vac[i]} \ket{y}_{\vac[j]} \mapsto
    \ket{x+ys/c}_{\vac[i]}\ket{y}_{\vac[j]} \mapsto
    \ket{cx+sy}_{\vac[i]}\ket{-y/c}_{\vac[j]} \mapsto
    \ket{cx+sy}_{\vac[i]}\ket{xs-yc}_{\vac[j]},
\end{equation}
on two quantum registers~$\vac[i]$ and~$\vac[j]$ by performing six basic arithmetic operations.
We therefore
have~$T_\textsc{hs} \in \order{1}$ for time complexity to perform~\textsc{hs}.

We now discuss the overall time complexity for the QFHT algorithm.
This algorithms starts with reordering the size-$N$ input state-vector and proceeds with~$\log_2(N)$ stages.
Each stage requires performing~$N/2$ quantum butterfly operations and~$N/4$ primitive quantum Hartley-shift operations;
see~\cref{fig:QFHT} and implementation of \textsc{qhs}~\eqref{eq:qhs} by \textsc{hs}~\eqref{eq:hs} in Library~\ref{lib:QFHT}.
Thus, by the time complexities for the key quantum operations,
we have
\begin{equation}
    T_\textsc{qfht} = T_\textsc{qrev} +
    \left((N/2) T_\textsc{qbf} + (N/4) T_\textsc{hs}
    \right) \log_2N \in \order{N\log_2N},
\end{equation}
for time complexity to perform the quantum fast Hartley transform~\eqref{eq:QFHT}.

%%%%%%%%%%%%%%%%%%%%%%%%%%%%%%%%%%%%%%%%%%%%%%%%%%%%%%%%%
\subsubsection{Quantum shear transform}
\label{subsubsec:QST_complexity}

Here we analyze time complexity for performing a shear transform on a quantum computer.
First we discuss time complexity to perform shear transform with a dense shear matrix.
Then we analyze time complexity for our QST algorithm presented in Algorithm~\ref{alg:QST}, which performs the shear transform required for generating the free-field ground state on a quantum register.

% time complexity for QST with a dense shear matrix
We use the decomposition in~\cref{eq:shearDecomp} to analyze time complexity for performing a QST with a dense shear matrix.
As~per this decomposition, performing a general shear transform is accomplished by performing a sequence of shear transforms with exactly one shear element.
The number of terms in the sequence is equal to the NNZ shear elements in the shear matrix, and the sequence's order is specified by~\cref{eq:shearDecomp}.

To perform the shear transform specified by each term of the sequence, we write the shear element into an ancillary quantum register denoted \shear\ and perform a multiplication as per~\cref{eq:shear_primitive}.
Then we erase the \shear\ register by re-writing the shear element.
Therefore, performing a shear transform with exactly one shear element requires three primitive operations:
writing into a quantum register, performing one multiplication and erasing the quantum register.
Consequently, the time complexity for performing a shear transform scales linearly with the NNZ shear elements in the shear-transform matrix.
The NNZ shear elements for a dense $N$-by-$N$ shear matrix is~$N(N-1)/2$, so we have~$\Theta(N^2)$ for time complexity to perform a shear transform with a dense shear matrix.

% time complexity for required QST
We now specify the time complexity to perform the required quantum shear transform for ground-state generation.
The shear matrix for the required shear transform is sparse.
We use the identified relationship between the time complexity to perform a shear transform and the NNZ shear elements to obtain the time complexity for performing a shear transform with a sparse shear matrix.
In particular, the complexity for performing a sparse shear transform is achieved by counting the NNZ shear elements in the shear-transform matrix.
The NNZ shear elements for the sparse shear matrix~$\upbm{U}$ in the UDU decomposition of the approximate ICM scales as $\order{Nw}$, where $w$~\eqref{eq:bandwidth} is logarithmic in the number of modes~$N$.
Therefore, time complexity for Algorithm~\ref{alg:QST}, which performs the QST with the sparse shear matrix~$\upbm{U}$, is quasilinear in the number of modes.

%%%%%%%%%%%%%%%%%%%%%%%%%%%%%%%%%%%%%%%%%%%%%%%%%%%%%%%%%
\subsubsection{Overall complexity for Fourier- and wavelet-based algorithms}
\label{subsubsec:overall_complexity}

We now determine the overall time complexity for the Fourier- and wavelet-based algorithms for ground-state generation.
Both algorithms have a classical preprocessing and quantum routine.
We begin by discussing the overall time complexity for each algorithm's quantum routine.
We then put together the overall time complexities for classical preprocessing and quantum routine to establish the overall time complexity for each of the two algorithms.

% Complexity for quantum routine
The quantum routine for each of the two state-generation algorithms comprises two subroutines:
generating~$N$ different~1DG states and performing a basis transformation to transform the 1DG states to the free-field ground state.
The basis transformation in the Fourier-based algorithm is performed by executing a quantum Hartley transform, and the basis transformation in the wavelet-based algorithm is performed by executing a quantum shear transform.

The time complexity for generating each 1DG state is logarithmic in the number modes~$N$ as per~\cref{subsubsec:1DG_complexity}, so time complexity for the first subroutine in each algorithm's quantum routine is quasilinear in~$N$.
The time complexities for performing the quantum Hartley transform and the quantum shear transform are also quasilinear in~$N$, as discussed in~\cref{subsubsec:QFFT} and~\cref{subsubsec:QST_complexity},~respectively.
Specifically, the time complexity for quantum routine of the Fourier-based algorithm~(FBA) and the wavelet-based algorithm~(WBA) is
\begin{equation}
\label{eq:qComplexity}
    T^{(\text{quantum})}_\text{FBA} \in \order{N\log_2N +N\log^3_2{(N/\sqrt{m_0\varepsilon_\text{vac}})}},
    \quad
    T^{(\text{quantum})}_\text{WBA} \in \order{Nw +N\log^3_2{(N/\sqrt{m_0\varepsilon_\text{vac}})}},
\end{equation}
respectively, where $w$ is given in~\cref{eq:bandwidth}.
Therefore, the overall time complexity for the quantum routine of each ground-state-generation algorithm is quasilinear in the number of modes~$N$.

% Overall complexity
The sum of overall time complexities for each algorithm's classical preprocessing and quantum routine yields the algorithm's overall time complexity.
By the complexity analysis in~\cref{subsubsec:FBA_classical_complexity} and~\cref{subsubsec:WBA_classical_complexity},
the overall time complexity for classical preprocessing of the Fourier- and wavelet-based algorithms is
\begin{equation}
\label{eq:cComplexity}
    T^{(\text{classical})}_\text{FBA}
    \in \order{\dbIndex^3+
    \dbIndex N\log_2\log_2(N/\sqrt{m_0\varepsilon_\text{vac}})},
    \quad
    T^{(\text{classical})}_\text{WBA}
    \in \order{\dbIndex^3+ \dbIndex^2N\log_2N},
\end{equation}
respectively.
By combining overall time complexities for each algorithm's classical preprocessing and quantum routine, we conclude that the overall time complexity for each of the two state-generation algorithms is quasilinear in the number of~modes.

%%%%%%%%%%%%%%%%%%%%%%%%%%%%%%%%%%%%%%%%%%%%%%%%%%%%%%%%%
\subsection{Lower bound for ground-state generation}
\label{subsec:lowebound}

In this subsection, we discuss a lower bound on the gate complexity for generating the free-field ground state with respect to the discretized-QFT number of modes.
We show that one cannot generate a good approximation for the ground state in a sublinear time.
In particular, we argue that any sublinear algorithm results in an exponentially bad approximation for the free-field ground state regarding the number of modes.

The free-field ground state for an infinite-mass~($m_0\to \infty$) theory is the tensor product of the all-zero state for every~mode.
The infinite-mass theory, however, is not a physical theory.
For a physical field theory with a finite but very large mass, the free-field ground state is a tensor product of one-dimensional Gaussian states with small variances; 
the variance of each Gaussian is $\sigma^2=1/m_0$, which is small due to the large mass $m_0$.
Now suppose a sublinear algorithm exists that generates an approximation for the ground state of a large-mass theory;
i.e., let us assume that the algorithm's gate complexity scales as~$N^\alpha$ for $0<\alpha<1$.
In this case, the algorithm must leave the state of $N^{1-\alpha}$ modes untouched in order to have sublinear gate complexity, meaning that the Gaussian states for the untouched modes are approximated by the all-zero state.
The fidelity between the qubit representation of a one-dimensional Gaussian state with variance $\sigma^2$, i.e., the discrete 1DG~\eqref{eq:lattice1DG} state over a lattice with spacing $\delta$ and $2^m$ points, and the all-zero state is
\begin{equation}
\label{eq:allzerofidelity}
    F_\text{1DG}:=\braket{0\cdots 0}{\G_\text{lattice}(\tilde{\sigma}, \delta, m)} = \frac{\delta}{\tilde{\mathcal{N}}}=
    \left(\sum_j \e^{-\frac{j^2}{2\tilde{\sigma}^2}}\right)^{-1/2},
\end{equation}
where $\tilde{\sigma}=\sigma/\delta$ and~$\tilde{\mathcal{N}}$ is the normalization in~\cref{eq:lattice1DG}.
For fixed $\delta$, if $\sigma^2=1/m_0\geq \delta$,
then $\tilde{\sigma}^2=\sigma^2/\delta^2\geq 1/\sigma^2=m_0$.
We therefore have
\begin{equation}
    F_\text{1DG} \leq \left(\sum_j \e^{-\frac{j^2}{2m_0}}\right)^{-1/2}
    \leq \left(1+2 \e^{-1/(2m_0)}\right)^{-1/2}
   \leq \e^{-1/(4m_0)},
\end{equation}
where the last inequality follows from $1+2\exp(-x/2)\geq \exp(x/2)$ for any~$x\in(0,1)$.
In this case, the fidelity between the ground state and the generated state falls off exponentially with respect to the number of modes that are untouched.
Specifically, the fidelity between the ground state~$\ket{\G}$ and the generated state $\ket{\tilde{\G}}$ is $\braket{\G}{\tilde{\G}} = F_\text{1DG}^{N^{1-\alpha}} \leq \exp(-N^{1-\alpha}/(4m_0))$, which falls off exponentially with respect to $N$.

For all other finite-mass field theories, including the zero-mass~$(m_0\to0)$ theory, generating the free-field ground state requires a non-trivial operation on every mode because the ground state, in this case, is a superposition of computational states with Gaussian amplitudes.
Therefore, generating a good approximation for the free-field ground state requires a number of gates that scales at least linearly with respect to the number of modes; any sublinear algorithm results in an exponentially bad approximation for the ground state.

%%%%%%%%%%%%%%%%%%%%%%%%%%%%%%%%%%%%%%%%%%%%%%%%%%%%%%%%%
\subsection{Fourier vs wavelet approach}
\label{subsec:Fouriervswavelet}

In this subsection, we compare the Fourier and wavelet approaches for ground-state generation.
We consider two cases where the wavelet approach could be advantageous over the Fourier approach:
(1)~QFTs with broken translational invariance and
(2)~generating states beyond the free-field ground state.
In~\cref{subsubsec:inhomogeneous_mass}, we consider a simple case of inhomogeneous-mass QFT and explain why the wavelet approach could be advantageous.
We then follow in~\cref{subsubsec:particlecreation} by comparing particle-state generation in both Fourier and wavelet approaches.

%%%%%%%%%%%%%%%%%%%%%%%%%%%%%%%%%%%%%%%%%%%%%%%%%%%%%%%%%
\subsubsection{Inhomogeneous-mass QFT}
\label{subsubsec:inhomogeneous_mass}

Here we perform a numerical study to justify why the wavelet-based approach could be preferred over the Fourier-based approach for QFTs with broken translational invariance.
We consider a simple case where the free-field Hamiltonian~\eqref{eq:free_Hamiltonian}
has a position-dependent mass.
In particular, we consider a point defect in the QFT where the mass is~$m_0$ plus a delta-function term.
The inhomogeneous mass can be seen as a space-dependent source term for a massive scalar QFT in one spatial dimension, a case for which estimating the vacuum-to-vacuum transition amplitude was shown to be a BQP-complete problem~\cite{JKL+18}.

We model the point defect by adding to the fixed-scale coupling matrix~\eqref{eq:fixedscaleK} a matrix containing all zeros except for one nonzero diagonal element that is significantly larger than the free-QFT mass~$m_0$.
In this case, the Fourier-based approach is not appropriate because a discrete Fourier transform cannot diagonalize the modified coupling matrix.
However, our numerical experiment demonstrates that the wavelet approach accommodates such a modification.

In our numerical experiment, we take the nonzero diagonal element to be~$100\times m_{0}$ and compute the ground-state ICM~\eqref{eq:multiscale_ICM} for a wide range of~$m_0$.
Figure~\ref{fig:modifiedICM}~(left) shows a visualization of the approximate ICM~\eqref{eq:cutoffcondition} derived from the modified coupling matrix and
\cref{fig:modifiedICM}~(right) shows the bandwidth of the approximate ICM's diagonal blocks with and without the point defect for a wide range of $m_{0}$.
As per these figures, the fingerlike sparse structure of the ground-state ICM in a multi-scale wavelet basis is not affected by the point defect.
Consequently, the wavelet-based algorithm is not affected by the point defect and successfully yields an approximation for the free-field ground-state with a quasilinear gate complexity.

\begin{figure}[H]
\centering
     \includegraphics[width=.9\textwidth]{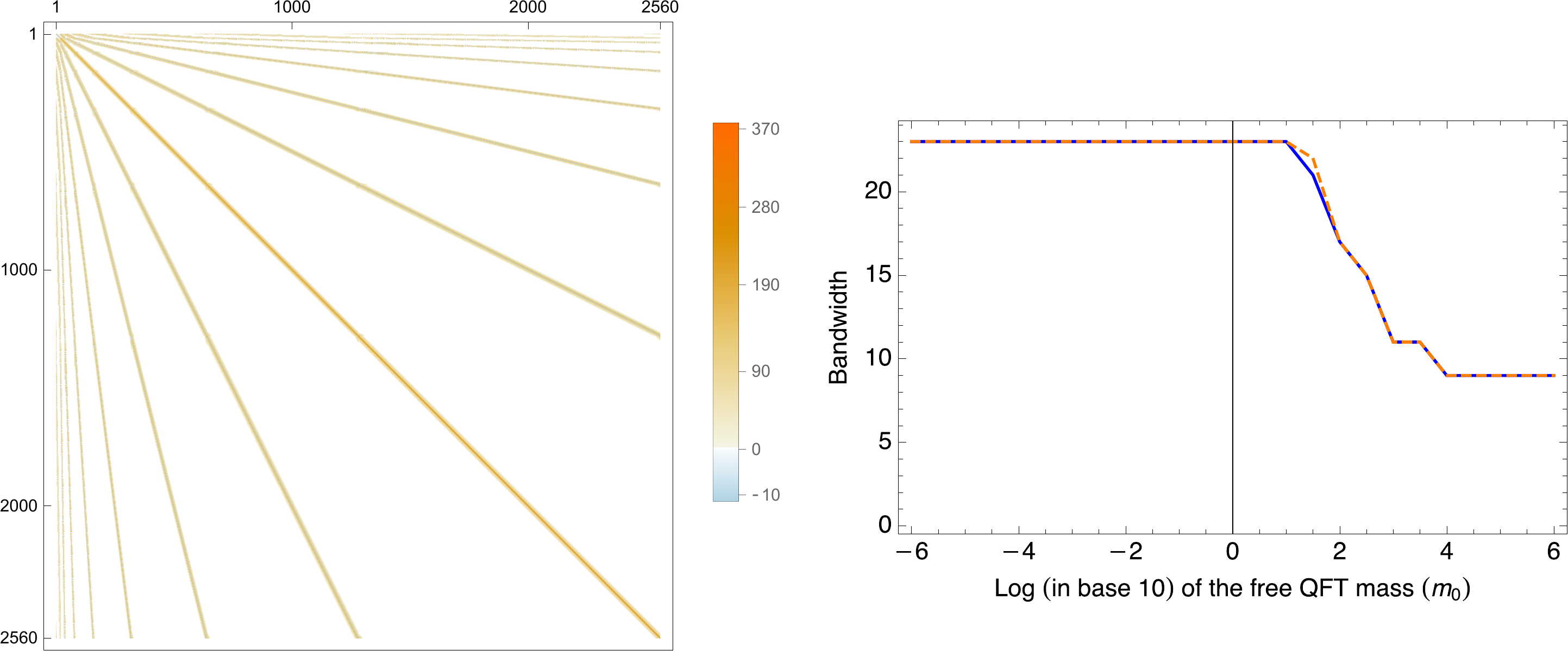}
\caption[Effect of a point defect]{
Effect of a mass defect on the ground-state ICM~\eqref{eq:multiscale_ICM} represented in a multi-scale wavelet basis.
Left: visual representation for approximation of the modified ICM~\eqref{eq:cutoffcondition} in a multi-scale wavelet basis where
elements with a magnitude less than $10^{-8}$ are replaced with exactly zero; rows and columns of the matrix are ordered as the matrix in~\cref{fig:truncICM}.
Right: bandwidth for diagonal blocks of the approximate ICM with mass defect (orange dashed line) and without mass defect (blue line) for a wide range of~mass:~$m_0 \in~[10^{-6},10^6]$.
\label{fig:modifiedICM}
}
 \end{figure}

%%%%%%%%%%%%%%%%%%%%%%%%%%%%%%%%%%%%%%%%%%%%%%%%%%%%%%%%% 
\subsubsection{Particle-state generation for the free QFT}
\label{subsubsec:particlecreation}

Here we describe a procedure used by the main server~(\cref{subsubsec:framework}) for generating a free-QFT particle state in the Fourier and wavelet approaches.
First we explain how to prepare a free wavepacket, i.e., a spatially localized free particle, in the Fourier approach.
Then we describe wavepacket preparation in the wavelet approach.
Finally, we compare time complexity for particle-state generation in the two approaches.

% Particle creation in Fourier approach
We begin with particle-state generation in the Fourier approach.
For simplicity, we consider preparing a free particle whose position-space wavefunction is the scaling function~$s_{\ell}^{(k)}\!(x)$ at scale~$k$.
Specifically, we aim to prepare the wavepacket state
\begin{equation}
\label{eq:scaleState}
\ket{s^{(k)}_\ell}:=
\hat{a}^{(k)\text\textdagger}_{\s;\ell}
\ket{\G^{(k)}_\text{scale}},    
\end{equation}
where~$\ket{\G^{(k)}_\text{scale}}$~\eqref{eq:fixedscale_groundstate} is the free-field ground state in the Fourier approach and~$\hat{a}^{(k)\text\textdagger}_{\s;\ell}$ is the creation operator constructed from the scale-field operator~$\hat{\Phi}_{\s;\ell}^{(k)}$~\eqref{eq:discrete_fields} and its conjugate
momentum~\cite{Note1}.
Following the Jordan-Lee-Preskill approach for preparing a free-particle state~\cite[p.~1027]{JLP14},
we introduce one ancillary qubit
denoted \anc, and construct the Hamiltonian
\begin{equation}
\label{eq:scaleHamiltonian}
    \hat{H}^{(k)}_{\s;\ell}:=
    \hat{a}^{(k)\text\textdagger}_{\s;\ell} \otimes (\ket{1}_\anc\!\!\bra{0})
    + \hat{a}^{(k)}_{\s;\ell} \otimes (\ket{0}_\anc\!\!\bra{1}).
\end{equation}
The time evolution generated by this Hamiltonian for time $t=\uppi/2$ is
\begin{equation}
    \e^{-\upi \hat{H}^{(k)}_{\s;\ell} \uppi/2}\ket{\G^{(k)}_\text{scale}}\ket{0}_\anc=-\upi \ket{s^{(k)}_\ell} \ket{1}_\anc,
\end{equation}
and we obtain the desired wavepacket state~\eqref{eq:scaleState}
up to a global phase with no entanglement between the wavepacket and ancilla-qubit states.
By expressing the constructed Hamiltonian~\eqref{eq:scaleHamiltonian} in terms of the scale-field
operator~$\hat{\Phi}_{\s;\ell}^{(k)}$~\eqref{eq:discrete_fields} and its conjugate momentum, we obtain a local Hamiltonian that can be simulated by a technique described in \S2.2.4 of~\cite{JLP14} with gate complexity that is logarithmic in the numbers of modes.
 
% Particle creation in wavelet approach
For particle-state generation in the wavelet approach, we consider a free particle whose position-space wavefunction is the wavelet function $w^{(r)}_\ell\!(x)$ at scale~$r$ for some integer~$r<k$.
In this case, the wavepacket state we wish to prepare is
\begin{equation}
\label{eq:waveletState}
\ket{w^{(r)}_\ell}:=\hat{a}^{(r)\text\textdagger}_{\s;\ell} \ket{\G^{(k)}_\text{wavelet}},  
\end{equation}
where~$\ket{\G^{(k)}_\text{wavelet}}$~\eqref{eq:multiscale_groundstate} is the free-field ground state in the wavelet approach and~$\hat{a}^{(r)\text\textdagger}_{\w;\,\ell}$ is the creation operator constructed from the wavelet-field operator~$\hat{\Phi}_{\w;\, \ell}^{(r)}$~\eqref{eq:discrete_fields} and its conjugate momentum at scale~$r$.
Analogous to particle creation in the Fourier approach, we construct the Hamiltonian
\begin{equation}
\label{eq:waveletHamiltonian}
    \hat{H}^{(r)}_{\w;\ell}:=
    \hat{a}^{(r)\text\textdagger}_{\w;\ell}
    \otimes (\ket{1}_\anc\!\!\bra{0})
    +\hat{a}^{(r)}_{\w;\ell}
    \otimes (\ket{0}_\anc\!\!\bra{1}),
\end{equation}
and simulate time evolution according to this Hamiltonian for time $t=\uppi/2$.
By the time evolution, we obtain the wavepacket state in~\cref{eq:waveletState}, up to a global phase, and an ancilla-qubit state that can be discarded.

% Comparison
We now compare time complexity for generating a single-particle state in the Fourier and wavelet approaches.
For comparison, we assign a unit cost to simulating time-evolution for a constant time induced by a local Hamiltonian in these two approaches.
Specifically, we assign a unit cost to simulating time evolution induced
by~$\hat{H}^{(k)}_{\s;\ell}$~\eqref{eq:scaleHamiltonian} in the Fourier approach and time evolution induced
by~$\hat{H}^{(r)}_{\w;\ell}$~\eqref{eq:waveletHamiltonian} in the wavelet approach for a constant time.

Without loss of generality, we discuss time complexity for preparing the wavepacket state in~\cref{eq:waveletState} for both approaches.
As described, preparing this state in the wavelet approach requires simulating time evolution generated 
by~$\hat{H}^{(r)}_{\w;\ell}$~\eqref{eq:waveletHamiltonian} for constant time~$t=\uppi/2$.
To prepare the same wavepacket state in the Fourier approach, we express time evolution generated by~$\hat{H}^{(r)}_{\w;\ell}$~\eqref{eq:waveletHamiltonian} in terms of time evolution generated
by~$\hat{H}^{(k)}_{\s;\ell^\prime}$~\eqref{eq:scaleHamiltonian} for various $\ell^\prime$. 
To this end, first we write the creation operation~$\hat{a}^{(r)\text\textdagger}_{\w;\, \ell}$ as a linear combination of the creation operators~$\hat{a}^{(k)\text\textdagger}_{\s;\, \ell}$.
Let $d:=k-s_0$, where $s_0\leq k$ is the scale index for the lowest scale,
and let
\begin{equation}
    \hat{\upbm{a}}_\s^{(k)\text\textdagger}:= 
    \left(
    \hat{a}_{\s;0}^{(k)\text\textdagger},
    \ldots,
    \hat{a}_{\s;2^k-1}^{(k)\text\textdagger}
    \right)^\T,
    \quad
    \hat{\upbm{a}}^{(k)\text\textdagger}:= 
    \left(
    \hat{a}_{\s;0}^{(s_0)\text\textdagger},
    \ldots,
    \hat{a}_{\s;2^{s_0}-1}^{(s_0)\text\textdagger},
    \hat{a}_{\w;0}^{(s_0)\text\textdagger},
    \ldots,
    \hat{a}_{\w;2^{s_0}-1}^{(s_0)\text\textdagger},
    \hat{a}_{\w;0}^{(s_0+1)\text\textdagger},
    \ldots,
    \hat{a}_{\w;2^k-1}^{(k-1)\text\textdagger},
    \right)^\T,
\end{equation}
be the vector of creation operators in the fixed- and multi-scale wavelet bases, respectively.
Then
\begin{equation}
\label{eq:creationtransform}
    \hat{\upbm{a}}^{(k)\text\textdagger}= \upbm{W}_d^{(k)}\hat{\upbm{a}}^{(k)\text\textdagger}_\s,
\end{equation}
where $\upbm{W}_d^{(k)}$ is the $d$-level wavelet-transform matrix at scale~$k$;
see~\cref{appx:wavelet_transform}.
This equation yields
\begin{equation}
\label{eq:waveletcreation}
    \hat{a}^{(r)\text\textdagger}_{\w;\ell} =
    \sum_{\ell^\prime} W^{(k)}_{d;\ell\ell^\prime}
    \hat{a}^{(k)\text\textdagger}_{\s;\ell^\prime},
\end{equation}
for any $s_0\leq r<k$.

The number of nonzero~(NNZ) coefficients in this summation depends on two parameters:
the wavelet index~\dbIndex\ and the difference~$(k-r)$ between the finest scale~$k$ and the scale~$r$.
Specifically, the NNZ coefficients is~$2^{k-r}\dbIndex$ for any~$r<k$.
This relation follows from the recursive relation for the scaling and wavelets function at different scales described by Eqs.~\eqref{eq:scaling_function} and~\eqref{eq:wavelet_function}, respectively.
By~\cref{eq:wavelet_function}, each wavelet function at a particular scale~$r$ is a linear combination of~$2\dbIndex$ scaling functions at one higher scale~$r+1$.
Similarly, by~\cref{eq:scaling_function}, each scaling function at a given scale can be written as a linear combination of~$2\dbIndex$ scaling functions at one higher scale.
Therefore, a wavelet function at scale~$r$ is a linear combination of~$2^{k-r}\dbIndex$ scaling functions at scale~$k>r$.
See~\cref{fig:waveletmtx} for a visual insight into the relationship between the NNZ coefficients, $(k-r)$ and~\dbIndex.

\begin{figure}[htb]
\centering
\includegraphics[width=.97\textwidth]{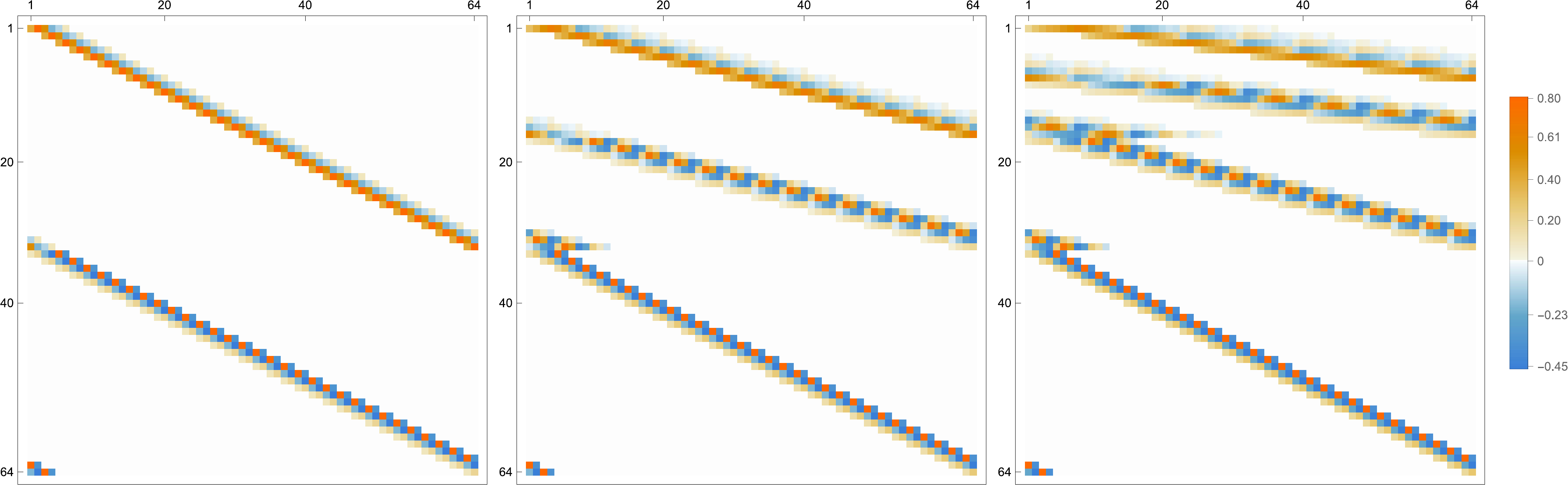}
\caption[A visualization for the wavelet-transform matrix]{
From left to right: visualization of the $d$-level wavelet-transform matrix~(WTM) with level~$d\in\{1,2,3\}$ at scale~$k=6$ for the Daubechies wavelet with index $\dbIndex=3$;
size of matrix is~$2^k\times 2^k$.
The number of nonzero~(NNZ) elements in each row of the 1-level WTM is~$2\dbIndex$.
The NNZ elements in those rows of the 2-level WTM that pertain to the first level, i.e., the bottom half rows, is~$2\dbIndex$,
and the NNZ elements in the rows pertaining to the second level, the top half rows, is~$4\dbIndex$.
For the 3-level wavelet transform, the NNZ elements in those rows pertaining to the third level is two times the NNZ elements in the rows pertaining to the second level, which itself is two times the NNZ on the rows pertaining to the first level.
The NNZ elements increases by a factor of two when we increase the level of wavelet transform by one. 
 \label{fig:waveletmtx}
  }
 \end{figure}

By the combination of
Eqs.~\eqref{eq:waveletcreation}, \eqref{eq:waveletHamiltonian} and \eqref{eq:scaleHamiltonian}, we have
\begin{equation}
\label{eq:HwDecomp}
    \hat{H}^{(r)}_{\w;\,\ell} =
    \sum_{\ell^\prime} W^{(k)}_{d;\ell \ell^\prime}
    \hat{H}^{(k)}_{\s;\ell^\prime},
\end{equation}
hence the time evolution induced by this Hamiltonian for time $t=\uppi/2$ in the Fourier approach yields
\begin{equation}
    \e^{-\upi \hat{H}^{(r)}_{\w;\ell} \uppi/2}
    \ket{\G^{(k)}_\text{scale}}\ket{0}_\anc =
    -\upi \ket{w^{(r)}_\ell} \ket{1}_\anc,
\end{equation}
and, therefore, we obtain the wavepacket state in~\cref{eq:waveletState} up to a global phase.
Because of the bosonic commutation relations between the bosonic creation and annihilation operators in~\cref{eq:scaleHamiltonian}, any term in the right-hand-side of~\cref{eq:HwDecomp} commutes with other terms.
Consequently, we have the decomposition
\begin{equation}
    \e^{-\upi \hat{H}^{(r)}_{\w;\, \ell} \uppi/2}
    = \prod_{\ell^\prime} \e^{-\upi \hat{H}^{(k)}_{\s;\, \ell^\prime} W^{(k)}_{d;\, \ell \ell^\prime}\uppi/2},
\end{equation}
for the time evolution induced by the wavelet Hamiltonian~\eqref{eq:HwDecomp}.
By this decomposition, simulating time evolution induced by the wavelet Hamiltonian for constant time~$\uppi/2$ in the Fourier approach is achieved by simulating time evolution induced by the scale Hamiltonian~\eqref{eq:scaleHamiltonian} with~$2^{k-r}\dbIndex$ different~$\ell^\prime$ for constant time $t_{\ell^\prime}:= W^{(k)}_{d;\, \ell \ell^\prime}\uppi/2$.
Therefore, generating a single-particle state at scale~$r$ in the Fourier approach is~$\Theta(2^{k-r}\dbIndex)$ times more expansive than generating the same state in the~wavelet~approach.

%%%%%%%%%%%%%%%%%%%%%%%%%%%%%%%%%%%%%%%%%%%%%%%%%%%%%%%%%
%%%%%% Discussion %%%%%%%%%%%%%%%%%%%%%%%%%%%%%%%%%%%%%%%
%%%%%%%%%%%%%%%%%%%%%%%%%%%%%%%%%%%%%%%%%%%%%%%%%%%%%%%%%
\section{Discussion}
\label{sec:discussion}

We have established two quasilinear quantum algorithms, one Fourier-based and the other wavelet-based, to generate an approximation for the ground state of a massive scalar bosonic free QFT.
Specifically, each of the two algorithms' time complexity is quasilinear with respect to the discretized-QFT number of modes.
Our algorithms deliver a super-quadratic speedup over the state-of-the-art quantum algorithm for ground-state generation and are optimal up to polylogarithmic factors.
The Fourier-based algorithm is limited to translationally invariant QFTs.
By numerical simulations, we have shown that the wavelet-based algorithm successfully yields the ground state for a QFT with a broken translational invariance.

We have also developed two quantum algorithms for generating one-dimensional~(1DG) Gaussian states.
Our first algorithm is based on the Kitaev-Webb method~\cite{KW09} for preparing a 1DG state, which itself is an application of the standard state-preparation method by Zalka~\cite{Zal96}, Grover and Rudolph~\cite{GR02}.
The Kitaev-Webb method, however, is restricted to 1DG states that possess an extremely large standard deviation, whereas our algorithm generates a 1DG state with any standard deviation.
Our second quantum algorithm for generating a 1DG state is based on inequality testing~\cite{SLSB19} that mitigates the number of arithmetic operations required by the standard state-preparation method.

Methodologically, in the Fourier-based algorithm, we discretize the continuum free QFT in a fixed-scale basis.
The ground-state ICM for the discretized QFT has a circulant structure, and we utilize this structure to establish a quasilinear quantum algorithm for ground-state generation.
In the wavelet-based algorithm, we discretize the continuum free QFT in a multi-scale wavelet basis.
In this case, the ground-state ICM is a quasi-sparse matrix.
Specifically, most elements of this matrix are nearly zero, and we replace these values with exactly zero.
This replacement enables a fingerlike sparse structure for the ground-state ICM that we exploit to achieve a quasilinear time complexity for the wavelet-based algorithm.

We went beyond ground-state generation and constructed procedures for preparing free-field wavepackets in the Fourier- and wavelet-based approaches.
We showed that, unlike the Fourier approach, the wavelet approach enables preparing particle states at different energy scales without an additional cost required for the Fourier approach.
Specifically, we showed that preparing a free-field single-particle state at scale with index~$r$ is~$\Theta(2^{k-r}\dbIndex)$ times more expensive than preparing the same state in the wavelet approach, where~$\dbIndex$ is the wavelet index and $k>r$ is the finest scale's index.
The wavelet approach's cheaper cost to preparing states beyond the free-field ground state suggests that this approach is advantageous over the Fourier approach in other aspects of simulating a QFT.
Moreover, as shown by the numerical simulations, the wavelet-based approach is applicable to field theories with broken translational invariance due to inhomogeneous mass, suggesting the wavelet-based approach allows simulating more general quantum field theories~\cite{BJV11}.

Our Fourier-based algorithm's key point is utilizing the circulant structure of the ground state's ICM due to the translational invariance of the free QFT.
We note that the ground-state ICM in the Jordan-Lee-Preskill approach~\cite{JLP12,JLP14} has the same structure as the ICM in our Fourier-based approach.
Hence our Fourier-based algorithm can be used for generating the ground state of the lattice QFT in the Jordan-Lee-Preskill approach and accelerate their algorithm for ground-state generation.
However, in contrast to their approach, the coupling-matrix elements in our Fourier approach are exact. 
The discretization error, due to approximating the derivative operator in the free-field Hamiltonian by a discretized derivative operator, results in a non-exact coupling matrix in the Jordan-Lee-Preskill approach~\cite{JLP12,JLP14}.
Consequently, the Fourier approach could be preferred over the lattice approach for generating the free-field ground state.

In developing our ground-state-generation algorithms, we only cared about producing quasilinear algorithms and opted to focus on their readability rather than optimizing their performance.
Therefore, our algorithms' time complexity could be improved, but any improvement will not change the quasilinear scaling of the algorithms' complexities.

Our Fourier- and wavelet-based algorithms have classical preprocessing and quantum routine.
The classical preprocessing of each algorithm produces a certain amount of classical information required for executing the quantum routine of the algorithm.
We analyzed not only the quantum complexity
but also the classical complexity of our state-generation algorithms in order to
ensure that the resulting procedures are indeed quasilinear in
the number of modes of the discretized QFT.
We established the classical time complexities
\begin{equation}
    T^{(\text{C})}_\text{FBA}
    \in \order{\dbIndex^3+
    \dbIndex N\log_2\log_2{\left(\frac{N}{\sqrt{m_0\varepsilon_\text{vac}}}\right)}},
    \quad
    T^{(\text{C})}_\text{WBA}
    \in \order{\dbIndex^3+ \dbIndex^2N\log_2N},
\end{equation}
for classical preprocessing of the Fourier- and wavelet-based algorithms, respectively, and the quantum time complexities
\begin{equation}
    T^{(\text{Q})}_\text{FBA} \in \order{N\log_2N +N\log^3_2{\left(\frac{N}{\sqrt{m_0\varepsilon_\text{vac}}}\right)}},
    \quad
    T^{(\text{Q})}_\text{WBA} \in \order{\left(\frac{N\dbIndex}{m_0}\right)\log^2_2{
    \left(\frac{N\dbIndex}{m_0\varepsilon_\text{vac}}\right)}
    +N\log^3_2{\left(\frac{N}{\sqrt{m_0\varepsilon_\text{vac}}}\right)}},
\end{equation}
for their quantum routine;
all parameters are specified in~\cref{table:inputsgroundstategen}.

In contrast to the usual approach of using gate complexity, i.e., the number of low-level quantum operations in an algorithm, as a metric to analyze time complexity for a quantum algorithm, 
we have analyzed our quantum algorithm's time complexity with respect to high-level operations.
The high-level quantum operations in our complexity analysis are similar to the high-level operations in the classical random-access machine model.
In particular, we have assigned a unit cost to basic arithmetic operations such as multiplication and addition on a quantum computer.
By the high-level operations, we avoid the implementation details of these operations in the compilation step.
Nevertheless, our algorithm's gate complexity remains quasilinear with respect to the number of modes, as the high-level operations are executed on quantum registers with size logarithmic in the number of modes.
However, the power of logarithmic factors in our algorithms' gate complexity depends on implementations of the high-level operations, specifically quantum multiplication, as other basic arithmetic operations are cheaper to implement than multiplication.
Whether to use schoolbook multiplication \'a la H\"aner, Roettler, and Svore~\cite{HRS18} or Karatsuba multiplication \'a la Gidney~\cite{Gid19} or even as-yet-undeveloped quantum multiplication algorithms based on asymptotically efficient classical multiplication algorithms such as the Sch\"onhage-Strassen algorithm~\cite[Sec.~4.3]{Knu97}, our algorithms' gate complexity stay quasilinear in the number of discretized-QFT modes.

%%%%%%%%%%%%%%%%%%%%%%%%%%%%%%%%%%%%%%%%%%%%%%%%%%%%%%%%%
%%%%%% conclusions %%%%%%%%%%%%%%%%%%%%%%%%%%%%%%%%%%%%%%
%%%%%%%%%%%%%%%%%%%%%%%%%%%%%%%%%%%%%%%%%%%%%%%%%%%%%%%%%
\section{Conclusions}
\label{sec:conclusions}

Free-field ground-state generation is a bottleneck for the prior approaches to simulating a massive scalar bosonic QFT on a quantum computer.
In this paper, we have established two quantum algorithms for generating an approximation for the free-field ground state with a quasilinear gate complexity in the discretized-QFT number of modes.
Our algorithms provide a super-quadratic speedup over the prior approaches and overcome the ground-state-generation bottleneck in simulating a massive scalar bosonic QFT.
We have shown that our ground-state-generation algorithms are optimal up to polylogarithmic factors.
In particular, we have proved that any state-generation algorithm with a sublinear time complexity will result in an exponentially bad approximation, with respect to the number of modes, for the free-field ground state.

We have compared the Fourier- and wavelet-based algorithms and shown that the wavelet-based algorithm is advantageous over the alternative Fourier-based algorithm for two cases where we go beyond ground-state generation and translationally invariant QFTs.
Specifically, for beyond ground-state generation, we have shown that the wavelet-based algorithm enables generating particle states at different energy scales directly, whereas the Fourier-based algorithm only allows direct preparation for particle states at a fixed energy scale.
Preparing a particle state at different scales by the Fourier-based algorithm requires further transformations that add to the cost of initial-state generation for simulating the QFT.
The Fourier-based algorithm is limited to translationally invariant QFTs.
We have shown by numerical simulation that our wavelet-based algorithm is applicable to field theories with broken translational invariance due to inhomogeneous mass, suggesting the wavelet-based approach allows simulating more general quantum field theories.

We have also developed two quantum algorithms for generating a one-dimensional Gaussian state, which is required for preparing the free-field ground state---a multidimensional Gaussian state.
Our first algorithm for 1DG-state generation is based on the standard state-preparation method~\cite{Zal96,GR02}, and the second algorithm is based on inequality testing~\cite{SLSB19} that mitigates the number of arithmetic operations required in the standard method for generating a 1DG state.
Our inequality-testing-based method is more practical than the Zalka-Grover-Rudolph method~\cite{Zal96,GR02}, i.e., the standard state-preparation method, and can be used broadly in the state-preparation subroutine of quantum-simulation algorithms.

The key point of our ground-state-generation algorithms is to utilize the circulant structure and the fingerlike sparse structure for a differential operator represented in a fixed- and multi-scale wavelet basis, respectively.
Our techniques for exploiting these structures could be used more broadly in the quantum simulation of continuous classical or quantum systems whose Hamiltonian involves differential operators, including fermionic QFT~\cite{JLP14-2}, as well as quantum-chemistry simulations~\cite{BGM+19}.
Our methods, particularly the wavelet-based method, could be replaced with finite-difference or finite-element methods to improve quantum algorithms for partial differential equations~\cite{CLO20,CJO19,MP16}.

We have focused on ground-state generation here, but our methods could be employed to improve other aspects of simulating a QFT, namely time evolution and measurement.
In particular, akin to particle-state generation at different energy scales, the wavelet approach could be advantageous over the Fourier approach for measuring mean momentum of a particle state at the final step of a QFT simulation. 

We have particularly used Daubechies wavelets in developing our state-generation algorithms.
These algorithms, however, work for any compactly supported wavelets that are differentiable.
A future direction is to see how other types of compactly supported wavelets could improve the state-generation algorithms.
In particular, least-asymmetric wavelets~\cite[pp.~254--257]{Dau92} also known as symlets~\cite[p.27]{AMC21}, a variant of Daubechies wavelets with nearly symmetrical basis functions, could result in a more sparse ground-state ICM for a given threshold value~\eqref{eq:cutoffcondition} and improve the wavelet-based algorithm.

%%%%%%%%%%%%%%%%%%%%%%%%%%%%%%%%%%%%%%%%%%%%%%%%%%%%%%%%%
%%%%%% Acknowledgements %%%%%%%%%%%%%%%%%%%%%%%%%%%%%%%%%
%%%%%%%%%%%%%%%%%%%%%%%%%%%%%%%%%%%%%%%%%%%%%%%%%%%%%%%%%
\section*{Acknowledgements}
This project is supported by the Government of Alberta and 
by the Natural Sciences and Engineering Research Council of Canada (NSERC).
MB and BCS acknowledge the traditional owners of the land on which some of this work was undertaken at the University of Calgary: the Treaty 7 First Nations.
YRS, DWB and GKB acknowledge the Wallamattagal people of the Dharug nation, whose cultures and customs have nurtured, and continue to nurture, the land on which some of this work was undertaken: Macquarie University.
YRS acknowledges the Gadigal and Guring-gai people of the Eora Nation upon whose ancestral lands the University of Technology Sydney now stands. 
MB thanks Mehdi Ahmadi for many helpful discussions.
YRS is supported by Australian Research Council Grant No: DP200100950.
DWB worked on this project under a sponsored research agreement with Google Quantum AI. DWB is also supported by Australian Research Council Discovery Projects DP190102633 and DP210101367.
GKB acknowledges support from the ARC from grant~DP200102152.

%%%%%%%%%%%%%%%%%%%%%%%%%%%%%%%%%%%%%%%%%%%%%%%%%%%%%%%%%
%%%%%% Bibliography %%%%%%%%%%%%%%%%%%%%%%%%%%%%%%%%%%%%%
%%%%%%%%%%%%%%%%%%%%%%%%%%%%%%%%%%%%%%%%%%%%%%%%%%%%%%%%%

\bibliography{references}

%%%%%%%%%%%%%%%%%%%%%%%%%%%%%%%%%%%%%%%%%%%%%%%%%%%%%%%%%
%%%%%% Appendix %%%%%%%%%%%%%%%%%%%%%%%%%%%%%%%%%%%%%%%%%
%%%%%%%%%%%%%%%%%%%%%%%%%%%%%%%%%%%%%%%%%%%%%%%%%%%%%%%%%
\appendix
\addtocontents{toc}{\protect\setcounter{tocdepth}{0}}

%%%%%%%%%%%%%%%%%%%%%%%%%%%%%%%%%%%%%%%%%%%%%%%%%%%%%%%%%
\section{Wavelet transform}
\label{appx:wavelet_transform}

In this appendix, we review a family of basis transforms called `the' discrete wavelet transform.
The specific member of that family of transforms will be clear from the context.
Here we focus only on the one-dimensional wavelet transforms.

Two parameters specify a one-dimensional wavelet transform.
The first parameter is a choice~$k$ of scale for~\sqintR.
The second parameter is a choice of decomposition level~$d\geq 1$.
The $d$-level wavelet transform at scale~$k$ is recursively defined as follows.
The $1$-level wavelet transform at any scale~$k$ is defined to be the canonical isomorphism
\begin{equation}
    \upbm{W}^{(k)} :
    \mathcal{S}_k \to \mathcal{S}_{k-1} \oplus \mathcal{W}_{k-1},
    \quad \upbm{W}_1^{(k)} := \upbm{W}^{(k)},
\end{equation}
where $\mathcal{S}_k$ and $\mathcal{W}_k$ are the scale and wavelet subspaces, respectively, defined in~\cref{subsec:wavelet_bases}.
The $d$-level wavelet transform, for $d>1$, is defined as
\begin{equation}
    \upbm{W}_d^{(k)} :=
    \left( \upbm{W}^{(k-d+1)} \oplus \id_{\mathcal{S}_{k-d+1}^\perp} \right)
    \cdot \upbm{W}_{d-1}^{(k)},
\end{equation}
where $\mathcal{S}_{k-d}^\perp:= \mathcal{W}_{k-d+1} \oplus 
\mathcal{W}_{k-d+2} \oplus \cdots \oplus \mathcal{W}_{k-1}$ is the orthogonal
complement of $\mathcal{S}_{k-d}$ considered as a subspace of $\mathcal{S}_k$.
Thus $\upbm{W}_d^{(k)}$ executes the sequence of transformations
\begin{equation}
\begin{split}
    \mathcal{S}_k &
    \xrightarrow{\upbm{W}^{(k)}} 
    \mathcal{S}_{k-1} \oplus \mathcal{W}_{k-1} \\ &
    \xrightarrow{\upbm{W}^{(k-1)} \oplus \id_{\mathcal{S}_{k-1}^\perp}}
    \mathcal{S}_{k-2} \oplus \mathcal{W}_{k-2} \oplus \mathcal{W}_{k-1} \\ &
    \xrightarrow{\upbm{W}^{(k-2)} \oplus \id_{\mathcal{S}_{k-2}^\perp}}
    \mathcal{S}_{k-3} \oplus \mathcal{W}_{k-3} \oplus \mathcal{W}_{k-2}
    \oplus \mathcal{W}_{k-1} \\ & \qquad \qquad \vdots \\ &
    \xrightarrow{\upbm{W}^{(k-d+1)} \oplus \id_{\mathcal{S}_{k-d+1}^\perp}}
    \mathcal{S}_{k-d} \oplus \mathcal{W}_{k-d} \oplus \cdots
    \oplus \mathcal{W}_{k-1}.
\end{split}
\end{equation}
The action of the canonical isomorphism $\upbm{W}^{(k)}$ can be expressed in terms of two infinite matrices
\begin{equation}
    \upbm{H} = 
    \begin{bmatrix}
               & \vdots & \vdots & \vdots & \vdots & \vdots &        \\
        \cdots & h_{ 0} & h_{ 1} & h_{ 2} & h_{ 3} & h_{ 4} & \cdots \\
        \cdots & h_{-2} & h_{-1} & h_{ 0} & h_{ 1} & h_{ 2} & \cdots \\
        \cdots & h_{-4} & h_{-3} & h_{-2} & h_{-1} & h_{ 0} & \cdots \\
               & \vdots & \vdots & \vdots & \vdots & \vdots &        \\
    \end{bmatrix},
    \quad
    \upbm{G} = 
    \begin{bmatrix}
               & \vdots & \vdots & \vdots & \vdots & \vdots &        \\
        \cdots & g_{ 0} & g_{ 1} & g_{ 2} & g_{ 3} & g_{ 4} & \cdots \\
        \cdots & g_{-2} & g_{-1} & g_{ 0} & g_{ 1} & g_{ 2} & \cdots \\
        \cdots & g_{-4} & g_{-3} & g_{-2} & g_{-1} & g_{ 0} & \cdots \\
               & \vdots & \vdots & \vdots & \vdots & \vdots &        \\
    \end{bmatrix},
\end{equation}
called the low-pass and high-pass filter matrix, respectively.
These infinite matrices are to be multiplied to vectors $\upbm{c}^{(k)} \in \mathcal{S}_k$ expressed in the basis $\left\{s_\ell^{(k)}\right\}$,
meaning that these vectors take the form
\begin{equation}
    \upbm{c}^{(k)} = \begin{bmatrix}
    \vdots \\ c_{-1}^{(k)} \\ c_0^{(k)} \\ c_1^{(k)} \\ \vdots
    \end{bmatrix},\quad
    c_\ell^{(k)} := \braket{s_\ell^{(k)}}{f}.
\end{equation}
The infinite matrix $\upbm{H}$ is then defined so that
$\upbm{H} \cdot \upbm{c}^{(k)} = \upbm{c}^{(k-1)}$.
Note that the form of $\upbm{H}$ and $\upbm{G}$ do not depend on~$k$.
Also note that $\upbm{H}$ and $\upbm{G}$ are both row-sparse operators if the wavelet is compactly supported, as only a finite number of~$h$ and~$g$ coefficients are nonzero.
With these definitions, we have
$ \upbm{W}^{(k)} \cdot \upbm{c}^{(k)} =
 \left( \upbm{H} \cdot \upbm{c}^{(k)} \right) \oplus
 \left( \upbm{G} \cdot \upbm{c}^{(k)} \right)$.

We do not work with an infinite matrix representation in this paper and instead restrict attention to finite-sized systems.
Specifically, we restrict attention to subset $\sqintS$ of $\sqintR$ that has support only on unit interval
$\mathds{S}:= \{x\in\reals \vert 0\leq x\leq 1 \}$.
In this restricted space, the scale coefficient
$c_\ell^{(k)} := \braket{s_\ell^{(k)}}{f}$
for the function $f \in \sqintS$
is guaranteed to be equal to zero for all but a finite number of
values of $\ell$.
Specifically, the fact that
\begin{equation}
  \operatorname{supp} \left( s_\ell^{(k)} \right) \subseteq
  \{ x \in \reals \vert 0 < 2^k x - \ell < 2\dbIndex - 1 \},
\end{equation}
where $\operatorname{supp}(\bullet)$ refers
to support of a function, implies that
$c_\ell^{(k)} = 0$ if $\ell < 1 - 2\dbIndex$ or $\ell \geq 2^k$.
That is to say, only $2^k + 2\dbIndex + 1$ scale coefficients
$c_\ell^{(k)}$ could be nonzero.
We thus treat the infinite-dimensional vector~$\upbm{c}^{(k)}$ as though it were finite-dimensional with dimension $2^k + 2\dbIndex + 1$.
Similarly, we treat $\upbm{H}$ and~$\upbm{G}$ as though they were
$(2^{k-1} + 2\dbIndex + 1) \times (2^k + 2\dbIndex + 1)$ matrices
and hence $\upbm{W}^{(k)}$ as though it were a 
$(2^k + 4\dbIndex + 2) \times (2^k + 2\dbIndex + 1)$ matrix.
Applying $\upbm{W}^{(k)}$ to a vector is described as
calculating the `discrete wavelet transform' of the vector.

The treatment described here leaves~$\upbm{W}^{(k)}$ as a non-square matrix, which cannot be
an automorphism of $\mathcal{S}_k \subset \sqintS$.
We deal with this issue in two ways.
In the first way, which we refer to as the `open boundaries' case,
we discard the negative index values $\ell < 0$.
Although this leaves us with a square ($2^k \times 2^k$) matrix, the result is not invertible and hence is not an automorphism of
$\mathcal{S}_k \subset \sqintS$ as desired.
Despite its weaknesses, this approach is good enough for numerical studies.
The second way we deal with the non-square $\upbm{W}^{(k)}$ is what we call the `(anti-)periodic boundaries' case.
In this approach, we treat the signal~$f$ as though it depicts a periodic (antiperiodic) system, so the problematic negative index values $\ell < 0$ would be identical to (the negative of) the index values $2^k - \ell$.
In this case, we apply the same condition to the derivative overlap coefficients~\eqref{eq:derivative_overlaps}.
By putting these two approaches together, the matrix $\upbm{W}^{(k)}$ for the db$3$ wavelets becomes the $2^k \times 2^k$ matrix
\begin{equation}
\label{eq:WTM}
   \upbm{W}^{(k)} = 
   \begin{bmatrix}
      h_0  &   h_1  &   h_2  &   h_3  &   h_4  &   h_5  &    0   &    0   & \cdots \\
       0   &    0   &   h_0  &   h_1  &   h_2  &   h_3  &   h_4  &   h_5  & \cdots \\
       0   &    0   &    0   &    0   &   h_0  &   h_1  &   h_2  &   h_3  & \cdots \\
    \vdots & \vdots & \vdots & \vdots & \vdots & \vdots & \vdots & \vdots & \ddots \\
    b h_4  & b h_5  &    0   &    0   &    0   &    0   &    0   &    0   & \cdots \\
    b h_2  & b h_3  & b h_4  & b h_5  &    0   &    0   &    0   &    0   & \cdots \\
      g_0  &   g_1  &   g_2  &   g_3  &   g_4  &   g_5  &    0   &    0   & \cdots \\
       0   &    0   &   g_0  &   g_1  &   g_2  &   g_3  &   g_4  &   g_5  & \cdots \\
       0   &    0   &    0   &    0   &   g_0  &   g_1  &   g_2  &   g_3  & \cdots \\
    \vdots & \vdots & \vdots & \vdots & \vdots & \vdots & \vdots & \vdots & \ddots \\
    b g_4  & b g_5  &    0   &    0   &    0   &    0   &    0   &    0   & \cdots \\
    b g_2  & b g_3  & b g_4  & b g_5  &    0   &    0   &    0   &    0   & \cdots \\
   \end{bmatrix}
   = \begin{bmatrix}
      \upbm{H}\\
      \upbm{G}
   \end{bmatrix},
\end{equation}
where $b = 0, +1, -1$ depending on whether the boundaries are open, periodic or antiperiodic.

%%%%%%%%%%%%%%%%%%%%%%%%%%%%%%%%%%%%%%%%%%%%%%%%%%%%%%%%%
%%%%%% Low-pass filter %%%%%%%%%%%%%%%%%%%%%%%%%%%%%%%%%%
%%%%%%%%%%%%%%%%%%%%%%%%%%%%%%%%%%%%%%%%%%%%%%%%%%%%%%%%%
\section{Low-pass filter for Daubechies wavelets}
\label{appx:lowpass_filter}

This appendix describes a method for computing the low-pass filter coefficients~$h_\ell$~\eqref{eq:scaling_function} for Daubechies wavelets.
We begin with describing the method and then construct a classical algorithm for computing the low-pass filter coefficients using the described method.
Finally, we discuss the algorithm's time complexity with respect to the classical primitive operations described in~\cref{subsubsec:complexity_measure}.

Several methods are used to compute the numerical values of the low-pass filters~$h_\ell$~\eqref{eq:scaling_function}.
A direct method is to solve the system of nonlinear equations~\cite[p.~3]{BP13}
\begin{equation}
    \sum_\ell h_\ell=\sqrt{2},
    \quad
    \sum_\ell h_\ell h_{\ell-2\ell'}=\updelta_{0,\ell'},
    \quad
    \sum_\ell (-1)^\ell \ell^{\ell'} h_{2\dbIndex-1-\ell}=0
    \quad
    \forall \ell'<\dbIndex,
\end{equation}
which are derived from various properties of the Daubechies wavelets, such as orthonormality.
This nonlinear system can be solved analytically for~$\dbIndex\leq 3$ and numerically otherwise.
However, the numerical methods for solving the nonlinear system become increasingly inefficient as~\dbIndex\ increases.
We use an alternative method proposed by Daubechies to compute the numerical values of the filter coefficients~\cite{Dau92};
also see~\cite[p.~81]{Cha97}.
The filter coefficients in Daubechies' method are obtained by first computing the~$2\dbIndex-2$ roots of the polynomial
\begin{equation}
\label{eq:fzPy}
   f(z):= z^{\dbIndex-1}P\left(\frac1{2}-\frac1{4z}-\frac{z}{4},\, \dbIndex\right),\quad
   P(y, \dbIndex):=\sum_{\ell=0}^{\dbIndex-1}\binom{\dbIndex-1+\ell}{\ell}y^\ell,
\end{equation}
where $P(y, \dbIndex)$ is a polynomial of degree $\dbIndex-1$.
Next those roots $z_\ell$ with $\abs*{z_\ell} < 1$ are selected.
The filter coefficients are then obtained by identifying the coefficients $H_\ell$ in
\begin{equation}
\label{eq:polyCoeff}
\left(1+z\right)^\dbIndex \prod_{\ell=0}^{\dbIndex-1}(z-z_\ell)
= \sum_{\ell=0}^{2\dbIndex-1} H_\ell z^{2\dbIndex-1-l},
\end{equation}
for selected roots and normalizing the coefficients as~$h_\ell=H_\ell/\sqrt{\upbm{H}\cdot \upbm{H}}$,
where~$\upbm{H}$ is a vector with elements~$H_\ell$.

The pseudocode in Algorithm~\ref{alg:lowPass} describes a formal algorithm for computing the low-pass filters using Daubechies' method.
This algorithm is used in classical preprocessing of the wavelet-based algorithm~(\cref{subsec:highlevel_description}), specifically as a subroutine in Algorithm~\ref{alg:invCovCircRows}, which requires the low-pass filters to compute unique elements of the ground-state ICM in a multi-scale wavelet~basis.

\begin{algorithm}[H]
  \caption{Classical algorithm for computing the low-pass filter for Daubechies wavelets}
  \label{alg:lowPass}
  \begin{algorithmic}[1]
  \Require{
  \Statex $\dbIndex \in \integers^{+}$
  \Comment{Daubechies wavelet index}
  \Statex $p \in \integers^{+}$
  \Comment{working precision}
    }
  \Ensure
  \Statex $\upbm{h} \in \reals^{2\dbIndex}$
  \Comment{low-pass filter~\eqref{eq:scaling_function} for the Daubechies wavelet with index~$\dbIndex$}
  \Function{poly}{$y,\dbIndex$}
  \Comment{defines the polynomial~$P(y,\dbIndex)$ in~\cref{eq:fzPy}}
  \State \Return $  \displaystyle \sum_{\ell=0}^{\dbIndex-1}\binom{\dbIndex-1+l}{\ell}y^\ell$
  \EndFunction
\Function{lowPassFilter}{$\dbIndex, p$}
    \State \label{line:rootf}
    $ \reals^{2\dbIndex-2} \ni \upbm{r}
    \gets \textsc{roots}\left(z^{\dbIndex-1}\textsc{poly}(\frac1{2}-\frac1{4z}-\frac{z}{4}, \dbIndex),\, p\right)$
    \Comment{computes roots of $f(z)$ in~\cref{eq:fzPy} with precision~$p$}
    \State $\upbm{z} \gets \varnothing$
    \Comment{initializes an empty list}
    \For{$\ell \gets 0$ to $2\dbIndex-3$}
    \Comment{appends roots~$r_\ell$ of~$f(z)$ with~$\abs{r_\ell} <1$ to~$\upbm{z}$}
        \If{$\abs{r_\ell} <1$}
            \State append $r_\ell$ to $\upbm{z}$
        \EndIf
    \EndFor
    \For{$\ell \gets 0$ to $\dbIndex-1$}
    \Comment{appends roots of $(1+z)^\dbIndex$ to $\upbm{z}$; see~\cref{eq:polyCoeff}}
        \State append $-1$ to $\upbm{z}$
    \EndFor
    \State $\reals^{2\dbIndex} \ni \upbm{H} \gets \textsc{polyCoeff}(\upbm{z})$
    \Comment{computes coefficients $H_\ell$~\eqref{eq:polyCoeff} of a polynomial with roots in $\upbm{z}$}
    \State $\upbm{h} \gets \upbm{H}/\sqrt{\upbm{H}\cdot\upbm{H}}$
\EndFunction
\State \Return $\upbm{h}$
\end{algorithmic}
\end{algorithm}

The time complexity of this algorithm is dominated by the time complexity for computing roots of~$f(z)$ in~\cref{eq:fzPy} because the root-finding subroutine in line~(\ref{line:rootf}) is the computationally expensive part of the algorithm.
The coefficients of~$f(z)$, a polynomial of degree $2\dbIndex-2$, have a magnitude less than~$2^{3\dbIndex}$.
By these properties, the arithmetic complexity for finding roots of~$f(z)$ with precision~$p$ is~$\order{\dbIndex \log^5 \dbIndex \log(3\dbIndex+p)}$
using a known root-finding algorithm~\cite{NR96}.
The classical primitive operations in the cost model described in~\cref{subsubsec:complexity_measure} include basic arithmetic operations.
Therefore, with respect to the classical primitives, the time complexity for computing the low-pass filter coefficients is quasilinear in the wavelet index.

%%%%%%%%%%%%%%%%%%%%%%%%%%%%%%%%%%%%%%%%%%%%%%%%%%%%%%%%%
%%%%%% Derivative overlaps %%%%%%%%%%%%%%%%%%%%%%%%%%%%%%
%%%%%%%%%%%%%%%%%%%%%%%%%%%%%%%%%%%%%%%%%%%%%%%%%%%%%%%%%

\section{Derivative overlap coefficients}
\label{appx:derivative_overlaps}

In this appendix, we describe Beylkin's method~\cite{Bey92} for computing the derivative overlap coefficients~\eqref{eq:derivative_overlaps} for the Daubechies wavelet with the wavelet index~\dbIndex.
For our application, we restrict the method to computing the overlaps coefficients for the second-order derivative operator.
Beylkin's method, however, is general and could be used to compute derivative overlaps for derivative operators with any integer or fractional order.
Having described the method, we then present a formal algorithm for computing the second-order derivative overlaps.

The explicit expression for the Daubechies scaling function and its derivatives are unknown.
Therefore, the numerical value of the derivative overlaps cannot be directly computed from~\cref{eq:derivative_overlaps}.
Beylkin's method is an indirect way to compute these coefficients.
The second-order derivative overlaps in this method
satisfy the system of linear algebraic equations~\cite[Proposition~2]{Bey92}
\begin{equation}
\sum_{\ell\in \integers}
\ell^2 x_\ell = 2,
\quad x_\ell = 4 x_{2\ell} + 2\sum_{k=1}^{\dbIndex} a_{2k-1}\left(x_{2\ell-(2k+1)} + x_{2\ell+2k+1}\right)
\quad \forall \ell \in \integers,
\end{equation}
where the coefficients
\begin{equation}
\label{eq:autocorrelation}
a_{2k-1}:=\dfrac{(-1)^{k-1}C}{(\dbIndex-k)!(\dbIndex+k-1)!(2k-1)}
\quad \forall k \in \{1, \ldots, \dbIndex\},
\quad
C:=\left[\dfrac{(2\dbIndex-1)!}{(\dbIndex-1)!4^{\dbIndex-1}} \right]^2,
\end{equation}
are auto-correlation of the low-pass filter $h_\ell$~\eqref{eq:scaling_function}.
If~$\dbIndex\geq 2$, then the system of linear equations has a unique solution with~$x_\ell\neq 0$
for~$2\dbIndex-2 \leq \ell\leq 2\dbIndex-2$,
and~$x_\ell=x_{-\ell}$.
Using these properties, we write the linear system in the matrix-vector form as~$\upbm{M} \bm{x} = \upbm{b}$,
where $\upbm{M}$ is a $(2\dbIndex-1)$-by-$(2\dbIndex-1)$ matrix with elements
\begin{equation}
\label{eq:M_elements}
    M_{mn}= n^2+4\updelta_{2m,n}-\updelta_{m,n}+2\sum_{k=1}^{\dbIndex} a_{2k-1}\left(\updelta_{n,\abs{2m-(2k+1)}}+\updelta_{n,2m+2k+1}\right),
\end{equation}
the solution vector is $\bm{x}:=(x_0,\ldots,x_{2\dbIndex-2})^\T$ and $\upbm{b}$ is the vector of all $1$'s, i.e.,~$\upbm{b}:=(1,\ldots,1)$.

The pseudocode in~\cref{alg:derivativeOverlaps} describes a formal algorithm for computing the derivative overlaps by Beylkin's method.
The formal algorithm presented here is only used to simplify the description of our ground-state-generation algorithms~(\cref{subsec:highlevel_description}).

\begin{algorithm}[H]
  \caption{Classical algorithm for computing the second-order derivative overlaps in~\cref{eq:derivative_overlaps}}
  \label{alg:derivativeOverlaps}
  \begin{algorithmic}[1]
  \Require{
  \Statex $\dbIndex \in \integers^+$
  \Comment{wavelet index}
  \Statex $p\in \integers^+$
  \Comment{working precision}
    }
  \Ensure
  \Statex $\bm{\Delta} \in \reals^{2\dbIndex-1}$
  \Comment{second-order derivative overlaps~\eqref{eq:derivative_overlaps}}
    \Function{derivativeOverlaps}{$\dbIndex, p$}
    \State $C \gets \left[\dfrac{(2\dbIndex-1)!}{(\dbIndex-1)!4^{\dbIndex-1}} \right]^{2}$
    \Comment{see~\cref{eq:autocorrelation}}
    \For{\text{$k$ from $1$ to $\dbIndex$}}
    \State $a_{2k-1} \gets \dfrac{(-1)^{k-1}C}{(\dbIndex-k)!(\dbIndex+k-1)!(2k-1)}$
    \Comment{computes the auto-correlation coefficients in~\cref{eq:autocorrelation}}
    \State $b_{k-1} \gets 1$
    \EndFor
    \State $\upbm{M}\gets \textsc{mtx}(\dbIndex)$
    \Comment{\textsc{mtx} is a function that constructs the matrix~$\upbm{M}$ by~\cref{eq:M_elements}}
    \State $\bm\Delta \gets \upbm{M}^{-1}\upbm{b}$
    \State \Return $\bm\Delta$
    \EndFunction
     \end{algorithmic}
\end{algorithm}

%%%%%%%%%%%%%%%%%%%%%%%%%%%%%%%%%%%%%%%%%%%%%%%%%%%%%%%%%
%%%%%% Proofs %%%%%%%%%%%%%%%%%%%%%%%%%%%%%%%%%%%%%%%%%%%
%%%%%%%%%%%%%%%%%%%%%%%%%%%%%%%%%%%%%%%%%%%%%%%%%%%%%%%%%

\section{Proofs}
\label{appx:proofs}
\noindent
This appendix contains the statement and proofs of several propositions used in this paper.

%%%%%%%%%%%%%%%%%%%%%%%%%%%%%%%%%%%%%%%%%%%%%%%%%%%%%%%%%
\begin{proposition}
\label{prop:derivativeMtx}
Let $N$ be a positive integer, $\{s_n(x) \,\vert\, n\in \{0,\ldots,N-1\}\}$ be a set of orthonormal, real-valued, differentiable and compactly supported functions that span a subspace of $\mathscr{L}^2(\reals)$, and let $\bm{\Gamma} \in \reals^{N\times N}$ be a matrix whose elements are
\begin{equation}
\label{eq:derivativeMtx}
 \Gamma_{nm}:=\int \dd{x} \dv{x}s_n(x) \dv{x}s_{m}(x)
 \quad \forall\, n,m\in \{0,\ldots,N-1\}.
\end{equation}
Then $\bm{\Gamma}$ is a symmetric and positive-semidefinite matrix.
\end{proposition}
%%%%%%%%%%%%%%%%%%%%%%%%%%%%%%%%%%%%%%%%%%%%%%%%%%%%%%%%%
\begin{proof}
The symmetry of $\bm{\Gamma}$ is immediate from the definition of its elements.
By definition, an $N$-by-$N$ symmetric real matrix $\upbm{A}$ is positive semidefinite if $\bm{v}^{\T}\upbm{A}\bm{v}\geq 0$ for all nonzero $\bm{v}\in\reals^N$.
We use this definition to prove that $\bm{\Gamma}$ is a positive-semidefinite matrix.
Let $\bm{v}$ be a nonzero vector in $\reals^N$ and let $f(x):=\sum\limits_{n=0}^{N-1} v_n s_n(x)$, then
\begin{equation}
    \bm{v}^{\T}\bm{\Gamma}\bm{v}=
    \sum_{n,m=0}^{N-1} v_n\Gamma_{nm}v_m
    =\int \dd{x} \dv{x}\left(\sum_{n=0}^{N-1}v_n s_n(x)\right) \dv{x}\left(\sum_{n=0}^{N-1}v_m s_m(x)\right)
    = \int \dd{x} \left(\dv{x}f(x)\right)^2 \geq 0,
\end{equation}
Thus, by definition, $\bm{\Gamma}$ is a positive-semidefinite matrix.
\end{proof}
%%%%%%%%%%%%%%%%%%%%%%%%%%%%%%%%%%%%%%%%%%%%%%%%%%%%%%%%%

\begin{proposition}
\label{prop:smallest_eigen_of_ICM}
Let $m_0 \in \reals^{+}$ be the free mass and $\upbm{A}\in\reals^{N\times N}$ be the ground-state ICM in a wavelet basis.
Then the smallest eigenvalue of $\upbm{A}$ is~$m_0$.
\end{proposition}
%%%%%%%%%%%%%%%%%%%%%%%%%%%%%%%%%%%%%%%%%%%%%%%%%%%%%%%%%

\begin{proof}
The spectrum of the ground-state ICM in a fixed- and multi-scale wavelet basis are identical because they are unitarily equivalent.
Therefore we only prove this proposition for a fixed-scale ICM~$\upbm{A}^{(k)}_\sS$~\eqref{eq:fixedscale_groundstate}.
This matrix is the principal square root of the fixed-scale coupling matrix in~\cref{eq:fixedscaleK}.
The fixed-scale coupling matrix can be written as
\begin{equation}
\label{eq:fixdedcoupMtxForm}
    \upbm{K}^{(k)}_\sS= m_0^2 \mathds{1} - 4^k \bm{\Delta}^{(2)} = m_0^2 \mathds{1} + 4^k \bm{\Gamma},
\end{equation}
where~$\bm{\Delta}^{(2)}$ is a matrix whose elements are the second-order derivative overlaps~$\Delta^{(2)}_{mn}$~\eqref{eq:derivative_overlaps}, and the second identity in~\cref{eq:fixdedcoupMtxForm} follows from~\cref{eq:derivativeMtx} and~\cref{eq:derivative_overlaps}.
Eigenvalues of~$\upbm{K}^{(k)}_\sS$ are
\begin{equation}
    \lambda^{(k)}_j := m_0^2 - 4^k \Delta^{(2)}_0
            -2\sum\limits_{\ell=1}^{2\dbIndex-1} 4^k \Delta^{(2)}_\ell \cos\left(\frac{2\uppi \ell}{N}j\right),
\end{equation}
which yields
\begin{equation}
\label{eq:smallestEigen}
    \lambda^{(k)}_0 = m_0^2 - 4^k \Delta^{(2)}_0
            -2\sum\limits_{\ell=1}^{2\dbIndex-1} 4^k \Delta^{(2)}_\ell =m^2_0,
\end{equation}
where the second identity follows from properties of the second-order derivative overlaps~\cite[Appendix~A]{SB16}.
By Eqs.~\eqref{eq:fixdedcoupMtxForm},~\eqref{eq:smallestEigen} and~\cref{prop:derivativeMtx}, $m^2_0$ is the smallest eigenvalue of the coupling matrix and therefore~$m_0$ is the smallest eigenvalue of the ground-state ICM in both fixed- and multi-scale wavelet basis.
\end{proof}
%%%%%%%%%%%%%%%%%%%%%%%%%%%%%%%%%%%%%%%%%%%%%%%%%%%%%%%%%

\begin{proposition}[Bound on determinant of near-identity matrices~\cite{BOS15}]
\label{prop:nearIdendityDetBound}
Let $N\in \integers^+$, $\varepsilon \in [0,1)$ and $\upbm{A}=\mathds1-\upbm{E} \in \reals^{N \times N}$ such that $\abs*{E_{ij}}\leq \varepsilon$ for all 
$i,j \in \{0,\ldots,N-1\}$.
If~$N\varepsilon\leq 1$, then
\begin{equation}
1-N\varepsilon \leq \det(\upbm{A})\leq  (1-N\varepsilon)^{-1}.
\end{equation}
\end{proposition}

%%%%%%%%%%%%%%%%%%%%%%%%%%%%%%%%%%%%%%%%%%%%%%%%%%%%%%%%%
\begin{proposition}
\label{prop:sum_integral_relation}
Let $J\in 2\integers^+,\, \sigma \in \reals^+$ and $\delta\leq\min(1/2,\sigma)$.
Define~$r:=\delta/(\sigma\sqrt{2})$ and
\begin{equation}
\label{eq:normalizations}
    \mathcal{N}^2:=\int_\reals \dd{x}\e^{-\frac{x^2}{2\sigma^2}}=
    \sigma\sqrt{2\uppi}, \quad
    \tilde{\mathcal{N}}^2:= \delta^2\sum_{j=-J/2}^{J/2-1}\e^{-j^2r^2}.
\end{equation}
Then
\begin{equation}
\label{eq:sum_integral_relation}
    \dfrac{1}{\mathcal{N}\tilde{\mathcal{N}}}
    \sum_{j=-J/2}^{J/2-1}\delta\,\e^{-j^2 r^2}
    \geq \frac1{\mathcal{N}^2}
    \int_{-J\delta/2}^{J\delta/2}\dd{x}\e^{-\frac{x^2}{2\sigma^2}}.
\end{equation}
\end{proposition}
%%%%%%%%%%%%%%%%%%%%%%%%%%%%%%%%%%%%%%%%%%%%%%%%%%%%%%%%%

\begin{proof}
Let
\begin{equation}
    p:=1+\e^{\uppi-2\uppi^2}\left(\frac{\uppi^{1/4}}{\Gamma(3/4)}-1\right),
\end{equation}
where $\Gamma$ is the Gamma function.
We first prove that
\begin{align}
    &
    \label{eq:NtildeNInequality}
    \tilde{\mathcal{N}}\leq \sqrt{p\delta}\mathcal{N},\\
    &
    \label{eq:sumIntegralInequality}
    \sqrt{\delta/p}
    \sum_{j=-J/2}^{J/2-1}\e^{-j^2r^2}
    \geq \int_{-J\delta/2}^{J\delta/2}\dd{x}\e^{-\frac{x^2}{2\sigma^2}}.
\end{align}
\cref{prop:sum_integral_relation} follows from Eqs.~\eqref{eq:NtildeNInequality} and~\eqref{eq:sumIntegralInequality};
\cref{eq:NtildeNInequality} implies that $\sqrt{p\delta}/(\tilde{\mathcal{N}}\mathcal{N}) \geq 1/\mathcal{N}^2$ and multiplying each side of this inequality by each side of~\cref{eq:sumIntegralInequality} yields~\cref{eq:sum_integral_relation}.
We now prove~\cref{eq:NtildeNInequality}.
Note that
\begin{align}
\label{eq:PoissonSummation}
    \sum_{j=-J/2}^{J/2-1}\e^{-j^2r^2}
    \leq \sum_{j\in\integers}\e^{-j^2r^2}
    = \frac{\sqrt{\uppi}}{r}
    \sum_{k\in\integers}\e^{-k^{2}\uppi^2/r^2},
\end{align}
where the equality follows from the Poisson summation formula~\cite[p.~385]{AD19}.
We now find an upper bound for the right-hand side of~\cref{eq:PoissonSummation}.
Note that~$r^2\leq1/2$ so $\e^{-k^{2}\uppi^2/r^2} \leq \e^{-2k^2\uppi^2}$,
also~$\e^{-2k^2\uppi^2}
    \leq \e^{\uppi-2\uppi^2}\e^{-\uppi k^2}$
for each $k\in\integers/\{0\}$, therefore,
\begin{align}
\label{eq:sumExansionUingThetaFunction}
    \sum_{k\in\integers}\e^{-k^2\uppi^2/r^2}
    \leq \sum_{k\in\integers}\e^{-2k^2\uppi^2}
    \leq 1+\e^{\uppi-2\uppi^2}
    \sum_{k\in\integers/\{0\}}\e^{-\uppi k^2}
    =1+\e^{\uppi-2\uppi^2}
    \left(\frac{\uppi^{1/4}}{\Gamma(3/4)}-1\right)
    =p,
\end{align}
where we used~$\sum_{k\in\integers}\e^{-\uppi k^2}
    =\uppi^{1/4}/\Gamma(3/4)$~\cite[p.~103]{Ber91}.
By Eqs.~\eqref{eq:normalizations}, \eqref{eq:PoissonSummation} and~\eqref{eq:sumExansionUingThetaFunction}
\begin{equation}
    \tilde{\mathcal{N}}^2=
    \delta^2\sum_{j=-J/2}^{J/2-1}\e^{-j^2r^2}
    \leq p\delta\mathcal{N}^2,
\end{equation}
which yields~\cref{eq:NtildeNInequality}.
We now prove~\cref{eq:sumIntegralInequality}.
Note that
\begin{align}
\label{eq:sumExpansion}
\sqrt{\delta/p}\sum_{j=-J/2}^{J/2-1}\e^{-j^2r^2}
= \sqrt{\delta/p}\left(1 + \e^{-J^2r^2/4} + 2\sum_{j=1}^{J/2-1}\e^{-j^2r^2}\right),
\end{align}
and, as $\e^{-x^2/(2\sigma^2)}$ is a convex function for $\abs*{x}\geq \sigma$,
\begin{align}
\label{eq:integralExpansion}
\int_{-J\delta/2}^{J\delta/2}\dd{x}\e^{-\frac{x^2}{2\sigma^2}}
&= 2\int_{0}^{\delta}\dd{x}\e^{-\frac{x^2}{2\sigma^2}}
+ 2\sum_{j=1}^{J/2-1}\int_{j\delta}^{(j+1)\delta}\dd{x}\e^{-\frac{x^2}{2\sigma^2}}\nonumber\\
&\leq 2\delta + 2\delta \sum_{j=1}^{J/2-1}\e^{-j^2r^2}-2(1/2)\delta\left((1/\sqrt{\e})-\e^{-J^2r^2/4}\right)\nonumber\\
&=\delta\left((2-1/\sqrt{\e})
+\e^{-J^2r^2/4}
+2\sum_{j=1}^{J/2-1}\e^{-j^2r^2}\right).
\end{align}
Note that if $\delta\leq (2-1/\sqrt{\e})^{-2}/p=0.514998$
then $\delta^2(2-1/\sqrt{\e})^2\leq \delta/p$ and therefore,
\begin{align}
\label{eq:deltaRestriction}
    \delta(2-1/\sqrt{\e})\leq\sqrt{\delta/p}.
\end{align}
This inequality implies that the right-hand side of~\cref{eq:sumExpansion} is greater than or equal to the right-hand side of~\cref{eq:integralExpansion}.
Therefore, Eqs.~\eqref{eq:sumExpansion}, \eqref{eq:integralExpansion}, \eqref{eq:deltaRestriction} and $\delta\leq 1/2$ yields \cref{eq:sumIntegralInequality}.
\end{proof}

%%%%%%%%%%%%%%%%%%%%%%%%%%%%%%%%%%%%%%%%%%%%%%%%%%%%%%%%%

\begin{proposition}
\label{prop:dmax_dmin_bound}
Let~$d_{\max}\in\reals^+$ and~$d_{\min}\in\reals^+$ be respectively the largest and the smallest diagonal elements of the diagonal matrix~$\upbm{D}$ in either the spectral or the UDU decomposition of a symmetric positive-definite matrix~$\upbm{A}\in\reals^{N\times N}$, with~$N \in \integers^+$.
Also let~$\kappa \in \reals^+$ be the condition number of~$\upbm{A}$,
quantified as the ratio of the largest to the smallest eigenvalues of~$\upbm{A}$.
Then
\begin{equation}
    \frac{d_{\max}}{d_{\min}} \in \order{\kappa}.
\end{equation}
\end{proposition}

%%%%%%%%%%%%%%%%%%%%%%%%%%%%%%%%%%%%%%%%%%%%%%%%%%%%%%%%%

\begin{proof}
If~$\upbm{D}$ is the diagonal matrix in the spectral decomposition of~$\upbm{A}$,
then~$d_{\max}$ and~$d_{\max}$ are respectively the largest and the smallest eigenvalue of~$\upbm{A}$.
Therefore, by definition,
\begin{align}
\label{eq:dmax_dmin_spectral}
     \frac{d_{\max}}{d_{\min}}
     = \kappa.
\end{align}
If~$\upbm{D}$ is the diagonal matrix in the UDU decomposition of~$\upbm{A}$, then
\begin{align}
\label{eq:dmax_dmin_UDU}
    \frac{d_{\max}}{d_{\min}}\leq \kappa.
\end{align}
The proof of this inequality is as follows.
Let $\lambda_{\max}(\upbm{A})$ and $\lambda_{\min}(\upbm{A})$ be respectively the largest and the smallest eigenvalues of~$\upbm{A}$ and let~$\upbm{e}_i$ be the~$i^\text{th}$ column of the $N$-by-$N$ identity matrix.
Then
\begin{align}
    &
    \label{eq:maxEigenA}
    \lambda_{\max}(\upbm{A})
    =\max_{\norm{\bm{x}}=1}\bm{x}^\T\upbm{A}\bm{x}
    \geq \max_i \upbm{e}^\T_i \upbm{A}\upbm{e}_i
    =\max_i \left(\upbm{UDU}^\T\right)_{ii}
    =\max_i \left(d_i +\sum_{j>i} d_j U^2_{ij}\right)
    \geq \max_i d_i = d_{\max},\\
    &
    \label{eq:minEigenA}
    \lambda^{-1}_{\min}(\upbm{A})
    =\lambda_{\max}\left(\upbm{A}^{-1}\right)
    \geq \max_i \left(\upbm{UDU^\T}\right)^{-1}_{ii}
    =\max_i \left(d^{-1}_i +\sum_{j>i} d^{-1}_j \left(\upbm{U}^{-1}\right)^2_{ij}\right)
    \geq \max d^{-1}_i
    =d^{-1}_{\min}.
\end{align}
The second equality in~\cref{eq:minEigenA} follows from the fact that the inverse of an upper unit-triangular matrix is an upper unit-triangular matrix~\cite[p.~220]{HJ12}.
The last inequality in Eqs.~\eqref{eq:maxEigenA} and~\eqref{eq:minEigenA} holds because $d_i>0$ for the positive-definite matrix~$\upbm{A}$.
\cref{prop:dmax_dmin_bound} follows from Eqs.~\eqref{eq:dmax_dmin_spectral} and~\eqref{eq:dmax_dmin_UDU}.
\end{proof}

%%%%%%%%%%%%%%%%%%%%%%%%%%%%%%%%%%%%%%%%%%%%%%%%%%%%%%%%%

\begin{proposition}
\label{prop:condition_bound}
Let~$\kappa\in \reals^+$ be the condition number of the ICM for the ground state of a $N$-mode massive real scalar bosonic free QFT with mass $m_0\in\reals^+$ in a fixed- or multi-scale wavelet basis.
Then
\begin{equation}
    \kappa \in \Theta (N/m_0).
\end{equation}
\end{proposition}

%%%%%%%%%%%%%%%%%%%%%%%%%%%%%%%%%%%%%%%%%%%%%%%%%%%%%%%%%

\begin{proof}
Note that the ICM $\upbm{A}$ in a multi-scale wavelet has the same condition number as the ICM $\upbm{A}_\sS$ in a fixed scale wavelet basis because $\upbm{A}$ is obtained from $\upbm{A}_\sS$ by a wavelet transform.
We therefore consider the condition number of $\upbm{A}_\sS$.
Now using~$\upbm{A}_\sS=\sqrt{\upbm{K}_\sS}$, where~$\upbm{K}_\sS$~\eqref{eq:fixedscaleK} is the coupling matrix in a fixed-scale wavelet basis, we have
\begin{equation}
\label{eq:condICM}
    \kappa=\sqrt{\frac{\lambda_{\max}(\upbm{K}_\sS)}{\lambda_{\min}(\upbm{K}_\sS)}}.
\end{equation}
We prove that
\begin{align}
    &\label{eq:minEigenK}
    \lambda_{\min}(\upbm{K}_\sS)= m^2_0,\\
    & \label{eq:maxEigenK}
    \lambda_{\max}(\upbm{K}_\sS)\in \Theta(m^2_0+N^2).
\end{align}
\cref{prop:condition_bound} follows form~Eqs.~\eqref{eq:condICM} to~\eqref{eq:maxEigenK}.
To prove~\cref{eq:maxEigenK}, we find a lower and an upper bound for the largest eigenvalue of $\upbm{K}_\sS$. 
We use the Gershgorin circle theorem to find the upper bound~\cite[p.~388]{HJ12}.
This theorem implies that
\begin{equation}
\label{eq:GreshgorinTheorem}
\lambda_{\max}(\upbm{K}_\sS)
\leq \max_i\left( K_{\sS;ii} + R_i\right),
\;\; R_i:=\sum_{j\neq i}\abs*{K_{\sS;ij}}.
\end{equation}
Note that for any positive-definite matrix~$\upbm{K}_\sS$~\cite[p.~434]{HJ12}
\begin{align}
\label{eq:SPDcondition}
    \abs*{K_{\sS;ij}}^2
    \leq K_{\sS;ii}K_{\sS;jj},
\end{align}
and by \cref{eq:fixedscaleK}
\begin{equation}
\label{eq:diagsK} 
K_{\sS;\, ii}=m^2_0+N^2 \Delta_0.
\end{equation}
Therefore,
\begin{align}
\label{eq:boundDiagsK}
    \abs*{K_{\sS;ij}}
    \leq m_0^2+N^2\Delta_0.
\end{align}
Furthermore,~$\upbm{K}_\sS$ is a banded circulant matrix with the lower and upper bandwidth~$2\dbIndex-2$.
Thus, using~\cref{eq:boundDiagsK},
\begin{align}
\label{eq:boundR}
    R_i\leq 2(2\dbIndex-2)(m_0^2+N^2 \Delta_0).
\end{align}
Eqs.~\eqref{eq:GreshgorinTheorem},~\eqref{eq:diagsK} and~\eqref{eq:boundR} yield
\begin{equation}
\label{eq:upperBoundMaxEigenK}
    \lambda_{\max}(\upbm{K}_\sS)
    \leq (4\dbIndex-3)(m_0^2+ N^2 \Delta_0).
\end{equation}
We now find a lower bound for the largest eigenvalue.
Note that
\begin{align}
\label{eq:lowerBoundMaxEigenK}
    \lambda_{\max}(\upbm{K}_\sS)
    =\max_{\norm{\upbm{x}}=1}\upbm{x}^\T\upbm{K}_\sS\upbm{x}
    \geq \max_i \upbm{e}^\T_i\upbm{K}_\sS\upbm{e}_i
    = m_0^2+N^2\Delta_0,
\end{align}
where~$\upbm{e}_i$ is the~$i^\text{th}$ column of~$\mathds{1}_{N\times N}$ and in the last equality we used \cref{eq:diagsK}.
Eqs.~\eqref{eq:upperBoundMaxEigenK} and~\eqref{eq:lowerBoundMaxEigenK} yield~\cref{eq:maxEigenK}.
\end{proof}

%%%%%%%%%%%%%%%%%%%%%%%%%%%%%%%%%%%%%%%%%%%%%%%%%%%%%%%%%

\begin{proposition}[exponentially decaying ICM at a fixed scale]
\label{prop:decayingICM}
Let $m_0 \in \reals^{+}$ be the free mass,
$\dbIndex \in \integers_{\geq 3}$ be the wavelet index and $r \in \integers_{\geq 0}$ be the scale index.
Let $\upbm{A}_\sS^{(r)}\in\reals^{L2^r\times L2^r}$~\eqref{eq:fixedscale_groundstate} with $L \in \integers_{\geq 2(2\dbIndex-1)}$ be the ground-state ICM at scale~$r$.
Then,
for any~$j$
\begin{equation}
\label{eq:decayingssBlocks}
    \abs{A^{(r)}_{\sS;\, 0, j}}
    \leq 4m_0\kappa^{(r)} 2^{-\abs{j}/\xi^{(r)}},
    \quad
    \xi^{(r)}:=(2\dbIndex-1)2^{r+1}/m_0,
\end{equation}
where $\kappa^{(r)}>1$ is the spectral condition number of~$\upbm{K}_\sS^{(r)}$~\eqref{eq:fixedscaleK}.
\end{proposition}

%%%%%%%%%%%%%%%%%%%%%%%%%%%%%%%%%%%%%%%%%%%%%%%%%%%%%%%%%

\begin{proof}
We employ the Benzi-Golub theorem~\cite[Theorem~2.2]{BG99} to bound off-diagonal entries of ICM~$\upbm{A}_\sS^{(r)}$; see~\cite[p.~243]{ben16} for an alternative statement of this theorem.
Let $\alpha>1$ and $\beta>0$, with $\alpha>\beta$, be the half axes of an ellipse in the complex plane with foci in~$\pm 1$.
This ellipse is specified by the sum of its half axes $\chi:=\alpha+\beta >1$, so we denote the ellipse by ${\cal{E}}_\chi$.
Let~$f(z)$ be an analytic function in the interior of~${\cal{E}}_\chi$ and continuous on~${\cal{E}}_\chi$ for any $1<\chi<\Bar{\chi}$, and let
\begin{equation}
\label{eq:BGK}
    M(\chi):=\max_{z\in {{\cal{E}}_\chi}} \abs{f(z)}
    \quad \text{and} \quad
    K:= \frac{2\chi M(\chi)}{\chi-1}.
\end{equation}
Also let~$\upbm{B}$ be a symmetric and banded matrix whose spectrum~$\operatorname{spec}\upbm{B}$
is contained in~$[-1,1]$ and let~$\cal{B}$ the upper bandwidth of~$\upbm{B}$.
Then by the Benzi-Golub theorem~\cite[Theorem~2.2]{BG99}
\begin{equation}
    \abs{\left[f(\upbm{B})\right]_{i,j}}
    \leq K 2^{-\gamma\abs{i-j}},
    \quad
    \gamma:= \frac{\log_2{\chi}}{{\cal{B}}+1}.
\end{equation}
The ICM matrix~$\upbm{A}^{(r)}_{\sS}$ is the principal square root of the coupling matrix~$\upbm{K}^{(r)}_{\sS}$~\eqref{eq:fixedscaleK}, and~$\upbm{K}^{(r)}_{\sS}$ is a symmetric and banded matrix with the upper bandwidth $2\dbIndex-2$.
However, 
$\operatorname{spec}\upbm{K}^{(r)}_{\sS}$
is not contained in~$[-1,1]$, so we cannot directly apply the Benzi-Golub theorem to this matrix.
Let $[a,b] \subset \reals^+$ be the interval containing $\operatorname{spec}\upbm{K}^{(r)}_{\sS}$.
To apply the Benzi-Golub theorem, by shifting and scaling the coupling matrix, we construct the matrix
\begin{equation}
    \upbm{B}:=\frac{2\upbm{K}^{(r)}_\sS
    -(b+a)\id}
    {b-a},
\end{equation}
whose spectrum is contained in~$[-1,1]$.
Then the function
\begin{equation}
\label{eq:mapfunc}
     f(z):= \sqrt{\frac{b-a}{2}z+\frac{b+a}{2}},
\end{equation}
maps~$\upbm{B}$ to the ICM; that is~$\upbm{A}^{(r)}_{\sS}=f(\upbm{B})$.
By the Benzi-Golub theorem, we then have
\begin{equation}
\label{eq:BGbound}
     \abs{A^{(r)}_{\sS;\, 0, j}}
     \leq K 2^{-\gamma\abs{j}},
\end{equation}
with
\begin{equation}
\label{eq:gamma}
    \gamma= \frac{\log_2{\chi}}{2\dbIndex-1}.
\end{equation}
We now obtain a lower bound for $\gamma$.
The function $f(z)$ in~\cref{eq:mapfunc} is analytic in the interior of any ellipse ${\cal{E}}_\chi$ with $\chi$ less than
\begin{equation}
    \Bar{\chi} = \frac{b+a}{b-a}+\sqrt{\left(\frac{b+a}{b-a}\right)^2-1}.
\end{equation}
Using the spectral condition number~$\kappa^{(r)}:=b/a$ of~$\upbm{K}^{(r)}_{\sS}$, we have
\begin{equation}
    \Bar{\chi} = \frac{\sqrt{\kappa^{(r)}}+1}{\sqrt{\kappa^{(r)}}-1}
    =1+ \frac{2}{\sqrt{\kappa^{(r)}}-1} 
    \geq 1+ \frac{2}{\sqrt{\kappa^{(r)}-1}},
\end{equation}
where the inequality holds because $\kappa^{(r)}\geq 1$.
Setting
\begin{equation}
\label{eq:chi}
\chi = 1+2/\sqrt{\kappa^{(r)}-1} \leq \Bar{\chi},
\end{equation}
we obtain
\begin{equation}
\label{eq:logChiBound}
    \log_2(\chi)= \log_2\left(1+2/\sqrt{\kappa^{(r)}-1}\right)\geq
    \sqrt{2/\kappa^{(r)}}
    \geq m_0/2^{r+1},
\end{equation}
where the first inequality holds because $\kappa^{(r)}\geq 1$ and the second inequality holds because
\begin{equation}
\label{eq:kappaBound}
    \kappa^{(r)} \leq 1+4^{r+2}/m_0^2 \leq
    2\times 4^{r+2}/m_0^2,
\end{equation}
as per~\cref{prop:condition_bound}.
The combination of~\cref{eq:gamma} and~\cref{eq:logChiBound} yields
\begin{equation}
    \label{eq:gammaBound}
    \gamma \geq 1/\xi^{(r)},
\end{equation}
where~$\xi^{(r)}$ is given in~\cref{eq:decayingssBlocks}.
We now find an upper bound for $K$ in~\cref{eq:BGbound}.
By~\cref{eq:chi} and $\kappa^{(r)}\geq 1$
\begin{equation}
\label{eq:Kbound0}
    \frac{2\chi}{\chi-1}
    = \chi\sqrt{\kappa^{(r)}-1}
    = \sqrt{\kappa^{(r)}-1}+2
    \leq 2\sqrt{2\kappa^{(r)}}.
\end{equation}
The function $f(z)$ in~\cref{eq:mapfunc} attains its maximum at the point $z_0:=\alpha$, where $\alpha$ is the greater half axis of the ellipse ${\cal{E}}_\chi$.
The combination of $\chi=\alpha+\beta$, $\alpha^2-\beta^2=1$ and~\cref{eq:chi} yields
\begin{equation}
    z_0:= \frac{\chi^2+1}{2\chi} = 1+\frac{(\chi-1)^2}{2\chi} \leq 1
    + \frac12 (\chi-1)^2
    = 1+ \frac{2}{\kappa^{(r)}-1},
\end{equation}
and
\begin{equation}
\label{eq:Kbound1}
    M(\chi)=\max_{z\in{\cal{E}}_\chi} \abs{f(z)}
    = f(z_0)
    =\sqrt{\frac{b-a}{2}\frac{\chi^2+1}{2\chi}
    +\frac{b+a}{2}}
    \leq \sqrt{a/2}\sqrt{2\kappa^{(r)}+2}
    \leq \sqrt{2a\kappa^{(r)}} = m_0  \sqrt{2\kappa^{(r)}},
\end{equation}
where the last equality follows because the smallest eigenvalue of~$\upbm{K}^{(r)}_{\sS}$ is $a=m_0^2$.
Eqs.~\eqref{eq:BGK},~\eqref{eq:Kbound0} and~\eqref{eq:Kbound1} yield
\begin{equation}
\label{eq:Kbound}
    K\leq 4m_0 \kappa^{(r)}.
\end{equation}
\cref{prop:decayingICM} follows from Eqs.~\eqref{eq:BGbound},~\eqref{eq:gammaBound}, and~\eqref{eq:Kbound}.
\end{proof}

%%%%%%%%%%%%%%%%%%%%%%%%%%%%%%%%%%%%%%%%%%%%%%%%%%%%%%%%%
%%%%%%%%%%%%%%%%%%%%%%%%%%%%%%%%%%%%%%%%%%%%%%%%%%%%%%%%%
\subsection{Proof of~\cref{prop:truncated_Gaussian}}
\label{proofprop:truncated_Gaussian}
To prove~\cref{prop:truncated_Gaussian},
first we show that the $\varepsilon_\text{th}$-approximate ICM $\upbm{A}_{\varepsilon_\text{th}}$~\eqref{eq:cutoffcondition} is a positive-definite matrix.
Then we prove that the infidelity between~$\ket{\G_N(\upbm{A})}$ and~$\ket{\G_N(\upbm{A}_{\varepsilon_\text{th}})}$ is no greater than~$\varepsilon$.

We use the Bauer-Fike theorem~\cite[p.~405]{HJ12} to show that $\upbm{A}_{\varepsilon_\text{th}}$ is a positive-definite matrix.
Let~$\upbm{Q}$ be a nonsingular matrix such that $\upbm{Q}^{-1}\upbm{A}\upbm{Q}=\bm{\Lambda}$ is a diagonal matrix and let~$\upbm{E}:=\upbm{A}-\upbm{A}_{\varepsilon_\text{th}}$.
According to the Bauer-Fike theorem,
for any eigenvalue~$\lambda(\upbm{A}_{\varepsilon_\text{th}})$ of $\upbm{A}_{\varepsilon_\text{th}}$, there is an eigenvalue $\lambda(\upbm{A})$ of $\upbm{A}$ such that
\begin{equation}
\label{eq:BauerFikeTheorem}
\abs{\lambda(\upbm{A}_{\varepsilon_\text{th}})-\lambda(\upbm{A})} \leq \kappa_p(\upbm{Q}) \norm{\upbm{E}}_{p}\,,
\end{equation}
where $\norm{\bullet}_{p}$ is the matrix norm induced by any $p$-norm on $\cmplex^N$ and $\kappa_p(\bullet)$ is the condition number of a matrix
with respect to this norm.
As $\upbm{A}$ is a real-symmetric matrix, $\upbm{Q}$ can be chosen to be an orthogonal matrix for which~$\norm{\upbm{Q}}_{2}=1$ and~$\kappa_2(\upbm{Q})=1$.
Therefore, using~\cref{eq:BauerFikeTheorem} with $p=2$, we obtain
\begin{align}
\label{eq:BauerFikeTheoremNorm2}
\abs*{\lambda_{\min}(\upbm{A}_{\varepsilon_\text{th}})-\lambda(\upbm{A})} \leq \norm{\upbm{E}}_2 \leq N\varepsilon_\text{th},
\end{align}
for the smallest eigenvalue $\lambda_{\min}(\bm{A}_{\varepsilon_\text{th}})$ of~$\upbm{A}_{\varepsilon_\text{th}}$.
The last inequality here comes from~$\abs{E_{ij}} \leq \varepsilon_\text{th}$ by the definition of~$\upbm{E}$ and using the matrix-norm inequalities
$\norm{\upbm{E}}_{2}\leq \norm{\upbm{E}}_{\text{F}} \leq N \max_{i,j}\abs{E_{ij}}$, where $\norm{\bullet}_\text{F}$ is the Frobenius norm of a matrix.
If $\lambda_{\min}(\upbm{A}_{\varepsilon_\text{th}})\geq \lambda_{\min}(\upbm{A})$, then $\upbm{A}_{\varepsilon_\text{th}}$ is already a positive-definite matrix.
Otherwise, using~\cref{eq:BauerFikeTheoremNorm2}
\begin{align}
\abs*{\lambda_{\min}(\upbm{A}_{\varepsilon_\text{th}})-\lambda_{\min}(\upbm{A})}\leq N\varepsilon_\text{th}
\leq \varepsilon \gamma/\sqrt{N},   
\end{align}
which implies that the~$\upbm{A}_{\varepsilon_\text{th}}$ is a positive-definite matrix. 

We now prove Eq.~(\ref{eq:infidGNAGNAth}).
Let~$\upbm{T}:=\upbm{E}\upbm{A}^{-1}$, then
\begin{equation}
\braket{\G_N(\upbm{A})}{\G_N(\upbm{A}_{\varepsilon_\text{th}})}
= \left(\dfrac{\det\upbm{A}\det\upbm{A}_{\varepsilon_\text{th}}}{(2\uppi)^{2N}} \right)^{1/4}
\int_{\reals^N} \dd[N]{\bm{x}}
\e^{-\tfrac{1}{4} \bm{x}^\T \left(\upbm{A}
+ \upbm{A}_{\varepsilon_\text{th}}\right)\bm{x}}
= \dfrac{\left[\det\left({\mathds1-\upbm{T}}\right)\right]^{1/4}}{\left[\det\left({\mathds1-\upbm{T}/2}\right)\right]^{1/2}},
\label{eq:fidelityBound}
\end{equation}
where we use $\det\upbm{A}_{\varepsilon_\text{th}} = \det\upbm{A} \det (\mathds1-\upbm{T})$ by employing the multiplicative property of the determinant function~\cite[p.~11]{HJ12}, and the Gaussian integral
\begin{equation}
\label{eq:Gaussian_integral}
    \int_{\reals^N} \dd[N]{\bm{x}}
    \e^{-\tfrac{1}{2} \bm{x}^\T \upbm{A}\bm{x}} = \left(\frac{(2\uppi)^N}{\det\upbm{A}}\right)^{1/2},
\end{equation}
for a symmetric and positive-definite matrix $\upbm{A}\in\reals^{N\times N}$.
We now use~\cref{prop:nearIdendityDetBound} to obtain a lower and an upper bound for determinant of the near-identity matrix $\mathds1-\upbm{T}$.
Let us first find an upper bound for $\abs*{T_{ij}}$.
Note that
\begin{align}
\label{eq:boundForElementsOfT}
    \abs{T_{ij}}
    = \abs{\sum_{k=0}^{N-1} E_{ik}\left(\upbm{A}^{-1}\right)_{kj}}
    \leq \varepsilon_\text{th} \max_{j} \sum_{k=0}^{N-1} \abs{\left(\upbm{A}^{-1}\right)_{kj}}
    = \varepsilon_\text{th}\norm{\upbm{A}^{-1}}_{1}
    \leq \varepsilon_\text{th} \gamma^{-1}\sqrt{N},
\end{align}
where the last inequality follows from the matrix-norm relations
\begin{align}
\label{eq:normRelations}
\norm{\upbm{A}}_1\leq \sqrt{N}\norm{\upbm{A}}_2,\;\;
\norm{\upbm{A}}_2=\lambda_{\max}(\upbm{A}),
\end{align}
for a real-symmetric matrix~$\upbm{A}$~\cite[pp.~346, 363]{HJ12},
and $\lambda_{\max}(\upbm{A}^{-1})= \lambda^{-1}_{\min}(\upbm{A})\leq \gamma^{-1}$;
note that~$\gamma > 0$ is a lower bound for the eigenvalues of $\upbm{A}$.
By~\cref{prop:nearIdendityDetBound} and~\cref{eq:boundForElementsOfT},
if $\varepsilon_\text{th} \gamma^{-1} N^{3/2} \leq \varepsilon$, then
\begin{equation}
\label{eq:detIneqaulities}
\det\left(\mathds1-\upbm{T}\right)\geq 1-\varepsilon, \;\;
\det\left(\mathds1-\upbm{T}/2\right)\leq (1-\varepsilon/2)^{-1}.
\end{equation}
Combination of \cref{eq:fidelityBound} and \cref{eq:detIneqaulities} yields
\begin{equation}
\braket{\G_N(\upbm{A})}{\G_N(\upbm{A}_{\varepsilon_\text{th}})}
\geq (1-\varepsilon)^{1/4}(1-\varepsilon/2)^{1/2}
\geq 1-\varepsilon,
\end{equation}
where, in the last inequality, we use $(1-\varepsilon/2)^{1/2}\geq (1-\varepsilon)^{1/2}$ and $(1-\varepsilon)^{3/4} \geq 1-\varepsilon$ for any $\varepsilon \in (0,1)$.
Therefore, the infidelity between~$\ket{\G_N(\upbm{A})}$ and~$\ket{\G_N(\upbm{A}_{\varepsilon_\text{th}})}$ is at most~$\varepsilon$.

%%%%%%%%%%%%%%%%%%%%%%%%%%%%%%%%%%%%%%%%%%%%%%%%%%%%%%%%%
%%%%%%%%%%%%%%%%%%%%%%%%%%%%%%%%%%%%%%%%%%%%%%%%%%%%%%%%%
\subsection{Proof of~\cref{prop:decayingDiagBlocks}}
\label{proofprop:decayingDiagBlocks}

This appendix presents a proof for~\cref{prop:decayingDiagBlocks}.
Following~\cref{eq:Aww} with $r=c$, we have
\begin{equation}
\label{eq:wws}
    \upbm{A}^{(r,r)}_\ww = \upbm{G} \upbm{A}^{(r+1)}_\sS\upbm{G}^\T,
\end{equation}
where $\upbm{G}$ is the lower half of the wavelet-transform matrix at scale $r$, see~\cref{eq:WTM}, with entries $G_{nm}=g_{n-2m}$.
Using~\cref{eq:wws} and circulant property of~$\upbm{A}^{(r+1)}_\sS$, we have
\begin{equation}
\label{eq:wwsr+1}
    A^{(r,r)}_{\ww;\, 0, j} = \sum_{m,n=0}^{2\dbIndex-1} g_n g_m A^{(r+1)}_{\sS;\, 0, (2j+m-n)}.
\end{equation}
By virtue of~\cref{prop:decayingICM}, $\abs{A^{(r+1)}_{\sS;\, 0,j}}$ decreases as~$j$ increases.
For any $m,n$ and $j\geq 2\dbIndex-1$, we have $2j+m-n\geq j$, so
\begin{equation}
\label{eq:ineq1}
    \abs{A^{(r+1)}_{\sS;\, 0, (2j+m-n)}} \leq 
    \abs{A^{(r+1)}_{\sS;\, 0,j}}
    \quad \forall j\geq 2\dbIndex-1.
\end{equation}
\cref{eq:wwsr+1} and~\cref{eq:ineq1} yield
\begin{equation}
\label{eq:ineq2}
    \abs{A^{(r,r)}_{\ww;\, 0, j}}
    \leq \abs{A^{(r+1)}_{\sS;\, 0,j}}
    \sum_{m,n=0}^{2\dbIndex-1} \abs{g_n g_m},
\end{equation}
and because $\sum g^2_n=1$, we have $\abs{g_n}\leq 1$.
Therefore
\begin{equation}
\label{eq:ineq3}
    \sum_{m,n=0}^{2\dbIndex-1} \abs{g_n g_m}
    \leq
    \left(\sum_{m=0}^{2\dbIndex-1} \abs{g_m}\right)
    \left(\sum_{n=0}^{2\dbIndex-1} \abs{g_n}\right)
    < 4\dbIndex^2.
\end{equation}
\cref{prop:decayingDiagBlocks} follows from Eqs.~\eqref{eq:decayingssBlocks}, \eqref{eq:ineq2} and~\eqref{eq:ineq3}.

%%%%%%%%%%%%%%%%%%%%%%%%%%%%%%%%%%%%%%%%%%%%%%%%%%%%%%%%%
%%%%%%%%%%%%%%%%%%%%%%%%%%%%%%%%%%%%%%%%%%%%%%%%%%%%%%%%%
\subsection{Proof of~\cref{corol:bandedCirculantBlocks}}
\label{appx:proofcorol:bandedCirculantBlocks}
Here we prove~\cref{corol:bandedCirculantBlocks}.
By virtue of~\cref{prop:decayingDiagBlocks}, the off-diagonal entries in diagonal blocks of the multi-scale ICM decay slower as the scale index $r$ increases. 
Therefore, for a given threshold value $\varepsilon_\text{th}$, the bottom-right block of the approximate ICM~\cref{eq:cutoffcondition} has the largest bandwidth among the diagonal blocks.
Plugging $r=k-1$ into~\cref{eq:decayingDiagBlocks} we obtain
\begin{equation}
   \abs{A^{(k-1,k-1)}_{\ww;\, 0, j}}
   \leq \varepsilon_\text{th},
\end{equation}
for
\begin{equation}
\label{eq:jBound}
    \abs{j} \geq \xi^{(k)} \log_2\left(\frac{16\dbIndex m_0\kappa^{(k)}}{\varepsilon_\text{th}}\right).
\end{equation}
We now use~\cref{prop:truncated_Gaussian} and $2^{k+1}\leq \log N$ to express the upper bandwidth in terms of $\dbIndex, m_0, \varepsilon_\text{vac}$ and $N$.
By~\cref{eq:decayingDiagBlocks}
\begin{equation}
    \xi^{(k)}= (2\dbIndex-1)2^{k+1}/m_0
    \leq (2\dbIndex-1)\left(\log N\right)/m_0 .
\end{equation}
The smallest eigenvalue of the ICM is $m_0$, so by~\cref{prop:truncated_Gaussian} we take
$\varepsilon_\text{th}=m_0\varepsilon_\text{vac}N^{-3/2}$.
Therefore
\begin{equation}
\label{eq:bandwidthIneqs}
    \frac{16\dbIndex^2 m_0\kappa^{(k)}}{\varepsilon_\text{th}}
    = \frac{16\dbIndex^2 N^{3/2}\kappa^{(k)}}{\varepsilon_\text{vac}}
    \leq\frac{32\dbIndex^2 N^{3/2}4^{k+2}}{m^2_0\varepsilon_\text{vac}}
    \leq \frac{32\dbIndex^2 N^{3/2}\log^2N}{m_0^2\varepsilon_\text{vac}}
    \leq \frac{16\dbIndex^2N^2}{m_0^2\varepsilon^2_\text{vac}},
\end{equation}
where we use~\cref{eq:kappaBound} in the first inequality.
In the last two inequalities we use
\begin{equation}
    2^{k+1}\leq \log N, \quad
    \varepsilon_\text{vac} \geq \varepsilon^2_\text{vac}
    \quad \text{and} \quad
    32N^{3/2}\log^{2}N \leq 16 N^2,
\end{equation}
for $\varepsilon_\text{vac} \in (0,1)$ and $N \in \integers^+$.
Equation~\ref{eq:bandwidth} follows from~\crefrange{eq:jBound}{eq:bandwidthIneqs}.

%%%%%%%%%%%%%%%%%%%%%%%%%%%%%%%%%%%%%%%%%%%%%%%%%%%%%%%%%
%%%%%%%%%%%%%%%%%%%%%%%%%%%%%%%%%%%%%%%%%%%%%%%%%%%%%%%%%
%%%%%%%%%%%%%%%%%%%%%%%%%%%%%%%%%%%%%%%%%%%%%%%%%%%%%%%%%
\section{Computing rotation angle for 1DG-state generation}
\label{appx:rotation_angle}

In this appendix, we analyze the time complexity for computing the rotation angle~$\theta_\ell$ in~\cref{eq:rot_angle} on a quantum computer.
We first simplify computing~$\theta_\ell$ using the double-angle identity as
\begin{equation}
    \label{eq:rot_angle_simplified}
    \theta_\ell = \frac{1}{2} \arccos(u_\ell)
    = \frac{\uppi}{4} -\frac{1}{4} \arcsin(u_\ell),
\end{equation}
where
\begin{equation}
\label{eq:rot_arg}
    u_\ell := 2\dfrac{f\left(\tilde{\sigma}_\ell/2,\, \mu_\ell/2,\, m_\ell-1\right)}{f\left(\tilde{\sigma}_\ell,\, \mu_\ell,\, m_\ell\right)}-1,
\end{equation}
and $f\left(\tilde{\sigma}_\ell,\, \mu_\ell,\, m_\ell\right)$ is defined in~\cref{eq:1DG_over_integers}.
Using~\cref{eq:rot_angle_simplified}, computing~$\theta_\ell$ requires computing~$u_\ell$ and~$\arcsin(u_\ell)$, and performing one multiplication and one addition.
Therefore, time complexity for computing the rotation angle~$\theta_\ell$, denoted by~$T_\textsc{angle}$, is
\begin{equation}
\label{eq:T_angle}
    T_\textsc{angle} = T_u + T_{\arcsin} + 2,
\end{equation}
where~$T_u $ and~$T_{\arcsin}$ are time complexities for computing $u_\ell$ and $\arcsin(u_\ell)$, respectively.
By~\cref{eq:rot_arg}, computing $u_\ell$ requires computing $f\left(\tilde{\sigma}_\ell/2,\, \mu_\ell/2,\, m_\ell-1\right)$ and
$f\left(\tilde{\sigma}_\ell,\, \mu_\ell,\, m_\ell\right)$, and performing one division, one multiplication and one addition.
Therefore,
\begin{equation}
\label{eq:T_arg}
    T_u = 2T_f + 3,
\end{equation}
where $T_f$ is time complexity for computing $f\left(\tilde{\sigma}_\ell,\, \mu_\ell,\, m_\ell\right)$.
By~\cref{eq:T_angle} and~\cref{eq:T_arg}, we only need~$T_f$ and~$T_{\arcsin}$ to obtain~$T_\textsc{angle}$.
First we provide a high-level description of how to obtain~$T_f$ and~$T_{\arcsin}$.
We show, in \cref{prop:rot_angle}, that computing $f\left(\tilde{\sigma}_\ell,\, \mu_\ell,\, m_\ell\right)$ to $t$ bit of precision requires computing at most~$4t-1$ exponentials and adding them all.
The argument of each exponential needs performing one addition, one multiplication, one division and calculating the square of two numbers.
Calculating the square of a number can be performed by one multiplication~\cite[p.~7]{HRS18}.
Altogether, time complexity for computing~$f\left(\tilde{\sigma}_\ell,\, \mu_\ell,\, m_\ell\right)$ is at most~$(4t-1) T_{\exp} + (4t-2) + 5$, and therefore
\begin{equation}
\label{eq:T_f}
    T_f \in \order{tT_{\exp}}, 
\end{equation}
where~$T_{\exp}$ is time complexity for computing an exponential.
We use~\cref{prop:exp_approx} to show~$T_{\exp}$ scales linearly with~$t$.
By this proposition, to compute $\e^{-x}$ with $t$ bits of precision, first the polynomial
\begin{equation}
\label{eq:poly}
    P(x):= a_0+ a_1 x + a_2 x^2,
\end{equation}
with $a_0=1, a_1=1/4^t$ and $a_2=1/4^{2t}$ is evaluated at~$x$ and then the result is iteratively squared~$2t$ times.
Computing square of a number needs one multiplication, so $T_{\exp} = T_P + 2t$, where~$T_P$ is time complexity for computing $P(x)$ in~\cref{eq:poly}.

We use Horner's method to compute a polynomial on a quantum computer~\cite{HRS18}, and to obtain~$T_P$.
To evaluate
the polynomial in~\cref{eq:poly} at~$x$, we store the coefficients $a_0, a_1$ and $a_2$ on a ancillary quantum register and implement the operation
$a_2 x + a_1\mapsto (a_2 x + a_1) x + a_0$.
Implementing this operation needs two multiplications and two additions before uncomputing the intermediate results, so~$T_P$ is a constant and
\begin{equation}
\label{eq:T_exp}
    T_{\exp} \in \Theta(t).
\end{equation}

We use~\cref{prop:arg_bound} and~\cref{prop:arcsin_approx} to obtain~$T_{\arcsin}$.
In \cref{prop:arg_bound}, we show that the argument~$x$ of $\arcsin(x)$ is a positive number less than~$1/2$ and, in~\cref{prop:arcsin_approx},
we show a polynomial of degree~$t$
suffices to approximate $\arcsin(x)$ with~$t$ bits of precision.
As described, computing a polynomial of degree~$t$ by Horner's method requires~$\Theta(t)$ multiplications and additions, so
\begin{equation}
\label{eq:T_arcsin}
   T_{\arcsin} = \Theta(t),
\end{equation}
The combination of~\crefrange{eq:T_angle}{eq:T_f} and~\crefrange{eq:T_exp}{eq:T_arcsin} yields
\begin{equation}
    T_\textsc{angle} \in \order{t^2}.
\end{equation}
We now state and prove the propositions that we use to obtain~$T_\textsc{angle}$.

\begin{proposition}
\label{prop:rot_angle}
Let $m\in \integers^+, t\in \integers^+$, $\ell \in\{0,\ldots,m-1\}$, $\mu_\ell \in [0,1)$, $m_\ell:=m-\ell$, $\tilde{\sigma}_\ell\in(0,2^{m_\ell/2})$ and
\begin{align}
    & \label{eq:f}
    f(\tilde{\sigma}_\ell,\mu_\ell,m_\ell):= \sum_{j=-2^{m_\ell-1}}^{2^{m_\ell-1}-1}
    \exp\left(-
    \tfrac{(j+\mu_\ell)^2}{2\tilde{\sigma}_\ell^2}
    \right),\\
    & \label{eq:g}
    g(\tilde{\sigma}_\ell,\mu_\ell):= \sum_{j=-2t-1}^{2t}
    \exp\left(
    -\tfrac{(j+\mu_\ell)^2}{2\tilde{\sigma}_\ell^2}
    \right).
\end{align}
For $2^{m_\ell}>8(t+3)$, if $\tilde{\sigma}_\ell^2> t$ then
\begin{equation}
    \abs{f(\tilde{\sigma}_\ell,\mu_\ell,m_\ell)-\tilde{\sigma}_\ell\sqrt{2\uppi}}\leq \frac1{2^{t+1}},
\end{equation}
else
\begin{equation}
    \abs{f(\tilde{\sigma}_\ell,\mu_\ell,m_\ell)-g(\tilde{\sigma},\mu_\ell)}\leq \frac1{2^{t+1}}.
\end{equation}
\end{proposition}

%%%%%%%%%%%%%%%%%%%%%%%%%%%%%%%%%%%%%%%%%%%%%%%%%%%%%%%%%

\begin{proof}
We first prove the first part of~\cref{prop:rot_angle}.
Using  the triangle inequality,
\begin{equation}
\label{eq:f_approx}
    \abs{
    f(\tilde{\sigma}_\ell,\mu_\ell,m_\ell)
    -\sigma_\ell\sqrt{2\uppi}
    }
    \leq 
    \abs{
    f(\tilde{\sigma}_\ell,\mu_\ell,m_\ell)
    -\sum_{j\in \integers}
    \exp\left(
    -\tfrac{(j+\mu_\ell)^2}{2\tilde{\sigma}_\ell^2}
    \right)
    }
    +
    \abs{\sum_{j\in \integers}
    \exp\left(
    -\tfrac{(j+\mu_\ell)^2}{2\tilde{\sigma}_\ell^2}
    \right)
    - \tilde{\sigma}_\ell\sqrt{2\uppi}}.
\end{equation}
Let us now find and upper bound for the first term in the right-hand side of~\cref{eq:f_approx}.
By~\cref{eq:f},
\begin{align}
\label{eq:bound_for_first_term}
\abs{f(\tilde{\sigma}_\ell,\mu_\ell,m_\ell)
    -\sum_{j\in \integers}
    \exp\left(-\tfrac{(j+\mu_\ell)^2}{2\tilde{\sigma}_\ell^2}\right)
    }
    &= \sum_{2^{m_\ell-1}}^{\infty}
    \left(\exp\left(-\tfrac{(j+\mu_\ell)^2}{2\tilde{\sigma}_\ell^2}\right)
    +
    \exp\left(-\tfrac{(j+1-\mu_\ell)^2}{2\tilde{\sigma}_\ell^2}\right)
    \right)\nonumber\\
    &\leq 2\sum_{2^{m_\ell-1}}^{\infty}
    \exp\left(-\tfrac{j^2}{2\tilde{\sigma}_\ell^2}\right)
    \leq 2\sum_{2^{m_\ell-1}}^{\infty}
    \e^{-j/4} = \dfrac{2\e^{-2^{m_\ell-3}}}{1-\e^{-1/4}}
    \leq \frac1{2}\e^{-t}.
\end{align}
The first inequality here holds because $\mu_\ell \in [0,1)$.
The second inequality is obtained using
\begin{equation}
    \sum_{2^{m_\ell-1}}^{\infty}
    \exp\left(-\tfrac{j^2}{2\tilde{\sigma}_\ell^2}\right)
    \leq
    \sum_{2^{m_\ell-1}}^{\infty}
    \exp\left(-\tfrac{2^{m_\ell-1}}{2\tilde{\sigma}_\ell^2}j\right),
\end{equation}
and $\tilde{\sigma}_\ell^2\leq 2^{m_\ell}$.
The last inequality is obtained using $2^{m_\ell}>8(t+3)$ and $2\e^{-3}/(1-\e^{-1/4})< 1/2$.

We now establish an upper bound for the second term in the right-hand side of~\cref{eq:f_approx}.
By Poisson's summation formula~\cite[p.~385]{AD19}, 
\begin{align}
    \sum_{j\in \integers}
    \exp\left(-\tfrac{(j+\mu_\ell)^2}{2\tilde{\sigma}_\ell^2}\right)
    = \tilde{\sigma}_\ell\sqrt{2\uppi} \sum_{j\in \integers}
    \e^{-2\uppi^2\tilde{\sigma}_\ell^2j^2 - 2\uppi \upi\mu_\ell j}
    = \tilde{\sigma}_\ell\sqrt{2\uppi}
    \left(1+2\sum_{j=1}^{\infty}
    \e^{-2\uppi^2\tilde{\sigma}_\ell^2j^2} \cos(2\uppi\mu_\ell j)
    \right).
\end{align}
Therefore,
\begin{align}\label{eq:bound_for_second_term}
   \abs{\sum_{j\in \integers}
    \exp\left(-\tfrac{(j+\mu_\ell)^2}{2\tilde{\sigma}_\ell^2}\right)
    - \tilde{\sigma}_\ell\sqrt{2\uppi}}
    &=
    2\sqrt{2\uppi}\tilde{\sigma}_\ell
    \abs{\sum_{j=1}^{\infty}
    \e^{-2\uppi^2\tilde{\sigma}_\ell^2j^2}
    \cos(2\uppi\mu_\ell j)}\nonumber\\
    &\leq
    2\sqrt{2\uppi}\tilde{\sigma}_\ell
    \sum_{j=1}^{\infty}
     \e^{-2\uppi^2\tilde{\sigma}_\ell^2j^2}\nonumber\\
    &\leq
    2\sqrt{2\uppi}\tilde{\sigma}_\ell
    \sum_{j=1}^{\infty}
    \e^{-2\uppi^2\tilde{\sigma}_\ell^2j}
    = 2\sqrt{2\uppi}\tilde{\sigma}_\ell
    \dfrac{\e^{-2\uppi^2\tilde{\sigma}_\ell^2}}{1-\e^{-2\uppi^2\tilde{\sigma}_\ell^2}}
    \leq \e^{-\uppi^2\tilde{\sigma}_\ell^2}
    \leq \e^{-\uppi^2t}.
\end{align}
Here we use the triangle inequality and $\abs{\cos(2\uppi\mu_\ell j)}\leq 1$ to obtain the first inequality.
The second inequality is obtained by
\begin{align}
    \sum_{j=1}^{\infty}
    \e^{-2\uppi^2\tilde{\sigma}_\ell^2j^2}
    \leq 
    \sum_{j=1}^{\infty}
    \e^{-2\uppi^2\tilde{\sigma}_\ell^2j}.
\end{align}
The third inequality is obtained using $1-\e^{-2\uppi^2\tilde{\sigma}_\ell^2}\geq \e^{-\tilde{\sigma}_\ell^2}$ for $\tilde{\sigma}_\ell^2 \geq t\geq 1$,
and $2\sqrt{2\uppi}x \e^{-(\uppi^2-1)x^2}\leq 1$ for any $x\in \reals^+$.
The last inequality holds because
$\tilde{\sigma}_\ell^2 \geq t$.
The first part of~\cref{prop:rot_angle} follows from Eqs.~\eqref{eq:f_approx},~\eqref{eq:bound_for_first_term},~\eqref{eq:bound_for_second_term} and
\begin{align}
    \frac1{2}\e^{-t}+\e^{-\uppi^2t} \leq \frac1{2^{t+1}} \;\;\; \forall t\in \integers^+.
\end{align}
We now prove the second part of~\cref{prop:rot_angle}.
By~Eqs.\eqref{eq:f} and~\eqref{eq:g},
\begin{align}
    \abs{f(\tilde{\sigma}_\ell,\mu_\ell,m_\ell)-g(\tilde{\sigma},\mu_\ell)}
    &=\sum_{j=2t+1}^{2^{m_\ell-1}-1}
    \left(\exp\left(-\tfrac{(j+\mu_\ell)^2}{2\tilde{\sigma}_\ell^2}\right)
    +
    \exp\left(-\tfrac{(j+1-\mu_\ell)^2}{2\tilde{\sigma}_\ell^2}\right)
    \right)\nonumber\\
    &\leq 2\sum_{j=2t+1}^{2^{m_\ell-1}-1} \exp\left(-\tfrac{j^2}{2\tilde{\sigma}_\ell^2}\right)
    \leq 2\sum_{j=2t+1}^{\infty} \e^{-j}
    = 2\frac{\e^{-(2t+1)}}{1-\e^{-1}}
    \leq \frac1{2^{t+1}}.
\end{align}
The first inequality here is obtained using $\mu_\ell \in [0,1)$; the second inequality is obtained
using
\begin{align}
    \sum_{j=2t+1}^{2^{m_\ell-1}-1}
    \exp\left(-\tfrac{j^2}{2\tilde{\sigma}_\ell^2}\right)
    \leq
    \sum_{j=2t+1}^{\infty}
    \exp\left(-\tfrac{2t+1}{2\tilde{\sigma}_\ell^2}j\right),
\end{align}
and
$\tilde{\sigma}_\ell^2\leq t$.
The last inequality is valid for any~$t\in \integers^+$.
\end{proof}

%%%%%%%%%%%%%%%%%%%%%%%%%%%%%%%%%%%%%%%%%%%%%%%%%%%%%%%%%

\begin{proposition}
\label{prop:exp_approx}
Let $t\in \integers^+, x\in [0, t]$ and $P(x/\ell):=1-x/\ell+x^2/\ell^2$.
Then, for any $\ell \geq 4^t$,
\begin{equation}
    \abs{\e^{-x}-P^\ell(x/\ell)}\leq 1/4^t.
\end{equation}
\end{proposition}

%%%%%%%%%%%%%%%%%%%%%%%%%%%%%%%%%%%%%%%%%%%%%%%%%%%%%%%%%

\begin{proof}
Note that
\begin{equation}
\label{eq:exp_expansion}
    \e^{-x} = \left(\e^{-x/\ell}\right)^\ell =
    \left[P(x/\ell)+E(x/\ell)\right]^\ell,
    \quad E(x/\ell) := \sum_{n\geq 3} \frac{(-)^n}{n!} (x/\ell)^n.
\end{equation}
Using binomial expansion
\begin{equation}
    \e^{-x} = P^\ell(x/\ell) + \sum_{n=1}^\ell \binom{\ell}{n} P^{\ell-n}(x/\ell) E^{n}(x/\ell).
\end{equation}
Therefore,
\begin{equation}
\label{eq:approx_errorbound}
    \abs{\e^{-x}-P^{\ell}(x/\ell)}
    \leq \sum_{n=1}^\ell \binom{\ell}{n}
    \abs{P(x/\ell)}^{\ell-n} \abs{E(x/\ell)}^n
    \leq \sum_{n=1}^\ell \binom{\ell}{n}
    \abs{P(x/\ell)}^{\ell-n} \abs{E(x/\ell)}^n,
\end{equation}
where, in the last inequality, we use $P(x/\ell) \leq 1$ for $x/\ell \leq 1$.
Let us now find an upper bound for~$\abs{E(x/\ell)}$.
Using~\cref{eq:exp_expansion} and the triangle inequality,
\begin{equation}
\label{eq:exp_errorbound}
    \abs{E(x/\ell)} \leq \sum_{n\geq 3}
    \frac{1}{n!} (x/\ell)^n
    = (x/\ell)^3 \sum_{n\geq 0}
    \frac{1}{(n+3)!} (x/\ell)^n
    \leq \frac{1}{6}(x/\ell)^3 \e^{x/\ell}
    \leq \frac{1}{2} (x/\ell)^3,
\end{equation}
the last inequality holds because $\e^{x/\ell}/3 \leq 1$ for $x/\ell \leq 1$.
Using Eqs.~\eqref{eq:approx_errorbound} and~\eqref{eq:exp_errorbound}
\begin{equation}
     \abs{\e^{-x}-P^{\ell}(x/\ell)}
      \leq \sum_{n=1}^\ell \binom{\ell}{n}
      \frac{1}{2^n}(x/\ell)^{3n}
      \leq \sum_{n=1}^\ell \frac{\ell^n}{2^n} (x/\ell)^{3n}
      = \frac{x^3}{2\ell^2} + \sum_{n=2}^\ell \frac{x^{3n}}{2^n \ell^{2n}}
      \leq \frac{x^3}{2\ell^2} \left(1+ \frac{x^3}{2\ell^2}\right)
      \leq \frac{1}{\ell} 
      \leq \frac{1}{4^t}.
\end{equation}
Here we use~$x^3/\ell \leq 1/2$ for $x\in [0,t]$ and $\ell>4^t$.
\end{proof}

%%%%%%%%%%%%%%%%%%%%%%%%%%%%%%%%%%%%%%%%%%%%%%%%%%%%%%%%%

\begin{proposition}
\label{prop:arg_bound}
Let $\tilde{\sigma}\in [1,\infty), \mu \in [0,1), m\in \integers^+$ and
\begin{equation}
  f(\tilde{\sigma},\mu,m)=
  \sum_{j=-2^{m-1}}^{2^{m-1}-1}
  \e^{-\frac{(j+\mu)^2}{2\tilde{\sigma}^2}}.
\end{equation}
Then
\begin{equation}
\label{eq:arg_bound}
     \abs{2\frac{f(\tilde{\sigma}/2,\mu/2,m-1)}{f(\tilde{\sigma},\mu,m)}-1}\leq 1/2.
\end{equation}
\end{proposition}
\begin{proof}
Let $f_\text{even}:= f(\tilde{\sigma}/2,\mu/2,m-1)$ and $f_\text{odd}= f(\tilde{\sigma}/2,(\mu+1)/2,m-1)$.
Then $f(\tilde{\sigma},\mu,m)= f_\text{even} +f_\text{odd}$, and~\cref{eq:arg_bound} becomes
\begin{equation}
\label{eq:arg_bound_oddeven}
    \abs{\frac{f_\text{even}-f_\text{odd}}{f_\text{even}+f_\text{odd}}} \leq 1/2.
\end{equation}
We show $3f_\text{even}\geq f_\text{odd}$ and $3f_\text{odd}\geq f_\text{even}$.
These two inequalities yield~\cref{eq:arg_bound_oddeven} which proves~\cref{prop:arg_bound}.
Notice that
\begin{align}
    & \label{eq:f_even}
    f_\text{even} = \sum_{j=0}^{2^{m-2}-1} 
    \e^{-\frac{(2j+\mu)^2}{2\tilde{\sigma}^2}}
    + \sum_{j=1}^{2^{m-2}}
    \e^{-\frac{(2j-\mu)^2}{2\tilde{\sigma}^2}}
    = \sum_{j=0}^{2^{m-2}-1} 
    \e^{-\frac{(2j+\mu)^2}{2\tilde{\sigma}^2}}
    + \sum_{j=2}^{2^{m-2}+1}
    \e^{-\frac{(2j-2-\mu)^2}{2\tilde{\sigma}^2}},\\
    & \label{eq:f_odd}
    f_\text{odd} = \sum_{j=0}^{2^{m-2}-1} 
    \e^{-\frac{(2j+1+\mu)^2}{2\tilde{\sigma}^2}}
    + \sum_{j=1}^{2^{m-2}}
    \e^{-\frac{(2j-1-\mu)^2}{2\tilde{\sigma}^2}}
    = \sum_{j=1}^{2^{m-2}} 
    \left(
     \e^{-\frac{(2j-1+\mu)^2}{2\tilde{\sigma}^2}}
     + \e^{-\frac{(2j-1-\mu)^2}{2\tilde{\sigma}^2}}
    \right).
\end{align}
Using these equations
\begin{align}
    & f_\text{even}-f_\text{odd} =
    \sum_{j=0}^{2^{m-2}-1}
    \left(
    \e^{-\frac{(2j+\mu)^2}{2\tilde{\sigma}^2}}
    - \e^{-\frac{(2j+\mu+1)^2}{2\tilde{\sigma}^2}}
    \right)
    + \sum_{j=2}^{2^{m-2}}
    \left(
    \e^{-\frac{(2j-2-\mu)^2}{2\tilde{\sigma}^2}}
    - \e^{-\frac{(2j-1-\mu)^2}{2\tilde{\sigma}^2}}
    \right)
    + \e^{-\frac{(2^{m-1}-\mu)^2}{2\tilde{\sigma}^2}}
    - \e^{-\frac{(1-\mu)^2}{2\tilde{\sigma}^2}}
    \\
    & f_\text{odd} -f_\text{even}=
    \sum_{j=1}^{2^{m-2}}
    \left(
    \e^{-\frac{(2j-\mu-1)^2}{2\tilde{\sigma}^2}}
    - \e^{-\frac{(2j-\mu)^2}{2\tilde{\sigma}^2}}
    \right)
    + \sum_{j=1}^{2^{m-2}-1}
    \left(
    \e^{-\frac{(2j-1+\mu)^2}{2\tilde{\sigma}^2}}
    - \e^{-\frac{(2j+\mu)^2}{2\tilde{\sigma}^2}}
    \right)
    + \e^{-\frac{(2^{m-1}-1+\mu)^2}{2\tilde{\sigma}^2}}
    - \e^{-\frac{\mu^2}{2\tilde{\sigma}^2}}.
\end{align}
The first three terms in the right-hand-side of these equations are non-negative, so we have
\begin{equation}
    f_\text{even}-f_\text{odd}
    \geq -\e^{-\frac{(1-\mu)^2}{2\tilde{\sigma}^2}}, \quad
    f_\text{odd}-f_\text{even}
    \geq -\e^{-\frac{\mu^2}{2\tilde{\sigma}^2}}.
\end{equation}
By these inequalities,~\cref{eq:f_even} and~\cref{eq:f_odd}
\begin{align}
    & 3f_\text{even}-f_\text{odd}
    \geq 2f_\text{even}
    -\e^{-\frac{(1-\mu)^2}{2\tilde{\sigma}^2}}
    \geq 2\left(
    \e^{-\frac{\mu^2}{2\tilde{\sigma}^2}}
    +\e^{-\frac{(2-\mu)^2}{2\tilde{\sigma}^2}}
    \right)
    -\e^{-\frac{(1-\mu)^2}{2\tilde{\sigma}^2}}
    \geq 0,\\
    & 3f_\text{odd} -f_\text{even}
    \geq 2f_\text{odd}
    -\e^{-\frac{\mu^2}{2\tilde{\sigma}^2}}
    \geq 2\left(
    \e^{-\frac{(1+\mu)^2}{2\tilde{\sigma}^2}}
    +\e^{-\frac{(1-\mu)^2}{2\tilde{\sigma}^2}}
    \right)
    - \e^{-\frac{\mu^2}{2\tilde{\sigma}^2}}
    \geq 0,
\end{align}
for all~$\tilde{\sigma}\geq 1$ and~$\mu \in [0,1)$.
These inequalities yield
$3f_\text{even}\geq f_\text{odd}$
and
$3f_\text{odd}\geq f_\text{even}$,
and hence~\cref{eq:arg_bound_oddeven}.
\end{proof}

%%%%%%%%%%%%%%%%%%%%%%%%%%%%%%%%%%%%%%%%%%%%%%%%%%%%%%%%%

\begin{proposition}
\label{prop:arcsin_approx}
Let $t\in \integers^+, \abs{x} \leq 1/2$ and
\begin{equation}
P(x):=\sum_{\ell=0}^{t-1} a_{2\ell+1} x^{2\ell+1}, \quad
a_{2\ell+1}:=\dfrac{(2\ell-1)!!}{(2\ell)!!}\dfrac{1}{2\ell+1}.
\end{equation}
Then $\abs{\arcsin(x)-P(x)}\leq 1/2^{2t+1}$.
\end{proposition}
\begin{proof}
Using the binomial series of $(1-x^{2})^{-1/2}$,
\begin{equation}
    \arcsin(x) =\int_0^x \frac{\dd{x}}{\sqrt{1-x^2}}
    = P(x) + \int_0^x \dd{x} \sum_{\ell=t}^\infty \frac{(2\ell-1)!!}{(2\ell)!!} x^{2\ell}.
\end{equation}
By this equation and $(2\ell-1)!!/(2\ell)!!<1$ we have
\begin{equation}
    \abs{\arcsin(x)-P(x)}
    < \int_0^x \dd{x} \sum_{\ell=t}^\infty x^{2\ell}
    = \int_0^x \dd{x} x^{2t} \sum_{\ell=0}^\infty x^{2\ell}
    \leq \frac{4}{3} \frac{x^{2t+1}}{2t+1}
    \leq \frac{1}{2^{2t+1}},
\end{equation}
where we use $\abs{x}\leq 1/2$ and $t\geq 1$ in the last two inequalities.
\end{proof}

%%%%%%%%%%%%%%%%%%%%%%%%%%%%%%%%%%%%%%%%%%%%%%%%%%%%%%%%%
%%%%%%%%%%%%%%%%%%%%%%%%%%%%%%%%%%%%%%%%%%%%%%%%%%%%%%%%%
%%%%%%%%%%%%%%%%%%%%%%%%%%%%%%%%%%%%%%%%%%%%%%%%%%%%%%%%%
\section{UDU decomposition}
\label{appx:UDU}

In this appendix, we present a classical algorithm for computing the UDU decomposition of a dense real-symmetric matrix.
The algorithm presented here elucidates our classical UDU-decomposition algorithm for a matrix with a fingerlike sparse structure.
First we describe the UDU matrix decomposition and then present the algorithm for a dense matrix as pseudocode.

% UDU decomposition for a dense matrix
In the UDU matrix decomposition, a symmetric matrix~$\upbm{A}$ is decomposed into the product of an upper unit-triangular matrix~$\upbm{U}$, a diagonal matrix~$\upbm{D}$ and transpose of the upper unit-triangular matrix.
The UDU decomposition is closely related to the LDL decomposition, where a symmetric matrix is decomposed into the product of a lower unit-triangular matrix~$\upbm{L}$, a diagonal matrix and transpose of the lower unit-triangular matrix.
The LDL decomposition algorithm starts from the top-left corner of the matrix $\upbm{L}$ and proceeds to compute entries of this matrix row by row~\cite{Hig09}.
The UDU-decomposition algorithm, however, starts from the top-right corner of~$\upbm{U}$ and proceeds to compute its entries column by column.

To elucidate the algorithm, we write the UDU decomposition of a real-symmetric matrix~$\upbm{A}$ as
\begin{equation}
\label{eq:UDU}
    \upbm{A} = \upbm{{UDU}}^\T = \bm{\mathrm{UV}}^\T,
    \quad
    \upbm{V}:= \upbm{UD},
\end{equation}
where~$\upbm{V}$ is an upper triangular matrix.
By definition, the nonzero elements in the $i^\text{th}$ row of~$\upbm{V}$ are
\begin{equation}
    \upbm{v}_{i,\, i:N-1} = \upbm{u}_{i,\, i:N-1} \odot \upbm{d}_{i:N-1},
\end{equation}
where $\odot$ denotes the Hadamard product.
By~\cref{eq:UDU}, elements of the $i^\text{th}$ column in the upper-triangular part of $\upbm{A}$ are
\begin{equation}
\label{eq:ithColNonzero}
    \upbm{a}_{0:i,\, i} = \upbm{u}_{0:i,\, i:N-1} \cdot \upbm{v}_{i,\, i:N-1}.
\end{equation}
This equation along with $u_{ii}=1$ and~$v_{ii} = d_i$ yield
\begin{align}
    &\label{eq:diagsOfD}
    d_i = a_{ii} - \upbm{u}_{i,\, i+1:N-1} \cdot \upbm{v}_{i,\, i+1:N-1},\\
    & \label{eq:shearsOfU}
    \upbm{u}_{0:i-1,\, i} = \left[\upbm{a}_{0:i-1,i}
    -\upbm{u}_{0:i-1,\, i+1:N-1} \cdot \upbm{v}_{i,\, i+1:N-1}\right]/d_i
    \quad \forall\, i\neq 0,
\end{align}
for the diagonal $d_i$ and shear elements in the~$i^\text{th}$ column of $\upbm{U}$, respectively.

The procedure of the UDU-decomposition algorithm is as follows.
We start from the last column $i=N-1$ and proceed to the first column~$i=0$.
For each index~$i$, first we compute $\upbm{v}_{i,\, i:N-1}$, i.e., the nonzero elements in the $i^\text{th}$ row of $\upbm{V}$ by~\cref{eq:ithColNonzero}.
Then we compute the $i^\text{th}$ diagonal $d_i$ of $\upbm{D}$ by~\cref{eq:diagsOfD} and the shear elements in the $i^\text{th}$ column of $\upbm{U}$ by~\cref{eq:shearsOfU}.
The inputs, outputs and explicit procedure of the algorithm is presented in~\cref{alg:UDUDecomp}

\begin{algorithm}[H]
  \caption{Classical algorithm for UDU decomposition of a dense real-symmetric matrix}
  \label{alg:UDUDecomp}
  \begin{algorithmic}[1]
  \Require{
  \Statex $N \in \integers^+$
  \Comment{order of a square matrix}
  \Statex $\upbm{A} \in \reals^{N\times N}$
  \Comment{a real-symmetric matrix of order $N$}
  }
  \Ensure
  \Statex $\upbm{U} \in \reals^{N\times N}$
    \Comment{upper unit-triangular matrix in UDU decomposition of $\upbm{A}$}
    \Statex $\upbm{d} \in \reals^N$
    \Comment{diagonals of the diagonal matrix~$\upbm{D}$ in UDU decomposition of $\upbm{A}$}
\Function{UDU}{$N, \upbm{A}$}
\State $\upbm{U} \gets \mathds1$
\Comment{initializes $\upbm{U}$ as identity matrix}
\For{$i \gets N-1$ to $0$}
    \Comment{iterates over columns/rows of $\upbm{A}$ from right/bottom to left/top}
    \State $\upbm{v}_{i,i+1:N-1}
    \gets \upbm{u}_{i,i+1:N-1}
    \odot \upbm{d}_{i+1:N-1}$
    \Comment{$\odot$ denotes the Hadamard product}
    \State $d_i \gets a_{ii}
    -\upbm{u}_{i,i+1:N-1} \cdot \upbm{v}_{i,i+1:N-1}$
    \If{$i>0$}
    \State $\upbm{u}_{0:i-1,i} \gets
    \left[\upbm{a}_{0:i-1,i}
    -\upbm{u}_{0:i-1,i+1:N-1} \cdot \upbm{v}_{i,i+1:N-1}
    \right]/d_i$
    \EndIf
\EndFor
\EndFunction
\State \Return $\{\upbm{d}, \upbm{U}\}$
\end{algorithmic}
\end{algorithm}

%%%%%%%%%%%%%%%%%%%%%%%%%%%%%%%%%%%%%%%%%%%%%%%%%%%%%%%%%
%%%%%% End document %%%%%%%%%%%%%%%%%%%%%%%%%%%%%%%%%%%%%
%%%%%%%%%%%%%%%%%%%%%%%%%%%%%%%%%%%%%%%%%%%%%%%%%%%%%%%%%
\end{document}